%% file: main_text.tex
\documentclass{Dissertate}
\usepackage{amsmath,amssymb,xspace,amsthm}

\newtheorem{thm}{Theorem}
\newtheorem{lem}[thm]{Lemma}
\newtheorem{cor}[thm]{Corollary}
\usepackage{amsfonts}
\usepackage{array, booktabs}
\usepackage{xr}
\usepackage{zref-xr}
\usepackage{hyperref}
\usepackage{algorithm}
\usepackage{algpseudocode}

\usepackage{blindtext}
\usepackage{enumitem}
\usepackage{sectsty}
\allsectionsfont{\normalfont\sffamily\bfseries}

\begin{document}

\frontmatter
\setstretch{\dnormalspacing}
\authorlist


\include{chapters/ch1_intro_v3} 
\include{chapters/ch2_spheres}

\include{chapters/ch3_m4}

\include{chapters/ch4_genManifolds_v4}

\include{chapters/ch5_extensions_v7}

\include{chapters/chz_appendix0_notations}

\include{chapters/chz_appendix1_gardner}

\setstretch{\dnormalspacing}

\backmatter
\clearpage
\begin{spacing}{\dcompressedspacing}
\addcontentsline{toc}{chapter}{References}
\bibliographystyle{unsrtnat}
\end{spacing}

\end{document}

%% file: chapters/ch1_intro_v3.tex
\chapter{Introduction\label{cha:Introduction}}

\section{Invariant Object Computation}

Object recognition is one of the most fundamental cognitive tasks
performed by the brain. A successful object recognition requires a
brain to discriminate between different classes of objects despite
variabilities in the stimulus space. For example, a mammalian visual
system can recognize objects despite a variation in the orientation,
position, and background context, etc. Such impressive robustness
to noise is not only specific to visual object recognition, but also
similar tasks done by other brain modalities. Auditory systems are
able to recognize auditory 'objects' such as songs, and languages,
despite variabilities in the sound intensity, relative pitches, or
sound textures (such as voice of a person). In general, human perception
has to operate with discrete entities such as objects, faces, words,
smells, and tasks. Hence, it is of fundamental interest to understand
to evaluate the emergence of neural representations of these entities
along the sensory hierarchies. Artificial intelligent systems aim
to solve similar perceptual tasks. The recent success of Deep Networks
has been foremost their ability to perform object recognition tasks
despite the immense variabilities in the signals input representations,
in both training and testing examples \cite{krizhevsky2012imagenet}.
An artificial face recognition tasks have to be done despite variabilities
of facial expressions, image scale and occlusion, etc. Autonomous
driving systems have to recognize objects in the driving environment
fast and accurately, despite the various conditions such as speed
of approach, location, confounding objects. Likewise, voice recognition
systems need to overcome enormous variability in many stimulus dimensions.
Indeed, it is a common practice in Machine Learning to augment the
training set by performing a variety of transformations representing
the natural inherent variability in the relevant object domain ('data
augmentation'). Therefore, understanding how brain achieves an invariant
object recognition tasks is not only important scientific challenge,
but may also provide insight on how to improve artificial intelligent
systems. 

\begin{figure}[h]

\begin{centering}
\includegraphics[width=0.9\textwidth]{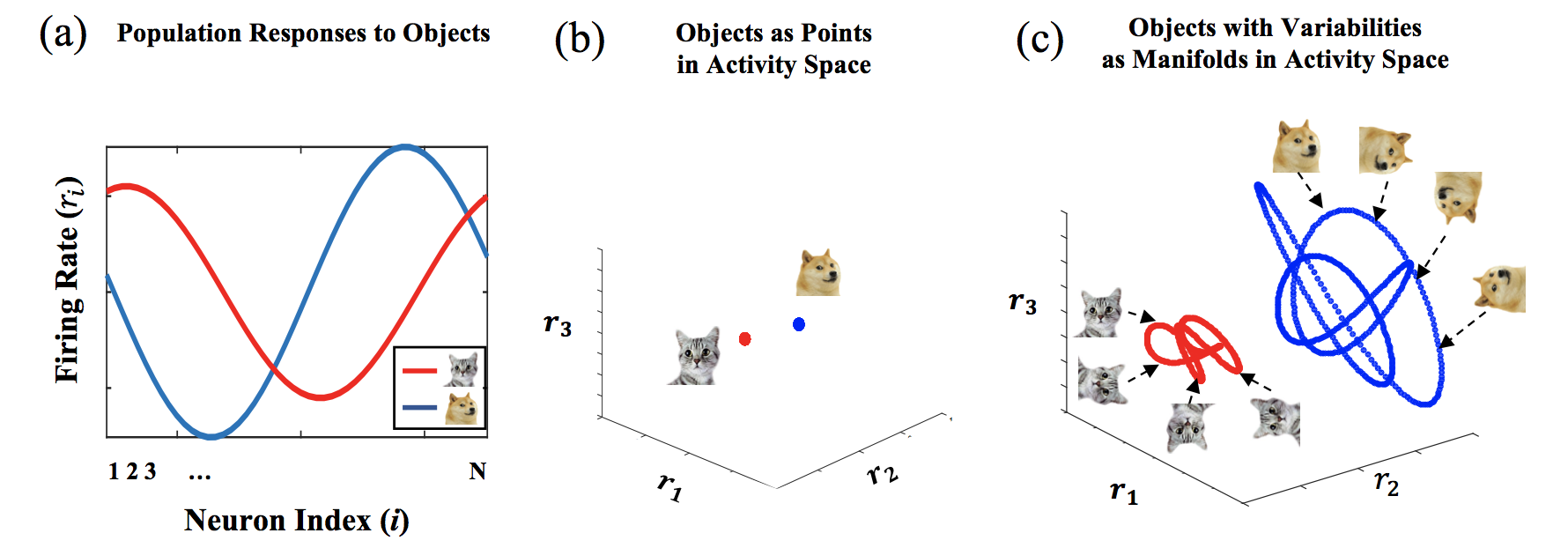}
\par\end{centering}
\caption{\textbf{Activity Patterns as Points and Manifolds in the Neural State
Space. }(a) (Illustration) Firing Rate of N neurons (Neuron Index:
1,...,N), responding to two objects. Neural activity pattern shown
as red line is a population response to a cat, and blue to a dog.
(b) In the $N$ dimensional neural state space, the population response
is a $N$-dimensional vector, representing a point. The blue line
in (a) , representing a response to a dog, is a point, $\vec{r}_{dog}$,
in $\mathbb{R}^{N}$space. Likewise, the red line in $(a)$ representing
a cat is a point $\vec{r}_{cat}$ in $\mathbb{R}^{N}$space (only
3 axis are shown for illustration). (c) When the stimulus variabilities
are introduced, such as change in orientation, the neural responses
undergo a smooth change, causing the point representing each object
in the state space move around, forming an object manifold in the
neural state space. Blue manifold is a set of neural activities representing
a dog at different orientations, and a red manifold is a set of neural
activities representing a cat in different orientations. \label{fig:PatternAsPoint}}

\end{figure}

\section{Object Manifolds}

Consider a set of neurons responding to different objects (Figure
\ref{fig:Manifolds_Illustration}(a)). Without additional variabilities,
two stimuli belonging to different classes are mapped into two points
in the neural state space, $R^{N}$ (Figure. \ref{fig:Manifolds_Illustration}(b)).
We will occasionally call each such point, a neural state or an \emph{activity
pattern}. If however, one varies continuously the physical parameters
in the stimulus space which do not change the object class, e.g.,
orientation, location, distortion, the neural state vector will vary
so that the set of neural states or activity patterns that correspond
to an object can be thought of as a manifold in the neural state space
(Figure \ref{fig:Manifolds_Illustration}(a)). In this geometrical
perspective, object recognition and discrimination can be viewed as
the the task of discriminating or recognition of manifolds. These
manifolds vary as the signals propagate from one processing stage
to another. We will therefore refer to these manifolds also as neural
manifolds or manifold representations, when dealing with object manifolds
as they are reflected in the state space of a specific neural stage. 

\begin{figure}[h]
\begin{centering}
\includegraphics[width=0.8\textwidth]{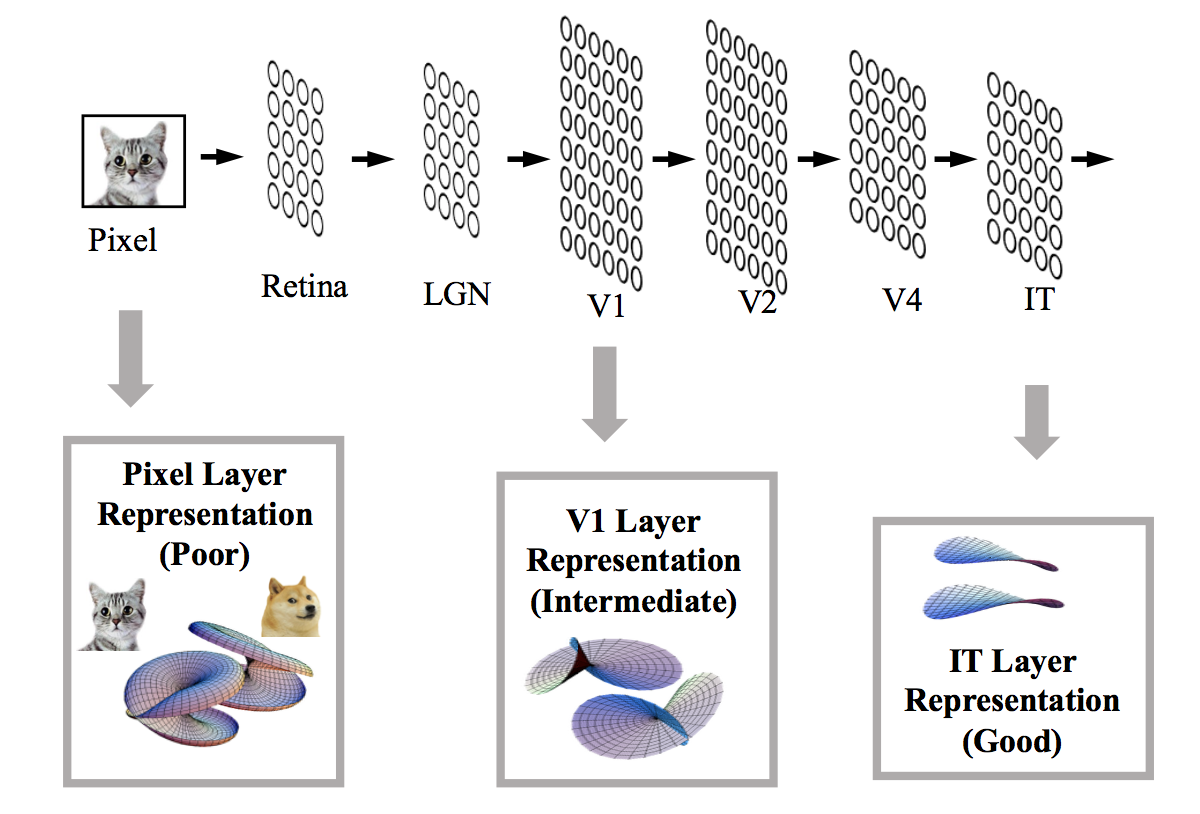}
\par\end{centering}
\caption{\textbf{Untanglement of Object Manifolds in Sensory Hierarchy.} Manifolds
corresponding to different objects are thought to be highly tangled
in the first stage of sensory processing (such as pixel layer representation),
and undergo transformations across different layers in the sensory
hierarchy and become more linearly separable in the downstream of
the sensory processing. \label{fig:UntanglementManifold} }

\end{figure}

\section{Linear Separability of Manifolds}

So, how does the brain and the deep networks overcome stimulus variabilities
in object recognition? In the feedforward visual hierarchy, it has
been suggested that the stages of nonlinear transformations reformat
the object manifolds so that they become increasingly easier to be
readout out by a simple downstream neural systems \cite{dicarlo2007untangling}.
The downstream circuit is assumed to implement a biologically plausible
linear readout. Hence, the reformatting of object manifolds is translated
as 'untangling' them so that they are eventually amenable to be separated
by a \emph{linear classifier}. The idea that 'intermediate' neural
representations help to discriminate complex stimuli by a linear readout,
has been applied to explain features of a variety of sensory representations
in the brain (including 'mixed representations' in prefrontal cortex\cite{rigotti2013importance},
sparse expansions in neocortical\cite{babadi2014sparseness}, memory
allocations in hippocampal and cerebellar systems\cite{valiant2012hippocampus}).
Deep Networks for object recognition has similarly employed an architecture
where at the top layer a linear classifier operates as a readout of
the networks. 

Linear separation of neural manifolds can be described by a decision
hyperplane that separates entire manifolds to one of the two sides
of the hyperplane, fig. \ref{fig:Manifolds_Illustration}. The separating
hyperplane is determined by the vector $\mathbf{w}$, a direction
vector normal to the hyperplane. The components of this vectors are
the synaptic weights of the Linear Readout, also known as the Perceptron,
as it computes the weighted sum of each vector and thresholds the
result to produce a binary output. One of the focus of this work is
to evaluate what aspects of the the neural manifolds representation
gives a better linear separability. Before continuing it is important
to emphasize that by separating manifolds we mean separating all points
on the manifolds according to a rule that assigns to all points belonging
to a single manifold the same label. Thus, at any given time, the
system classifies \emph{a single input vector. }

\begin{figure}
\begin{centering}
\includegraphics[width=0.8\textwidth]{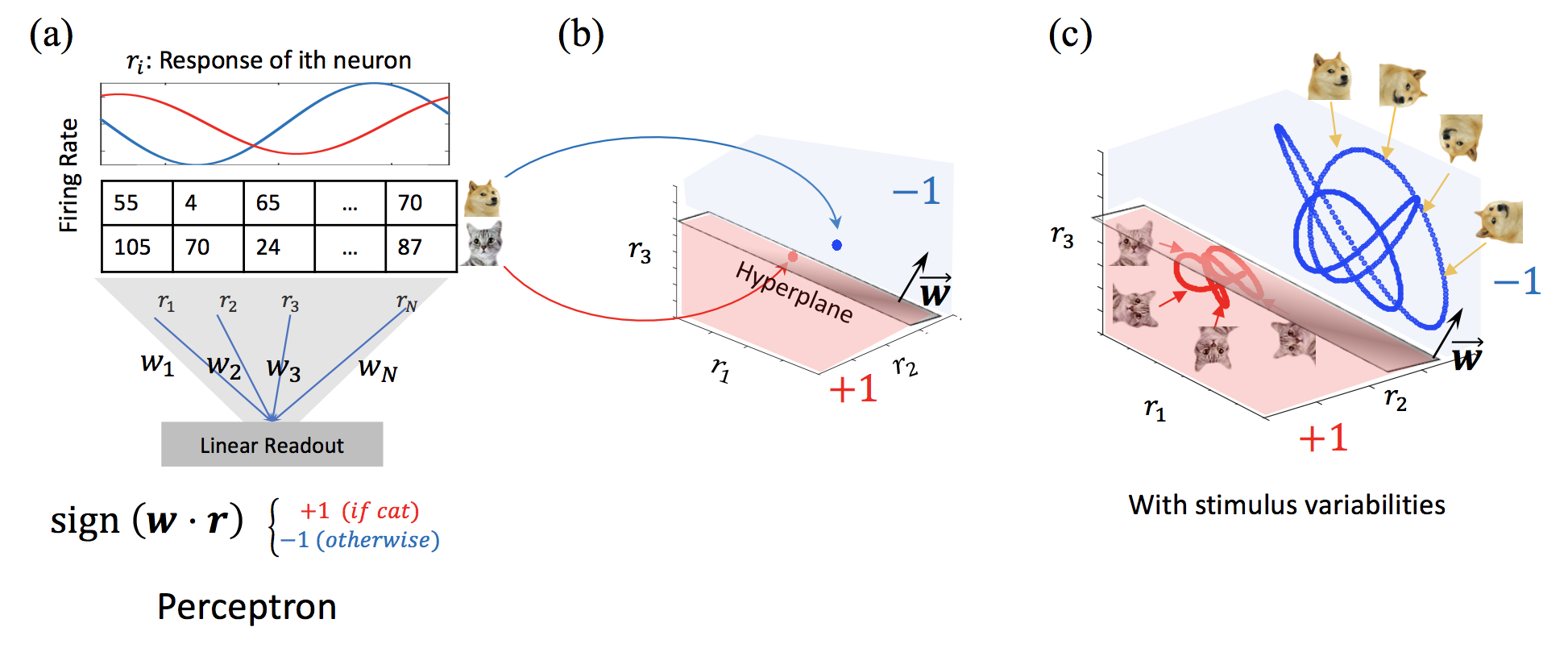}
\par\end{centering}
\caption{\textbf{Invariant Object Discrimination as Linear Separation of Manifolds.}
(a) A perceptron with a weight $\mathbf{w}$, which can classify different
objects corresponds to a hyperplane characterized by the orthogonal
vector $\mathbf{w}$ which separates the two object points in the
state space (b). (c) In this perspective, a perceptron weight $\mathbf{w}$,
which can classify different objects with their variabilities, corresponds
to a hyperplane characterized by the orthogonal vector $\mathbf{w}$,
which separates the two object manifolds in the state space. \label{fig:Manifolds_Illustration}}
\end{figure}

\section{Theory of Linear Classification}

Quantifying linear separability has been extensively studied in the
context of linearly classifying isolated points. Perceptron capacity,
first introduced by Cover \cite{cover1965geometrical}. He asked the
following question in his formulation of Cover's Theorem. Suppose
there are $P$ points in an N-dimensional ambient space, and they
are in general position. Each of the point represent a distinct pattern,
and half of the points are labeled positive, and the other half negative.
Then, what is the maximum number of $P$ where most of the dichotomies
are linearly separable? If there are only a few points, it is easy
to find a linearly separable solution, and with an increasing number
of points, it becomes harder to find linearly separable solution.
When there are too many points, they become linearly non-separable.
He derived an analytic formula for the probability that a random classification
of P points in N dimensions can be implemented by perceptron as 

\begin{equation}
\frac{C(P,N)}{2^{P}}=\frac{2\sum_{k=0}^{N-1}\left(\begin{array}{c}
P-1\\
k
\end{array}\right)}{2^{P}}\label{eq:CoverTheorem}
\end{equation}

and the notion of perceptron capacity deals with the question of what
is the maximum number of patterns allowed for linear such that almost
all dichotomies are linearly separable (Figure \ref{fig:CoverCapacity}).
Cover's perceptron capacity refers to the maximum number of patterns
$P_{max}$ allowed per ambient dimension $N$, also known as load
($\alpha=P_{max}/N$),such that the probability of linear separability
is larger than 0.5. VC dimension refers to the maximum load $\alpha$
such that the probability of linear separability is 1(Figure \ref{fig:CoverCapacity}). 

\begin{figure}

\begin{centering}
\includegraphics[width=0.9\textwidth]{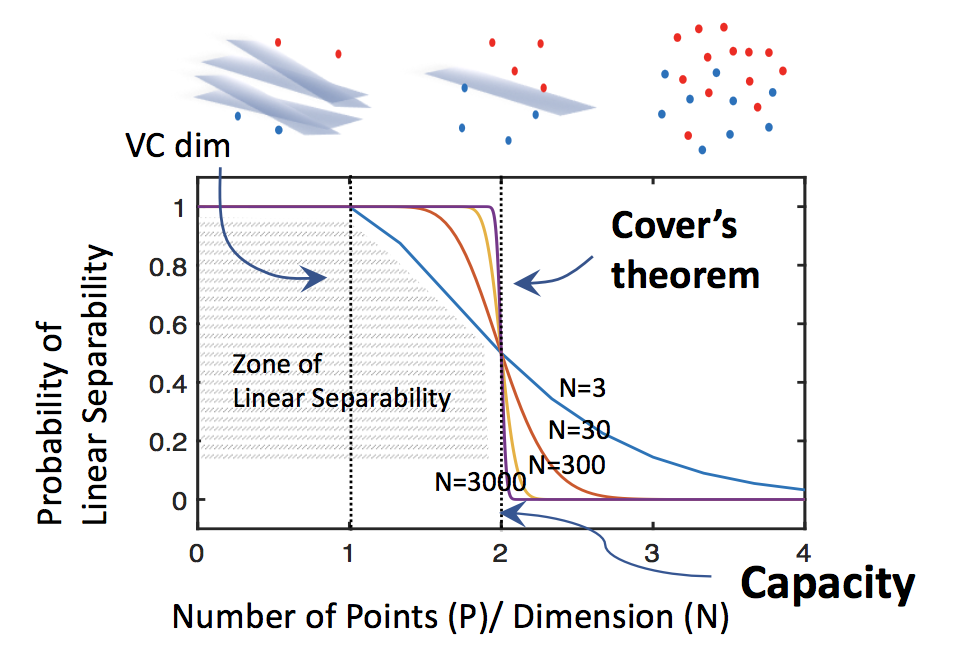}
\par\end{centering}
\caption{\textbf{Linear Separability of Points: Cover's Theorem of Perceptron
Capacity.} Cover's theorem specifies the perceptron capacity for isolated
points in general position, as the maximum number of points $P$ (in
dimension $N$) for which at least half of the possible dichotomies
have a linear classifier. \label{fig:CoverCapacity}}

\end{figure}

A statistical mechanical theory of the perceptron was first introduced
by Elizabeth Gardner\cite{gardner1988space,gardner87}. Gardner's
theory is extremely important as it provides accurate estimates of
the Perceptron capacity beyond the Cover theorem. In particular, the
theory allows to evaluate the capacity for solutions with a given
robustness measures. Similar to Support Vector Machines \cite{vapnik1998statistical}.
robustness of linear classifiers can be quantified by the margin,
ie., the distance between the separating hyperplane and the closest
point. And the solutions with maximum margins are known as the SVM
solutions. 

Formally, Gardner's theory evaluates the maximal number of points
in $R^{N}$ for which there is a vector $w$ that obeys the following
set of inequalities

\begin{equation}
y^{\mu}\left(\vec{w}\cdot\vec{x}^{\mu}+b\right)/||\vec{w}||>\kappa\geq0;\:y^{\mu}=\pm1
\end{equation}

Unlike Cover result, the answer to this question depends on the statistics
of the inputs and labels. The simplest case is where all components
$x_{i}^{\mu}$are iid with zero mean and finite variance (which can
be taken as 1). (The shape of the distribution are less important
as long as mild conditions are obeyed). The labels are randomly assigned
to these points each with probability $1/2$ . Finally, the theory
becomes exact in the the thermodynamic limit $N,P\rightarrow\infty$,
while $\alpha,\kappa,=O(1)$ . Using replica theory in the theory
of spin glasses (more detailed treatment is in the appendix to the
chapter), Gardner has evaluated analytically the volume of possible
solutions for a given load $\alpha$ and margin $\kappa$ . The volume
is exponentially large (in $N$) below the capacity, $\alpha_{G}(\kappa)$
and is zero above it. The maximal margin solution is right at the
border between the two regimes. Using the vanishing volume condition,
Gardner obtained an elegant expression for the perceptron capacity
with finite margin $\kappa$ 

\begin{equation}
\alpha_{G}(\kappa)=\left(\int_{-\kappa}^{\infty}Dt\left(t+\kappa\right)^{2}\right)^{-1}\label{eq:GardnerCapacity}
\end{equation}

where $Dt=\frac{\text{exp}\left(-\frac{1}{2}t^{2}\right)}{\left(2\pi\right)^{1/2}}\text{d}t$
(Figure. \ref{fig:GardnerResults}(a)). Furthermore, the Gardner framework
allows for the calculation of fraction of support vectors on the margin,
which has an important bearing on its robustness and generalization
performance (Ref) (Figure. \ref{fig:GardnerResults}(b)). 

\begin{figure}
\begin{centering}
\includegraphics[width=0.8\textwidth]{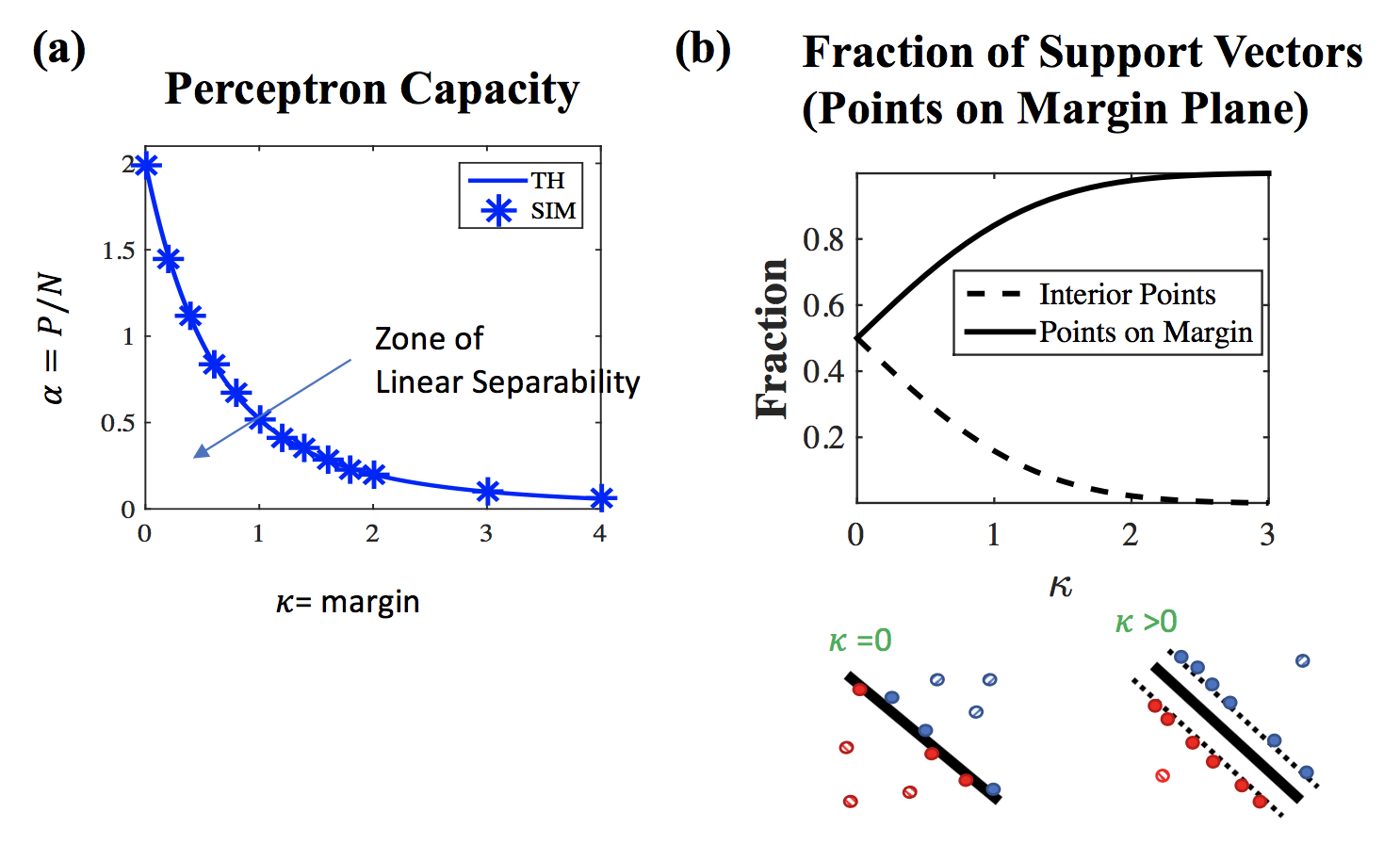}
\par\end{centering}
\caption{\textbf{Gardner's Perceptron Theory: Capacity and Support Vectors.}
(a) Gardner's replica analysis specifies the perceptron capacity $\alpha=P/N$
as a function of margin $\kappa$. $P$ is number of points, $N$
is the network size. (b) The fraction of support vectors amongst the
total number of points can be calculated as a function of margin $\kappa$
. At zero margin, half of the points are support vectors (black solid),
and the other half are interior points, that are in the space shattered
by the hyperplane (black dashed). The fraction of support vectors
in crease with increasing margin $\kappa$. \label{fig:GardnerResults}}

\end{figure}

Gardner theory is also applicable to more complex statistical ensembles,
such as the case of sparse labels where the labels are not uniformly
distributed. However, the current theory is inapplicable to the problem
of manifold classification, where the strong correlations between
points belonging to the same manifold is of primary importance. The
thesis addresses the following questions:

1. What is the capacity of manifolds, and the nature of solution?
What geometric features of the manifolds are relevant for the manifold
capacity? 

2. How to implement the practical aspects of analyzing data manifolds
numerically? In order to test the manifold capacity with simulation,
what is the most efficient algorithm to find a classification solution
for manifolds? To get an estimate of the manifold capacity, how to
numerically solve it? 

3. What are the necessary extensions required to understand and analyze
more realistic problems? We extend it to manifold classification problem
with sparse labeling, correlation, classification with nonlinearities
such as multilayer and nonlinear kernels, and apply the theory to
realistic data. 

\section{Outline of Thesis}

This thesis introduces a theory that generalizes Gardner's analysis
of perceptron capacity for isolated points to the perceptron capacity
for manifolds. The theory assumes (most of the time) that the manifolds
span a low dimensional hyperspace (strictly speaking the embedding
dimension is held finite as $N\rightarrow\infty).$ In the following
chapters, we introduce a set of investigations that lays groundwork
for a comprehensive theory of linear manifold classification. In chapter
2, we provides the basic tools for applying the replica theory to
compute linear classification of manifolds. Here we focus on the simple
manifolds: line segments, $L_{2}$ balls, and $L_{p}$ balls. In chapter
3, we address the numerical question of how to solve max margin problems
on manifolds, which consists of uncountable set of training examples.
We use methods from Quadratic Semi-Infinite Programming (QSIP) to
develop a novel algorithm denoted M4 (Max Margin Manifolds Machines).
In chapter 4, we generalize the theory of chapter 2 to address more
complex manifold geometries, for both smooth and non-smooth manifolds.
In chapter 5, we present a set of important extensions of the theory
to cover more realistic conditions, such as correlated manifolds,
and sparse coding tasks. We also discuss extensions to nonlinear manifold
classifications. Finally, we demonstrate how the theory can be applied
to analyze deep networks for in visual object recognition. 

\section{Chapter 2: Linear Separation of $L_{p}$Balls}

\begin{figure}[h]
\begin{centering}
\includegraphics[width=0.7\textwidth]{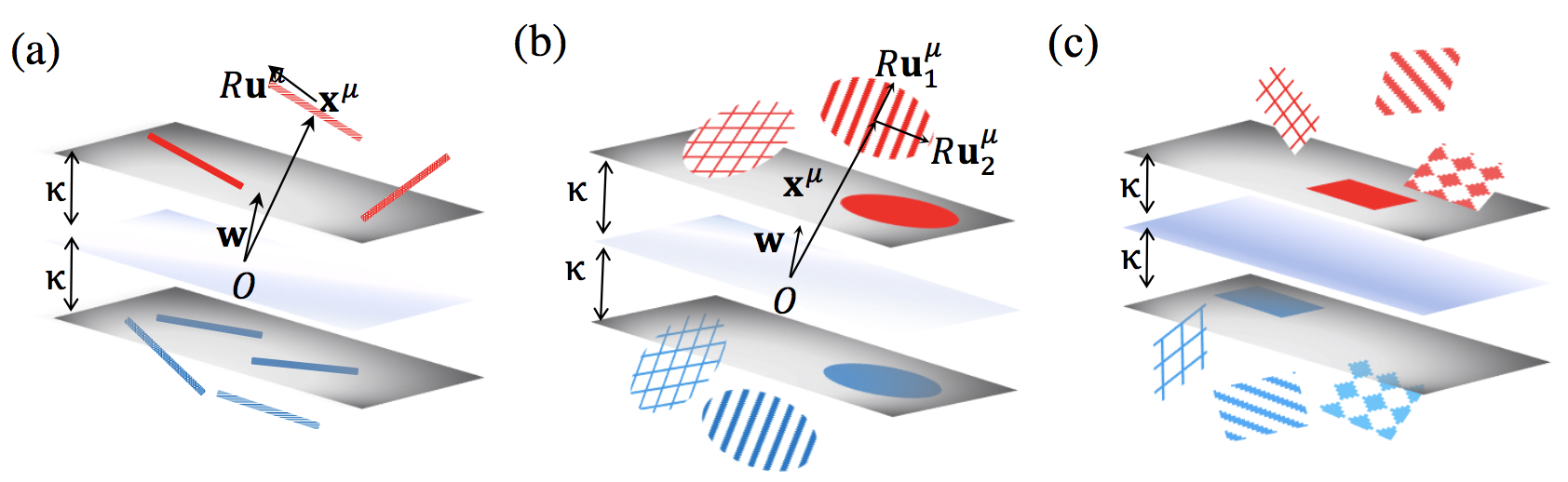}
\par\end{centering}
\caption{\textbf{Linear Classification of $L_{p}$ balls.} (a) Classification
of line segments with length $2R$. (Example of $D=1$ balls with
radius $R$). (b) Classification of $D$-dimensional $L_{2}$ balls
with radius $R$. (c) Classification of $D$-dimensional $L_{1}$
balls with radius $R$. \label{fig:LpClassification} }
\end{figure}

In this chapter we lay the ground for the statistical mechanical theory
of linear classification of manifolds. We consider manifolds which
can be described as $L_{p}$ balls in $D$ dimensions with a radius
$R$ . We write points on the manifolds as,

\begin{equation}
x^{\mu}=\left\{ x_{0}^{\mu}+R\sum_{i=1}^{D}s_{i}u_{i}^{\mu}\right\} \label{eq:BallsEqnIntro}
\end{equation}

where $x_{0}^{\mu}$ is the center of the $\mu$th ball, $\mu=1,...,P$.
 The axes of the balls are given by the D vectors $u_{i}^{\mu}$where
$i=1,...,D$. The vector $\vec{s}$ parameterizes the point on the
ball and obeys the constraint $||s||_{p}\leq1$ . The case of $p=2$
corresponds to the usual Euclidean balls in $D$ dimensions. The case
of $D=1$ is the special case of line segments with length $2R$.
Other examples are shown in fig. \ref{fig:LpClassification}. As we
show in this chapter, linear classification of these balls corresponds
to the requirements that the closest points on each manifolds obeys
inequalities, eq \ref{eq:BallsEqnIntro} above. For the $L_{p}$ balls,
with $1<p\leq\infty$1, this amounts to the following constraints
(where we consider zero bias for simplicity)

\begin{equation}
h_{0}^{\mu}-R||\vec{h}^{\mu}||_{q}\geq\kappa,\label{eq:minS}
\end{equation}

\begin{equation}
q=p/(p-1),\;1<p<\infty\label{eq:minS-1}
\end{equation}

\begin{equation}
q=\infty\;0<p\leq1\label{eq:minS-1-1}
\end{equation}

where $h_{0}^{\mu}=||\mathbf{w}||^{-1}y^{\mu}\mathbf{w}\cdot\mathbf{x}^{\mu}$
are the fields induced by the centers and $h_{i}^{\mu}=||\mathbf{w}||^{-1}y^{\mu}\mathbf{w}\cdot\mathbf{u_{i}}^{\mu}$
are the fields induced by the $i$ th basis vectors of $\mu$th manifold,
$\kappa$ is the margin of the linear classifier. 

Importantly, linear classification of manifolds depends on the geometric
properties of the \textbf{convex hulls} of the data manifolds. Thus,
when $p\le1$, the convex hull of the manifold becomes faceted, consisting
of vertices, flat edges and faces. For these geometries, the constraints
on the fields associated with a solution vector $\mathbf{w}$ becomes:
$h_{0}^{\mu}-R\max_{i}\left\Vert h_{i}^{\mu}\right\Vert \ge\kappa$
for all $p<1$ (Fig. \ref{fig:LpClassification}(c)). 

The statistical mechanical theory evaluates the average of the log
of solution volume, 

\begin{equation}
V=\int_{\left\Vert \mathbf{w}\right\Vert ^{2}=N}d^{N}\mathbf{w}\:\prod_{\mu=1}^{P}\Theta\left(h_{0}^{\mu}-R||\vec{h}^{\mu}||_{q}-\kappa\right).\label{eq:Volume}
\end{equation}

and identifying the point of vanishing volume allowed us to evaluate
the capacity, in the form of $\alpha_{B_{p}}(\kappa,R,D)$ for various
norms $L_{p}$ . Beyond the capacity, the theory provides an important
insight into the nature of the max margin solution. In particular
it generalizes the notion of support vectors to support manifolds.
As we show, some of the support manifolds are fully embedded in the
margin hyperplanes, some are touching the planes in a single point,
while in the case of $L_{1}$balls, they may have edges or faces in
the hyperplanes. These properties have important implications for
noise robustness of the solutions. Finally, these examples already
reveal the tradeoff between $D$ and $R$, and the effect of large
$R$ and large $D$. Specifically, we show that for large $D$ $L_{2}$
balls,

\begin{equation}
\alpha_{B}(\kappa,R,D\gg1)=(1+R^{2})\alpha_{G}(\kappa+R\sqrt{D}),\,\ D\gg1
\end{equation}
relating linear separation of balls to linear separation of points
with an additional effective margin $R\sqrt{D}$. 

\section{Chapter 3: \textcolor{black}{The Max Margin Manifold Machine}}

\begin{figure}[h]

\begin{centering}
\includegraphics[width=0.9\textwidth]{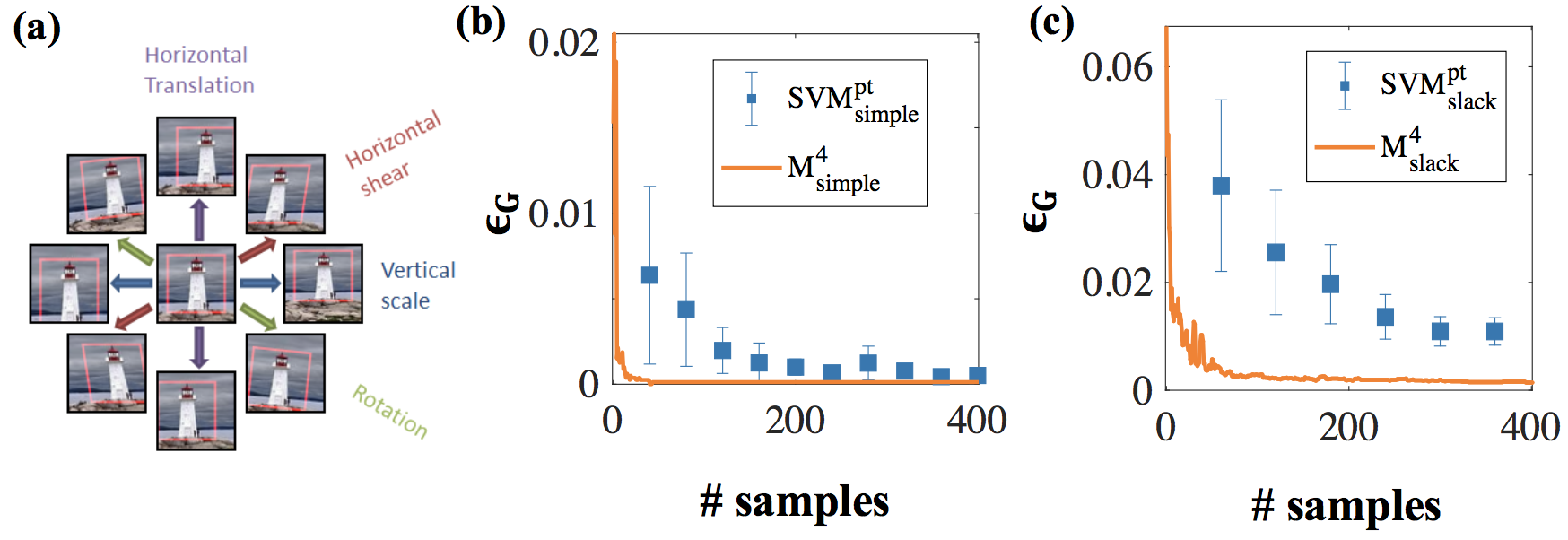}
\par\end{centering}
\caption{\textbf{The Maximum Margin Manifold Machine Performance.} (a) Object-based
manifolds are created by affine transformations of images from ImageNet
dataset. (b) Generalization error versus number of training samples,
for the linear $SVM$ (blue markers) and the $M_{simple}^{4}$ algorithm,
when the manifolds are in the separable regime. (b) Generalization
error versus the number of training samples for slack-SVM (blue) and
$M_{slack}^{4}$ algorithm when manifolds are in the non-separable
regime (above capacity). Our $M_{simple}^{4}$ and $M_{slack}^{4}$
show superior generalization error performance for the same number
of training samples. \label{fig:M4performance} }

\end{figure}

\textcolor{black}{Most learning algorithms assume the number of training
examples is finite. In this work, we consider the problem of classifying
data manifolds utilizing the underlying manifold structure consisting
of an uncountable number of points. We propose an efficient iterative
algorithm called $M^{4}$ that solves a quadratic semi-infinite programming
problem to find the maximum margin solution. Our method is based upon
a cutting-plane approach which converges to an approximate solution
in a finite number of iterations. We provide a proof of convergence
as well as a polynomial bound on the number of iterations and training
examples required for a desired tolerance in the objective function.
The efficiency and performance of $M^{4}$ are demonstrated on high-dimensional
synthetic data in addition to object manifolds generated by continuous
transformations of images from the ImageNet dataset. Our results indicate
that $M^{4}$ is able to rapidly learn good classifiers and shows
superior generalization performance than traditional support vector
machines using data augmentation methods (Fig. }\ref{fig:M4performance}\textcolor{black}{). }

\section{Chapter 4: Linear Classification of General Low Dimensional Manifolds}

In this chapter we generalize the perceptron capacity for the classification
of manifolds further, to classification of general manifolds. The
theory is exact in the thermodynamic limit, i.e., $N,P\rightarrow\infty$,
$\alpha=P/N$ is finite as in the Gardner's analysis. In addition,
for the mean field theory to be exact, the dimensionality of the manifolds
$D$ has to be finite in this limit (note: this holds except for the
special case of parallel manifolds, section \ref{subsec:ParallelSpheres},
where $D$ is proportional to $N$). To set the stage, we first consider
linear classification capacity of $L_{2}$ ellipsoids. We present
explicit analytical solution to the classification problem, and show
that the capacity and solution properties depend in general on all
$D$ radii. Like the balls, the max margin solution is characterized
by two types of support ellipsoids (touching or fully embedded). 

\begin{figure}[h]

\begin{centering}
\includegraphics[width=0.8\textwidth]{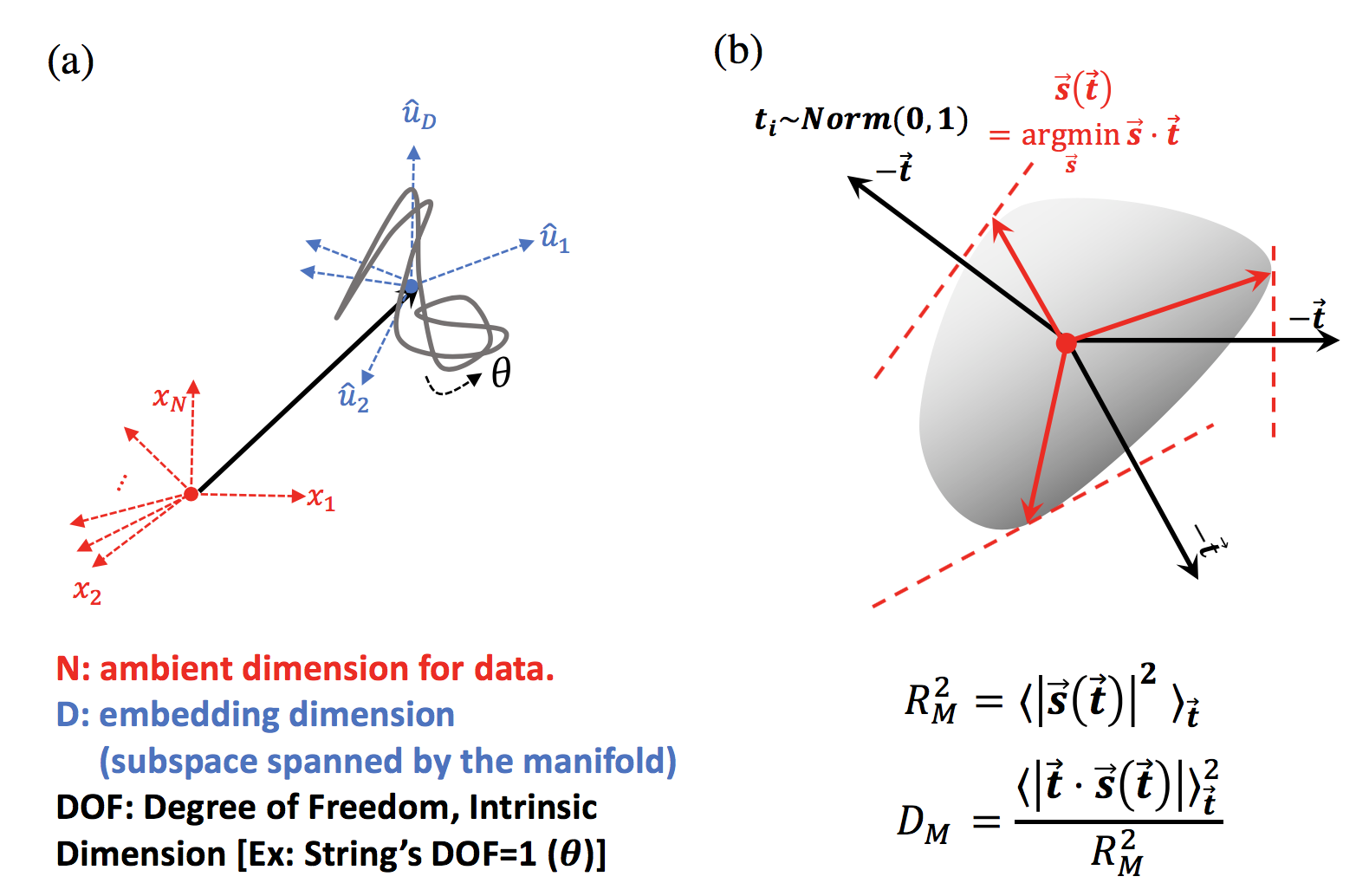}
\par\end{centering}
\caption{\textbf{Manifold's Sizes and Dimensions. }(a) Dimensions of a Random
String. A random string's degree of freedom (intrinsic dimension)
is 1, but is spanning $D$-dimensional (embedding dimension) and defined
in $N$ ambient dimension. (b) Effective manifold radius $R_{M}$
and effective manifold dimension $D_{M}$, which are relevant properties
for the manifold's linear classification capacity. \label{fig:ManifoldDimensions}}

\end{figure}

\emph{Effective Sizes and Dimensions: }In general, manifolds considered
here are characterized by several dimensionalities. All points on
all manifolds are in $\mathbb{R}^{N}$, so $N$ is the ambient dimension.
All points on a given manifold (relative to its center) span $D$
dimensions, thus $D$ is the manifold embedding dimension. In addition,
manifolds may be characterized by intrinsic dimensionality which may
be much smaller than $D$. See Fig. \ref{fig:ManifoldDimensions}(a)
for an example of a string in $D$ dimension. This intrinsic dimension
is important practically, but will not play an important role in the
theory of linear classification. In addition to the above, the manifold
classification properties may be described in certain regime by effective
dimensions and effective size (Fig. \ref{fig:ManifoldDimensions}(b)). 

Here we present the results for ellipsoids in the important limit
of large $D$. In this limit we find that the capacity can be well
approximated as,

\begin{equation}
\alpha_{E}(\kappa,R)=(1+R_{E}^{2})\alpha_{G}(\kappa+R_{E}\sqrt{D_{E}}),\,\ D\gg1
\end{equation}

where $E$ stands for ellipsoids, and with effective ellipsoid radius
$R_{E}$ and effective ellipsoid dimension $D_{E}$ given by,

\begin{equation}
R_{E}^{2}=\sum_{i=1}^{D}\frac{R_{i}^{4}}{(1+R_{i}^{2})^{2}}/\sum_{j=1}^{D}\frac{R_{j}^{2}}{(1+R_{j}^{2})^{2}}\label{eq:Reff}
\end{equation}

\begin{equation}
D_{E}=\left(\sum_{i=1}^{D}\frac{R_{i}^{2}}{1+R_{i}^{2}}\right)^{2}/\sum_{i=1}^{D}\frac{R_{i}^{4}}{(1+R_{i}^{2})^{2}}\label{eq:Deff}
\end{equation}
where $R_{i}$ are the different radii of the ellipsoid. Finally,
when the radii are small, $R_{i}\ll1$ (i.e., relative to the center
norms which is normalized here to $1$). these quantities reduce to
the simple formulae

\begin{equation}
R_{E}^{2}=\frac{\sum_{i}R_{i}^{4}}{\sum_{i}R_{i}^{2}}\label{eq:Reff_elps_scale}
\end{equation}

\begin{equation}
D_{E}=\frac{(\sum_{i}R_{i}^{2})^{2}}{\sum_{i}R_{i}^{4}}=D_{svd}\label{eq:Deff_elps_scale}
\end{equation}
where $D_{svd}$ is the participation ratio evaluated from the SVD
of the ellipsoids (with a uniform measure). These results set the
stage for a derivation of a theory applicable to \emph{general low
dimensional manifolds}. Briefly, general smooth convex manifolds behave
qualitatively the same as the ellipsoids, for the geometric reason
that they can either be interior to, fully embedded in or touching
the margin planes. 

Non-smooth manifold can have a large spectrum of overlaps with the
planes (as the example of $L_{1}$ ball indicates). Nevertheless,
we have derived self consistent mean field equations that describe
the capacity (and solution properties) for a general manifold, and
present numerical procedures to solve these equations iteratively.
Here we briefly discuss the theoretical prediction for the limit of
large $D$. In this regime, capacity is well approximated by 

\begin{equation}
\alpha_{M}(\kappa)=(1+R_{M}^{2})\alpha_{G}(\kappa+R_{M}\sqrt{D_{M}}),\,\ D\gg1
\end{equation}
with self consistent equations for $R_{M}$ and $D_{M}$, which need
to be solved numerically by iterative mean field methods. Remarkably,
in the regime where $R_{M}\ll1$, $R_{M}$ and $D_{M}$ simplify to
the quantities shown in Fig. \ref{fig:ManifoldDimensions}(b) and
are related to the well known Gaussian Mean Width of convex bodies
(Fig. \ref{fig:MeanWidthOfManifold}). 

An important application of this theory is finite point cloud manifolds
that arise when subsampled points of each potentially continuous manifold
is given. In this case, $R_{M}$ and $D_{M}$ (of the training manifolds)
can be estimated from the given finite training set. The interesting
question of how these quantities are related to the effective radius
and dimension underlying full manifold is touched upon in the following
section. An interesting example is the case of $L_{1}$balls in $D$
dimensions with radius $R$. In the limit of large $D$ and small
$R$, the effective radius is simply $R$ but the dimension is 
\begin{equation}
D_{M}=2\log D\label{eq:L1DM_2LogD}
\end{equation}
In general, in other point cloud manifolds we expect that $D_{M}\propto\log m$
where $m$ is the number of samples per manifold. 

\emph{Infinite size manifolds:}\textbf{ }Finally, it should be noted
that as the manifold size grows to infinity (in all dimensions), their
geometric details don't matter; only the number of dimensions they
span. Here we obtain

\begin{equation}
D_{M}\rightarrow D\label{eq:L1DM_2LogD-1}
\end{equation}

reflecting the need of the classifying weight vector to be orthogonal
to all the $DP$ dimensional hyperspace that the manifolds span, namely
the capacity reduces to

\begin{equation}
\alpha_{M}=\frac{1}{D}
\end{equation}

where $D$ denotes the embedding dimensions of the manifolds (where
we assume for simplicity that the manifolds are not bounded in any
of the $D$ directions). 

\section{Chapter 5: Extensions }

\begin{figure}
\begin{centering}
\includegraphics[width=0.9\textwidth]{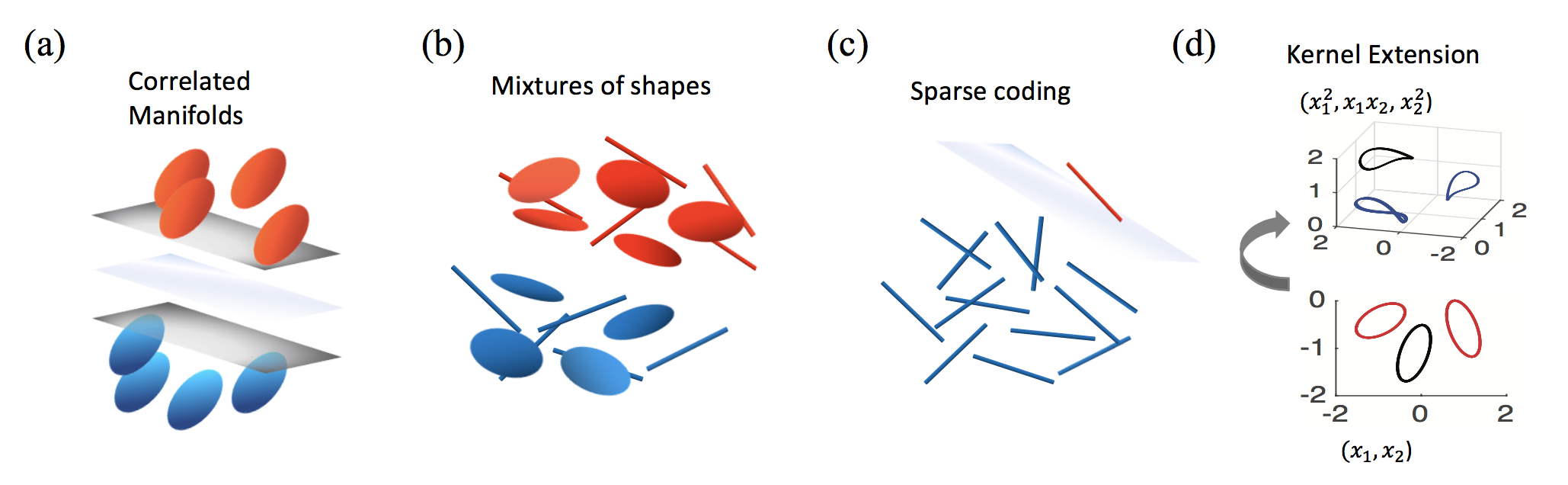}
\par\end{centering}
\caption{\textbf{Extensions of Manifold Classification Theory. }(a) Classification
of Correlated Manifolds. (b) Classification of Mixtures of Manifolds.
(c) Sparse Coding (Classification with Sparse Labels) and Object Recognition
Limit (One versus All Classification). (d) Extension to Kernel Framework.
(Red/black) $2D$ disks in the kernel input space, transformed to
(blue/black) $5D$ manifolds in the quadratic kernel's feature space.
\label{fig:ExtensionsManifolds}}
\end{figure}

In Chapter 5, we further extend the theory in directions likely relevant
to applications to real data. We have extended our general manifold
classification theory to incorporate correlated manifolds, mixtures
of manifold geometries, sparse labels and nonlinear classification,
see Fig. \ref{fig:ExtensionsManifolds}. We highlight here briefly
several important results. 
\medskip{}
\emph{1. Correlated manifolds:} when manifold axes are strongly parallel
(fig. \ref{fig:ExtensionsManifolds}(a)) we expect the capacity to
be relatively large. For example if their spanning spaces are fully
aligned but they are large in extent, $\mathbf{w}$ can solve the
problem by orthogonalize to the $D$ common directions (rather than
$DP$ in the uncorrelated case). Interestingly, for high dimensional
parallel balls we find a phase transition whereby above some finite
critical radius the max margin solution fully orthogonalize to the
manifolds subspace. In real data we expect positive correlations but
not full alignment of the different manifolds. 

\medskip{}
\emph{2. Sparse labels:} In this case, the fraction of say plus manifolds,
$f$, is smaller than that of the minus ones. In many real life tasks
this is to be expected. An extreme case is that of object recognition
task defined as classifying one manifold as one and the rest as minus
one. This can be viewed as a binary classification with $f=\frac{1}{P}$
. As in Gardner's theory the capacity grows as $f\rightarrow0$. However,
we show that the size of the manifolds substantially limits this growth.
For instance, in balls with large radius $R$ , the $f$ is small
but larger than $1/R$ the capacity remains of order unity.

\medskip{}
\emph{3. Nonlinear manifold separation}: We consider two schemes of
two layer classification of manifolds in cases where they are not
linearly separable. One is in the form of a nonlinear kernel, similar
to Kernel SVM. For this we present a version of the $M^{4}$algorithm
in a \emph{'dual' }form, appropriate for kernels. We briefly discuss
the effect of the kernel on the geometry of the manifold and the classification
capacity. The second architecture is that of hidden layer of binary
units, forming a sparse intermediate representation of the manifolds.
We show how this extra layer formed by unsupervised learning can enhance
the capacity and robustness of the classification of the manifolds.

\medskip{}
\emph{4. Generalization properties :} Computation with manifolds raises
a specific type of generalization problem, namely how training with
a subsampled training points perform when new points from the same
underlying manifolds are presented in the test phase. Exact analytical
expression for the generalization error is complicated; also the error
depends on the assumed sampling measure on the manifold (whereas the
separability problem is measure invariant). However, in the case of
linearly separable manifolds with high $D$ we can use the insight
from the above theory (the notions of effective dimensions and radii)
to derive a particularly simple approximation. Assume $\alpha$ is
such that the full manifolds are linearly separable with a max margin
$\kappa$. Then the generalization error will eventually vanish as
more samples per manifold , $m$ , are presented. In the limit of
large $m$ , we obtain, 

\begin{equation}
\epsilon_{g}(m)\propto\frac{\exp[-\kappa\sqrt{2\log m}]}{m}
\end{equation}
Interestingly, this decay is faster than the generic power law, $\epsilon_{g}(m)\propto m^{-1}$
of generalization bounds in linearly separable problem and reflects
the presence of finite margin of the entire manifold. We also discuss
the generalization error of subsampled manifolds in the case where
the full manifolds are not linearly separable. 

\medskip{}
\emph{5. Application to Deep Networks: }We close this section by applying
some of the theoretical concepts to Deep Networks trained to perform
visual object recognition tasks. We show how the theory can be used
to characterize the change in the geometry of the manifolds, and changes
in the manifold correlation structure at different stages of the network
(using ImageNet\cite{deng2009imagenet} as an example). 

\section{Conclusion and Future Direction }

In this thesis, we generalized Gardner's theory of linear classification
of points to the classification of general randomly oriented low dimensional
manifolds. The theory, exact in the thermodynamic limit, describes
the relation between the detailed geometry of the convex hulls of
the data manifolds and the ability to linearly classify them. The
problem simplifies considerably when the manifold dimension is high.
In this limit, the classification properties depend on two geometric
parameters of the convex hulls: the effective dimension $D_{M}$ and
effective radius $R_{M}$. In high dimensional manifold with small
sizes, capacity depends on $R_{M}$ and $D_{M}$ mainly through the
scaling relation $R_{M}\sqrt{D_{M}}$. This quantity is related to
the well known Gaussian Mean Width of convex bodies. Optimal solution
exhibits support manifold structures with potential consequences for
noise robustness. We developed a novel efficient training algorithm,
the Maximum Margin Manifold Machines, for finding the maximum margin
solution for classifying manifolds with uncountable number of training
samples, and provide convergence proof with polynomial bounds on the
number of iterations required for convergence. Our theory has been
extended to  incorporate correlations in the manifolds, mixtures of
shapes, sparse coding, nonlinear processing such as multilayer network
or kernel framework, as well as an analysis of manifold generalization
error. With these extensions, our theory provides qualitative and
quantitative measures for assessing the ability to decode object information
from the different stages of Deep biological and artificial neural
networks. 

Ongoing work includes suggesting design principles for deep networks
by taking into account the network size, dimension, sparsity, as well
as role of nonlinearities in reformatting of the manifolds such that
the capacity is increased. Whether manifold capacity can be used as
an object function of the training of a network is an interesting
question to pursue. We are exploring applications of our theory on
several neural data bases from IT and other areas in visual cortex,
responding to different object stimuli with a variety of physical
transformations. We hope that our theory will provide new insights
into the computational principles underlying processing of sensory
representations in the brain. As manifold representations of the sensory
world are ubiquitous in both biological and artificial neural systems,
exciting future work lies ahead.

%% file: chapters/ch2_spheres.tex
\chapter{Linear Classification of Spherical Manifolds \label{cha:spheres}}

High-level perception in the brain involves classifying or identifying
objects which are represented by continuous manifolds of neuronal
states in all stages of sensory hierarchies \cite{dicarlo2007untangling,pagan2013signals,alemi2013multifeatural,bizley2013and,meyers2015intelligent,schwarzlose2008distribution,gottfried2010central}
Each state in an object manifold corresponds to the vector of firing
rates of responses to a particular variant of physical attributes
which do not change object's identity, e.g., intensity, location,
scale, and orientation. It has been hypothesized that object identity
can be decoded from high level representations, but not from low level
ones, by simple downstream readout networks \cite{hung2005fast,dicarlo2007untangling,pagan2013signals,freiwald2010functional,cadieu2014deep,kobatake1994neuronal,rust2010selectivity,schwarzlose2008distribution}.
A particularly simple decoder is the perceptron, which performs classification
by thresholding a linear weighted sum of its input activities \cite{minsky1987perceptrons,gardnerEPL}.
However, it is unclear what makes certain representations well suited
for invariant decoding by simple readouts such as perceptrons. Similar
questions apply to the hierarchy of artificial deep neural networks
for object recognition \cite{serre2005object,goodfellow2009measuring,ranzato2007unsupervised,bengio2009learning,cadieu2014deep}.
Thus, a complete theory of perception requires characterizing the
ability of linear readout networks to classify objects from variable
neural responses in their upstream layer.

A theoretical understanding of the perceptron was pioneered by Elizabeth
Gardner who formulated it as a statistical mechanics problem and analyzed
it using replica theory \cite{gardner1988space,engel2001statistical,advani2013statistical,brunel2004optimal,sompolinsky1990learning,opper1991generalization,rubin2010theory,amit1989perceptron,monasson1992properties}.
In this work, we generalize the statistical mechanical analysis and
establish a theory of linear classification of manifolds synthesizing
statistical and geometric properties of high dimensional signals.
We apply the theory to simple classes of manifolds and show how changes
in the dimensionality, size, and shape of the object manifolds affect
their readout by downstream perceptrons.

\section{Line Segments}

One-dimensional object manifolds arise naturally from variation of
stimulus intensity, such as visual contrast, which leads to approximate
linear modulation of the neuronal responses of each object. We model
these manifolds as line segments and consider classifying $P$ such
segments in $N$ dimensions, expressed as $\left\{ \mathbf{x}_{0}^{\mu}+Rs\mathbf{u}^{\mu}\right\} $,
$-1\le s\le1$, $\mu=1,...,P$. The $N$-dimensional vectors $\mathbf{x}^{\mu}\in\mathbb{R}^{N}$
and $\mathbf{u}^{\mu}\in\mathbb{R}^{N}$ denote respectively, the
\emph{centers} and \emph{directions }of\emph{ }the $\mu$-th segment,
and the scalar $s$ parameterizes the continuum of points along the
segment. The parameter $R$ measures the extent of the segments relative
to the distance between the centers (Fig. \ref{fig:PerceptronLines}).

We seek to partition the different line segments into two classes
defined by binary labels $y^{\mu}=\pm1$ . To classify the segments,
a weight vector $\mathbf{w}\in\mathcal{R}^{N}$ must obey $y^{\mu}\mathbf{w}\cdot\left(\mathbf{x}^{\mu}+Rs\mathbf{u}^{\mu}\right)\ge\kappa$
for all $\mu$ and $s$. The parameter $\kappa\ge0$ is known as the
margin; in general, a larger $\kappa$ indicates that the perceptron
solution will be more robust to noise and display better generalization
properties \cite{vapnik1998statistical}. Hence, we are interested
in maximum margin solutions, i.e., weight vectors $\mathbf{w}$ that
yield the maximum possible value for $\kappa$. Since line segments
are convex, only the endpoints of each line segment need to be checked,
namely $\min\,h_{0}^{\mu}\pm Rh^{\mu}=h_{0}^{\mu}-R\left|h^{\mu}\right|\ge\kappa$
where $h_{0}^{\mu}=||\mathbf{w}||^{-1}y^{\mu}\mathbf{w}\cdot\mathbf{x}^{\mu}$
are the fields induced by the centers and $h^{\mu}=||\mathbf{w}||^{-1}y^{\mu}\mathbf{w}\cdot\mathbf{u}^{\mu}$
are the fields induced by the line directions.

\begin{figure}
\noindent \begin{centering}
\includegraphics[width=0.9\textwidth]{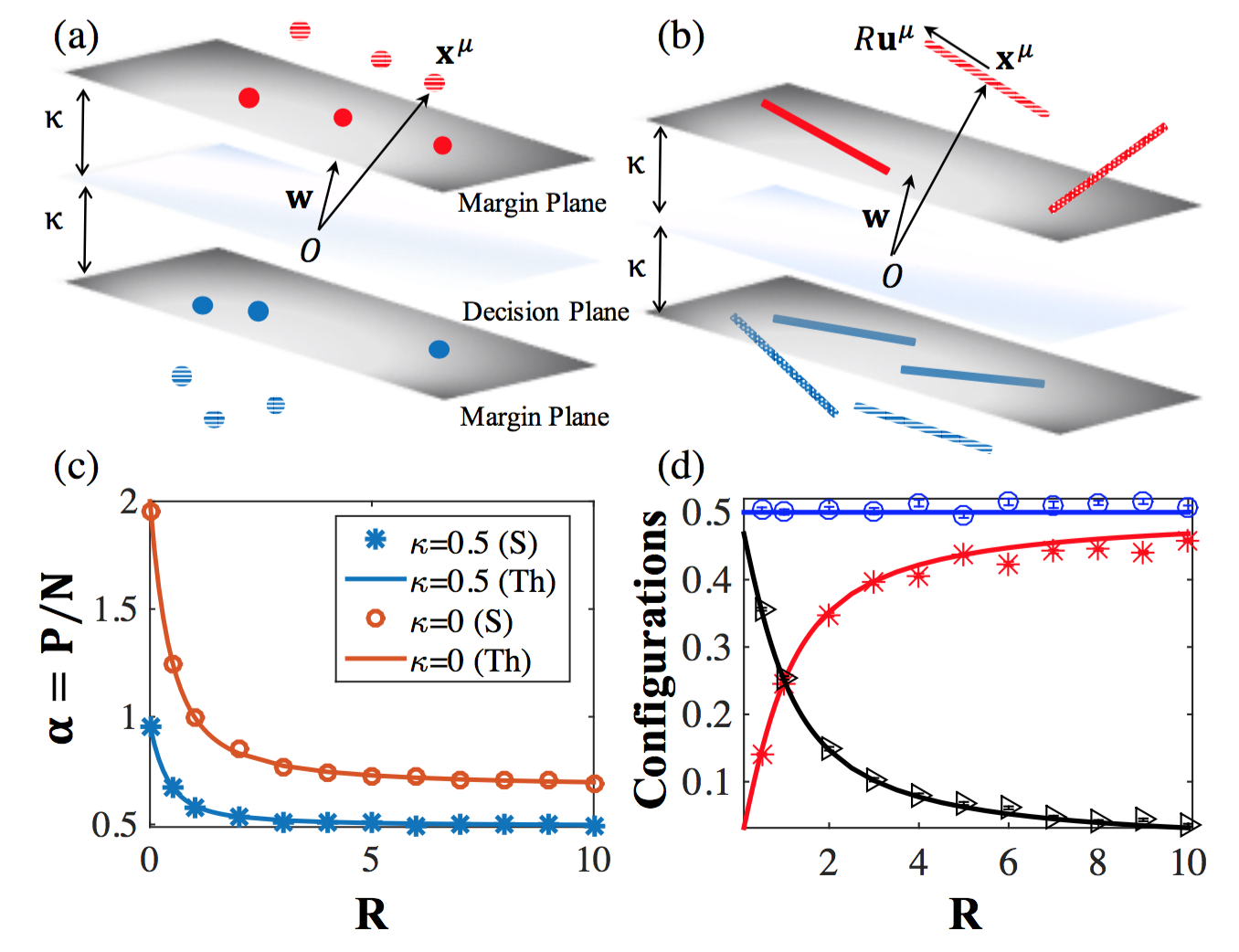} 
\par\end{centering}
\caption{(a) Linear classification of points. (solid) points on the margin,
(striped) internal points. (b) Linear classification of line segments.
(solid) lines embedded in the margin, (dotted) lines touching the
margin, (striped) interior lines. (c) Capacity $\alpha=P/N$ of a
network $N=200$ as a function of $R$ with margins $\kappa=0$ (red)
and $\kappa=0.5$ (blue). Theoretical predictions (lines) and numerical
simulation (markers, see Appendix for details) are shown. (d) Fraction
of different line configurations at capacity with $\kappa=0$. (red)
lines in the margin, (blue) lines touching the margin, (black) internal
lines. \label{fig:PerceptronLines} }
\end{figure}

\subsection{Replica Theory}

The existence of a weight vector $\mathbf{w}$ that can successfully
classify the line segments depends upon the statistics of the segments.
We consider random line segments where the components of $\mathbf{x}^{\mu}$
and $\mathbf{u}^{\mu}$ are i.i.d. Gaussians with zero mean and unit
variance, and random binary labels $y^{\mu}$. We study the thermodynamic
limit where the dimensionality $N\rightarrow\infty$ and number of
segments $P\rightarrow\infty$ with finite $\alpha=P/N$ and $R$.
Following Gardner \cite{gardner1988space} we compute the average
of $\log V$ where $V$ is the volume of the space of perceptron solutions:
\begin{equation}
V=\int_{\left\Vert \mathbf{w}\right\Vert ^{2}=N}d^{N}\mathbf{w}\:\prod_{\mu=1}^{P}\Theta\left(h_{0}^{\mu}-R\left\Vert h^{\mu}\right\Vert -\kappa\right).\label{eq:Volume}
\end{equation}
$\Theta(x)$ is the Heaviside step function. According to replica
theory, the fields are described as sums of random Gaussian fields
$h_{0}^{\mu}=t_{0}^{\mu}+z_{0}^{\mu}$ and $h^{\mu}=t^{\mu}+z^{\mu}$
where $t_{0}$ and $t$ are quenched components arising from fluctuations
in the input vectors $\mathbf{x}^{\mu}$ and $\mathbf{u}^{\mu}$ respectively,
and the $z_{0}$, $z$ fields represent the variability in $h_{0}^{\mu}$
and $h^{\mu}$ resulting from different solutions of $\mathbf{w}$.
These fields must obey the constraint $z_{0}+t_{0}-R\left|z+t\right|\ge\kappa.$
The capacity function $\alpha_{L}(\kappa,R)$ (the subscript $L$
denotes the line) describes for which $P/N$ ratio the perceptron
solution volume shrinks to a unique weight vector. The reciprocal
of the capacity is given by the replica symmetric calculation (details
provided in the Appendix \ref{subsec:AppendixReplicaLines}): 
\begin{equation}
\alpha_{L}^{-1}(\kappa,R)=\left\langle \min_{z_{0}+t_{0}-R\left|z+t\right|\ge\kappa}\frac{1}{2}\left[z_{0}^{2}+z^{2}\right]\right\rangle _{t_{0},t}\label{eq:alphaLinesAverage}
\end{equation}
where the average is over the Gaussian statistics of $t_{0}$ and
$t$. To compute Eq. \eqref{eq:alphaLinesAverage}, three regimes
need to be considered. First, when $t_{0}$ is large enough so that
$t_{0}>\kappa+R\left|t\right|$, the minimum occurs at $z_{0}=z=0$
which does not contribute to the capacity. In this regime, $h_{0}^{\mu}>\kappa$
and $h^{\mu}>0$ implying that neither of the two segment endpoints
reach the margin. In the other extreme, when $t_{0}<\kappa-R^{-1}|t|$,
the minimum is given by $z_{0}=\kappa-t_{0}$ and $z=-\left|t\right|$,
i.e. $h_{0}^{\mu}=\kappa$ and $h^{\mu}=0$ indicating that both endpoints
of the line segment lie on the margin planes. In the intermediate
regime where $\kappa-R^{-1}\left|t\right|<t_{0}<\kappa+R\left|t\right|$,
$z_{0}=\kappa-t_{0}+R|z+t$|, i.e., $h_{0}^{\mu}-R|h^{\mu}|=\kappa$
but $h_{0}^{\mu}>\kappa$, corresponding to only one of the line segment
endpoints touching the margin. In this regime, the solution is given
by minimizing the function $(R\left|z+t\right|+\kappa-t_{0})^{2}+z^{2}$
with respect to $z$. Combining these contributions, we can write
the perceptron capacity of line segments: 
\begin{eqnarray}
\alpha_{L}^{-1}(\kappa,R) & = & \int_{-\infty}^{\infty}Dt\int_{\kappa-R^{-1}|t|}^{\kappa+R|t|}Dt_{0}\frac{\left(R\left|t\right|+\kappa-t_{0}\right)^{2}}{R^{2}+1}\nonumber \\
 & + & \int_{-\infty}^{\infty}Dt\int_{-\infty}^{\kappa-R^{-1}|t|}Dt_{0}\left[(\kappa-t_{0})^{2}+t^{2}\right]\quad\quad\label{eq:alphaCLine}
\end{eqnarray}

with integrations over the Gaussian measure, $Dx\equiv\frac{1}{\sqrt{2\pi}}e^{-\frac{1}{2}x^{2}}dx$.
It is instructive to consider special limits. When $R\rightarrow0,$
Eq. \eqref{eq:alphaCLine} reduces to $\alpha_{L}(\kappa,0)=\alpha_{0}(\kappa)$
where $\alpha_{0}(\kappa)$ is Gardner's original capacity result
for perceptrons classifying $P$ points (the subscript $0$ stands
for zero-dimensional manifolds) with margin $\kappa$ \ref{fig:PerceptronLines}-(a).
Interestingly, when $R=1$, then $\alpha_{L}(\kappa,1)=\frac{1}{2}\alpha_{0}(\kappa/\sqrt{2})$.
This is because when $R=1$ there are no statistical correlations
between the line segment endpoints and the problem becomes equivalent
to classifying $2P$ random points with average norm $\sqrt{2N}$
.

Finally, when $R\rightarrow\infty$, the capacity is further reduced:
$\alpha_{L}^{-1}(\kappa,\infty)=\alpha_{0}^{-1}(\kappa)+1$. This
is because when $R$ is large, the segments become unbounded lines.
In this case, the only solution is for $\mathbf{w}$ to be orthogonal
to all $P$ line directions. The problem is then equivalent to classifying
$P$ center points in the $N-P$ null space of the line directions,
so that at capacity $P=\alpha_{0}(\kappa)(N-P)$.

We see this most simply at zero margin, $\kappa=0$. In this case,
Eq. \eqref{eq:alphaCLine} reduces to a simple analytic expression
for the capacity: $\alpha_{L}^{-1}(0,R)=\frac{1}{2}+\frac{2}{\pi}\arctan R$
(Appendix \ref{subsec:AppendixReplicaLines}). The capacity is seen
to decrease from $\alpha_{L}(0,R=0)=2$ to $\alpha_{L}(0,R=1)=1$
and $\alpha_{L}(0,R=\infty)=\frac{2}{3}$ for unbounded lines. We
have also calculated analytically the distribution of the center and
direction fields $h_{0}^{\mu}$ and $h^{\mu}$ \cite{abbott1989universality}.
The distribution consists of three contributions, corresponding to
the regimes that determine the capacity. One component corresponds
to line segments fully embedded in these planes. The fraction of these
manifolds is simply the volume of phase space of $t$ and $t_{0}$
in the last term of Eq. \eqref{eq:alphaCLine}. Another fraction,
given by the volume of phase space in the first integral of \eqref{eq:alphaCLine}
corresponds to line segments touching the margin planes at only one
endpoint. The remainder of the manifolds are those interior to the
margin planes. Fig. \ref{fig:PerceptronLines} shows that our theoretical
calculations correspond nicely with our numerical simulations for
the perceptron capacity of line segments, even with modest input dimensionality
$N=200$. Note that as $R\rightarrow\infty$, half of the manifolds
lie in the plane while half only touch it; however, the angles between
these segments and the margin planes approach zero in this limit.
As $R\rightarrow0$ , half of the points lie in the plane \cite{abbott1989universality}.

\section{$D$-dimensional Balls}

Higher dimensional manifolds arise from multiple sources of variability
and their nonlinear effects on the neural responses. An example is
varying stimulus orientation, resulting in two-dimensional object
manifolds under the cosine tuning function (Fig. \ref{fig:Disks}(a)).
Linear classification of these manifolds depends only upon the properties
of their convex hulls \cite{de2000computational}. We consider simple
convex hull geometries as $D$-dimensional balls embedded in $N$-dimensions:
$\left\{ \mathbf{x_{0}}^{\mu}+R\sum_{i=1}^{D}s_{i}\mathbf{u}_{i}^{\mu}\right\} $,
so that the $\mu$-th manifold is centered at the vector $\mathbf{x}^{\mu}\in\mathbb{R}^{N}$
and its extent is described by a set of $D$ basis vectors $\left\{ \mathbf{u}_{i}^{\mu}\in\mathbb{R}^{N},\:i=1,...,D\right\} $.
The points in each manifold are parameterized by the $D$-dimensional
vector $\vec{s}\in\mathbb{R}^{D}$ whose Euclidean norm is constrained
by: $\left\Vert \vec{s}\right\Vert \leq1$ and the radius of the balls
are quantified by $R$ .

Statistically, all components of $\mathbf{x}_{0}^{\mu}$ and $\mathbf{u}_{i}^{\mu}$
are i.i.d. Gaussian random variables with zero mean and unit variance.
We define $h_{0}^{\mu}=N^{-1/2}y^{\mu}\mathbf{w}\cdot\mathbf{x}^{\mu}$
as the field induced by the manifold centers and $h_{i}^{\mu}=N^{-1/2}y^{\mu}\mathbf{w}\cdot\mathbf{u}_{i}^{\mu}$
as the $D$ fields induced by each of the basis vectors and with normalization
$\left\Vert \mathbf{w}\right\Vert =\sqrt{N}$. To classify all the
points on the manifolds correctly with margin $\kappa$, $\mathbf{w}\in\mathbb{R}^{N}$
must satisfy the inequality $h_{0}^{\mu}-R||\vec{h}^{\mu}||\geq\kappa$
where $||\vec{h}^{\mu}||$ is the Euclidean norm of the $D$-dimensional
vector $\vec{h}^{\mu}$ whose components are $h_{i}^{\mu}$ . This
corresponds to the requirement that the field induced by the points
on the $\mu$-th manifold with the smallest projection on $\mathbf{w}$
be larger than the margin $\kappa$.

We solve the replica theory in the limit of $N,\,P\rightarrow\infty$
with finite $\alpha=P/N$, $D$, and $R$. The fields for each of
the manifolds can be written as sums of Gaussian quenched and entropic
components, $\left(t_{0}\in\mathbb{R},\:\vec{t}\in\mathbb{R}^{D}\right)$
and $\left(z_{0}\in\mathbb{R},\:\vec{z}\in\mathbb{R}^{D}\right)$
, respectively. The capacity for $D$-dimensional manifolds is given
by the replica symmetric calculation (Appendix \ref{subsec:AppendixReplicaBalls}):

\begin{equation}
\alpha_{B}^{-1}(\kappa,R,D)=\left\langle \min_{t_{0}+z_{0}-R\left\Vert \vec{t}+\vec{z}\right\Vert >\kappa}\frac{1}{2}\left[z_{0}^{2}+\left\Vert \vec{z}\right\Vert ^{2}\right]\right\rangle _{t_{0},\vec{t}}
\end{equation}

where $B$ stands for $L_{2}$balls. The capacity calculation can
be partitioned into three regimes. For large $t_{0}>\kappa+Rt$, where
$t=\left\Vert \vec{t}\right\Vert $, $z_{0}=0$ and $\vec{z}=0$ corresponding
to manifolds which lie interior to the margin planes of the perceptron.
On the other hand, when $t_{0}<\kappa-R^{-1}t$, the minimum is obtained
at $z_{0}=\kappa-t_{0}$ and $\vec{z}=-\vec{t}$ corresponding to
manifolds which are fully embedded in the margin planes. Finally,
in the intermediate regime, when $\kappa-R^{-1}t<t_{0}<\kappa+Rt$,
$z_{0}=R\left\Vert \vec{t}+\vec{z}\right\Vert -t_{0}+\kappa$ but
$\vec{z}\ne-\vec{t}$ indicating that these manifolds only touch the
margin plane. Decomposing the capacity over these regimes and integrating
out the angular components, the capacity of the perceptron can be
written as: 
\begin{eqnarray}
\alpha_{B}^{-1}(\kappa,R,D) & = & \int_{0}^{\infty}dt\,\chi_{D}(t)\int_{\kappa-\frac{1}{R}t}^{\kappa+Rt}Dt_{0}\frac{\left(Rt+\kappa-t_{0}\right)^{2}}{R^{2}+1}\nonumber \\
 & + & \int_{0}^{\infty}dt\,\chi_{D}(t)\int_{-\infty}^{\kappa-\frac{1}{R}t}Dt_{0}\left[\left(\kappa-t_{0}\right)^{2}+t^{2}\right]\qquad\label{eq:alphaCDisks}
\end{eqnarray}

where $\chi_{D}(t)=\frac{2^{1-\frac{D}{2}}}{\Gamma(\frac{D}{2})}t^{D-1}e^{-\frac{1}{2}t^{2}}$
is the \emph{D-}Dimensional Chi probability density function. For
large $R\rightarrow\infty$, Eq. \eqref{eq:alphaCDisks} reduces to:
$\alpha_{B}^{-1}(\kappa,R=\infty,D)=\alpha_{0}^{-1}(\kappa)+D$ which
indicates that $\mathbf{w}$ must be in the null space of the $PD$
basis vectors $\left\{ \mathbf{u}_{i}^{\mu}\right\} $ in this limit.
This case is equivalent to the classification of $P$ points (the
projections of the manifold centers) by a perceptron in the $N-PD$
dimensional null space.

To probe the fields, we consider the joint distribution of the field
induced by the center, $h_{0}$, and the norm of the fields induced
by the manifold directions, $h\equiv\left\Vert \vec{h}\right\Vert $
. There are three contributions. The first term corresponds to $h_{0}-Rh>\kappa$,
i.e. balls that lie interior to the perceptron margin planes; the
second component corresponds to $h_{0}-Rh=\kappa$ but $h>0$, i.e.
balls that touch the margin planes; and the third contribution represents
the fraction of balls obeying $h_{0}=\kappa$ and $h=0$, i.e. balls
fully embedded in the margin. The dependence of these contributions
on $R$ for $D=2$ is shown in Fig. \ref{fig:Disks}(b). Interestingly,
when $\kappa=0$ , the case of $R=1$ is particularly simple for all
$D$ . The capacity is $\alpha_{B}(R=1,D)=2/(D+1)$ ; in addition,
the fraction of embedded and interior balls are equal and the fraction
of touching balls have a maximum, see Fig. \ref{fig:Disks}(b) and
Appendix.

\begin{figure}
\noindent \begin{centering}
\includegraphics[width=0.8\textwidth]{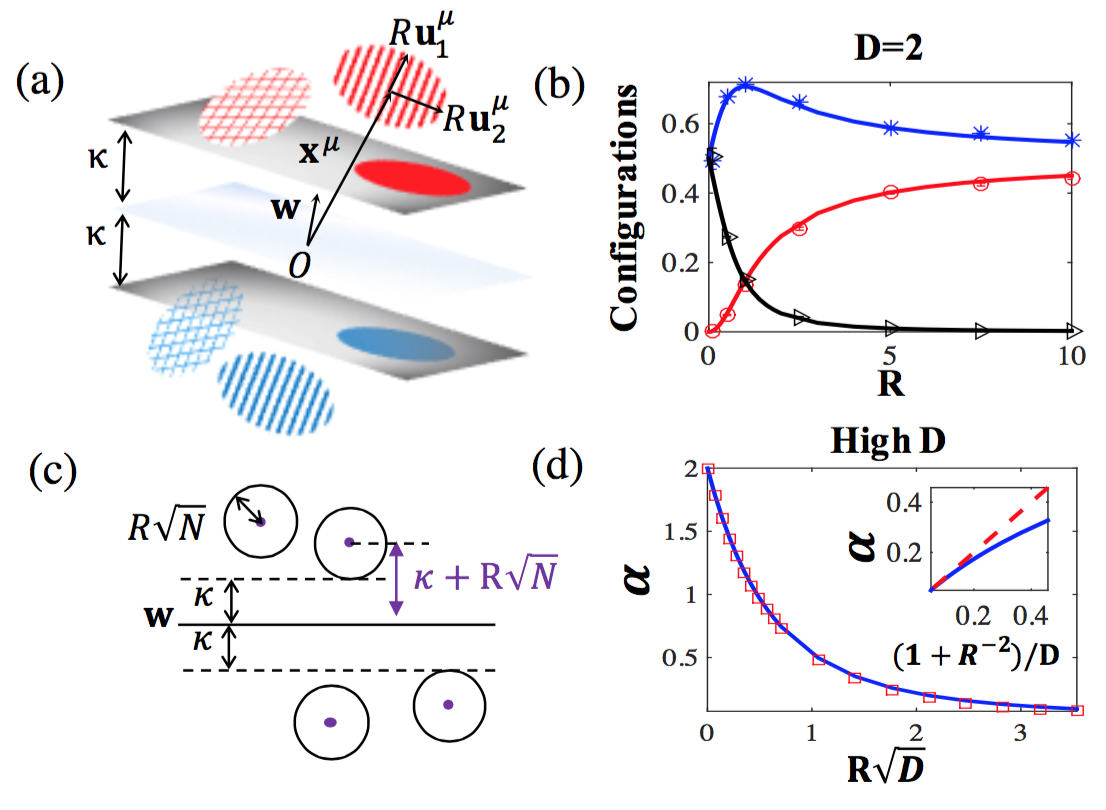} 
\par\end{centering}
\caption{Random $D$-dimensional balls: (a) Linear classification of $D=2$
balls. (b) Fraction of 2-$D$ ball configurations as a function of
$R$ at capacity with $\kappa=0$, comparing theory (lines) with simulations
(markers). (red) balls embedded in the plane, (blue) balls touching
the plane, (black) interior balls. (c) Linear classification of balls
with $D=N$ at margin $\kappa$ (black circles) is equivalent to point
classification of centers with effective margin $\kappa+R\sqrt{N}$
(purple points). (d) Capacity $\alpha=P/N$ for $\kappa=0$ for large
$D=50$ and $R\propto D^{-1/2}$ as a function of $R\sqrt{D}$. (blue
solid) $\alpha_{B}(\kappa=0,R,D)$ compared with $\alpha_{0}(\kappa=R\sqrt{D})$
(red square). (Inset) Capacity $\alpha$ at $\kappa=0$ for $0.35\leq R\leq20$
and $D=20$: (blue) theoretical $\alpha$ compared with approximate
form $(1+R^{-2})/D$ (red dashed).\label{fig:Disks}}
\end{figure}

In a number of realistic problems, the dimensionality $D$ of the
object manifolds could be quite large. Hence, we analyze the limit
$D\gg1$. In this situation, for the capacity to remain finite, $R$
has to be small, scaling as $R\propto D^{-\frac{1}{2}}$, and the
capacity is $\alpha_{B}(\kappa,R,D)\approx\alpha_{0}(\kappa+R\sqrt{D})$.
In other words, the problem of separating $P$ high dimensional balls
with margin $\kappa$ is equivalent to separating $P$ points but
with a margin $\kappa+R\sqrt{D}$. This is because when the distance
of the closest point on the $D$-dimensional ball to the margin plane
is $\kappa$, the distance of the center is $\kappa+R\sqrt{D}$ (see
Fig. \ref{fig:Disks}). When $R$ is larger, the capacity vanishes
as $\alpha_{B}(0,R,D)\approx\left(1+R^{-2}\right)/D$. When $D$ is
large, making $\mathbf{w}$ orthogonal to a significant fraction of
high dimensional manifolds incurs a prohibitive loss in the effective
dimensionality. Hence, in this limit, the fraction of manifolds that
lie in the margin plane is zero. Interestingly, when $R$ is sufficiently
large, $R\propto\sqrt{D}$, it becomes advantageous for $\mathbf{w}$
to be orthogonal to a finite fraction of the manifolds.

\section{$L_{p}$ Balls}

To study the effect of changing the geometrical shape of the manifolds,
we replace the Euclidean norm constraint on the manifold boundary
by a constraint on their $L_{p}$ norm. Specifically, we consider
$D$-dimensional manifolds $\left\{ \mathbf{x}_{0}^{\mu}+R\sum_{i=1}^{D}s_{i}\mathbf{u}_{i}^{\mu}\right\} $
where the $D$ dimensional vector $\vec{s}$ parameterizing points
on the manifolds is constrained: $\left\Vert \vec{s}\right\Vert _{p}\leq1$.
For $1<p<\infty$, these $L_{p}$ manifolds are smooth and convex.
Their linear classification by a vector $\mathbf{w}$ is determined
by the field constraints $h_{0}^{\mu}-R||\vec{h}{}^{\mu}||_{q}\geq\kappa$
where, as before, $h_{0}^{\mu}$ are the fields induced by the centers,
and $||\vec{h}^{\mu}||_{q}$, $q=p/(p-1)$, are the $L_{q}$ dual
norms of the $D$-dimensional fields induced by $\mathbf{u}_{i}^{\mu}$
(Appendix \ref{fig:AppendixLp}). The resultant solutions are qualitatively
similar to what we observed with $L_{2}$ ball manifolds.

However, when $p\le1$, the convex hull of the manifold becomes faceted,
consisting of vertices, flat edges and faces. For these geometries,
the constraints on the fields associated with a solution vector $\mathbf{w}$
becomes: $h_{0}^{\mu}-R\max_{i}\left\Vert h_{i}^{\mu}\right\Vert \ge\kappa$
for all $p<1$ . We have solved in detail the case of $D=2$ (Appendix
\ref{subsec:AppendixReplicaLp}). There are four manifold classes:
interior; touching the margin plane at a single vertex point; a flat
side embedded in the margin; and fully embedded. The fractions of
these classes are shown in Fig. \ref{fig:L1manifolds}.

\begin{figure}
\noindent \begin{centering}
\includegraphics[width=0.8\textwidth]{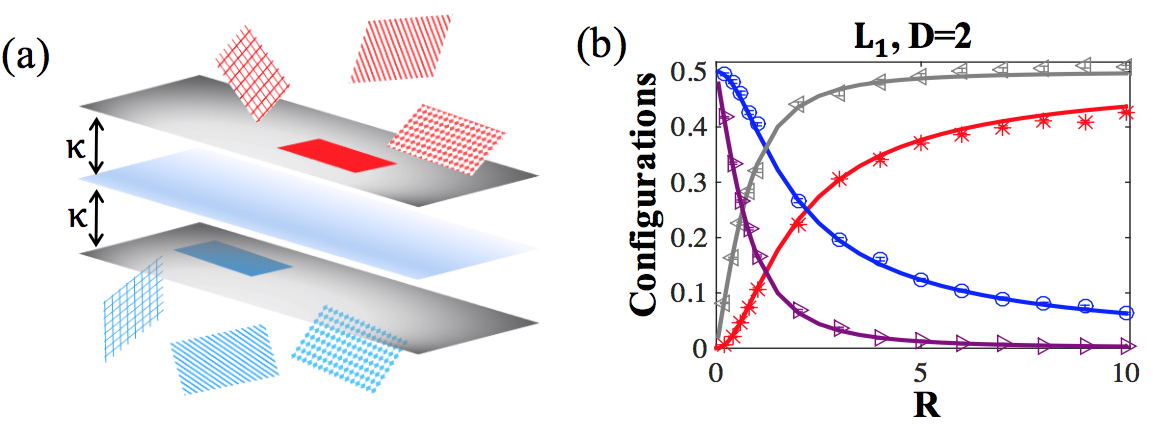} 
\par\end{centering}
\caption{$L_{1}$ balls: (a) Linear classification of 2-$D$ $L_{1}$ balls.
(b) Fraction of manifold configurations as a function of radius $R$
at capacity with $\kappa=0$ comparing theory (lines) to simulations
(markers). (red) entire manifold embedded, (blue) manifold touching
margin at a single vertex, (gray) manifold touching with two corners
(one side), (purple) interior manifold.\label{fig:L1manifolds}}
\end{figure}

\paragraph*{Discussion:}

We have extended Gardner's theory of the linear classification of
isolated points to the classification of continuous manifolds. Our
analysis shows how linear separability of the manifolds depends intimately
upon the dimensionality, size and shape of the convex hulls of the
manifolds. Some or all of these properties are expected to differ
at different stages in the sensory hierarchy. Thus, our theory enables
systematic analysis of the degree to which this reformatting enhances
the capacity for object classification at the higher stages of the
hierarchy.

We focused here on the classification of fully observed manifolds
and have not addressed the problem of generalization from finite input
sampling of the manifolds. Nevertheless, our results about the properties
of maximum margin solutions can be readily utilized to estimate generalization
from finite samples. The current theory can be extended in several
important ways. Additional geometric features can be incorporated,
such as non-uniform radii for the manifolds as well as heterogeneous
mixtures of manifolds. The influence of correlations in the structure
of the manifolds as well as the effect of sparse labels can also be
considered. The present work lays the groundwork for a computational
theory of neuronal processing of objects, providing quantitative measures
for assessing the properties of representations in biological and
artificial neural networks. 

\section{Appendix \label{sec:AppendixSpheres}}

\subsection{Perceptron Capacity of Line Segments. \label{subsec:AppendixReplicaLines}}

The simplest example of linear separability of manifolds is when the
manifolds are line segments. Specifically,

we consider the problem of classification of $P$ line segments of
length $2R$, given by

\begin{equation}
\left\{ \mathbf{x}_{0}^{\mu}+Rs\mathbf{u}^{\mu}\right\} ,\:|s|\leq1,\:\mu=1,...P
\end{equation}

the $N-$dimensional vectors $\mathbf{x_{0}}^{\mu}$ and $\mathbf{u}^{\mu}$,
which are, respectively, the centers and the directions of the $\mu$
segment. {[}We use the boldface style to denote $N$-dim vectors{]}.
We consider random line segments, specifically assume that the components
of all $\mathbf{x_{0}}^{\mu}$ and $\mathbf{u}^{\mu}$ are i.i.d.
normally distributed (with zero mean and unit variance). The target
classification labels of the manifolds are $y^{\mu}=\pm1$ and are
drawn randomly with equal probability of $\pm1.$

We search for an $N$-dimenional weight vector $\mathbf{w}$ that
classifies correctly the line segments. Since the line segments are
convex this is equivalent to the requirement that $\mathbf{w}$ classifies
correctly the end points of each segments, This condition can be written
using two local fields for each segment. One is the field induced
by the center of the line $x^{\mu},$ giving

\begin{equation}
h_{0}^{\mu}=||\mathbf{w}||^{-1}y^{\mu}\mathbf{w}\cdot\mathbf{x_{0}}^{\mu}
\end{equation}

The other is the field induced by the direction vector $\mathbf{u}^{\mu},$

\begin{equation}
h^{\mu}=||\mathbf{w}||^{-1}y^{\mu}\mathbf{w}\cdot\mathbf{u}^{\mu}
\end{equation}

Note that all the fields are defined with the target label $y^{\mu}$,
and they are normalized by the norm of $\mathbf{w}$. With these definitions,
$h_{0}^{\mu}\pm Rh^{\mu}$ are the signed distance of the endpoints
of the segment from the separating plane which is the plane orthogonal
to $\mathbf{w}$ . Thus, $\mathbf{w}$ has to obey,

\begin{equation}
h_{0}^{\mu}-R\left\Vert h_{i}^{\mu}\right\Vert \geq\kappa\label{eq:lines}
\end{equation}

where $h_{0}^{\mu}-R\left\Vert h_{i}^{\mu}\right\Vert $ is the field
of the endpoint with the smallest (signed) distance to the plane.
The parameter $\kappa>0$ is a parameter defining two margin planes.
According to Eq. (\ref{eq:lines}) all the positively labeled inputs
must lie either above the 'positive' margin plane. Conversely the
negatively labeled points must lie below the negative margin plane
(See Fig. 1).

\subsubsection{Replica Theory}

We consider a thermodynamic limit where $N,\,P\rightarrow\infty$
whereas $\alpha=P/N$$,$ and $R$ are finite. We use the Gardner
framework to compute the volume of space of solutions. 
\begin{equation}
V=\int d^{N}\mathbf{w_{\alpha}}\delta(\mathbf{w}^{2}-N)\Pi_{\mu=1}^{P}\Theta(h_{0}^{\mu}-R\left\Vert h^{\mu}\right\Vert -\kappa)\label{eq:V}
\end{equation}

where $\Theta$ is the Heaviside function. According to replica theory,
$\langle\ln V\rangle=\lim_{n\rightarrow0}\frac{\langle V^{n}\rangle-1}{n}$,
where $V^{n}$ can be written as,

\begin{eqnarray}
V^{n} & = & \prod_{\alpha=1}^{n}\int d\mathbf{w}_{\alpha}\delta(\mathbf{w}_{\alpha}^{2}-N)\prod_{\mu=1}^{P}\int_{\kappa}^{\infty}dh_{\alpha}^{\mu+}\int_{\kappa}^{\infty}dh_{\alpha}^{\mu-}\int_{-\infty}^{\infty}d\tilde{h}_{\alpha}^{\mu+}\int_{-\infty}^{\infty}d\tilde{h}_{\alpha}^{\mu-}X\label{eq:Vn}
\end{eqnarray}

\begin{equation}
X=e^{\sum_{\pm}i\tilde{h}_{\alpha}^{\mu\pm}\left(h_{\alpha}^{\mu\pm}-\frac{1}{\sqrt{N}}\left\{ y^{\mu}\mathbf{w}_{\alpha}^{T}(\mathbf{x_{0}}^{\mu}\pm R\mathbf{u}^{\mu})\right\} \right)}
\end{equation}

where $h^{\mu+}=h_{0}^{\mu}+Rh^{\mu}$, $h^{\mu-}=h_{0}^{\mu}-Rh^{\mu}$.
Averaging over the random inputs $\mathbf{x}^{\mu}$ and $\mathbf{u}^{\mu}$
, the above fields can be written as sums of two random fields, where
$t_{0}$ and $t$ are the quenched component resulting from the quenched
random variables, namely the input vectors $\mathbf{x}_{0}^{\mu}$and
$\mathbf{u}_{i}^{\mu}$, while the $z_{0}$and $z$ fields represent
the variability of different $\mathbf{w}$'s within the volume of
solutions for each realization of inputs and labels,

\begin{equation}
h_{0}^{\mu}=\sqrt{q}t_{0}^{\mu}+\sqrt{1-q}z_{0}^{\mu},\;h^{\mu}=\sqrt{q}t^{\mu}+\sqrt{1-q}z^{\mu}
\end{equation}

where the replica symmetric order parameter $q$ is $q=\frac{1}{N}\mathbf{w}_{\alpha}\cdot\mathbf{w}_{\beta},\,\alpha\neq\beta$
. The resultant 'free energy' $G$ is:

\begin{equation}
\langle V^{n}\rangle_{t_{0},t}=e^{Nn\left[G(q)\right]}=e^{Nn\left[G_{0}(q)+\alpha G_{1}(q)\right]}\label{eq:G}
\end{equation}

where,

\begin{equation}
G_{0}(q)=\frac{1}{2}\ln(1-q)+\frac{q}{2(1-q)}\label{eq:G0}
\end{equation}

is the entropic term representing the volume of $\mathbf{w}_{\alpha}$
subject to the constraint that $q=\frac{1}{N}\mathbf{w}_{\alpha}\cdot\mathbf{w}_{\beta}$
. The classification constraints contributes

\begin{equation}
G_{1}(q)=\langle\mbox{ln}Z(q,t_{0,}t)\rangle_{t_{0},t}\label{eq:G1}
\end{equation}
\begin{equation}
Z(q,t_{0},t)=\int_{-\infty}^{\infty}Dz_{0}\int_{-\infty}^{\infty}Dz\Theta\left[\left(\sqrt{q}t_{0}+\sqrt{1-q}z_{0}\right)-R\left|\sqrt{q}t+\sqrt{1-q}z\right|-\kappa\right]\label{eq:Z}
\end{equation}

where $Dx\equiv\frac{dx}{\sqrt{2\pi}}\exp-\frac{x^{2}}{2}$ and the
average wrt $t_{0}$, $t$ denotes integrals over the gaussian variables
$t_{0}$, $t$ with measures $Dt_{0}$ and $Dt$, respectively. Finally,
$q$ is determined by solving $\frac{\partial G}{\partial q}=0$ .
Solution with $q<1$ indicates a finite volume of solutions. For each
$\kappa$ there is a maximum value of $\alpha$ where a solution exists.
As $\alpha$ approaches this maximal value, $q\rightarrow1$ indicating
the existence of a unique solution, which is the max margin solution
for this $\alpha$.

In this chapter we focus on the properties of the \textit{max margin}
solution, i.e., on the limit $q\rightarrow1.$

\paragraph*{$q\rightarrow1$ Limit }

We define

\begin{equation}
Q=\frac{q}{1-q}
\end{equation}
and study the limit of $Q\rightarrow\infty$. In this limit the leading
order is $G_{0}=\frac{Q}{2}$.

\begin{equation}
\langle\ln V\rangle=\frac{Q}{2}\left[1-\alpha\langle g(t_{0},t)\rangle_{t_{0},t}\right]\label{eq:logVLargeQ}
\end{equation}

where, $g\equiv-\frac{2}{Q}\log Z$ is independent of $Q$ and is
given by replacing the integrals in Eq. (\ref{eq:Z}) by their saddle
point, yielding

\begin{equation}
g(t_{0},t)=\min_{z_{0}+t_{0}-R|z+t|\geq\kappa}[z_{0}{}^{2}+z^{2}]\label{eq:gLines}
\end{equation}

Note that here we have scaled variables $z_{0}$ and $z$ such that
$z_{0}\rightarrow\sqrt{Q}z_{0}$ and similarly for $z$.

Finally, at the capacity, $\ln V$ vanishes, hence

\begin{equation}
\alpha_{1}^{-1}(\kappa,R)=\langle g(t_{0},t)\rangle_{t_{0},t}\label{eq:alphaLines}
\end{equation}

where we have denoted the capacity for one dimensional manifolds as
$\alpha_{1}$.

\subsubsection{Capacity }

The nature of solution of Eq. (\ref{eq:gLines}) depends on the values
of $t$ and $t_{0}$. There are three regimes.

\paragraph*{a) Regime 1: }

\begin{equation}
t_{0}-\kappa>R\left\Vert t\right\Vert \label{eq:regime1Lines}
\end{equation}

in which case the solution is $z_{0}=z=0$ which does not contribute
to Eq. (\ref{eq:alphaLines}). 

For values of $t_{0}-\kappa\leq R||t||$, the solution obeys $z_{0}+t_{0}-R||z+t||=\kappa$
, meaning that one of the endpoints touches the margin plane. This
regime is further divided into two cases.

\paragraph*{b) Regime 2: 
\begin{equation}
-R^{-1}||t||<t_{0}-\kappa<R||t||\label{eq:Regime2Lines}
\end{equation}
}

Here, the center field $z_{0}+t_{0}$ is larger than the margin (i.e.,
the center points are interior) and the fields can be determined by
minimizing Eq. (\ref{eq:gLines}) $(R|z+t|-t_{0}+\kappa)^{2}+z^{2}$
w.r.t. $z$ yielding

\begin{equation}
z=\frac{R^{2}||t||+R(\kappa-t_{0})}{1+R^{2}}
\end{equation}

\begin{equation}
z_{0}=\frac{R||t||+\kappa-t_{0}}{1+R^{2}}
\end{equation}

and its contribution to Eq. (\ref{eq:alphaLines}) is

\begin{equation}
g=\frac{(\kappa-t_{0}+R||t||)^{2}}{1+R^{2}}
\end{equation}

\paragraph*{c) Regime 3: 
\begin{equation}
t_{0}-\kappa<-R^{-1}||t||
\end{equation}
}

Here the center points are also on the margin plane, hence $z=-||t||$
and $z_{0}+t_{0}=\kappa$, contributing

\begin{equation}
g=(t_{0}-\kappa)^{2}+t^{2}
\end{equation}

Finally, combining the contributions from Regimes 2 and 3 yields,

\textsl{ 
\begin{equation}
\alpha_{L}^{-1}(\kappa,R)=\int_{-\infty}^{\infty}Dt\left[\int_{\kappa-|t|R^{-1}}^{\kappa+R|t|}Dt_{0}\frac{(|t|R-(t_{0}-\kappa))^{2}}{(1+R^{2})}+\int_{-\infty}^{\kappa-|t|R^{-1}}Dt_{0}((t_{0}-\kappa)^{2}+t^{2})\right]\label{eq:alphaCLine-1}
\end{equation}
}

For $\kappa=0$, this expression reduces to\textsl{ 
\begin{equation}
\alpha_{L}^{-1}(0,R)=\int_{-\infty}^{\infty}Dt\left[\int_{-|t|R^{-1}}^{R|t|}Dt_{0}\frac{(|t|R-t_{0})^{2}}{(1+R^{2})}+\int_{-\infty}^{-|t|R^{-1}}Dt_{0}(t_{0}^{2}+t^{2})\right]\label{eq:alphaCLinek0}
\end{equation}
}

By switching to polar coordinates: $t=r\cos\phi$, $t_{0}=r\sin\phi,$
these integrals reduce to

\begin{equation}
\alpha_{L}^{-1}(R)=\frac{1}{2}+\frac{2}{\pi}\arctan R\label{eq:tanR}
\end{equation}

\subsubsection{Limits of R }

In the limit of $R\rightarrow0,$ Eq. (\ref{eq:alphaCLine-1}) reduces
to $\alpha_{L}(\kappa,R=0)=\alpha_{0}(\kappa)$ where $\alpha_{0}(\kappa)$
is the Gardner's result for classifying $P$ random points. 

Interestingly, $\alpha_{L}(\kappa,1)=\frac{1}{2}\alpha_{0}(\kappa/\sqrt{2})$
. This is because when $R=1$ the distance between edge points on
the line segments is statistically the same as that between points
of different segments, hence the problem is equivalent to classifying
randomly $2P$ points with norms $\sqrt{2N}$. 

Finally, when $R\rightarrow\infty,$ the capacity becomes 

\textsl{ 
\begin{equation}
\alpha_{L}^{-1}(\kappa,R=\infty)=\int_{-\infty}^{\infty}Dt\left[\int_{\kappa}^{\infty}Dt_{0}t^{2}+\int_{-\infty}^{\kappa}Dt_{0}((t_{0}-\kappa)^{2}+t^{2})\right]\label{eq:InftyRLine1}
\end{equation}
}

\textsl{ 
\begin{equation}
=1+\int_{-\infty}^{\kappa}Dt_{0}((t_{0}-\kappa)^{2}=1+\alpha_{0}^{-1}(\kappa)\label{eq:InftyRline}
\end{equation}
}

The reason for this is that when $R$ is large, the manifolds are
essentially unbounded lines. The only way to classify them correctly
is for $\mathbf{w}$ to be orthogonal to all $P$ lines, reducing
the problem to classifying $P$ points which are the projections of
the centers on the null space of the lines. Thus, this is equivalent
to classifying random points in a space with dimensionality $N-P=N(1-\alpha)$
from which Eq. (\ref{eq:InftyRLine1}) follows. These limits can be
readily seen in the simple case of $\kappa=0$. It is readily seen
from Eq. (\ref{eq:tanR}) that $\alpha=2$, $1,$ and $2/3$ for $R=0,$1,
and $\infty$ respectively.

\subsubsection{Distribution of Fields }

It is instructive to calculate the distribution of fields $P(h_{0},h)$
induced by the manifolds with the max margin solution $\mathbf{w}$.
Using the above theory, we find that

\begin{eqnarray}
P(h_{0},h) & = & \langle\frac{1}{Z}\int_{-\infty}^{\infty}Dz_{0}\int_{-\infty}^{\infty}Dz\Theta\left[\left(\sqrt{q}t_{0}+\sqrt{1-q}z_{0}\right)-R\left|\sqrt{q}t+\sqrt{1-q}z\right|-\kappa\right]\\
 &  & \delta(h_{0}-\sqrt{q}t_{0}-\sqrt{1-q}z_{0})\delta(h-\sqrt{q}t-\sqrt{1-q}z)\rangle_{t,t_{0}}
\end{eqnarray}

Considering the three above regimes for $(t,t_{0})$, we obtain the
dominant contribution in the limit of $Q\rightarrow\infty$,

\begin{equation}
P(h_{0},h)=A(h_{0},h)\Theta(h_{0}-R||h||-\kappa)+B(h_{0})\delta(||h||-R^{-1}(h_{0}-\kappa))+C\delta(h_{0}-\kappa)\delta(h)
\end{equation}

\begin{equation}
A(h_{0},h)=\frac{\exp(-\frac{1}{2}(h_{0}^{2}+h^{2})}{2\pi},\,h_{0}-R|h|-\kappa\geq0
\end{equation}

\begin{equation}
B(h_{0})=2\sqrt{\frac{1+R^{-2}}{2\pi}}H(-R^{-1}\kappa')\exp\left[-\frac{(1+R^{-2})(h_{0}-R^{-2}\kappa')^{2}}{2}\right],\,h_{0}\geq0
\end{equation}

\begin{equation}
C=\int Dz\int_{-\infty}^{\kappa-R^{-1}||z||}Dt
\end{equation}

where $H(x)=\int_{x}^{\infty}Dz$ , and

\begin{equation}
\kappa'=\frac{\kappa}{1+R^{-2}}
\end{equation}

The integrated weights are:

\begin{equation}
\int dh_{0}dhA(h_{0},h)=2\int_{0}^{\infty}DtH(\kappa+Rt)
\end{equation}

\begin{equation}
\int dh_{0}B(h_{0})=\int_{-\infty}^{\infty}Dt\int_{\kappa-|t|R^{-1}}^{\kappa+R|t|}Dt_{0}=2\int_{0}^{\infty}Dt\left[H(\kappa-tR^{-1})-H(\kappa+Rt)\right]
\end{equation}

\begin{equation}
C=\int_{-\infty}^{\infty}Dt\int_{-\infty}^{\kappa-|t|R^{-1}}Dt_{0}=1-2\int_{0}^{\infty}DtH(\kappa-tR^{-1})
\end{equation}

The first term represents the fraction of line segments that are interior
to the margin plane (corresponding to Regime 1); the second component
corresponds to segments that touch the margin planes but do not lie
on the margin plane (Regime 2); the third term corresponds to the
segments that lie completely on the margin planes (see Fig. 1 in main
text). When $R\rightarrow\infty,$ we obtain,

\begin{equation}
\int dh_{0}B(h_{0})=H(\kappa)
\end{equation}

\begin{equation}
C=1-H(\kappa)
\end{equation}

The reason for this is as follows. when $R\rightarrow\infty$, $\mathbf{w}$
becomes increasingly orthogonal to all the directors, hence the fraction
of interior points vanish. The value of $B$ represents the fraction
of segments that touch the margin planes. The fields associated with
the centers is finite, larger than $\kappa$. However, the angle between
the segments and \textbf{$\mathbf{w}$} vanish, since the angle is
roughly $||h||$ which is $R^{-1}(h_{0}-\kappa)$ . In contrast, the
fields of the segments represented by $C$ equal $\kappa$, hence
they lie in the margin planes. Thus, in this limit, the fields are
the same as the separation of the centers in the null space (of dimension
$N-P)$.

\subsection{Perceptron Capacity of $D$-dimensional Balls \label{subsec:AppendixReplicaBalls}}

We now consider linear classification of higher dimensional manifolds,
modeling them as $D$ dimensional balls with radius $R$, 
\begin{equation}
\mathbf{x_{0}}^{\mu}+R\sum_{i=1}^{D}s_{i}\mathbf{u}_{i}^{\mu},\:\forall s,\:||\vec{s}||\le1\label{eq:manifoldDef}
\end{equation}
{[}We use $\vec{}\:$ sign to denote $D$ - dimensional vectors and
$||...||$ for $L_{2}$ norm {]}. For each manifold, the center $\mathbf{x}_{0}^{\mu}$,
and the $D$ basis vectors $\{\mathbf{u}_{i}^{\mu}\}$ are $N$ dimensional
vectors ($i=1,..,D$), the components of which are all independent
Gaussian random variables with zero mean and unit variance. The target
labels of the manifolds are random assignments of $y^{\mu}=\pm1$.
To classify all the points on the manifolds correctly (with a given
margin), the weight vector $\mathbf{w}$ (normalized for convenience
by $||\mathbf{w}||=\sqrt{N}),$ must satisfy 
\begin{equation}
h_{0}^{\mu}+R\min_{\vec{s},\,||\vec{s}||^{2}=1}\sum_{i=1}^{D}s_{i}h_{i}^{\mu}\geq\kappa\label{eq:minS}
\end{equation}

where $h_{0}^{\mu}=N^{-1/2}y^{\mu}\mathbf{w}\cdot\mathbf{x^{\mu}}$
is the field induced by the manifold centers and $h_{i}^{\mu}=N^{-1/2}y^{\mu}\mathbf{w}\cdot\mathbf{u}_{i}^{\mu}\:i=1,...,D$
are $D$ fields induced by each of the basis vectors. Differentiating
$\sum_{i=1}^{D}s_{i}h_{i}^{\mu}+\lambda\sum_{i}s_{i}^{2}$ (where
$\lambda$ is a Lagrange multiplier enforcing the norm constraint)
wrt $s_{i}$, we obtain,

\begin{equation}
s_{i}=-\frac{h_{i}^{\mu}}{||\vec{h}^{\mu}||}
\end{equation}

where $||\vec{h}^{\mu}||$ is the $L_{2}$ norm of the $D$-dimensional
vector $h_{i}^{\mu}$, hence $\sum_{i}s_{i}h_{i}^{\mu}=-||\vec{h}^{\mu}||$
and the constraints can be written as

\begin{equation}
h_{0}^{\mu}-R||\vec{h}^{\mu}||\geq\kappa\label{eq:L2 Inequality}
\end{equation}

Geometrically, the LHS corresponds to the field induced by the point
on the manifold $\mu$ which has the smallest (signed) projection
on $\mathbf{w}$. We consider a thermodynamic limit where $N,\,P\rightarrow\infty$
while $\alpha=P/N$, $D,$ and $R$ are finite.

\subsubsection{Capacity}

The replica theory as outlined above, yields

\begin{equation}
\langle V^{n}\rangle_{t_{0},t}=e^{Nn\left[G(q)\right]}=e^{Nn\left[G_{0}(q)+\alpha G_{1}(q)\right]}\label{eq:G-1}
\end{equation}

where as before,

\begin{equation}
G_{0}(q)=\frac{1}{2}\ln(1-q)+\frac{q}{2(1-q)}\label{eq:G0-1}
\end{equation}

and

\begin{equation}
G_{1}(q)=\langle\ln Z(q,t_{0,}\vec{t})\rangle_{t_{0},\vec{t}}\label{eq:G1-1}
\end{equation}
\begin{equation}
Z(q,t_{0},\mathbf{t})=\int_{-\infty}^{\infty}Dz_{0}\int_{-\infty}^{\infty}D\vec{z}\Theta\left[\left(\sqrt{q}t_{0}+\sqrt{1-q}z_{0}\right)-R||\sqrt{q}\vec{t}+\sqrt{1-q}\vec{z}||-\kappa\right]\label{eq:Z-1}
\end{equation}

where

\begin{equation}
h_{0}^{\mu}=\sqrt{q}t_{0}^{\mu}+\sqrt{1-q}z_{0}^{\mu},\;h_{i}^{\mu}=\sqrt{q}t_{i}^{\mu}+\sqrt{1-q}z_{i}^{\mu},\,i=1,...,D
\end{equation}
and $||...||$ is the $L_{2}$ norm of the $D$-dimensional vectors.
All variables $z_{0},t_{0},\vec{z},\vec{t}$ are normally distributed.

\[
\ln Z(q,t_{0},\vec{t})=\ln\int_{-\infty}^{\infty}Dz_{0}\int_{-\infty}^{\infty}D\vec{z}\,\Theta\left(\sqrt{Q}t_{0}+z_{0}-R||\sqrt{q}\vec{t}+\sqrt{1-q}\vec{z}||-\kappa\right)
\]

where the saddle point behavior in the limit of $q\rightarrow1,Q\rightarrow\infty$
gives $g\equiv-\frac{2}{Q}\log Z$,

\begin{equation}
g(t_{0,}t)=\min_{t_{0}+z_{0}-R||\vec{t}+\vec{z}||>\kappa}\frac{1}{2}\left[z_{0}^{2}+\left\Vert \vec{z}\right\Vert ^{2}\right]\label{eq:logZSpeheres}
\end{equation}

and the capacity is given by

\begin{equation}
\alpha_{B}^{-1}(\kappa,R,D)=\langle g(t_{0},\vec{t})\rangle_{t_{0},\mathbf{\vec{t}}}\label{eq:alphaSpheres}
\end{equation}

Again, there are three regimes.

\paragraph*{a) Regime 1: }

Defining $t=||\vec{t}||$, when $t_{0}-\kappa>Rt$: then $z_{0}\approx0$,
$\vec{z}\approx0$ $g\approx0$ corresponding to manifolds which obey
the inequality (not equality) of Eq. (\ref{eq:minS}), hence are interior
to the plane.

\paragraph*{b) Regime 2:}

When $-R^{-1}t<t_{0}-\kappa<Rt$: then $z_{0}=R||\vec{t}+\vec{z}||-t_{0}+\kappa$
and

\begin{equation}
\vec{z}=-z\vec{t}/t
\end{equation}

the scalar $z$ can be calculated by 
\begin{equation}
g\approx\min_{z}\frac{1}{2}\left[\left(R(t-z)-t_{0}+\kappa\right)^{2}+z^{2}\right]\label{eq:logZSpheres}
\end{equation}

\begin{equation}
z=\frac{R^{2}t-R(\kappa-t_{0})}{1+R^{2}}
\end{equation}
\begin{equation}
z_{0}=\frac{Rt+\kappa-t_{0}}{1+R^{2}}
\end{equation}

\begin{equation}
g=\frac{(\kappa-t_{0}-Rt)^{2}}{1+R^{2}}
\end{equation}

\paragraph*{c) Regime 3:}

When $t_{0}<-\frac{1}{R}t$: then $z=t$ and $z_{0}=\kappa-t_{0}$
so that $g\approx(t_{0}-\kappa)^{2}+t^{2}$ .

Combining these contributions, the capacity is: 
\begin{equation}
\alpha_{B}^{-1}(\kappa,R,D)=\int_{0}^{\infty}dt\chi{}_{D}(t)\left[\int_{\kappa-\frac{1}{R}t}^{\kappa+Rt}Dt_{0}\frac{(Rt-t_{0}+\kappa)^{2}}{R^{2}+1}+\int_{-\infty}^{\kappa-\frac{1}{R}t}Dt_{0}([t_{0}-\kappa]^{2}+t^{2})\right]\label{eq:alphaCSpheres}
\end{equation}

where $\chi_{D}$ is the $D$-dim \emph{Chi distribution, }

\begin{equation}
\chi_{D}(t)=\int D\vec{t}\delta\left(t-||\vec{t}||\right)=\frac{2^{1-\frac{D}{2}}t^{D-1}e^{-\frac{1}{2}t^{2}}}{\Gamma(\frac{D}{2})}\label{eq:PD}
\end{equation}

\subsubsection{Distribution of Fields}

We consider the joint distribution of two fields: $h_{0}$ which is
the field induced by the manifold centers, and $h\equiv||\vec{h}||$,
namely the $L_{2}$ norm of the $D$ dimensional vector of fields
induced by the $D$ $\mathbf{u}_{i}$'s. Taking into account the above
three regimes, we have,

\begin{equation}
P(h_{0},h)=A(h_{0},h)\Theta(h_{0}-Rh-\kappa)+B(h_{0})\delta(h-R^{-1}(h_{0}-\kappa))+C\delta(h_{0}-\kappa)\delta(h)
\end{equation}

\paragraph*{1. Field Distribution for $\kappa=0$.}

\begin{equation}
A(h_{0},h)=\frac{\exp(-\frac{1}{2}h_{0}^{2})}{\sqrt{2\pi}}\chi_{D}(h),\,h_{0}-R|\vec{h}|\geq0
\end{equation}

\begin{equation}
B(h_{0})=(1+R^{2})^{-1}\int Dt_{0}\int_{h}^{\infty}dt\chi_{D}(t)\,\delta(t_{0}-(1+R^{-2})h_{0}-R^{-1}t)
\end{equation}

\begin{equation}
C=\int Dz\int_{-\infty}^{\kappa-R^{-1}|\vec{z}|}Dt
\end{equation}

\paragraph*{2. Integrated Weights:}

\begin{equation}
\int dh_{0}dhA(h_{0},h)=\int_{0}^{\infty}dtP_{D}(t)H(\kappa+Rt)
\end{equation}

\begin{equation}
\int dh_{0}B(h_{0})=\int_{0}^{\infty}dt\chi_{D}(t)\int_{\kappa-\frac{1}{R}t}^{\kappa+Rt}Dt_{0}=\int_{0}^{\infty}dt\chi_{D}(t)\left[H(\kappa-\frac{t}{R})-H(\kappa+Rt)\right]
\end{equation}

\begin{equation}
C=\int_{0}^{\infty}dt\chi_{D}(t)\int_{-\infty}^{\kappa-\frac{t}{R}}Dt_{0}=1-\int_{0}^{\infty}dt\chi_{D}(t)H(\kappa-\frac{t}{R})
\end{equation}

As in the case of line segments, the first term corresponds to the
fraction of $D$-dim balls that lie in the interior space; the second
component corresponds to the fraction of balls that touch the margin
planes, whereas $C$ stands for the fraction of balls that are fully
embedded in these planes.

\subsubsection{$R=1$ }

\paragraph*{1. Capacity for $\kappa=0$ }

In the case of $R=1$, the capacity obtains a simple form:

\[
\alpha_{B}^{-1}(\kappa=0,R=1,D)=\int_{0}^{\infty}dt\chi{}_{D}(t)\left[\int_{-t}^{+t}Dt_{0}\frac{t^{2}+t_{0}^{2}}{2}-\int_{-t}^{+t}Dt_{0}tt_{0}+\int_{-\infty}^{-t}Dt_{0}(t_{0}{}^{2}+t^{2})\right]
\]

\[
\alpha_{B}^{-1}(\kappa=0,R=1,D)=\int_{0}^{\infty}dt\chi{}_{D}(t)\left[\int_{0}^{\infty}Dt_{0}(t_{0}{}^{2}+t^{2})\right]
\]

\[
\alpha_{B}^{-1}(\kappa=0,R=1,D)=\frac{D+1}{2}\:\square
\]

\paragraph*{2. Manifold Geometry Configurations for $\kappa=0$ }

\paragraph*{a) Interior vs. Embedded:}

The fraction of embedded manifolds:

\[
p_{\mbox{embedded}}=\int_{0}^{\infty}dt\chi{}_{D}(t)\left[\int_{-\infty}^{-t}Dt_{0}\right]
\]

Fraction of interior manifolds:

\[
p_{\mbox{interior}}=\int_{0}^{\infty}dt\chi{}_{D}(t)\int_{t}^{\infty}Dt_{0}=p_{\mbox{embedded}}
\]

The fraction of touching manifolds:

\[
p_{\mbox{touching}}=\int_{0}^{\infty}dt\chi{}_{D}(t)\int_{-t}^{t}Dt_{0}
\]

Thus, the fraction of interior manifolds and embedded manifolds are
equal. $\square$.

\paragraph*{b) Touching Manifolds: }

In general,

\[
p_{\mbox{touching}}=\int_{0}^{\infty}dt\chi{}_{D}(t)\left[\int_{-\frac{1}{R}t}^{+Rt}Dt_{0}\right]
\]

\[
=\int_{0}^{\infty}dt\chi{}_{D}(t)\left[1-H(Rt)-H(t/R)\right]
\]

The radius $R$ at which $p_{\mbox{touching}}$is at maximum can be
found by

\[
\frac{\partial}{\partial R}\left(p_{\mbox{touching}}\right)=\int_{0}^{\infty}dt\chi{}_{D}(t)\left[-tH'(Rt)+tR^{-2}H'(t/R)\right]=0
\]

The solution for above is $R=1$ for all $D$. For $D=2$ , $p_{\mbox{touching}}(R=1,D=2)=\int_{0}^{\infty}dt\chi_{2}(t)[1-2H(t)]\sim0.7$.

Therefore, at $R=1$, the fraction of touching disks is at maximum,
and for $D=2$, the value is about 0.7. $\square.$

\subsubsection{Large $R$ Limit}

In the limit of large $R$, Eq. (\ref{eq:alphaCSpheres}) reduces
to:

\begin{equation}
\alpha_{B}^{-1}(\kappa,R=\infty,D)=\int_{0}^{\infty}dt\chi{}_{D}(t)\left[\int_{\kappa}^{\infty}Dt_{0}t^{2}+\int_{-\infty}^{\kappa}Dt_{0}([t_{0}-\kappa]^{2}+t^{2})\right]\label{eq:alphaCSpheres-1}
\end{equation}

\begin{equation}
=\int_{0}^{\infty}dt\chi{}_{D}(t)t^{2}+\int_{-\infty}^{\kappa}Dt_{0}(t_{0}-\kappa)^{2}=\alpha_{0}^{-1}(\kappa)+D\label{eq:alphaCSpheres-1-1}
\end{equation}

which reflects the fact that when $R$ is large $\mathbf{w}$ must
be in the null space of the $PD$ vectors $\mathbf{u}_{i}^{\mu}$;
thus, the classification problem is that of $P$ points (i.e., the
projections of the centers onto the null space) in $N-PD$ dimensions.
Likewise, in this limit $A$ vanishes and the angle between the manifold
centers and the margin planes vanish.

\subsubsection{Limit of Large $D$ }

In many realistic problems it is expected that the dimension of the
object manifolds is large, hence it is of interest to examine the
results in the limit of $D\gg1$. In his limit, $\chi_{D}(t)$ is
centered around $t=\sqrt{D}$, yielding

\begin{equation}
\alpha_{B,D\gg1}^{-1}=\int_{\kappa-\frac{\sqrt{D}}{R}}^{\kappa+R\sqrt{D}}Dt_{0}\frac{(R\sqrt{D}-t_{0}+\kappa)^{2}}{R^{2}+1}+\int_{-\infty}^{\kappa-\frac{\sqrt{D}}{R}}Dt_{0}([t_{0}-\kappa]^{2}+D)\label{eq:alphacSpheresLargeD}
\end{equation}
As long as $R\ll\sqrt{D}$ , the second term in Eq. (\ref{eq:alphacSpheresLargeD})
vanishes and yields

\begin{equation}
\alpha_{B,D\gg1}^{-1}=\int_{-\infty}^{\kappa+R\sqrt{D}}Dt_{0}\frac{(R\sqrt{D}-t_{0}+\kappa)^{2}}{R^{2}+1}=\frac{\alpha_{0}^{-1}(\kappa+R\sqrt{D})}{1+R^{2}}\label{eq:scalingAlpha}
\end{equation}

Thus, $\alpha$remains finite in the limit of large $D$ only if $R$
is not larger than the order of $D^{-1/2}$. If, on the other hand,
$R\sqrt{D}\gg1$, Eq. (\ref{eq:scalingAlpha}) implies

\begin{equation}
\alpha_{B,D\gg1}^{-1}=\frac{R^{2}D}{1+R^{2}}\label{eq:alphaD}
\end{equation}
(where we have used the asymptote $\alpha_{0}^{-1}(x)\rightarrow x$
for large $x$).

Numerically, this approximation works very well for $R\geq0.5$ and
all $D$ (as long as $R\ll\sqrt{D}$).

\subparagraph{Field Distribution in Large $D$ : }

In the limit of large $D$ , the fraction of manifolds that lie on
the margin plane, $C$, is zero. The overall fraction of interior
manifolds is $H(\kappa+R\sqrt{D})$ whereas the fraction of manifolds
that touch the margin planes is $1-H(\kappa+R\sqrt{D})$ .

\subparagraph*{Large $R$ : }

In the limit of $R\propto\sqrt{D}$,

\begin{equation}
\alpha_{B,R\gg1}^{-1}=D\int_{\kappa-\frac{\sqrt{D}}{R}}^{\infty}Dt_{0}+D\int_{-\infty}^{\kappa-\frac{\sqrt{D}}{R}}Dt_{0}=D\label{eq:alphacSpheresLargeD2}
\end{equation}
Note that in this case, both terms in Eq. (\ref{eq:alphacSpheresLargeD})
contribute. This reflects the fact that when $R$ is $O(\sqrt{D})$
it is again advantageous for $\mathbf{w}$ to be orthogonal to some
of the spheres. This is seen in the field distribution. In this limit,
it consists of a fraction of $H(\sqrt{D}/R)$ lying on the plane whereas
the fraction of touching balls is $1-H(\sqrt{D}/R)$. Finally, when
$R/\sqrt{D}$ is large, most of the spheres lie on the margin, as
expected.

\subsection{Perceptron Capacity of $L_{p}$ Manifolds \label{subsec:AppendixReplicaLp}}
\noindent \begin{center}
\begin{figure}
\noindent \begin{centering}
\includegraphics[clip,width=0.9\textwidth]{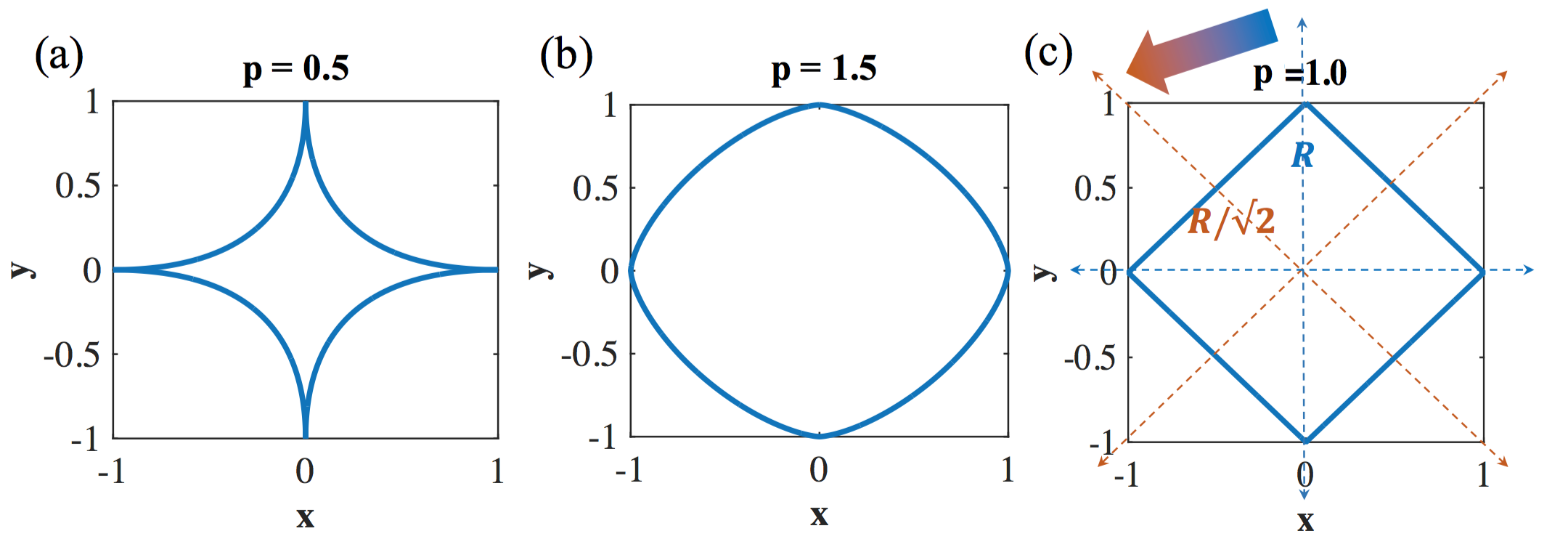}
\par\end{centering}
\caption{\textbf{$L_{p}$ balls.} Illustration of $L_{p}$ balls of norm $p$
for (a) $p$ = 0.5, (b) $p$ =1.5, (c) $p$ = 1. Notice that rotation
of the axis introduces the reduction of effective $R$ by factor of
$1/\sqrt{2}$. \label{fig:AppendixLp}}
\end{figure}
\par\end{center}

We consider manifolds defined with $L_{p}$ norm,

\begin{equation}
\mathbf{x_{0}}^{\mu}+R\sum_{i=1}^{D}s_{i}\mathbf{u}_{i}^{\mu},\:\forall s,\:\left\Vert \vec{s}\right\Vert _{p}\le1\label{eq:manifoldDefLp}
\end{equation}

where $\vec{\left\Vert s\right\Vert }{}_{p}$ is the $L_{p}$norm
of $\vec{s}$ . Linear classification requires

\begin{equation}
h_{0}^{\mu}+R\min_{s,\,\left\Vert \vec{s}\right\Vert _{p}=1}\sum_{i=1}^{D}s_{i}h_{i}^{\mu}\geq\kappa\label{eq:minSLp}
\end{equation}

\paragraph*{$1<p<\infty$:}

Differentiating $\sum_{i}s_{i}h_{i}^{\mu}+\lambda\sum_{i}\left\Vert s_{i}\right\Vert {}^{p}$
wrt $s_{i}$ yields,

\begin{equation}
s_{i}=-\mbox{sign}(h_{i}^{\mu})\frac{||h_{i}^{\mu}||^{1/(p-1)}}{(||h^{\mu}||_{q})^{1/p}}\label{eq:SLp}
\end{equation}

where $q=\frac{p}{p-1}$ is the dual norm of $p$ , hence, $\min_{\vec{s},\,\vec{\left\Vert s\right\Vert }_{p}=1}\sum_{i=1}^{D}s_{i}h_{i}^{\mu}=-\left\Vert \vec{h}^{\mu}\right\Vert {}_{q}$
. Thus, linear classification of $L_{p}$ manifolds is equivalent
to the constraints on the fields,

\begin{equation}
h_{0}^{\mu}-R\left\Vert \vec{h}^{\mu}\right\Vert {}_{q}\geq\kappa\label{eq:Lp Inequality}
\end{equation}

Smoothness of the $L_{q}$ norm guarantees that the solution will
be qualitatively similar to spheres (i.e., $p=q=2$). (See Fig. \ref{fig:AppendixLp}(a)
in the Appendix)

\paragraph*{$0<p\leq1$: }

In this regime differentiating with respect to $s_{i}$ does not minimize
$\sum_{i}s_{i}h_{i}^{\mu}$ . Instead, the minima are at the $D$
extremal points: $s_{i}=1,\,s_{j}=0,\,j\neq i$ corresponding to the
corners of the manifolds (see Fig. \ref{fig:AppendixLp}(b)). Thus,
for all $p\leq1$ the linear classification constraint is the same
and is given by the corner with the smallest projection on $w$ ,
namely

\begin{equation}
h_{0}^{\mu}-R\max_{i}h_{i}^{\mu}\geq\kappa
\end{equation}

We can now use the replica theory, where now the capacity is given
by 

\begin{equation}
\alpha_{B_{p}}^{-1}(\kappa,R,D)=\langle g(t_{0},\vec{t})\rangle_{t_{0},\vec{t}}\label{eq:alphaLp}
\end{equation}

where $B_{p}$ stands for Balls with $L_{p}$ norm, 

\begin{equation}
g(t_{0,}t)=\min_{t_{0}+z_{0}-R\max_{i}(z_{i}+t_{i})>\kappa}\frac{1}{2}\left[z_{0}^{2}+\left\Vert \vec{z}\right\Vert ^{2}\right]\label{eq:gLp}
\end{equation}

where,

\begin{equation}
z_{0}+t_{0}-R\max_{i}(z_{i}+t_{i})\geq\kappa
\end{equation}

\subsubsection{$L_{1}$ in $D=2$ }

\paragraph{Rotated Coordinates}

Without loss of generality, we assume the $t_{i}$ are ordered: $t_{2}\ge t_{1}\ge0$
and similarly for $z_{i}+t_{i}$ .

It is easier to consider the following transformation

\begin{equation}
t_{1}'=\frac{1}{\sqrt{2}}(t_{2}-t_{1})
\end{equation}

\begin{equation}
t_{2}'=\frac{1}{\sqrt{2}}(t_{2}+t_{1})
\end{equation}

\begin{equation}
z_{1}'=\frac{1}{\sqrt{2}}(z_{2}-z_{1})
\end{equation}
\begin{equation}
z_{2}'=\frac{1}{\sqrt{2}}(z_{2}+z_{1})
\end{equation}

\begin{equation}
R'=R/\sqrt{2}
\end{equation}

see the geometry of the rotation in Fig. \ref{fig:AppendixLp} (c).

In these coordinates and convention, $\max_{i}(z_{i}+t_{i})=z_{2}+t_{2}=\frac{1}{\sqrt{2}}\sum_{i}(z'_{i}+t'_{i})$
hence,

\begin{equation}
g(t_{0,}t)=\min_{t_{0}+z_{0}-R\left\Vert \vec{z}+\vec{t}\right\Vert {}_{1}>\kappa}\frac{1}{2}\left[z_{0}^{2}+\left\Vert \vec{z}\right\Vert {}^{2}\right]
\end{equation}

where we have dropped the primes.

General solution:

\textsf{
\begin{equation}
z=-a,\:a>0
\end{equation}
}

The sign is well defined only for $t+z\ne0$ . Hence the general solution
takes the form

\begin{equation}
z_{i}=-\min(a,t_{i}),\:a\geq0
\end{equation}

a) $t_{0}-\kappa>R||\vec{t}||_{1}:$

\begin{equation}
z_{0},\:z=0
\end{equation}

b) $t_{0}-\kappa<R|\vec{t}|_{1}:$

\begin{equation}
z_{0}=\kappa-t_{0}+R|\left\Vert \vec{z}+\vec{t}\right\Vert {}_{1}
\end{equation}

\begin{equation}
a=\frac{R(R\left\Vert \vec{t}\right\Vert {}_{1}-t_{0}+\kappa)}{1+2R^{2}}
\end{equation}

\begin{equation}
g=\frac{(R\left\Vert \vec{t}\right\Vert {}_{1}-t_{0}+\kappa)^{2}}{(1+2R^{2})}
\end{equation}

This is consistent if 
\begin{equation}
0<\frac{R(R\left\Vert \vec{t}\right\Vert {}_{1}-t_{0}+\kappa)}{1+2R^{2}}<t_{1}
\end{equation}
\begin{equation}
t_{0}-\kappa>R\left\Vert \vec{t}\right\Vert {}_{1}-R^{-1}(1+2R^{2})t_{1}=t^{1}
\end{equation}

\begin{equation}
t_{0}-\kappa>t^{1}
\end{equation}

c) $t_{0}<t^{1}$

\begin{equation}
t^{1}=R(t_{1}+t_{2})-R^{-1}(1+2R^{2})t_{1}=R(t_{2}-t_{1})-R^{-1}t_{1}
\end{equation}

Assume

\begin{equation}
z_{2}=-a,\:z_{1}=-t_{1},\:t_{2}>a>t_{1}
\end{equation}

\begin{equation}
a=\frac{R(Rt_{2}-t_{0}+\kappa)}{1+R^{2}}
\end{equation}

\begin{equation}
g=t_{1}^{2}+\frac{(Rt_{2}-t_{0}+\kappa)^{2}}{(1+R^{2})}
\end{equation}

\begin{equation}
t_{0}-\kappa>-R^{-1}t_{2}=t^{2}
\end{equation}

d) $t_{0}<t^{2}$

\begin{equation}
z=-t
\end{equation}
\begin{equation}
g=(t_{0}-\kappa)^{2}+t^{2}
\end{equation}

\paragraph{Capacity }

Finally, converting back the above regimes and values of $g$ to the
original coordinates, we have

\begin{equation}
\alpha_{B_{1}}^{-1}=8\int_{0}^{\infty}Dt_{2}\int_{0}^{t_{2}}Dt_{1}\int_{\kappa+Rt_{1}-R^{-1}(t_{2-}t_{1})}^{\kappa+Rt_{2}}Dt_{0}\frac{(Rt_{2}-t_{0}+\kappa)^{2}}{(1+R^{2})}
\end{equation}

\begin{equation}
+8\int_{0}^{\infty}Dt_{2}\int_{0}^{t_{2}}Dt_{1}\int_{\kappa-R^{-1}(t_{1}+t_{2})}^{\kappa+Rt_{1}-R^{-1}(t_{2-}t_{1})}Dt_{0}\left[\frac{(t_{2}-t_{1})^{2}}{2}+\frac{(\frac{1}{2}R(t_{1}+t_{2})-t_{0}+\kappa)^{2}}{(1+\frac{1}{2}R^{2})}\right]
\end{equation}
\begin{equation}
+8\int_{0}^{\infty}Dt_{2}\int_{0}^{t_{2}}Dt_{1}\int_{-\infty}^{\kappa-R^{-1}(t_{1}+t_{2})}Dt_{0}\left[(t_{0}-\kappa)^{2}+t_{1}^{2}+t_{2}^{2}\right]
\end{equation}

where the subscript for $\alpha_{B_{1}}^{-1}$ is used to denote capacity
for balls with $L_{1}$ norm.

1. $R\rightarrow0$

\begin{equation}
\alpha_{B_{1}}^{-1}=8\int_{0}^{\infty}Dt_{2}\int_{0}^{t_{2}}Dt_{1}\int_{-\infty}^{\kappa}Dt_{0}(-t_{0}+\kappa)^{2}=\alpha_{0}^{-1}(\kappa)
\end{equation}

2. $R\rightarrow\infty$

\begin{equation}
\alpha_{B_{1}}^{-1}=8\int_{0}^{\infty}Dt_{2}\int_{0}^{t_{2}}Dt_{1}\int_{\kappa}^{\infty}Dt_{0}\left[\frac{(t_{2}-t_{1})^{2}+(t_{1}+t_{2})^{2}}{2}\right]
\end{equation}

\begin{equation}
+8\int_{0}^{\infty}Dt_{2}\int_{0}^{t_{2}}Dt_{1}\int_{-\infty}^{\kappa}Dt_{0}\left[[t_{0}-\kappa]^{2}+t_{1}^{2}+t_{2}^{2}\right]
\end{equation}

\begin{equation}
=\int_{-\infty}^{\infty}Dt_{2}\int_{-\infty}^{\infty}Dt_{1}\int_{-\infty}^{\infty}Dt_{0}\left[t_{1}^{2}+t_{2}^{2}\right]+\int_{-\infty}^{\infty}Dt_{2}\int_{-\infty}^{\infty}Dt_{1}\int_{-\infty}^{\kappa}Dt_{0}[t_{0}-\kappa]^{2}
\end{equation}

\begin{equation}
\alpha_{B_{1}}^{-1}=2+\alpha_{0}^{-1}(\kappa)
\end{equation}

as expected in this case of $D=2$. The effective dimensionality is
$N-2P$ .

\paragraph{Fields.}

The integrated weight of manifolds that touch the margin planes is 

\begin{equation}
8\int_{0}^{\infty}Dt_{2}\int_{0}^{t_{2}}Dt_{1}\int_{\kappa+Rt_{1}-R^{-1}(t_{2-}t_{1})}^{\kappa+Rt_{2}}Dt_{0}
\end{equation}

The integrated weight of manifold that have a side on the planes is

\begin{equation}
8\int_{0}^{\infty}Dt_{2}\int_{0}^{t_{2}}Dt_{1}\int_{\kappa-R^{-1}(t_{1}+t_{2})}^{\kappa+Rt_{1}-R^{-1}(t_{2-}t_{1})}Dt_{0}
\end{equation}

The fraction of manifolds that lie on the planes is 
\begin{equation}
8\int_{0}^{\infty}Dt_{2}\int_{0}^{t_{2}}Dt_{1}\int_{-\infty}^{\kappa-R^{-1}(t_{1}+t_{2})}Dt_{0}
\end{equation}

\subsection{Simulation Details }

\subsubsection{Linear Classification of Line Segments}

\subparagraph*{Linear Classification of Line Segments. }

The classification problem of $P$ line segments is cast in the form
of linear classification of the $2P$ endpoints where each pair of
endpoints receive the same target label. These labeled inputs were
classified using IBM cplex package which uses quadratic programming
solving the primal support vector problem. To compute the network
capacity $\alpha=P/N$, 100 trials were used for each $P$ and the
fraction of converged trials was computed. $P$ was gradually increased.
Maximum capacity was defined as the value of $P$ for which the convergence
rate reached 0.5. In Fig. 1(c) (main text) $N=200$ was used. To obtain
the capacity for $\kappa=0.5$, $P$ was varied until SVM's maximum
margins averaged over $100$ runs was close to $\kappa=0.5$.

\subparagraph{Fraction of Line Segment Configurations.}

Once the data was determined to be separable, the fraction of the
different line segment configurations was computed. Each line segment's
configuration was determined based on the number of endpoints on the
main plane. Endpoints were considered to be on the margin plane iff
their field was larger than the margin by an amount smaller than a
tolerance of $\epsilon=10^{-8}$ .

\subsubsection{Linear Classification of D-dimensional $L_{2}$ balls }

\subparagraph*{Sampling of $L_{2}$ balls of Dimension $D$.}

Unlike the case of the line segments, where it is sufficient to consider
the endpoints, finding SVM solution for classifying $D$ dimensional
$L_{2}$ norm balls requires an iterative algorithm to sample the
points on the $L_{2}$ balls so that the decision plane is efficiently
determined. First, we sample randomly a number of points on all manifolds
and find the max margin solution $\mathbf{w}$ and its margin $\kappa$,
for this set of points. Next, for each manifold we find analytically
the point on the boundary which has the minimum (signed) distance
from the decision plane given by \textbf{$\mathbf{w}$}. If the field
of this point lies below the margin this point is added to the training
data and a new $\mathbf{w}$ is computed. This iterative procedure
stops when all the minimal points lie above or on the current margin,
guaranteeing the correct classifications of the entire manifolds.
For details, see Alg.\ref{alg:Procedure-for-sampling}.

\begin{algorithm}
\textbf{Initialize}:

$\quad$ $\mathbf{x}^{\mu},\mathbf{\vec{u}}^{\mu}\sim Norm(0,1)$
( $\mu=1,...,P$) {[}Sample centers and direction vectors{]}

$\quad$ $y^{\mu}\sim\mbox{sign}\left\{ \mbox{unif}(-1,1)\right\} $
( $\mu=1,...,P$) {[}Sample labels for manifolds{]}

$\quad$ $s_{i}^{k}\sim\mbox{unif}\left(-1,1\right)$ and $||\vec{s}^{k}||=1$
$k=1,...,m$. {[}Sample $m$ coefficient vectors{]}

$\quad$ $\mathbf{x}^{\mu}+R\sum_{i=1}^{D}s_{i}^{k}\mbox{\textbf{u}}_{i}^{\mu}$$\in D_{\mbox{data}}$
{[}Construct $m$ points on each manifold{]}

$\quad$ t=0; $w_{svm}^{t}=\mbox{svmsolver}(D_{data},Y)$ {[}Check
separability, find SVM solution{]}

$\quad$ t=0; $h_{min}^{t}=\mbox{argmin}_{\mu,i,k}||\mathbf{w}_{svm}^{t}||^{-1/2}y^{\mu}\mathbf{w}^{T}\left\{ \mathbf{x}^{\mu}+R\sum_{i=1}^{D}s_{i}^{k}\mathbf{u}_{i}^{\mu}\right\} $
{[}Get margin{]}

\textbf{Repeat:} while $t<t_{max}$

$\quad$ $t=t+1$

$\quad$ \textbf{Repeat: }for $\mu=1:P$ {[}For each manifold{]}

$\quad$$\quad$ $s_{i}^{\mbox{min}}=-\frac{h_{i}^{\mu}(w_{svm}^{t})}{||\vec{h}^{\mu}(w_{svm}^{t})||}$
{[}Coefficients of point with a minimum field{]}

$\quad$ $\quad$$\mathbf{d_{min}^{\mu}}=\mathbf{x}^{\mu}+R\sum_{i=1}^{D}s_{i}^{\mbox{min}}\mbox{\textbf{u}}_{i}^{\mu}$
{[}Point with smallest (signed) distance to the current margin plane{]}

$\quad$ $\quad$If $y^{\mu}\frac{\mathbf{w}_{svm}^{t}\cdot\mathbf{d_{min}^{\mu}}}{|\mathbf{w}_{svm}^{t}|}<h_{\mbox{min}}^{t}$
then add $\mathbf{d}_{min}^{\mu}$ to $D_{\mbox{data}}$

$\quad$ \textbf{End}

$\quad$ $w_{svm}^{t}=\mbox{svmsolver}(D_{data},Y)$ {[}Check separability,
find new SVM solution{]}

$\quad$ $h_{min}^{t}=\mbox{argmin}_{\mu,i,k}||\mathbf{w}_{svm}^{t}||^{-1/2}y^{\mu}\mathbf{w}^{T}\left\{ \mathbf{x}^{\mu}+R\sum_{i=1}^{D}s_{i}^{k}\mathbf{u}_{i}^{\mu}\right\} $
{[}Get new margin{]}

\textbf{End}

\textbf{Continue:} until no more points are added

\caption{Pseudocode for linear classification of $L_{2}$ spheres. \label{alg:Procedure-for-sampling}}
\end{algorithm}

Fraction of ball configurations were computed similar to the line
case.

\subparagraph{Simulation Results: }

Fig. 2-(b) (main text): We have used network of size $N=200$, and
$m=20$ initial points on each manifold. Each point (marker) displayed
is an average over 50 trials.

\subparagraph*{Numerical Results for High $D$ $L_{2}$ Balls.}

The test of the network capacity for large dimensional balls, we performed
simulations to evaluate the capacity for balls with $1<D<25$ and
$R=1,5,10$ . Here the capacity was estimated using 20 trials. Good
agreement was achieved with the theory, See Fig. \ref{fig:AppendixCapacityD}.

\begin{figure}
\noindent \begin{centering}
\includegraphics[width=0.7\textwidth]{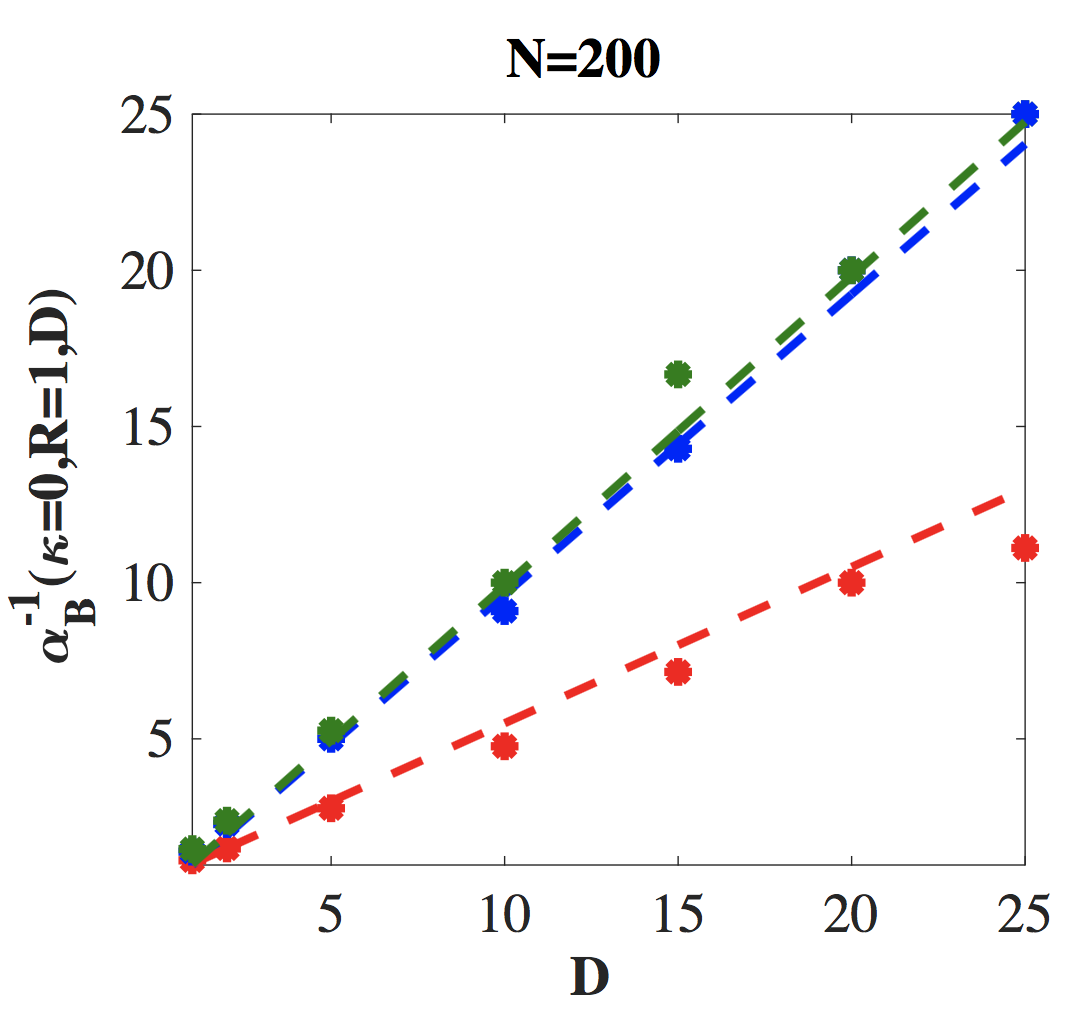} 
\par\end{centering}
\caption{\textbf{Capacity of high dimensional $L_{2}$ balls.} $\alpha=P/N$
at capacity with $\kappa=0$ as a function of \emph{D}. (red) $R=1$
(blue) $R=5$ (green) $R=10$. (markers) Simulation results. (dashed)
full evaluation of $\alpha_{B}(\kappa=0,R,D)$. Note that for $R>5$,
$\alpha_{B}^{-1}(\kappa=0,R\gg1,D)\sim D$ . \label{fig:AppendixCapacityD}}
\end{figure}

\subsubsection{Linear Classification of $L_{1}$ Balls }

\subparagraph*{Sampling of $L_{1}$ Balls.}

Because the sides of $L_{1}$ balls are straight lines, if all the
vertices are on the same side of the plane, all the points in the
interior of $L_{1}$ ball are on the same side of the plane as well.
Therefore linear classifications of the entire $L_{1}$ball is equivalent
to linear classification of all the vertices. In Fig. 3 (main text),
we consider $D=2,$ thus, we simulated SVM solutions of the $4P$
points, $\mathbf{z}^{\mu}=\mathbf{x_{0}}^{\mu}\pm R\mathbf{u}_{1}^{\mu}\pm R\mathbf{u}_{2}^{\mu}$
where each set of $4$ points on the same manifold receive the same
label. In the simulations shown in the figure, network size of $N=200$
was used and the simulation was repeated 100 times to get the convergence
rate (of 0.5 for estimating capacity). The fractions of manifold geometry
configurations were computed similarly to the previous cases. Here,
however, there are $4$ configurations, corresponding to configurations
with $0,1,2,$ or $4$ number of vertices on the margin plane.

%% file: chapters/ch3_m4.tex
\chapter{The Maximum Margin Manifold Machines\label{cha:m4}}

\section{\textcolor{black}{Introduction }}

\textcolor{black}{Handling object variability is a major challenge
for machine learning systems. For example, in visual recognition tasks,
changes in pose, lighting, identity or background can result in large
variability in the appearance of objects \cite{hinton1997modeling}.
Techniques to deal with this variability has been the focus of much
recent work, especially with convolutional neural networks consisting
of many layers. The manifold hypothesis states that natural data variability
can be modeled as lower-dimensional manifolds embedded in higher dimensional
feature representations \cite{bengio2013representation}. A deep neural
network can then be understood as disentangling or flattening the
data manifolds so that they can be more easily read out in the final
layer \cite{brahma2016deep}. Manifold representations of stimuli
have also been utilized in neuroscience, where different brain areas
are believed to untangle and reformat their representations \cite{riesenhuber1999hierarchical,serre2005object,hung2005fast,dicarlo2007untangling,pagan2013signals}.}

\textcolor{black}{This chapter addresses the problem of efficiently
utilizing manifold structures to learn classifiers. The manifold structures
may be known from prior knowledge, or could be estimated from data
using a variety of manifold learning algorithms \cite{tenenbaum1998mapping,roweis2000nonlinear,tenenbaum2000global,belkin2003laplacian,belkin2006manifold,canas2012learning}.
Based upon knowledge of these structures, some areas of prior research
have focused on building invariant representations \cite{anselmi2013unsupervised}
or constructing invariant metrics \cite{simard1994memory}. On the
other hand, most approaches today rely upon data augmentation by explicitly
generating ``virtual'' examples from these manifolds \cite{niyogi1998incorporating,scholkopf1996incorporating}.
Unfortunately, the number of samples needed to successfully learn
the underlying manifolds may increase the original training set by
more than a thousand-fold \cite{krizhevsky2012imagenet}.}
\begin{figure}[h]
\noindent \begin{centering}
\textcolor{black}{\includegraphics[width=0.4\columnwidth]{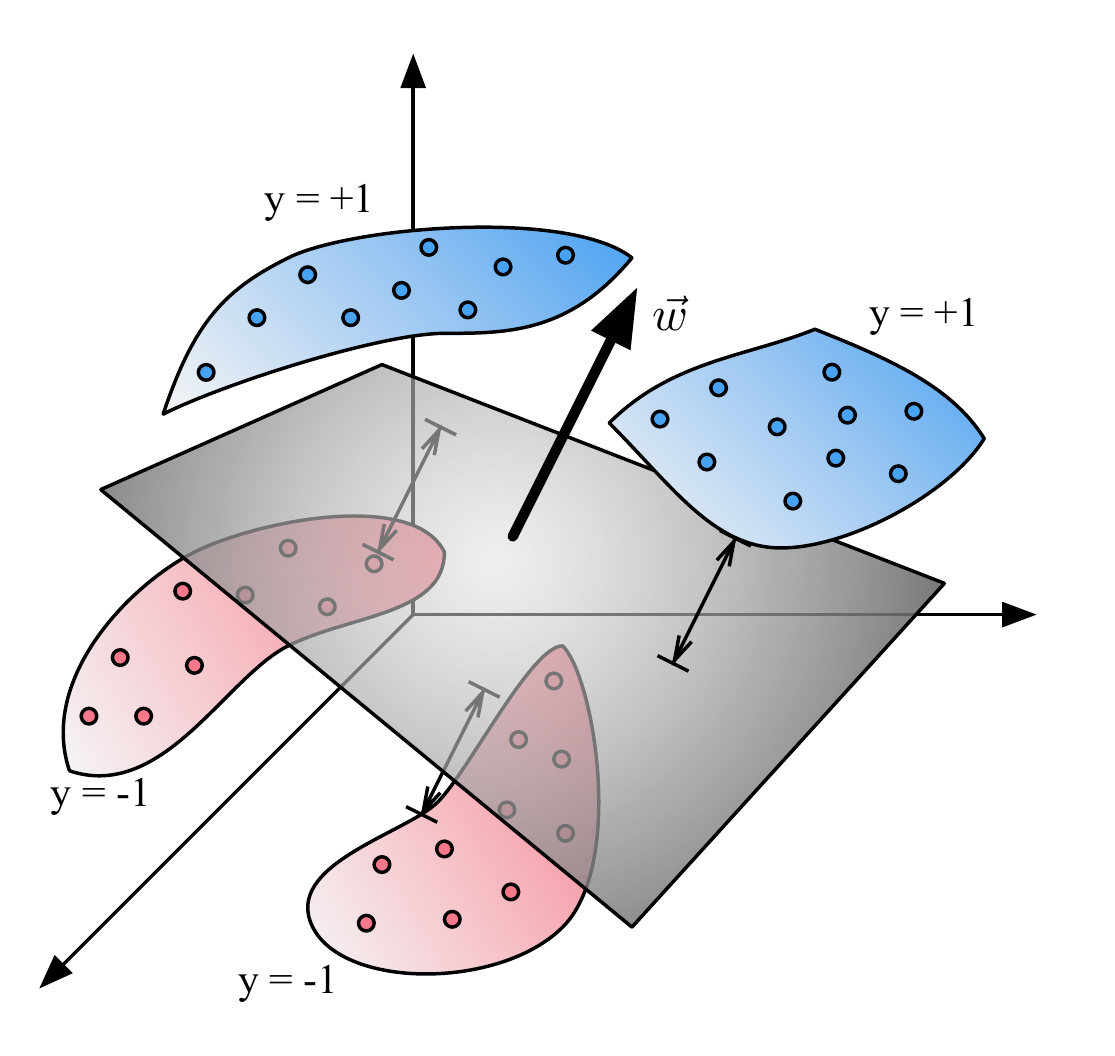}}
\par\end{centering}
\textcolor{black}{\caption{The maximum margin manifold binary classification problem. The optimal
linear hyperplane is parameterized by the weight vector $\vec{w}$
which separates positively labeled manifolds from negatively labeled
manifolds. Traditional data augmentation techniques would sample a
large number of points from each manifold to train a conventional
SVM. \label{fig:ManifoldClassification} }
}
\end{figure}

We propose a new method, called the Maximum Margin Manifold Machine
or $M^{4}$, that uses knowledge of the manifolds to efficiently learn
a maximum margin classifier. Figure \ref{fig:ManifoldClassification}
illustrates the problem in its simplest form, binary classification
of manifolds with a linear hyperplane. Given a number of manifolds
embedded in a feature space, the $M^{4}$ algorithm learns a weight
vector $\vec{w}$ that separates positively labeled manifolds from
negatively labeled manifolds with the maximum margin. Although the
manifolds consist of uncountable sets of points, the $M^{4}$ algorithm
is able to find a good solution in a provably finite number of iterations
and training examples.

Support vector machines (SVM) are widely used method to learn a maximum
margin classifier based upon a set of training examples \textcolor{black}{\cite{vapnik1998statistical}}.
However, the standard SVM algorithm quickly becomes computationally
intractable in time and memory as the number of training examples
increases, rendering data augmentation methods impractical for SVMs.
Methods to reduce the space complexity of SVM have been studied before,
in the context of dealing with large-scale datasets. Chunking makes
large problems solvable by breaking up the problem into subproblems
\cite{smola1998learning}, but the resultant kernel matrix may still
be very large. One method that subsamples the training data include
the reduced SVM (RSVM), which utilize a random rectangular subset
of the kernel matrix\textcolor{black}{{} \cite{lee2001rsvm}}. But this
approach and other methods that attempt to reduce the number of training
samples \textcolor{black}{\cite{wang2005training,smola1998learning}}
may result in suboptimal solutions that do not generalize well. 

Our $M^{4}$ algorithm directly handles the uncountable set of points
in the manifolds by solving a quadratic semi-infinite programming
problem (QSIP). $M^{4}$ is based upon a cutting-plane method which
iteratively refines a finite set of training examples to solve the
underlying QSIP \cite{fang2001solving,kortanek1993central,liu2004new}.
The cutting-plane method was also previously shown to efficiently
handle learning problems with a finite number of examples but an exponentially
large number of constraints \textcolor{black}{\cite{joachims2006training}}.
We provide a novel analysis of the convergence of $M^{4}$ with both
hard and soft margins. When the problem is realizable, the convergence
bound explicitly depends upon the margin value whereas with a soft
margin and slack variables, the bound depends linearly on the number
of manifolds.

\textcolor{black}{The chapter is organized as follows. We first consider
the hard margin problem and analyze the simplest form of the $M^{4}$
algorithm. Next, we introduce slack variables in $M^{4}$, one for
each manifold, and analyze its convergence with those additional auxiliary
variables. We then demonstrate application of $M^{4}$ to both synthetic
data where the manifold geometry is known as well as to actual object
images undergoing a variety of warpings. We compare its performance,
both in efficiency and generalization error, with conventional SVMs
using data augmentation techniques. Finally, we discuss some natural
extensions and potential future work on $M^{4}$ and its applications.}

\section{\textcolor{black}{Maximum Margin Manifold Machines with Hard Margin
\label{sec:M4-simple}}}

\textcolor{black}{In this section, we first consider the problem of
classifying a set of manifolds when they are linearly separable. This
allows us to introduce the simplest version of the $M^{4}$ algorithm
along with the appropriate definitions and QSIP formulation. We analyze
the convergence of the simple algorithm and prove an upper bound on
the number of errors the algorithm can make in this setting.}

\subsection{\textcolor{black}{Hard Margin QSIP \label{subsec:simpeM4-optprob}}}

\textcolor{black}{Formally, we are given a set of $P$ manifolds $M_{p}\subset\mathbb{R}^{N}$,
$p=1,\ldots,P$ with binary labels $y_{p}=\pm1$ (all points in the
same manifold share the same label). Each $M_{p}$ is defined by a
parametrization $\vec{x}=M_{p}(\vec{s}$) where $\vec{s}\in S_{p}$,
$S_{p}$ is a compact subset of $\mathbb{R}^{D}$, $M_{p}(\vec{s}):\mathbb{R}^{D}\rightarrow\mathbb{R}^{N}$
is a continuous function of $\vec{s}\in S_{p}$ so that the manifolds
are bounded: $\forall\vec{s},\,\left\Vert M_{p}(\vec{s})\right\Vert <L$
by some $L$. We would like to solve the following semi-infinite quadratic
programming problem for the weight vector $\vec{w}\in\mathbb{R}^{N}$
: }

\begin{equation}
\begin{array}{c}
SVM_{simple}:\underset{\vec{w}}{\text{argmin}}\frac{1}{2}\left\Vert \vec{w}\right\Vert ^{2}\\
s.t.\;\forall p,\forall\vec{x}\in M_{p}:\;y_{p}\left\langle \vec{w},\vec{x}\right\rangle \ge1
\end{array}\label{eq:SVMopt}
\end{equation}

\textcolor{black}{This is the primal formulation of the problem, where
maximizing the margin $\kappa=\frac{1}{||\vec{w}||}$ is equivalent
to minimizing the squared norm $\frac{1}{2}||\vec{w}||^{2}.$ We denote
the maximum margin attainable by $\kappa^{*}$, and the optimal solution
as $\vec{w}^{*}$. For simplicity, we do not explicitly include the
bias term here. A non-zero bias can be modeled by adding an additional
feature of constant value as a component to all the $\vec{x}$. Note
that the dual formulation of this QSIP is more complicated, involving
optimization of non-negative measures over the manifolds. In order
to solve the hard margin QSIP, we propose the following simple $M^{4}$
algorithm. }

\subsection{\textcolor{black}{$M_{simple}^{4}$ Algorithm}}

\textcolor{black}{The $M_{simple}^{4}$ algorithm is an iterative
algorithm to find the optimal $\vec{w}$ in \eqref{eq:SVMopt}, based
upon a cutting plane method for solving the QSIP. The general idea
behind $M_{simple}^{4}$ is to start with a finite number of training
examples, find the maximum margin solution for that training set,
augment the training set by looking for a point on the manifolds that
most violates the constraints, and iterating this process until a
tolerance criterion is reached.}

\textcolor{black}{At each stage $k$ of the algorithm there is a finite
set of training points and associated labels. The training set at
the $k$-th iteration is denoted by the set: $T_{k}=\left\{ \left(\vec{x}^{i}\in M_{p_{i}},y_{p_{i}}\right)\right\} $
with $i=1,\ldots,\left|T_{k}\right|$ examples. For the $i$-th pattern
in $T_{k}$, $p_{i}$ is the index of the manifold, and $y_{p_{i}}$
is its associated label.}

\textcolor{black}{On this set of examples, we solve the following
finite quadratic programming problem:}

\begin{equation}
\begin{array}{c}
SVM_{T_{k}}:\underset{\vec{w}}{\text{argmin}}\frac{1}{2}\left\Vert \vec{w}\right\Vert ^{2}\\
s.t.\forall\vec{x}^{i}\in T_{k}:\;y_{p_{i}}\left\langle \vec{w},\vec{x}^{i}\right\rangle \ge1
\end{array}\label{eq:SVMTk}
\end{equation}

\textcolor{black}{to obtain the optimal weights $\vec{w}^{(k)}$ on
the training set $T_{k}$. We then find a worst constraint-violating
point $\vec{x}_{k+1}\in M_{p_{k+1}}$ from one of the manifolds such
that 
\begin{equation}
\left\{ p_{k+1},\,\vec{x}_{k+1}\right\} :\underset{p,\,\vec{x}\in M_{p}}{\text{argmin}}y_{p}\left\langle \vec{w}^{(k)},\vec{x}\right\rangle <1-\text{\ensuremath{\delta}}\label{eq:violation_point-1}
\end{equation}
}

\textcolor{black}{with a required tolerance $\delta>0.$ If there
is no such point, the $M_{simple}^{4}$ algorithm terminates. If such
a point exists, it is added to the training set, defining the new
set $T_{k+1}=T_{k}\cup\left\{ \left(\vec{x}_{k+1},y_{p_{k+1}}\right)\right\} $.
The algorithm then proceeds at the next iteration to solve $SVM_{T_{k+1}}$to
obtain $\vec{w}^{(k+1)}$. For $k=1$, the set $T_{1}$ is initialized
with at least one point from each manifold. The pseudocode for $M_{simple}^{4}$
is shown in \ref{alg:M4}. }

\textcolor{black}{}
\begin{algorithm}[h]
\begin{enumerate}
\item \textcolor{black}{Input: $\delta$ (tolerance), $P$ manifolds and
labels $\left\{ M_{p},\,y_{p}=\pm1\right\} $, $p=1,...,P$.}
\item \textcolor{black}{Initialize the iteration number $k=1$, and the
set $T_{1}=\left\{ \left(\vec{x}^{i}\in M_{p_{i}},\,y_{p_{i}}\right)\right\} $
with at least one sample from each manifold $M_{p}$.}
\item \textcolor{black}{Solve for $\vec{w}^{(k)}$ in $SVM_{T_{k}}$: $\min\frac{1}{2}\left\Vert \vec{w}\right\Vert ^{2}$
such that $y^{p_{i}}\langle\vec{w,}\vec{x}^{i}\rangle\ge1$ for all
$\left(x^{i},y^{p_{i}}\right)\in T_{k}$.}
\item \textcolor{black}{Find a worst point $x^{k+1}\in M_{p_{k+1}}$ among
the manifolds $\left\{ p_{k+1}=1,...,P\right\} $ with a margin smaller
than $1-\delta$:}

\textcolor{black}{
\[
\vec{x}_{k+1}=\underset{\vec{x}\in M_{p_{k+1}},p_{k+1}=1,...P}{\text{argmin}}y_{p_{k+1}}\left\langle \vec{w}^{(k)},\vec{x}\right\rangle <1-\text{\ensuremath{\delta}}
\]
}
\item \textcolor{black}{If there is no such point, then stop. Else, augment
the point set: $T_{k+1}=T_{k}\cup\left\{ \left(\vec{x}_{k+1},y_{p_{k+1}}\right)\right\} $.}
\item \textcolor{black}{$k\leftarrow k+1$ and go to 3. }
\end{enumerate}
\textcolor{black}{\caption{Pseudocode for the $M_{simple}^{4}$ algorithm.\label{alg:M4} }
}
\end{algorithm}

\textcolor{black}{In step 4 of the $M_{simple}^{4}$ algorithm, a
point among the manifolds needs to be found with the worst margin
constraint violation. This is particularly convenient if the manifolds
are given by analytic parametric forms, where this point could be
computed analytically as for the case of manifolds with $L_{2}$ balls
or ellipses. However, for the algorithm to converge it is sufficient
that a constraint violation point is found. Thus, local optimization
procedures such as gradient descent may be used to search for such
a point. However, the speed of convergence in the latter stages of
$M_{simple}^{4}$ might be improved by a larger difference in the
constraint violation of the point found in this step.}

\subsection{\textcolor{black}{Convergence of $M_{simple}^{4}$ }}

\textcolor{black}{The $M_{simple}^{4}$ algorithm will converge asymptotically
to an optimal solution when it exists. Here we show that the $\vec{w}^{(k)}$
asymptotically converges to an optimal $\vec{w}^{\star}$. Denote
the change in the weight vector in the $k$-th iteration as $\Delta\vec{w}^{(k)}=\vec{w}^{(k+1)}-\vec{w}^{(k)}$.
First we have the following lemma:}
\begin{lem}
\textcolor{black}{\label{lem:delta_projection_simple}The change in
the weights satisfies $\langle\Delta\vec{w}^{(k)},\vec{w}^{(k)}\rangle\geq0$
. }
\end{lem}
\begin{proof}
\textcolor{black}{Define $\vec{w}(\lambda)=\vec{w}^{(k)}+\lambda\Delta\vec{w}^{(k)}$.
Then for all $0\le\lambda\le1$, $\vec{w}(\lambda)$ satisfies the
constraints on the point set $T_{k}$: $y_{p_{i}}\left\langle \vec{w}(\lambda),\vec{x}_{i}\right\rangle \ge1$
for all $\left(\vec{x}_{i},y_{p_{i}}\right)\in T_{k}$. However, if
$\langle\Delta\vec{w}^{(k)},\vec{w}^{(k)}\rangle<0$, there exists
a $0<\lambda^{\prime}<1$ such that $\left\Vert \vec{w}(\lambda^{\prime})\right\Vert ^{2}<\left\Vert \vec{w}^{(k)}\right\Vert ^{2}$,
contradicting the fact that $\vec{w}^{(k)}$ is a solution to $SVM_{T_{k}}$.}
\end{proof}
\textcolor{black}{Next, we show that the norm $\left\Vert \vec{w}^{(k)}\right\Vert ^{2}$
must monotonically increase by a finite amount at each iteration.} 
\begin{thm}
\textcolor{black}{\label{thm:delta_finite_simple}In the $k_{th}$
iteration of $M_{simple}^{4}$ algorithm, the increase in the norm
of $\vec{w}^{(k)}$ is lower bounded by $\left\Vert \vec{w}^{(k+1)}\right\Vert ^{2}\ge\left\Vert \vec{w}^{(k)}\right\Vert ^{2}+\frac{\delta_{k}^{2}}{L^{2}}$,
where $\delta_{k}=1-y_{p_{k+1}}\left\langle \vec{w}^{(k)},\vec{x}_{k+1}\right\rangle $
and }$\left\Vert \vec{x}_{k+1}\right\Vert \leq L$\textcolor{black}{{}
. }
\end{thm}
\begin{proof}
\textcolor{black}{First, note that $\delta_{k}>\delta\geq0$ , otherwise
the algorithm stops. We have: }$\left\Vert \vec{w}^{(k+1)}\right\Vert ^{2}=\left\Vert \vec{w}^{(k)}+\Delta\vec{w}^{(k)}\right\Vert ^{2}$$=\left\Vert \vec{w}^{(k)}\right\Vert ^{2}+\left\Vert \Delta\vec{w}^{(k)}\right\Vert ^{2}+2\left\langle \vec{w}^{(k)},\Delta\vec{w}^{(k)}\right\rangle $
$\geq\left\Vert \vec{w}^{(k)}\right\Vert ^{2}+\left\Vert \Delta\vec{w}^{(k)}\right\Vert ^{2}$\textcolor{black}{{}
due to Lemma \ref{lem:delta_projection_simple}. Now we consider the
point added to set $T_{k+1}=T_{k}\cup\left(\vec{x}_{k+1},y_{p_{k+1}}\right)$.
At this point, we know $y_{p_{k+1}}\left\langle \vec{w}^{(k+1)},\overrightarrow{x}_{k+1}\right\rangle \ge1$,
$y_{p_{k+1}}\left\langle w^{(k)},\overrightarrow{x}_{k+1}\right\rangle =1-\delta_{k}$,
hence $y_{p_{k+1}}\left\langle \Delta\vec{w}^{(k)},\vec{x}_{k+1}\right\rangle \ge\delta_{k}$.
Then, from the Cauchy-Schwartz inequality,
\begin{equation}
\left\Vert \Delta\vec{w}^{(k)}\right\Vert ^{2}\ge\frac{\delta_{k}^{2}}{\left\Vert \vec{x}_{k+1}\right\Vert ^{2}}>\frac{\delta_{k}^{2}}{L^{2}}>\frac{\delta^{2}}{L^{2}}\label{eq:deltaWLowerBound}
\end{equation}
}

\textcolor{black}{Since the optimal solution $\vec{w}^{\star}$ satisfies
the constraints for $T_{k}$, we have $\left\Vert \vec{w}^{(k)}\right\Vert \le\frac{1}{\kappa^{*}}$.
We thus have a sequence of iterations whose norms monotonically increase
and are upper bounded by $\frac{1}{\kappa^{\star}}$. Due to convexity,
there is a single global optimum and the $M_{simple}^{4}$ algorithm
is guaranteed to converge, asymptotically if the tolerance $\delta=0$,
and in a finite number of steps if $\delta>0$ . }
\end{proof}
\textcolor{black}{As a corollary, we see that this procedure is guaranteed
to find a realizable solution if it exists in a finite number of steps.}
\begin{cor}
\textcolor{black}{\label{cor:ZeroTrErrConv}The $M_{simple}^{4}$
algorithm converges to a zero error classifier in less than $\frac{L^{2}}{\left(\kappa^{\star}\right)^{2}}$
iterations, where $\kappa^{\star}$ is the optimal margin and $L$
bounds the norm of the points on the manifolds.}
\end{cor}
\begin{proof}
\textcolor{black}{When there is an error, we have $\delta_{k}>1$,
and $\left\Vert \vec{w}^{(k+1)}\right\Vert ^{2}\ge\left\Vert \vec{w}^{(k)}\right\Vert ^{2}+\frac{1}{L^{2}}$
(See \eqref{eq:deltaWLowerBound}). This implies the total number
of possible errors is upper bounded by $\frac{L^{2}}{\left(\kappa^{\star}\right)^{2}}$. }
\end{proof}
\textcolor{black}{With a finite tolerance $\delta>0$, we obtain a
bound on the number of iterations required for convergence:}
\begin{cor}
\textcolor{black}{The $M_{simple}^{4}$ algorithm for a given tolerance
$\delta>0$ will terminate after a finite number of iterations, with
less than $\frac{L^{2}}{\left(\kappa^{\star}\delta\right)^{2}}$ iterations
where $\kappa^{\star}$ is the optimal margin and $L$ bounds the
norm of the points on the manifolds.}
\end{cor}
\begin{proof}
\textcolor{black}{Again, $\left\Vert \vec{w}^{k}\right\Vert ^{2}\le\left\Vert \vec{w}^{\star}\right\Vert ^{2}=\frac{1}{\left(\kappa^{\star}\right)^{2}}$
and each iteration increases the squared norm by at least $\frac{\delta^{2}}{L^{2}}$.}
\end{proof}
\textcolor{black}{We can also bracket the error in the objective function
after $M_{simple}^{4}$ terminates:}
\begin{cor}
\textcolor{black}{With tolerance $\delta>0$, after $M_{simple}^{4}$
terminates with solution $\vec{w}_{M^{4}}$, the optimal value $\left\Vert \vec{w}^{\star}\right\Vert $
of $SVM_{simple}$ is bracketed by:
\begin{equation}
\left\Vert \vec{w}_{M^{4}}\right\Vert ^{2}\le\left\Vert \vec{w}^{\star}\right\Vert ^{2}\le\frac{1}{\left(1-\delta\right)^{2}}\left\Vert \vec{w}_{M^{4}}\right\Vert ^{2}.
\end{equation}
}
\end{cor}
\begin{proof}
\textcolor{black}{The lower bound on $\left\Vert \vec{w}^{\star}\right\Vert ^{2}$
is as before. Since $M_{simple}^{4}$ has terminated, setting $\vec{w}^{\prime}=\frac{1}{(1-\delta)}\vec{w}_{M^{4}}$
would make $\vec{w}^{\prime}$ feasible for $SVM_{simple}$, resulting
in the upper bound on $\left\Vert \vec{w}^{\star}\right\Vert ^{2}$. }
\end{proof}

\section{\textcolor{black}{$M^{4}$ with Slack Variables \label{sec:M4-slack}}}

\textcolor{black}{In many classification problems, the manifolds may
not be linearly separable due to their dimensionality, size, and/or
correlations. In these situations, $M_{simple}^{4}$ will not even
be able to find a feasible solution. To handle these problems, the
classic approach is to introduce slack variables. Naively, we could
introduce a slack variable for every point on the manifolds as below:}

\[
\begin{array}{c}
SVM_{naive}^{slack}:\underset{\vec{w},\xi_{p}(\vec{x})}{\text{argmin}}\frac{1}{2}\left\Vert \vec{w}\right\Vert ^{2}+C\sum_{p=1}^{P}\int_{\vec{x}\in M_{p}}\xi_{p}(\vec{x})\\
s.t.\;\forall p,\forall\vec{x}\in M_{p}:\;y_{p}\left\langle \vec{w},\vec{x}\right\rangle +\xi_{p}(\vec{x})\ge1,\\
\xi_{p}(\vec{x})\ge0
\end{array}
\]

\textcolor{black}{The parameter $C$ represents the tradeoff between
fitting the manifolds to obey the margin constraints while allowing
some set of points to be misclassified. This approach cannot be used
when training data consists of entire manifolds as in general, it
would require replacing the sum over a finite number of training points
in the cost function, to an integral with an appropriate measure over
the manifolds. Thus, we formulate an alternative version of the QSIP
with slack variables below.}

\subsection{\textcolor{black}{QSIP with Manifold Slacks }}

In this work, we propose using only one slack variable per manifold
for classification problems with non-separable manifolds. This formulation
demands that all the points on each manifold $\vec{x}\in M_{p}$ obey
an inequality constraint with one manifold slack variable, $y_{p}\left\langle \vec{w},\vec{x}\right\rangle +\xi_{p}\ge1$.
As we see below, solving for this constraint is tractable and the
algorithm has good convergence guarantees.

However, this single slack requirement for each manifold by itself
may not be sufficient for good generalization performance. Indeed,
our empirical studies show that generalization performance can be
improved if we additionally demand that some representative points
$\vec{x}_{p}\in M_{p}$ on each manifold also obey the margin constraint:
$y_{p}\left\langle \vec{w},\vec{x}_{p}\right\rangle \ge1$. In this
work, we implement this intuition by specifying appropriate center
points $\vec{x}_{p}^{c}$ for each manifold $M_{p}$. This center
point could be the center of mass of the manifold, a representative
point, or an exemplar used to generate the manifolds \textcolor{black}{\cite{krizhevsky2012imagenet}}.
For simplicity, we demand that these points strictly obey the margin
inequalities corresponding to their manifold label, but we could have
alternatively introduced additional slack variables for these constraints.
Formally, the QSIP optimization problem is summarized below, where
the objective function is minimized over the\textcolor{black}{{} weight
vector $\vec{w}\in\mathbb{R}^{N}$ and slack variables $\vec{\xi}\in\mathbb{R}^{P}$: }

\[
\begin{array}{c}
SVM_{manifold}^{slack}:\,\underset{\vec{w},\vec{\xi}}{\text{argmin}}F(\vec{w},\vec{\xi})=\frac{1}{2}\left\Vert \vec{w}\right\Vert ^{2}+C\sum_{p=1}^{P}\xi_{p}\\
s.t.\;\forall p,\forall\vec{x}\in M_{p}:\;y_{p}\left\langle \vec{w},\vec{x}\right\rangle +\xi_{p}\ge1\ \text{(manifolds)}\\
\forall p:\ y_{p}\left\langle \vec{w},\vec{x}_{p}^{c}\right\rangle \ge1\ \text{(centers)}\\
\forall p:\;\xi_{p}\ge0
\end{array}
\]

\subsection{\textcolor{black}{$M_{slack}^{4}$ Algorithm }}

With these definitions, we introduce our $M_{slack}^{4}$ algorithm
with slack variables below. 

\textcolor{black}{}
\begin{algorithm}[h]
\begin{enumerate}
\item \textcolor{black}{Input: $\delta$ (tolerance), $P$ manifolds and
labels $\left\{ M_{p},y_{p}=\pm1\right\} $, and centers $\vec{x}_{p}^{c}$}
\item \textcolor{black}{Initialize the iteration number $k=1$, and the
set $T_{1}=\left\{ \left(\vec{x}^{i}\in M_{p_{i}},\,y_{p_{i}}\right)\right\} $
with at least one sample from each manifold $M_{p}$.}
\item \textcolor{black}{Solve for $\vec{w}^{(k)}$,$\vec{\xi}^{(k)}$: $\min\frac{1}{2}\left\Vert \vec{w}\right\Vert ^{2}+C\sum_{p=1}^{P}\xi_{p}$
such that $y_{p_{\mu}}\langle\vec{w,}\vec{x}^{\mu}\rangle+\xi_{p_{\mu}}\ge1$
for all $\left(\vec{x}^{\mu},y_{p_{\mu}}\right)\in T_{k}$ and $y_{p}\langle\vec{w,}\vec{x}_{p}^{c}\rangle\ge1$
for all $p$.}
\item \textcolor{black}{Find a point $\vec{x}^{k+1}\in M_{p_{k+1}}$ among
the manifolds $\left\{ p=1,...,P\right\} $ with slack violation larger
than the tolerance $\delta$:}

\textcolor{black}{
\[
y_{p_{k+1}}\left\langle \vec{w}^{(k)},\vec{x}_{k+1}\right\rangle +\xi_{p_{k+1}}<1-\text{\ensuremath{\delta}}
\]
}
\item \textcolor{black}{If there is no such point, then stop. Else, augment
the point set: $T_{k+1}=T_{k}\cup\left\{ \left(\vec{x}_{k+1},y_{p_{k+1}}\right)\right\} $.}
\item \textcolor{black}{$k\leftarrow k+1$ and go to 3. }
\end{enumerate}
\textcolor{black}{\caption{Pseudocode for the $M_{slack}^{4}$ algorithm.\label{alg:M4-slack-L1} }
}
\end{algorithm}

\textcolor{black}{The proposed $M_{slack}^{4}$ algorithm modifies
the cutting plane approach to solve a semi-infinite, semi-definite
quadratic program. Each iteration involves a finite set: $T_{k}=\left\{ \left(\vec{x}^{i}\in M_{p_{i}},y_{p_{i}}\right)\right\} $
with $i=1,\ldots,\left|T_{k}\right|$ examples that is used to define
the following soft margin SVM:}

\[
\begin{array}{c}
SVM_{T_{k}}^{slack}:\underset{\vec{w},\vec{\xi}}{\text{argmin}}\frac{1}{2}\left\Vert \vec{w}\right\Vert ^{2}+C\sum_{p=1}^{P}\xi_{p}\\
s.t.\;\forall\left(\vec{x}^{i},y_{p_{i}}\right)\in T_{k}:\;y_{p_{i}}\left\langle \vec{w},\vec{x}^{i}\right\rangle +\xi_{p_{i}}\ge1\\
\forall p:\ y_{p}\left\langle \vec{w},\vec{x}_{p}^{c}\right\rangle \ge1\;(centers)\\
\forall p:\;\xi_{p}\ge0
\end{array}
\]

\textcolor{black}{to obtain the weights $\vec{w}^{(k)}$ and slacks
$\vec{\xi}^{(k)}$ at each iteration. We then find a point $\vec{x}_{k+1}\in M_{p_{k+1}}$
from one of the manifolds so that:
\begin{equation}
y_{p_{k+1}}\left\langle \vec{w}^{(k)},\vec{x}_{k+1}\right\rangle +\xi_{p_{k+1}}^{(k)}=1-\delta_{k}\label{eq:violation_point}
\end{equation}
where $\delta_{k}>\delta$. If there is no such a point, the $M_{slack}^{4}$
algorithm terminates. Otherwise, the point $\vec{x}_{k+1}$ is added
as a training example to the set $T_{k+1}=T_{k}\cup\left\{ \left(\vec{x}_{k+1},y_{p_{k+1}}\right)\right\} $
and the algorithm proceeds to solve $SVM_{T_{k+1}}^{slack}$ to obtain
$\vec{w}^{(k+1)}$ and $\vec{\xi}^{(k+1)}$. Note that }$M_{slack}^{4}$
\textcolor{black}{embodies the fact that for the algorithm to converge,
it is not necessary to find the point with the }\textcolor{black}{\emph{worst
}}\textcolor{black}{constraint violation at each iteration. }

\subsection{\textcolor{black}{Convergence of $M_{slack}^{4}$ }}

\textcolor{black}{Here we show that the objective function $F\left(\vec{w},\vec{\xi}\right)=\frac{1}{2}\left\Vert \vec{w}\right\Vert ^{2}+C\sum_{p=1}^{P}\xi_{p}$
is guaranteed to increase by a finite amount with each iteration.
This result is similar to \cite{tsochantaridis2005large}, but here
we demonstrate a proof in the primal domain over an infinite number
of examples. We first have the following lemmas, }
\begin{lem}
\textcolor{black}{\label{lem:delta_projection_slack} The change in
the weights and slacks satisfy:
\begin{equation}
\left\langle \Delta\vec{w}^{(k)},\vec{w}^{(k)}\right\rangle +C\sum_{p}\Delta\vec{\xi}_{p}^{(k)}\ge0\label{eq:delta_projection}
\end{equation}
}

where\textcolor{black}{{} $\Delta\vec{w}^{(k)}=\vec{w}^{(k+1)}-\vec{w}^{(k)}$
and $\Delta\vec{\xi}^{(k)}=\vec{\xi}^{(k+1)}-\vec{\xi}^{(k)}$.}
\end{lem}
\begin{proof}
\textcolor{black}{Define $\vec{w}(\lambda)=\vec{w}^{(k)}+\lambda\Delta\vec{w}^{(k)}$
and $\vec{\xi}(\lambda)=\vec{\xi}^{(k)}+\lambda\Delta\vec{\xi}^{(k)}$.
Then for all $0\le\lambda\le1$, $\vec{w}(\lambda)$ and $\vec{\xi}(\lambda)$
satisfy the constraints for $SVM_{T_{k}}^{slack}$. The resulting
change in the objective function is given by:
\begin{multline}
F\left(\vec{w}(\lambda),\vec{\xi}(\lambda)\right)-F\left(\vec{w}^{(k)},\vec{\xi}^{(k)}\right)=\\
\lambda\left[\left\langle \Delta\vec{w}^{(k)},\vec{w}^{(k)}\right\rangle +C\sum_{p}\Delta\xi_{p}^{(k)}\right]+\frac{1}{2}\lambda^{2}\left\Vert \Delta\vec{w}^{(k)}\right\Vert ^{2}
\end{multline}
If \eqref{eq:delta_projection} is not satisfied, then there is some
$0<\lambda^{\prime}<1$ such that $F\left(\vec{w}(\lambda^{\prime}),\vec{\xi}(\lambda^{\prime})\right)<F\left(\vec{w}^{(k)},\vec{\xi}^{(k)}\right)$,
which contradicts the fact that $\vec{w}^{(k)}$ and $\vec{\xi}^{(k)}$
are a solution to $SVM_{T_{k}}$.}
\end{proof}
\begin{lem}
\textcolor{black}{\label{lem:new_support_vector}In each iteration
of $M_{slack}^{4}$ algorithm, the added point $\left(\vec{x}_{k+1},y_{p_{k+1}}\right)$
must be a support vector for the new weights and slacks, s.t. 
\begin{equation}
y_{p_{k+1}}\left\langle \vec{w}^{(k+1)},\vec{x}_{k+1}\right\rangle +\xi_{p_{k+1}}^{(k+1)}=1\label{eq:new_support}
\end{equation}
and}

\textcolor{black}{
\begin{equation}
y_{p_{k+1}}\left\langle \Delta\vec{w}^{(k)},\vec{x}_{k+1}\right\rangle +\Delta\xi_{p_{k+1}}^{(k)}=\delta_{k}\label{eq:new_support_delta}
\end{equation}
}
\end{lem}
\begin{proof}
\textcolor{black}{Suppose $y_{p_{k+1}}\left\langle \vec{w}^{(k+1)},\vec{x}_{k+1}\right\rangle +\xi_{p_{k+1}}^{(k+1)}=1+\epsilon$
for some $\epsilon>0$. Then we can choose $\lambda^{\prime}=\frac{\delta_{k}}{\delta_{k}+\epsilon}<1$
so that $\vec{w}(\lambda^{\prime})=\vec{w}^{(k)}+\lambda^{\prime}\Delta\vec{w}^{(k)}$
and $\vec{\xi}(\lambda^{\prime})=\vec{\xi}^{(k)}+\lambda^{\prime}\Delta\vec{\xi}^{(k)}$
satisfy the constraints for $SVM_{T_{k+1}}^{slack}$. But, from Lemma
\ref{lem:delta_projection_slack}, we have $F\left(\vec{w}(\lambda^{\prime}),\vec{\xi}(\lambda^{\prime})\right)<F\left(\vec{w}^{(k+1)},\vec{\xi}^{(k+1)}\right)$
which contradicts the fact that $\vec{w}^{(k+1)}$ and $\text{\ensuremath{\vec{\xi}}}^{(k+1)}$
are a solution to $SVM_{T_{k+1}}$. Thus, $\epsilon=0$ and the point
$\left(\vec{x}_{k+1},y_{p_{k+1}}\right)$ must be a support vector
for $SVM_{T_{k+1}}.$ \eqref{eq:new_support_delta} results from subtracting
\eqref{eq:violation_point} from \eqref{eq:new_support}.}
\end{proof}
\textcolor{black}{We also derive a bound on the following quadratic
function over nonnegative values:}
\begin{lem}
\textcolor{black}{\label{lem:quadratic_bound}Given $K>0$,$\delta>0$,
then $\forall x\ge0$}

\textcolor{black}{
\begin{equation}
\frac{1}{2}(x-\delta)^{2}+Kx\ge\min\left(\frac{1}{2}\delta^{2},\frac{1}{2}K\delta\right)
\end{equation}
}
\end{lem}
\begin{proof}
\textcolor{black}{The minimum value occurs when $x^{\star}=\left[\delta-K\right]_{+}$.
When $K\ge\delta$, then $x^{\star}=0$ and the minimum is $\frac{1}{2}\delta^{2}$.
When $K<\delta$, the minimum occurs at $K\left(\delta-\frac{1}{2}K\right)\ge\frac{1}{2}K\delta$.
Thus, the lower bound is the smaller of these two values.}
\end{proof}
\begin{thm}
\textcolor{black}{\label{prop:objective_increase} In each iteration
$k$ of $M_{slack}^{4}$ algorithm, the increase in the objective
function for $SVM_{manifold}^{slack}$ is lower bounded by
\begin{equation}
\Delta F^{(k)}\ge\min\left(\frac{1}{8}\frac{\delta_{k}^{2}}{L^{2}},\,\frac{1}{2}C\delta_{k}\right)\label{eq:dFbound}
\end{equation}
}

where $\Delta F^{(k)}=F\left(\vec{w}^{(k+1)},\vec{\xi}^{(k+1)}\right)-F\left(\vec{w}^{(k)},\vec{\xi}^{(k)}\right)$.
\end{thm}
\begin{proof}
(Sketch) First, if\textcolor{black}{{} $\Delta\vec{w}^{(k)}\ne0$, Lemma
\ref{lem:delta_projection_slack} completes the proof. If $\Delta\vec{w}^{(k)}=0$,
then $\Delta\xi_{p_{k}}^{(k)}=\delta_{k}$ from Lemma \ref{lem:new_support_vector}
and $\Delta\xi_{p\ne p_{k}}^{(k)}=0$ since $\vec{\xi}^{(k)}$ is
the solution for $SVM_{T_{k}}$. So for $\Delta\vec{w}^{(k)}=0$,
$\Delta F^{(k)}=C\delta_{k}$. }

\textcolor{black}{The added point $\vec{x}_{k+1}$ is from a particular
manifold $M_{p_{k+1}}$. If $\Delta\xi_{p_{k+1}}^{(k)}\le0$, we have
$y_{p_{k+1}}\left\langle \Delta\vec{w}^{(k)},\vec{x}_{k+1}\right\rangle \ge\delta_{k}$
($\because$ Lemma \ref{lem:new_support_vector}). Then, $\left\Vert \Delta\vec{w}^{(k)}\right\Vert ^{2}\ge\frac{\delta_{k}^{2}}{L^{2}}$
, yielding $\Delta F^{(k)}\ge\frac{1}{2}\frac{\delta_{k}^{2}}{L^{2}}$.}

\textcolor{black}{We next analyze $\Delta\xi_{p_{k+1}}^{(k)}>0$ and
consider the finite set of points $\left(\vec{x}^{\nu},\,p_{k+1}\right)\in T_{k}$,
coming from the $p_{k+1}$ manifold. Each of these points obeys the
constraints:}

\textcolor{black}{
\begin{align}
y_{p_{k+1}}\left\langle \vec{w}^{(k)},\vec{x}^{\nu}\right\rangle +\xi_{p_{k+1}}^{(k)} & =1+\epsilon_{\nu}^{(k)}\\
y_{p_{k+1}}\left\langle \vec{w}^{(k+1)},\vec{x}^{\nu}\right\rangle +\xi_{p_{k+1}}^{(k+1)} & =1+\epsilon_{\nu}^{(k+1)}\\
\epsilon_{\nu}^{(k)},\,\epsilon_{\nu}^{(k+1)} & \ge0
\end{align}
}

\textcolor{black}{We consider the minimum value of the thresholds:
$\eta=\min_{\nu}\epsilon_{\nu}^{(k)}$. }

\textcolor{black}{If }\textbf{\textcolor{black}{$\eta>0$, }}none
of the $\vec{x}^{\nu}$ points are support vectors for \textcolor{black}{$SVM_{T_{k}}^{slack}$.
In this case, we define a linear set of slack variables: $\xi_{p}(\lambda)=\xi_{p_{k}}^{(k)}$
for $p=p_{k}$, and $\xi_{p}(\lambda)=\xi_{p}^{(k)}+\lambda\Delta\xi_{p}^{(k)}$
for $p\neq p_{k}$, and $\vec{w}(\lambda)=\vec{w}^{(k)}+\lambda\Delta\vec{w}^{(k)}$,
which satisfy the constraints for $SVM_{T_{k}}$. Then, this implies
\begin{equation}
\left\langle \Delta\vec{w}^{(k)},\vec{w}^{(k)}\right\rangle +C\sum_{p\ne p_{k}}\Delta\vec{\xi}_{p}^{(k)}\ge0\label{eq:nonsv_gradient_bound}
\end{equation}
}

which implies $\Delta F^{(k)}\ge\min\left(\frac{1}{2L^{2}}\delta_{k}^{2},\,\frac{1}{2}C\delta_{k}\right)$. 

\textcolor{black}{If $\eta=0$,}\textbf{\textcolor{black}{{} }}\textcolor{black}{at
least one support vector lies in $M_{p_{k+1}}$. Consider $\varepsilon=\min_{\epsilon_{\nu}^{(k)}=0}\epsilon_{\nu}^{(k+1)}\ge0$. }

\textcolor{black}{We then define $\xi_{p}(\lambda)=\xi_{p_{k}}^{(k)}+\lambda\left(\Delta\xi_{p}^{(k)}-\varepsilon\right)$
for $p=p_{k}$, and $\xi_{p}(\lambda)=\xi_{p}^{(k)}+\lambda\Delta\xi_{p}^{(k)}$
for $p\ne p_{k}$, and $\vec{w}(\lambda)=\vec{w}^{(k)}+\lambda\Delta\vec{w}^{(k)}$.
Then, there exists $0\le\lambda\le\lambda_{min}$ for which $\vec{\xi}(\lambda)$
and $\vec{w}(\lambda)$ satisfy the constraints for $SVM_{T_{k}}$,
so that }

\begin{equation}
\left\langle \Delta\vec{w}^{(k)},\vec{w}^{(k)}\right\rangle +C\sum_{p}\Delta\vec{\xi}_{p}^{(k)}\ge C\varepsilon\label{eq:sv_gradient_bound}
\end{equation}

\textcolor{black}{We also have a support vector $\left(\vec{x}^{\nu},\,p_{k+1}\right)\in T_{k}$,
with $y_{p_{k}}\left\langle \Delta\vec{w}^{(k)},\vec{x}^{\nu}\right\rangle +\Delta\xi_{p_{k}}^{(k)}=\varepsilon$,
then $\left\Vert \Delta\vec{w}^{(k)}\right\Vert ^{2}\ge\frac{1}{4L^{2}}\left(\delta_{k}-\varepsilon\right)^{2}$
by using Lemma \ref{lem:new_support_vector}. }

Then, by using Lemma \eqref{lem:quadratic_bound} on $\varepsilon$,
we get

\begin{equation}
\Delta F^{(k)}\ge\min\left(\frac{1}{8L^{2}}\delta_{k}^{2},\,\frac{1}{2}C\delta_{k}\right)\label{eq:sv_quadratic_bound}
\end{equation}

\textcolor{black}{Thus, we obtain the final bound combining results
from two cases of $\eta$.}
\end{proof}
\textcolor{black}{Since the $M_{slack}^{4}$ algorithm is guaranteed
to increase the objective by a finite amount, it will terminate in
a finite number of iterations if we require $\delta_{k}>\delta$ for
some positive $\delta>0$.}
\begin{cor}
\textcolor{black}{The $M_{slack}^{4}$ algorithm for a given $\delta>0$
will terminate after at most $P\cdot\max\left(\frac{8CL^{2}}{\delta^{2}},\,\frac{2}{\delta}\right)$
iterations where $P$ is the number of manifolds, L bounds the norm
of the points on the manifolds.}
\end{cor}
\begin{proof}
\textcolor{black}{$\vec{w}=0$ and $\xi_{p}=1$ is a feasible solution
for $SVM_{manifold}^{slack}$. Therefore, the optimal objective function
is upper-bounded by $F\left(\vec{w}=0,\vec{\xi}=1\right)=PC$. The
upper bound on the number of iterations is then provided by Theorem
\eqref{prop:objective_increase}. }
\end{proof}
\textcolor{black}{We can also bound the error in the objective function
after $M_{slack}^{4}$ terminates:}
\begin{cor}
\textcolor{black}{With $\delta>0$, after $M_{slack}^{4}$ terminates
with solution $\vec{w}_{M^{4}}$, slack $\vec{\xi}_{M^{4}}$ and value
$F_{M^{4}}=F\left(\vec{w}_{M^{4}},\vec{\xi}_{M^{4}}\right)$, then
the optimal value $F^{\star}$ of $SVM_{manifold}^{slack}$ is bracketed
by:
\begin{equation}
F_{M^{4}}\le F^{\star}\le F_{M^{4}}+PC\delta.
\end{equation}
}
\end{cor}
\begin{proof}
\textcolor{black}{The lower bound on $F^{\star}$ is apparent since
$SVM_{manifold}^{slack}$ includes $SVM_{T_{k}}^{slack}$ constraints
for all $k$. Setting the slacks $\xi_{p}=\xi_{M^{4},p}+\delta$ will
make the solution feasible for $SVM_{manifold}^{slack}$ resulting
in the upper bound. }
\end{proof}

\section{\textcolor{black}{Experiments \label{sec:Experiments}}}

\subsection{\textcolor{black}{Synthetic Manifolds }}

\paragraph*{Random $L_{2}$ balls }

\textcolor{black}{As an illustration of our method, we have generated
manifolds consisting of random $D$-dimensional Euclidean balls with
a given radius. Each manifold $M_{p}$ is described by a center vector
$\vec{x}_{p}^{c}\in\mathbb{R}^{N}$ and $D$ basis vectors $\left\{ \vec{u}_{i}^{p}\in\mathbb{R}^{N},\:i=1,...,D\right\} $.
The points on the manifold can be parameterized as $M_{p}=\left\{ \vec{x}\left|\vec{x}=\vec{x}_{p}^{c}+R\sum_{i=1}^{D}s_{i}\vec{u}_{i}^{p}\right.\right\} $
where $R$ is the radius of the ball and $\vec{s}\in\mathbb{R}^{D}$
are normalized so that $\sum_{i=1}^{D}s_{i}^{2}=1$. }

\paragraph*{Simulations}

We compare the performance of $M^{4}$ to the conventional point SVM
with samples uniformly drawn from the $L_{2}$ ball manifolds. Performance
is measured by generalization error as a function of the number of
samples used by the algorithm. 

\textcolor{black}{For these manifolds, the worst constraint-violating
point can easily be computed by taking the derivative of the constraint
$y^{p}\left[\vec{w}\cdot\left(\vec{x}_{0}^{p}+R\sum_{i=1}^{D}s_{i}\vec{u}_{i}^{p}\right)\right]+\xi_{p}\geq1$
with respect to $\vec{s}$ for all $p$. This results in the analytic
solution $s_{i}^{p,worst}=-\frac{y^{p}\vec{w}\cdot\vec{u_{i}^{p}}}{\sqrt{\sum_{i=1}^{D}\left(\vec{w}\cdot\vec{u_{i}^{p}}\right)^{2}}}$
. For problems with non-separable manifolds in $M_{slack}^{4}$, we
used an additional single margin constraint per manifold given by
the center $\vec{x}_{p}^{c}$. }

\textcolor{black}{We used the following parameters in the simulations
shown below: embedding dimension $N=500$, manifold dimension $D=10$,
radius $R=20$. With these parameters, the critical manifold capacity
for linear classification is estimated to be $P_{critical}=48.3$
\cite{chung2016linear}, hence we consider $P=46$ to test $M_{simple}^{4}$
and $P=50$ for the $M_{slack}^{4}$ simulations. }

The results are presented in figure \ref{fig:L2-balls-GenErr} for
the separable case and non-separable case.

\textcolor{black}{}
\begin{figure}[h]
\noindent \begin{centering}
\textcolor{black}{\includegraphics[width=1\columnwidth]{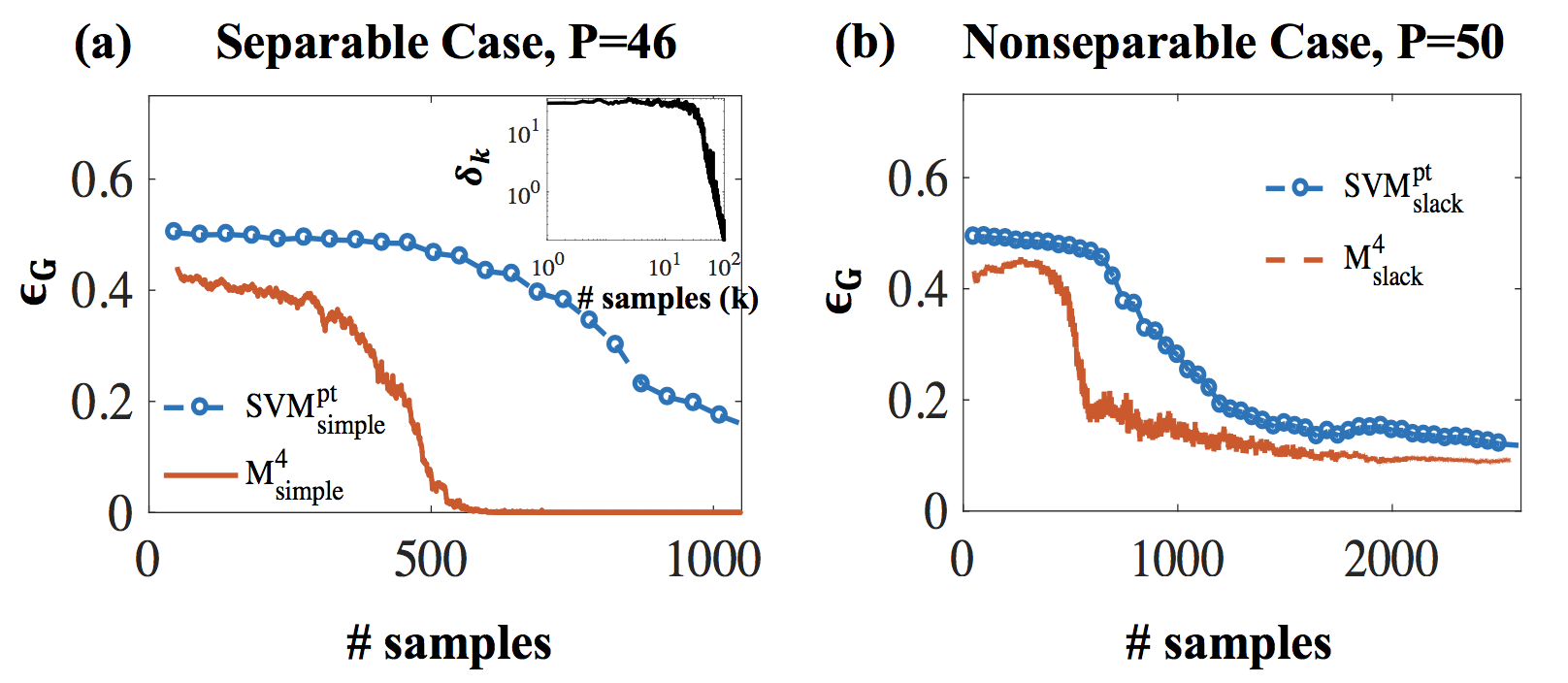}}
\par\end{centering}
\textcolor{black}{\caption{\textbf{Generalization error of the $M^{4}$ solution for $L_{2}$
ball manifolds}, shown as a function of the total number of training
samples per manifold (red solid) compared with that of conventional
point SVM (blue dashed). $N=500$, $D=10$, $R=20$, and (a) $P=46$
is used for $M_{simple}^{4}$and (b) $P=48$ for $M_{slack}^{4}$
with $C=100$. (a)-(Inset) $\delta_{k}$ is shown for the $k_{th}$
added point in $M_{simple}^{4}$. The critical capacity with these
parameters is $P_{c}\approx48$. \label{fig:L2-balls-GenErr}}
}
\end{figure}

\subsection{\textcolor{black}{ImageNet Dataset }}

\paragraph*{\textcolor{black}{Image-based Object Manifolds}}

We also apply the $M^{4}$ algorithm to a more realistic class of
object manifolds. Here each object manifold is defined by the infinite
set of images created by applying 2-D affine transformations on a
single template image. In order to create object manifolds, $P$ template
images were sampled from the ImageNet 2012 data set for which exact
object bounding boxes are available \cite{deng2009imagenet}, and
each image was cropped and scaled such that the object occupies the
middle $48\times48$ pixels of the template image. 

Each sample from the object manifold is a $64\times64$ gray-scale
image created by applying a 2-D affine transformation on the template
image. Those transformations are defined as a composition of seven
basic transformations: horizontal or vertical translation, horizontal
or vertical scaling, horizontal or vertical shear, and rotation. The
range of each basic transformations was chosen so that the largest
pixel displacement was equivalent to 8 pixels. The composition of
these seven basic transformations thus defines a 7-D highly non-linear
object manifold for each object.

\paragraph*{Simulations }

We compared the performance of $M^{4}$ to point SVM in classification
of samples from the object manifolds. Performance is measured as above
by generalization error as a function of the number of samples each
algorithm uses. We used object manifolds with up to $M=10000$ samples
drawn from each manifold, using $80\%$ of the samples as training
set and $20\%$ as a testing set. Rather than performing classification
directly on image pixels the samples were projected to the space defined
by their $N=500$ largest principal components. For this data set
the classification problem is separable for $P=4$ and non-separable
for $P=6$.

Point SVM (defined as $SVM_{naive}^{slack}$ above) was trained with
varying numbers of training set samples, with $C$ obtained through
cross validation. The training was repeated $11$ times with different
samples to estimate the variability of the generalization error. $M^{4}$
was trained \textcolor{black}{with a constraint per manifold given
by the center of mass at the training set}

and $C$ obtained through cross-validation. At each iteration of the
algorithm, the worst-violating constraint point was found using local
search. Initialized with a random sample from each manifold, it was
compared to a set of $K$ neighboring samples in the space of potential
transformation ($K=5$ was used throughout). This process is iterated
until a set of local minima were found, and these points were candidates
to be added to the active set of the $M^{4}$ algorithm. 

Figure \ref{fig:image-based-object-manifolds}-b compares the generalization
error for a separable classification problem (at $P=4$) while Figure
\ref{fig:image-based-object-manifolds}-c compare those for a non-separable
classification problem (at $P=6$). Those representative results illustrate
that in both cases $M^{4}$ achieve a very low generalization error
(compared to point SVM) already at very small number of samples.

\textcolor{black}{}
\begin{figure}[h]
\noindent \begin{centering}
\textcolor{black}{\includegraphics[width=1\columnwidth]{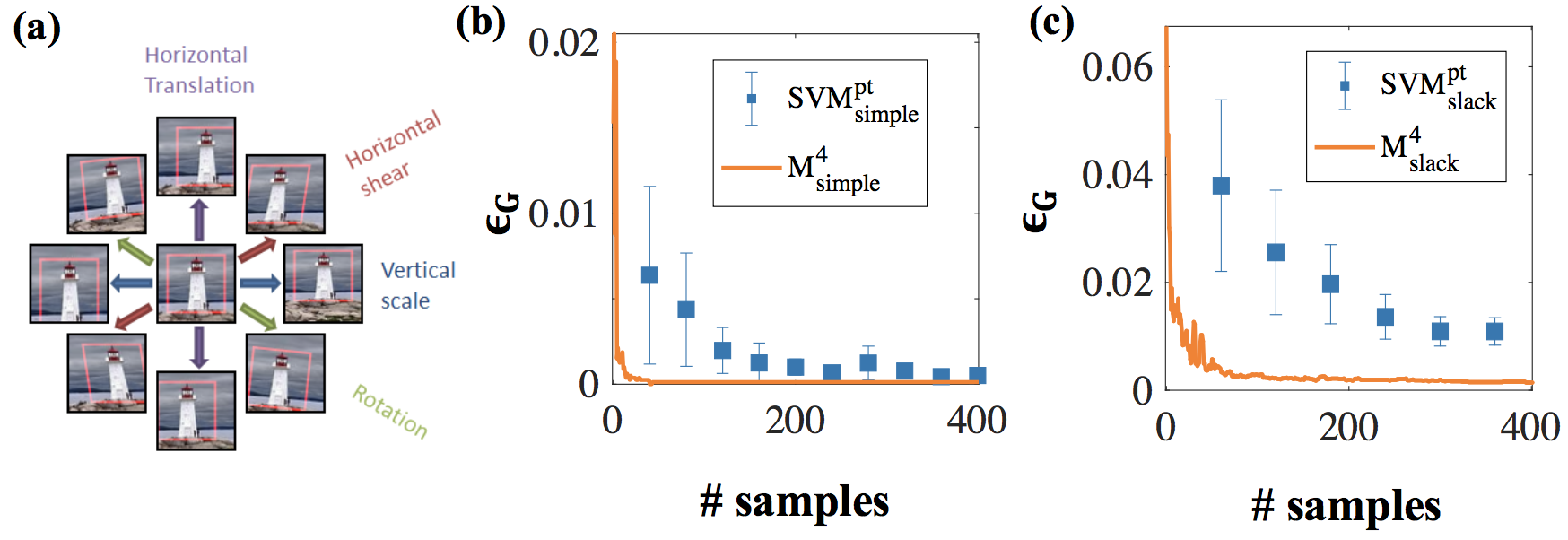}}
\par\end{centering}
\textcolor{black}{\caption{\label{fig:image-based-object-manifolds}\textbf{Image-based object
manifolds}. (a) Basic affine transformation: a template image (middle)
surrounded with changes along 4 axes defined by basic affine transformation.
A $48\times48$ square marking the object bounding box was added to
the template image for illustration purposes. (b-c) Generalization
error of the $M^{4}$ solution for 7-D image-based object manifolds,
shown as a function of the number of training samples per manifold
(solid line) compared with that of conventional point SVM (blue squares).
At $N=300$ the problem is separable for $P=4$ (b) and non-separable
for $P=6$ (c).}
}
\end{figure}

\section{\textcolor{black}{Discussion }}

\textcolor{black}{We described and analyzed a novel algorithm for
finding the maximum margin solution for classifying manifolds. The
algorithm, called $M^{4}$, is based upon a cutting-plane method and
iterates between adding the worst violating point to a finite training
set, and updating a maximum margin solution. There are two versions
of the algorithm, one without slack variable appropriate for separable
manifolds, and a slack version for non-separable manifolds. We proved
the convergence of $M^{4}$, and provided bounds on the number of
iterations required and the deviation from the optimal objective function.
On experiments with both synthetic manifolds and with actual image
manifolds, our empirical results demonstrate the efficiency of $M^{4}$
and the superior performance in terms of generalization error, compared
to conventional SVM's, using data augmentation techniques with many
virtual examples. Ongoing work includes theoretical research to understand
how $M^{4}$ explicitly scales with the number of manifolds and the
embedding dimensionality.}

\textcolor{black}{There is natural extension of $M^{4}$ to nonlinear
classifiers via the kernel trick, as all our operations involve dot
products between the weight vector $\vec{w}$ and manifold points
$M_{p}(\vec{s})$. At each iteration, the algorithm would solve the
dual version of the $SVM_{T_{k}}$ problem which is readily kernelized.
In addition, the $M^{4}$ algorithm relies upon finding a point on
a manifold with sufficiently strong violation of the constraints.
Since a local minimization of the constraint violation at each stage
is sufficient in the relaxed version of the algorithm, we expect that
this step of $M^{4}$ will be practical for simpler kernel functions.
However, we note that with infinite-dimensional kernels such as RBF's,
the full manifold optimization problem becomes a fully infinite quadratic
programming problem, no longer a QSIP which requires further theoretical
work to establish the existence and properties of optimal solutions.}

\textcolor{black}{Beyond binary classification, variations of $M^{4}$
can also be used to solve other machine learning problems including
multi-class classification, ranking, one-class learning, etc. In this
work, we have shown how $M^{4}$ can be used to classify image manifolds
at pixel input representations. We can also use this algorithm to
evaluate the computational benefits of manifold representations at
successive layers of deep networks in both machine learning and in
brain sensory hierarchies. We also anticipate using $M^{4}$ to build
novel hierarchical architectures that can incrementally reformat the
manifold representations through the layers for better overall performance
in machine learning tasks.}

\textcolor{black}{We anticipate this work will make an important contribution
to the understanding of how neural architectures can learn to process
high dimensional real-world signal ensembles and cope with large variability
due to continuous modulation of the underlying physical parameters.}

%% file: chapters/ch4_genManifolds_v4.tex
\chapter{Linear Classification of General Manifolds \label{cha:genManifolds}}

\section{Introduction}

In chapter \ref{cha:spheres}, we applied methods from statistical
mechanics of spin glasses to solve the problem of linear, max margin,
classification of manifolds with simple geometries such as lines,
as well as $L_{2}$and $L_{p}$ balls embedded in embedded in a linear
subspace with dimensions $D$, where $D$ is much smaller than the
ambient dimension $N$. In chapter \ref{cha:m4}, we presented a new
efficient algorithm for finding max margin linear classifier of manifolds. 

In this chapter, we return to the theory and consider the problem
of linear classification of general manifolds, again with embedding
dimension much smaller than $N$. To set the stage of more complex
geometries, we begin by considering the classification of $D$-dimensional
$L_{2}$ ellipsoids. The results from the analysis of ellipsoids are
readily extended to the case general smooth convex manifolds. We then
move to consider classification of non-smooth low dimensional manifolds,
which exhibit a more complex solution structure. characterized by
variety of 'support' structures. Nevertheless, we derive a set of
mean field equations that apply to general low dimensional smooth
as well as non-smooth manifolds, including also manifolds consisting
of finite number of points (point clouds). We identify key geometric
descriptors of the manifolds: the effective manifold dimension $D_{M}$
and the effective manifold radius $R_{M}$ , two geometric 'order
parameters' which determine the capacity of linear classification
of general manifolds (when their dimensionality is high), and provide
an iterative algorithm that can efficiently solve for $D_{M}$ and
$R_{M}$ , as well as general manifold capacity $\alpha_{M}$ . Finally,
we emphasize that although the data manifolds may not be convex, any
hyperplane that separates them must also separate their \emph{convex
hull}. Hence, all geometric properties discussed in this chapter refer
to \emph{convex manifolds}. 

Note that in general, the capacity of manifolds embedded in $D$ dimension
can be upper and lower bounded by 

\begin{equation}
\frac{2}{1+2D}<\alpha_{c}<2
\end{equation}

This is because in the limit where extents of a manifold in all $D$
embedding dimensions go to zero, a manifold becomes a point, whose
perceptron capacity is $\alpha_{0}=2$ \cite{gardner1988space}. In
the limit where extents of a manifold in all $D$ embedding dimensions
go to infinity, then the linear classifier $\mathbf{w}$ has be in
the subspace orthogonal to all $D$ directions of the $P$ manifolds
\cite{chung2016linear}. Since $P$ of the $D$-dimensional manifolds
occupy $PD$ dimension, the classification becomes point classification
in $N-PD$ dimension, resulting in the maximum number of manifolds
linear separable $P=2(N-PD)$, resulting in the capacity $P/N=\frac{2}{1+2D}$.
These asymptotic bounds of a manifold capacity apply for arbitrary
manifold shapes. Now let us focus on simplest extension of classification
of $L_{2}$ balls, classification of $D$ -dimensional ellipsoids. 

\section{$L_{2}$ Ellipsoids }

\subsection{Model}

Consider the problem of linearly classifying $D$-dimensional ellipsoids
(Figure. \ref{fig:EllipsoidsIllustratoin}) in $N$-dimensional ambient
space, where each point within the $\mu$th ellipsoid is expressed
as 

\begin{multline}
\left\{ \mathbf{x}_{0}^{\mu}+\sum_{i=1}^{D}s_{i}\mathbf{u}_{i}^{\mu},y^{\mu}=\pm1\right\} \label{eq:xpoint}
\end{multline}

\begin{figure}
\begin{centering}
\includegraphics[scale=0.5]{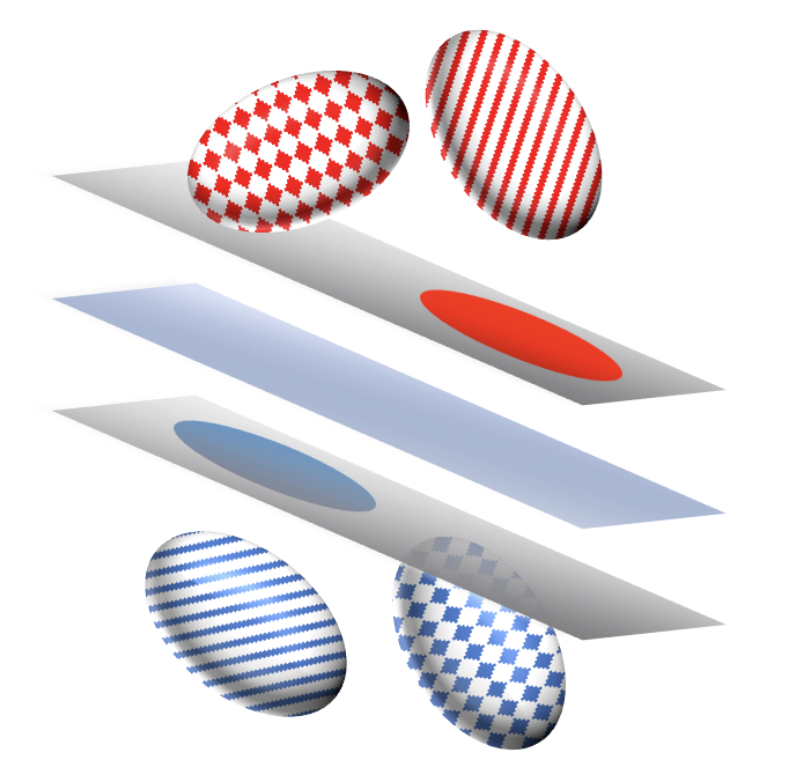}
\par\end{centering}
\caption{\textbf{Linear Classification of Ellipsoids (Illustration).} In $N$-dimensional
ambient space, the solution hyperplane (blue plane) has to separate
between red $D$-dimensional ellipsoids and blue ellipsoids with margin
$\kappa$. Margin $\kappa$ is the distance between the closest point
on the ellipsoids and the solution hyperplane. We refer to all points
of such distance to the solution hyperplane as ``margin planes''
(grey planes). Different patterns are used to denote different support
configuration of manifolds. Solid pattern: ellipsoids embedded in
the margin plane, diamond pattern: ellipsoids touching the margin
plane with one point, striped pattern: interior ellipsoids (ellipsoids
that are in the interior space shattered by margin planes). \label{fig:EllipsoidsIllustratoin}}
\end{figure}
For each $\mu$, $\mathbf{x}_{0}^{\mu}\in\mathbb{R}^{N}$ is N-dim
vector representing the center of the ellipsoid, the set of D N-dim
vectors, $\mathbf{u}_{i}^{\mu}\in\mathbb{R}^{N}$, for $i=1,...,D$
, are the ellipsoid's basis vectors. The vectors $\vec{s}\in\mathbb{R}^{D}$
parametrize the points on the manifolds and obey the constraint

\begin{equation}
f(\vec{s})\leq0\label{eq:f(s)}
\end{equation}
where, 

\begin{align}
f(\vec{s}) & =\sum_{i=1}^{D}s_{i}^{2}R_{i}^{-2}-1\label{eq:constraint_ellipsoid}
\end{align}
$R_{i}$ represent the ellipsoid's radii in the $i$ th direction,
relative to the center norm. In order to evaluate the ability of the
perceptron to classify the ellipsoids, we need to specify their statistical
properties. Here we assume that each component of $\mathbf{x}_{0}^{\mu}$,
$\mathbf{u}_{i}^{\mu}$ are independent Gaussian random variables
with unit variance. With these assumptions, and assuming large $N$,
the norm of the centers is (approximately) $\sqrt{N}$ and the $\mathbf{u}_{i}$'s
are (approximately) orthogonal vectors with norms $\sqrt{N}$ .

We assume the ellipsoids are assigned binary labels (which are therefore
the same for all points on the ellipsoid) denoted as $y^{\mu}=\pm1$
. We search of a set of weight vectors $\mathbf{w}\in\mathbb{R}^{N}$
that obey the following inequalities,
\begin{equation}
y^{\mu}\mathbf{w}^{T}\left(\mathbf{x_{\mathbf{0}}}^{\mu}+\sum_{i=1}^{D}s_{i}\mathbf{u}_{i}^{\mu}\right)\geq\kappa\left\Vert \mathbf{w}\right\Vert \quad\forall\vec{s},f(\vec{s})\leq0\label{eq:ineq}
\end{equation}
The maximum $\kappa$ that admits a solution $\mathbf{w}$ will be
called the margin of the system. Here we assume the labels for the
manifolds are assigned randomly i.i.d with probability half for $y^{\mu}=\pm1$
. The case where the fraction of positive and negative labels are
not equal (sparse labels) will be covered in Chapter \ref{cha:extensions}. 

\subsection{Fields of the Closest Point}

To classify all the points on the ellipsoids correctly, a necessary
and sufficient condition is that the weight vector $\mathbf{w}$,
satisfies the constraints on the 'worst' points on each ellipsoid
, namely the ones which are closest to the separating plane defined
by $\mathbf{w}$. To find this point for the $\mu$th manifold we
define the fields $h_{0}^{\mu}=\left\Vert \mathbf{w}\right\Vert ^{-1}y^{\mu}\mathbf{w}\cdot\mathbf{x_{0}^{\mu}}$,
which are the field (the protection on $\mathbf{w}$) induced by the
manifold's center, and $h_{i}^{\mu}=\left\Vert \mathbf{w}\right\Vert ^{-1}y^{\mu}\mathbf{w}\cdot\mathbf{u}_{i}^{\mu}\:i=1,...,D$
, which are the $D$ fields induced by the basis vectors of the manifold.
In terms of these fields, \ref{eq:ineq} can be written as

\begin{equation}
h_{0}^{\mu}+\Phi(\vec{h}^{\mu})\geq0\label{eq:phi}
\end{equation}
where

\begin{equation}
\Phi(\vec{h})=\tilde{s}(\vec{h})\cdot\vec{h}-\kappa\label{eq:minS}
\end{equation}
where

\begin{equation}
\tilde{s}(\vec{h})=\text{arg}\min_{\vec{s},\,f(\vec{s})=0}\{\vec{s}\cdot\vec{h}\}\label{eq:argminS}
\end{equation}
gives the point on the manifold which has the smallest (signed) projection
on the hyperplane $\mathbf{w}$. In order to evaluate $\tilde{s}(\vec{h})$,
we differentiate $\sum_{i=1}^{D}s_{i}h_{i}+\lambda f(\vec{s}$) with
respect to $s_{i}$, where $\lambda$ is a Lagrange multiplier enforcing
the manifold constraint, yielding in the case of ellipsoids,

\begin{equation}
\tilde{s}_{i}=-\frac{h_{i}R_{i}^{2}}{\sqrt{\sum_{j}h_{j}^{2}R_{j}^{2}}}=-\frac{h_{i}R_{i}^{2}}{\left\Vert \vec{h}\circ\vec{R}\right\Vert }\label{eq:stilde}
\end{equation}
where $\circ$ denotes element-wise product, and $||\vec{h}\circ\vec{R}||$
is the $L_{2}$ norm of the $D$-dimensional vector whose components
are $\{h_{i}R_{i}$\}. To evaluate $\Phi$, we note that $\tilde{s}\cdot\vec{h}=-\left\Vert \vec{h}\circ\vec{R}\right\Vert $,
hence, 

\begin{equation}
\Phi(\vec{h})=-\left\Vert \vec{h}\circ\vec{R}\right\Vert -\kappa\label{eq:L2 Inequality}
\end{equation}

\subsection{Mean Field Theory}

We consider a thermodynamic limit where $N,\,P\rightarrow\infty$
whereas $\alpha=P/N$$,$ D, and $\vec{R}$ are finite. Following
Gardner's framework, we compute the average of $\text{ln}V$ , where
$V$ is the volume of the space of the solutions, which in our case,
can be written as: 
\begin{equation}
V=\int d^{N}\mathbf{}\mathbf{w_{\alpha}}\delta(\mathbf{w}^{2}-N)\Pi_{\mu=1}^{P}\Theta_{\mu}(h_{0}^{\mu}+\Phi(\vec{h}^{\mu}))\label{eq:V}
\end{equation}
where $\Theta$ is the Heaviside function. We use replica theory,
$\langle\ln V\rangle=\lim_{n\rightarrow0}\frac{\langle V^{n}\rangle-1}{n}$,
where $\langle\rangle$ refers to the average over the 'quenched random
variables: the input parameters, $\mathbf{x_{0}}^{\mu}$ and $\mathbf{u}^{\mu}$
and the labels, to evaluate $\langle\ln V\rangle$ via the replica
symmetric saddle point equations. The saddle point approximation is
exact in the thermodynamic limit and the replica symmetric ansatz
holds for convex problems such as ours. These equations are expressed
in terms of the order parameter, $q=\frac{1}{N}\mathbf{\langle w}_{\alpha}\cdot\mathbf{w}_{\beta},\,\alpha\neq\beta$,
where $\mathbf{w}_{\alpha}$ and $\mathbf{w}_{\beta}$ are two typical
solutions of the classification problem.

The 'free energy' $G$ associated with $\langle V^{n}\rangle$ is
given by, 

\begin{equation}
\langle V^{n}\rangle_{x_{0},u,y}\sim e^{Nn\left[G(q)\right]}=e^{Nn\left[G_{0}(q)+\alpha G_{1}(q)\right]}\label{eq:G-1}
\end{equation}
where,

\begin{equation}
G_{0}(q)=\frac{1}{2}\ln(1-q)+\frac{q}{2(1-q)}\label{eq:G0-1}
\end{equation}
is the entropic term representing the volume of $\mathbf{w}$ subject
to the constraint that $q=\frac{1}{N}\mathbf{w}_{\alpha}\cdot\mathbf{w}_{\beta}$
. $G_{1}(q)$ embodies the constraints imposed by the classification
task and is expressed in terms of the fields $h_{0}^{\mu}$ and $\vec{h}^{\mu}$.
In the considered limit, these fields can be written as sums of two
random fields, where $t_{0}$ and $\vec{t}$ are the quenched component
resulting from the quenched random variables, namely the input vectors
$\mathbf{x}_{0}^{\mu}$ and $\mathbf{u}_{i}^{\mu}$, while the $z_{0}$and
$\vec{z}$ are the fields representing the variability of different
$\mathbf{w}$'s within the volume of solutions for each realization
of inputs and labels: 

\begin{equation}
h_{0}^{\mu}=\sqrt{q}t_{0}^{\mu}+\sqrt{1-q}z_{0}^{\mu},\;\vec{h}^{\mu}=\sqrt{q}\vec{t}^{\mu}+\sqrt{1-q}\vec{z}^{\mu}
\end{equation}

and, 

\begin{equation}
G_{1}(q)=\langle\text{ln}Z(q,t_{0,}\vec{t})\rangle_{t_{0},\vec{t}}\label{eq:G1-1}
\end{equation}
where the average wrt $t_{0}$, $\vec{t}$ denotes integrals over
the gaussian variables $t_{0}$, $\vec{t}$ with measures $Dt_{0}$
and $D\vec{t}=\pi_{i}Dt_{i}$, respectively, and 

\begin{equation}
Z(q,t_{0},t)=\int_{-\infty}^{\infty}Dz_{0}\int_{-\infty}^{\infty}D\vec{z}\Theta\left[\left(\sqrt{q}t_{0}+\sqrt{1-q}z_{0}\right)+\Phi\left(\sqrt{q}\vec{t}+\sqrt{1-q}\vec{z}\right)\right]\label{eq:Z}
\end{equation}
Finally, $q$ is determined by solving $\frac{\partial G}{\partial q}=0$
. Solution with $q<1$ indicates a finite volume of solutions. For
each $\kappa$ there is a maximum value of $\alpha$ where a solution
exists. As $\alpha$ approaches this maximal value, $q\rightarrow1$
indicating the existence of a unique solution, which is the max margin
solution for this $\alpha$. We focus on the properties of the \textit{max
margin} solution, i.e., on the limit $q\rightarrow1.$

\subsection{The Capacity Limit}

We define

\begin{equation}
Q=\frac{q}{1-q}
\end{equation}
and study the limit of $Q\rightarrow\infty$. In this limit, the leading
order for $G_{0}$ term is $G_{0}=\frac{Q}{2}$ and $G_{1}$ can be
evaluated by a saddle point approximation of the $z_{0}$and $\vec{z}$
integrals, 

\begin{equation}
\ln Z(t_{0},\vec{t})=-\min_{z_{0},\vec{z},\sqrt{Q}t_{0}+z_{0}+\Phi\left(\sqrt{Q}\vec{t}+\vec{z}\right)>0}\frac{1}{2}\left[z_{0}^{2}+\left\Vert \vec{z}\right\Vert ^{2}\right]
\end{equation}
Scaling the variables $z_{0}$ and $z$ such that $z_{0}\rightarrow\sqrt{Q}z_{0}$
and $z\rightarrow\sqrt{Q}z$ and using the fact that $\Phi(\vec{h})$
is linear in the magnitude of $\vec{h}$ to write $\text{\ensuremath{\ln}}Z(t_{0},\vec{t})=-\frac{Q}{2}F(t_{0},\vec{t})$
, 

\begin{equation}
F(t_{0,}t)=\min_{z_{0},\vec{z},t_{0}+z_{0}+\Phi\left(\vec{t}+\vec{z}\right)>0}\left[z_{0}^{2}+\left\Vert \vec{z}\right\Vert ^{2}\right]\label{eq:logZSpeheres-1-1}
\end{equation}
where $\Phi(\vec{h})=-\left\Vert \vec{h}\circ\vec{R}\right\Vert -\kappa$
(Eqn. \ref{eq:L2 Inequality}). Finally, 

\begin{equation}
\langle\ln V\rangle=\frac{Q}{2}\left[1-\alpha\langle F(t_{0},\vec{t})\rangle_{t_{0},\vec{t}}\right]\label{eq:logVLargeQ}
\end{equation}
so the capacity, defined by vanishing $\langle\ln V\rangle$is given
by,

\begin{equation}
\alpha_{E}^{-1}(\kappa)=\langle F(t_{0},\vec{t})\rangle_{t_{0},\vec{t}}\label{eq:EllipsoidCapacity}
\end{equation}
where the subscript $E$ stands for ellipsoids. For each $\vec{t}$
the nature of the solution to the $\min$ operation in \ref{eq:EllipsoidCapacity}
depends on $t_{0}$ yielding three regimes of $t_{0}$ with qualitatively
different contributions to the capacity, as described below. 

\subsubsection{Regime 1 (Interior Manifolds): $t_{0}-\kappa>\left|\vec{T}\right|$}

where,

\begin{equation}
\vec{T}=\vec{t}\circ\vec{R}\label{eq:T}
\end{equation}
In this case, the solution is $z_{0}=z=0$ and does not contribute
to Eq. \ref{eq:EllipsoidCapacity}. 

For values of $t_{0}-\kappa\leq\left|\vec{T}\right|$, the solution
obeys

\begin{equation}
t_{0}+z_{0}+\Phi\left(\vec{t}+\vec{z}\right)=0\label{eq:t0etc}
\end{equation}
meaning that the closest point is on the margin plane. This regimes
is divided into two cases: 

\subsubsection{Regime 2 (Touching Manifolds): $t_{C}<t_{0}-\kappa<\left|\vec{T}\right|$}

where,

\begin{equation}
t_{C}=-\sqrt{\sum_{i}R_{i}^{-2}t_{i}^{2}}\label{eq:tc}
\end{equation}
Here, $t_{0}+z_{0}+\Phi\left(\vec{t}+\vec{z}\right)=0$ but $h_{0}=t_{0}+z_{0}>\kappa$,
implying that the ellipsoid center is an interior point; in other
words, the ellipsoid touches the margin plane only at a single point.
Thus, for a given $t_{0}$ and $\vec{t}$ we need to solve

\begin{equation}
\min_{\vec{z}}\left[z_{0}^{2}+\left\Vert \vec{z}\right\Vert ^{2}\right]
\end{equation}
where $z_{0}=-t_{0}-\Phi\left(\vec{t}+\vec{z}\right)$. Differentiating
with respect to $\vec{z}$ yields, $\vec{z}=z_{0}\partial_{\vec{z}}\Phi=z_{0}\partial_{\vec{h}}\{\vec{s}\cdot\vec{h}\},$
namely,

\textcolor{black}{
\begin{equation}
\vec{z}=z_{0}\vec{s}\label{eq:z0s}
\end{equation}
 where from now on, unless otherwise specified, $\vec{s}$will be
a shorthand of $\tilde{s}(\vec{h})=\tilde{s}(\vec{t}+\vec{z}).$ Note
that this is a self consistent equation for $\vec{s}$ due to }\ref{eq:z0s}\textcolor{black}{.
This yields also, $\Phi=\vec{s}\cdot(\vec{t}+z_{0}\vec{s})-\kappa$,
hence}

\textcolor{black}{
\begin{equation}
z_{0}=\frac{(\kappa-t_{0}-\vec{t}\cdot\vec{s})}{(1+s^{2})}\label{eq:phi2}
\end{equation}
Finally, $z_{0}^{2}+z^{2}=z_{0}^{2}(1+s^{2})$ yielding,}

\textcolor{black}{
\begin{equation}
F(t_{0,}t)=\frac{(\kappa-\vec{t}\cdot\vec{s}-t_{0})^{2}}{1+s^{2}}\label{eq:F}
\end{equation}
}To conclude the evaluation of $F$ we need to calculate $\vec{s}$.
Eq, \ref{eq:stilde}, for the ellipsoid, yields,
\begin{equation}
s_{i}=-\frac{H_{i}R_{i}}{\left\Vert \vec{H}\right\Vert }
\end{equation}

with $\vec{H}\equiv\vec{h}\circ\vec{R}=(\vec{t}+z_{0}\vec{s}$$)\circ\vec{R}$.
Substituting in the above equation, one obtains,

\begin{equation}
s_{i}=-\frac{R_{i}^{2}t_{i}}{||\vec{H}||+z_{0}R_{i}^{2}}\label{eq:si}
\end{equation}
which yields an equation of $\vec{s}(\vec{t})$ in terms of $||\vec{H}||$
and $z_{0}$. These two scalars are related through 

\begin{equation}
z_{0}=-t_{0}-\Phi\left(\vec{t}+\vec{z}\right)=-\kappa-t_{0}+||\vec{H}||\label{eq:z0}
\end{equation}
where $||\vec{H}||=\vec{s}\cdot(\vec{t}+z_{0}\vec{s})$ and $\Phi$
is given by Eqn. \ref{eq:L2 Inequality}. Finally, an equation for
$z_{0}$ can be derived from the ellipsoid constraint $f(\vec{s})=0$, 

\begin{equation}
1=\sum_{i}s_{i}^{2}R_{i}^{-2}\label{eq:s2}
\end{equation}

To summarize, Eqns. \ref{eq:si} -\ref{eq:s2} yields $\vec{s}(\vec{t},t_{0})$
which we use to evaluate $F$ , Eq. \ref{eq:F}.

\subsubsection{Regime 3 (Embedded Manifolds): $t_{0}-\kappa<t_{C}$}

Here $\vec{h}=\vec{t}+\vec{z}=0$, and $h_{0}=t_{0}+z_{0}=\kappa,$
implying that the center as well as the entire manifold is on the
margin plane, hence 

\begin{equation}
F(t_{0},\vec{t})=\left(t_{0}-\kappa\right)^{2}+\left\Vert \vec{t}\right\Vert ^{2}
\end{equation}

Finally, combining contributions from regimes 2 and 3, the expression
of the capacity is 

\noindent\fbox{\begin{minipage}[t]{1\columnwidth \fboxsep \fboxrule}%
\begin{equation}
\alpha_{E}^{-1}(\kappa)=\int D\vec{t}\,\int_{\kappa+t_{C}(\vec{t})}^{\kappa+\vec{|T|}}Dt_{0}\left[\frac{(\kappa-\vec{t}\cdot\vec{s}-t_{0})^{2}}{1+s^{2}}\right]+\int D\vec{t}\,\int_{-\infty}^{\kappa+t_{C}(\vec{t})}Dt_{0}\left[\left(t_{0}-\kappa\right)^{2}+\left\Vert \vec{t}\right\Vert ^{2}\right]\label{eq:EllipsoidCapacityExpression}
\end{equation}

\begin{equation}
\vec{|T}|=\sqrt{\sum_{i}R_{i}^{2}t_{i}^{2}}\label{eq:T-1-1}
\end{equation}

\begin{equation}
t_{C}=-\sqrt{\sum_{i}R_{i}^{-2}t_{i}^{2}}\label{eq:tc-1-1}
\end{equation}

In the first integral, $\vec{s}$ is given by,

\begin{equation}
s_{i}=-\frac{R_{i}^{2}t_{i}}{\kappa+t_{0}+z_{0}(1+R_{i}^{2})}\label{eq:si_elps_analytic}
\end{equation}

and $z_{0}(\vec{t},t_{0})$ is evaluated by solving, 

\begin{equation}
1=\sum_{i}\frac{R_{i}^{2}t_{i}^{2}}{(\kappa+t_{0}+z_{0}(1+R_{i}^{2}))^{2}}\label{eq:z0_elps_analytic}
\end{equation}
\end{minipage}} 

\subsection{The Large D limit}

If the size of the ellipsoid is not small, we expect the capacity
to be small (of order $1/D$, see Eqn. \ref{eq:alphaEdE_largeR}).
On the other hand, when the radii are small the capacity should be
order 1 as in the case of points. We inquire how small $R_{i}$'s
should be in order to yield a finite capacity even when $D$ is large.
The answer is provided by a scaling analysis, below.

\subsubsection{Large $D$ , $R_{i}=O(1)$}

In this limit, $||\vec{T|}|,-t_{C}=O(D^{1/2})$, so integral bounds
in the first term of \ref{eq:EllipsoidCapacityExpression} can be
taken to $\pm\infty.$ From \ref{eq:s2}, $s_{i}=O(D^{-1/2}$) and
from Eqns \ref{eq:z0}, $z_{0}\approx||\vec{H}||=O(D^{1/2})$ . 

\begin{equation}
s_{i}=-\frac{R_{i}^{2}t_{i}}{z_{0}(1+R_{i}^{2})}
\end{equation}

and from the normalization,

\begin{equation}
z_{0}^{2}\approx\langle z_{0}^{2}\rangle=\sum_{i=1}^{D}\frac{R_{i}^{2}t_{i}^{2}}{(1+R_{i}^{2})^{2}}\approx\sum_{i=1}^{D}\frac{R_{i}^{2}}{(1+R_{i}^{2})^{2}}\label{eq:z0_LargeD_RO1}
\end{equation}

where we have replaced $t_{i}^{2}\approx1$ under the summation. Similarly,
\begin{equation}
s^{2}\approx\langle s^{2}\rangle=\frac{1}{z_{0}^{2}}\sum_{i=1}^{D}\frac{R_{i}^{4}}{(1+R_{i}^{2})^{2}}=O(1)
\end{equation}

\begin{equation}
\vec{t}\cdot\vec{s}\approx\langle\vec{t}\cdot\vec{s}\rangle=-\frac{1}{z_{0}}\sum_{i=1}^{D}\frac{R_{i}^{2}}{(1+R_{i}^{2})}=O(D^{1/2})
\end{equation}

Hence,

\begin{equation}
\alpha_{E}^{-1}\approx\frac{\langle\vec{t}\cdot\vec{s}\rangle^{2}}{1+\langle s^{2}\rangle}\quad\text{when }D\gg1,R_{i}=O(1)
\end{equation}

which is of order $D$ as expected. 

\textbf{Effective Dimensionality and Radius:} These results suggest
to express the capacity by introducing the ellipsoid effective dimension
($D_{E}$) and radius ($R_{E}$), as follows,

\begin{equation}
\alpha_{E}^{-1}=\frac{R_{E}^{2}D_{E}}{1+R_{E}^{2}}\quad\text{when }D\gg1,R_{i}=O(1)\label{eq:alphaE_highD}
\end{equation}

\begin{equation}
R_{E}^{2}=\langle s^{2}\rangle=\sum_{i}\frac{R_{i}^{4}}{(1+R_{i}^{2})^{2}}/\sum_{j}\frac{R_{j}^{2}}{(1+R_{j}^{2})^{2}}\quad\text{when }D\gg1,R_{i}=O(1)\label{eq:RE_Rorder1}
\end{equation}

\begin{equation}
D_{E}=\left(\sum_{i}\frac{R_{i}^{2}}{1+R_{i}^{2}}\right)^{2}/\sum_{i}\frac{R_{i}^{4}}{(1+R_{i}^{2})^{2}}\quad\text{when }D\gg1,R_{i}=O(1)\label{eq:DE_Rorder1}
\end{equation}

Thus, \emph{the capacity of ellipsoids in large $D$ is equivalent
to that of $L_{2}$balls with radii $R_{E}$and dimensionality $D_{E}.$}

\subsubsection{Large $D$, Large $R$ Regime }

Finally, when most of the $R_{i}$ are large, $R_{E}\gg1$ and 

\begin{equation}
\alpha_{E}^{-1}=D_{E}=D\quad\text{when }D\gg1,R_{i}\gg1\label{eq:alphaEdE_largeR}
\end{equation}

In this case, $\mathbf{w}$ is orthogonal to the basis vectors with
large $R_{i}$. 

\subsubsection{Scaling Regime: Large $D$, $R_{i}\propto D^{-1/2}$}

The above results suggest that when the radii are small such that,
$R_{E}\propto D_{E}^{-1/2}$ the capacity becomes order 1. Thus, the
scaling relation $R_{i}\propto D^{-1/2}$ denotes the regime of finite
capacity, namely the balance between large dimension and small size
that maintains a finite capacity. This regime requires its own analysis
of the various terms that contributes to the capacity. First, 

\begin{equation}
||\vec{T|}|\approx||\vec{R}||=O(1)
\end{equation}

\begin{equation}
-t_{C}=O(D^{1/2})
\end{equation}
So the integral bounds in the first term of \ref{eq:EllipsoidCapacityExpression}
is from $-\infty$ to $\kappa+||\vec{R}||$ and the second term is
negligible. From \ref{eq:s2}, $s_{i}=O(D^{-1})$, and $||\vec{H}||=\sqrt{\sum_{i}h_{i}^{2}R_{i}^{2}}=O(1)$
and from Eqns \ref{eq:z0}, $z_{0}\approx\kappa-t_{0}+||\vec{H}||=O(1)$.
Hence, $z_{0}R_{i}^{2}=O(D^{-1})$. Then, from \ref{eq:si}, 

\begin{equation}
s_{i}\approx-\frac{R_{i}^{2}t_{i}}{||\vec{H}||}=O(D^{-1})
\end{equation}

as expected. 

And from this normalization, 

\begin{equation}
s^{2}\approx\langle s^{2}\rangle=\frac{1}{||\vec{H}||^{2}}\sum_{i}R_{i}^{4}t_{i}^{2}=\frac{1}{||\vec{H}||^{2}}\sum_{i}R_{i}^{4}=O(1)O(D*D^{-2})=O(D^{-1})
\end{equation}

\begin{equation}
\vec{t}\cdot\vec{s}\approx\left\langle \vec{t}\cdot\vec{s}\right\rangle =-\sum_{i}\frac{R_{i}^{2}}{||\vec{H}||}=-\frac{||\vec{R}||^{2}}{||\vec{H}||}=0(1)
\end{equation}

and from normalization, 

\begin{equation}
1=\sum_{i}s_{i}^{2}R_{i}^{-2}\approx\frac{1}{||\vec{H}||^{2}}\sum_{i}R_{i}^{2}\label{eq:s2-1}
\end{equation}

implying $||\vec{H}||=||\vec{R}||$ and $\vec{t}\cdot\vec{s}=-||\vec{R}||$
. Hence,

\begin{equation}
\alpha_{E}^{-1}=\frac{\int_{-\infty}^{\kappa+||\vec{R}||}Dt_{0}(\kappa+||\vec{R}||-t_{0})^{2}}{1+\left\langle s^{2}\right\rangle }
\end{equation}

Although $\left\langle s^{2}\right\rangle =||\vec{R}\circ\vec{R}||^{2}/||\vec{R}||^{2}$
is a correction of order $D^{-1}$ , we will keep it because it turns
out to be important to keep in simulations. 

We can express these results in terms of the effective dimensionality
and radius introduced above. In the limit of small $R$s these quantities
reduce to,

\begin{equation}
R_{E}^{2}=\frac{\sum_{i}R_{i}^{4}}{\sum_{i}R_{i}^{2}}=O(D^{-1})\quad\text{when }D\gg1,R_{i}\lesssim O\left(D^{-1/2}\right)\label{eq:RE_elps_scale}
\end{equation}

\begin{equation}
D_{E}=\frac{(\sum_{i}R_{i}^{2})^{2}}{\sum_{i}R_{i}^{4}}=O(D)\quad\text{when }D\gg1,R_{i}\lesssim O\left(D^{-1/2}\right)\label{eq:DE_elps_scale}
\end{equation}

and, 

\begin{equation}
||\vec{R}||^{2}=R_{E}^{2}D_{E}
\end{equation}

Using these quantities, and the formula for the point capacity, $\alpha_{0}^{-1}(\kappa)=\int_{-\infty}^{\kappa}Dt(t-\kappa)^{2},$
we can write,

Capacity for Ellipsoids: 

\begin{equation}
\alpha_{E}(\kappa)=(1+R_{E}^{2})\alpha_{0}(\kappa+R_{E}\sqrt{D_{E}})\quad\text{when }D\gg1,R_{i}\lesssim O\left(D^{-1/2}\right)\label{eq:alphacScaling}
\end{equation}

with $R_{E}$ and $D_{E}$ are defined by Eqns \ref{eq:RE_elps_scale}-\ref{eq:DE_elps_scale}
when $D\gg1,R_{i}\lesssim D^{-1/2}$, and $R_{E}\sqrt{D_{E}}$ behaves
like an additional margin $\kappa_{E}=R_{E}\sqrt{D_{E}}$ introduced
by the ellipsoid structure. 

Interestingly, in the scaling regime, the effective dimension for
the ellipsoids is equivalent to another measure of dimension, $D_{svd}$,
called the participation ratio \cite{amir2016non,rajan2010stimulus,litwin2017optimal},
defined by 

\begin{equation}
D_{svd}=\frac{\left(\sum_{i=1}^{D}\lambda_{i}^{2}\right)^{2}}{\sum_{i=1}^{D}\lambda_{i}^{4}}\label{eq:Dsvd}
\end{equation}
 where $\lambda_{i}$ is an eigenvalue of single value decomposition
(not normalized), whereas for$D_{E}$, $R_{i}$ is a radius in $i$
th dimension of an ellipsoid. $R_{i}$ and $\lambda_{i}$ are closely
related, as both definitions are measures for how extended data are
in the $i$th dimension. Particularly when $s_{i}R_{i}^{-1}$ are
uniformly sampled from a sphere, then $R_{i}$ and $\lambda_{i}$
are proportional to each other (Lemma \ref{lem:RiLambdai} in Appendix
to the chapter) . 

\subsubsection{Combined Expression for the Capacity in Large $D$}

Finally, we note that we can combine the results for all the above
regimes can be expressed by a single set of equations. 

\noindent\fbox{\begin{minipage}[t]{1\columnwidth \fboxsep \fboxrule}%
For large $D$, 

\begin{equation}
\alpha_{E}(\kappa)=(1+R_{E}^{2})\alpha_{0}(\kappa+\kappa_{E})\label{eq:alphacLargeD}
\end{equation}

\begin{equation}
\kappa_{E}=R_{E}\sqrt{D_{E}}\label{eq:kappaE_LargeD}
\end{equation}

\begin{equation}
R_{E}^{2}=\langle s^{2}\rangle=\sum_{i}\frac{R_{i}^{4}}{(1+R_{i}^{2})^{2}}/\sum_{j}\frac{R_{j}^{2}}{(1+R_{j}^{2})^{2}}\label{eq:RE_LargeD}
\end{equation}

\begin{equation}
D_{E}=\left(\sum_{i}\frac{R_{i}^{2}}{1+R_{i}^{2}}\right)^{2}/\sum_{i}\frac{R_{i}^{4}}{(1+R_{i}^{2})^{2}}\label{eq:DE_LargeD}
\end{equation}
\end{minipage}}

Finally, we note that the definition of the ellipsoid dimension above,
$D_{E}$, is not invariant to a global scale of all radii, except
in the regime of small $R_{i}$ . The reason is that the separation
of the manifolds depend not only on their intrinsic geometry but also
on their distance from the common origin. Thus both dimensionality
and radii take into account the center norm. Indeed, the size scale
$1$ appearing in the definition of $D_{E}$ above is the scale of
the center norm. This is reflected in the numerical evaluation of
$D_{E}$ in Figure \ref{fig:EllipsoidsResults} below. 

\subsection{Support Manifold Structures }

It is instructive to consider the types of manifold support structures
that arise. In general, the fraction of touching ellipsoids is

\begin{equation}
p_{touching}^{E}=\int D\vec{t}\,\int_{\kappa+t_{C}(\vec{t})}^{\kappa+\vec{||T||}}Dt_{0}
\end{equation}

The fraction of embedded ellipsoids is

\begin{equation}
p_{embedded}^{E}=\int D\vec{t}\,\int_{-\infty}^{\kappa+t_{C}(\vec{t})}Dt_{0}
\end{equation}

The fraction of interior ellipsoids is

\begin{equation}
p_{interior}^{E}=\int D\vec{t}\,\int_{\kappa+\vec{||T||}}^{\infty}Dt_{0}
\end{equation}

\subsubsection{Large D Limit }

Here we consider the limit of large $D$, and analyze the behavior
of support structures in different regimes of $R_{i}$. 

\paragraph{Large $D$ , $R_{i}=O(1)$}

In this limit $||\vec{T|}|,-t_{C}=O(D^{1/2})$, so that 

\begin{equation}
p_{touching}^{E}=\int D\vec{t}\,\int_{\kappa+t_{C}(\vec{t})}^{\kappa+\vec{||T||}}Dt_{0}=\int D\vec{t}\,\int_{-\infty}^{\infty}Dt_{0}=1
\end{equation}
, and 
\begin{equation}
p_{embedded}^{E}=\int D\vec{t}\,\int_{-\infty}^{\kappa+t_{C}(\vec{t})}Dt_{0}=\int D\vec{t}\,\int_{-\infty}^{-\infty}Dt_{0}=0
\end{equation}
 and 
\begin{equation}
p_{interior}^{E}=\int D\vec{t}\,\int_{\kappa+\vec{||T||}}^{\infty}Dt_{0}=0
\end{equation}
implying that all of the manifolds are touching the hyperplane in
this regime.

\paragraph{Large $D$ , $R_{i}\propto D^{-1}$: Scaling Regime }

In this limit, $||\vec{T|}|\approx||\vec{R}||=O(1)$ and $-t_{C}=O(D^{1/2}).$
Therefore, 

\begin{equation}
p_{touching}^{E}=\int D\vec{t}\,\int_{\kappa+t_{C}(\vec{t})}^{\kappa+\vec{||T||}}Dt_{0}=\int D\vec{t}\,\int_{-\infty}^{\kappa+\vec{||T||}}Dt_{0}=\left\langle 1-H(\kappa+|\vec{t}\circ\vec{R}|\})\right\rangle _{\vec{t}}
\end{equation}

\begin{equation}
p_{embedded}^{E}=\int D\vec{t}\,\int_{-\infty}^{\kappa+t_{C}(\vec{t})}Dt_{0}=\int D\vec{t}\,\int_{-\infty}^{-\infty}Dt_{0}=0
\end{equation}

and 

\begin{equation}
p_{interior}^{E}=1-p_{embedded}^{E}-p_{touching}^{E}=\left\langle H(\kappa+|\vec{t}\circ\vec{R}|\})\right\rangle _{\vec{t}}
\end{equation}

implying that there is no manifold embedded, but most of the manifolds
are either touching the margin plane or in the interior space. 

\subsection{Remarks}

It is notable that the capacity of ellipsoids in the high $D$ limit
(Eqn. \ref{eq:alphaE_highD}) resembles that of $L_{2}$ balls (Eqn.
\ref{eq:alphaD}), with an effective dimension $D_{E}$ and radius
$R_{E}$. The support structures of the ellipsoids also behave similarly
to the spherical $L_{2}$ balls in Chapter \ref{cha:spheres}, exhibiting
three regimes of support structures (embedded, touching and interior)
and in high $D$, non of the manifolds are embedded in the margin
plane, and the fraction of touching manifolds increase like $1-H(\kappa+R\sqrt{D})$.
In the next section, this analogy extends to more general case of
arbitrary smooth convex manifolds, and we show the replica treatment
of the smooth convex manifolds. 

\section{General Smooth Convex Manifolds }

\subsection{Model }

We now consider the problem of linear binary classification of points
on convex smooth manifolds. We define a smooth convex manifold as
a compact convex manifold in Euclidean space with convex and twice
differentiable bounding curve. It is useful to parametrize such a
manifold as the set of all points, $\mathbf{x}$ in $\mathbb{R}^{N}$,
of the form 

\begin{equation}
\mathbf{x}_{0}+\sum_{i=1}^{D}s_{i}\mathbf{u}_{i}
\end{equation}
where $\mathbf{x}_{0}$ and $\mathbf{u}_{i}$ are (linearly independent)
vectors in $\mathbb{R}^{N}$ and the $D$-dimensional vector $\vec{s}$
obeys the constraint $f(\vec{s})\leq0,$ where $f$: $\mathbb{R}^{D}\rightarrow\mathbb{R}$
is a twice differentiable convex function. We will refer to $\mathbf{x}_{0}$
as the center of the manifold and to $\mathbf{u}_{i}$ as its $D$
axes. Examples of smooth and non-smooth convex manifolds are provided
(Smooth: Figures \ref{fig:EllipsoidsIllustratoin}-\ref{fig:GeneralSmooth_Illustration1},
Non-smooth: Figure \ref{fig:L1embedding}). Our data consists of $P$
such manifolds and their target binary labels denoted as $y^{\mu},\:\mu=1,...,P$.
We search of a set of weights $\mathbf{w}\in\mathbb{R}^{N}$ that
obey the following inequalities,
\begin{equation}
y^{\mu}\mathbf{w}^{T}\left(\mathbf{x_{\mathbf{0}}}^{\mu}+\sum_{i=1}^{D}s_{i}\mathbf{u}_{i}^{\mu}\right)\geq\kappa\left\Vert \mathbf{w}\right\Vert \quad\forall\vec{s},f(\vec{s})\leq0\label{eq:ineq-1}
\end{equation}
In order to evaluate the ability of the perceptron to classify the
manifolds, we need to specify their statistical properties. As before,
we assume that each component of $\mathbf{x}_{0}^{\mu}$, $\mathbf{u}_{i}^{\mu}$
are independent Gaussian random variables with unit variance. With
these assumptions, and assuming large $N$, the centers have norm
$\sqrt{N}$ and the $\mathbf{u}_{i}$'s are orthogonal vectors with
norms $\sqrt{N}$ . 

Similar to the replica calculation for ellipsoids, we consider the
thermodynamic limit $N,P\rightarrow\infty$. We assume the manifold
embedding dimension, $D$ is finite in the thermodynamic limit, and
that the function $f(\vec{s}$) is independent of $N$ . 

\subsection{Fields of the Closest Point }

Given $\mathbf{w}$, we define the fields induced by the centers $h_{0}^{\mu}$
and the fields induced by the basis vectors $\vec{h}^{\mu}$ as $h_{0}^{\mu}=\left\Vert \mathbf{w}\right\Vert ^{-1}y^{\mu}\mathbf{w}\cdot\mathbf{x_{0}^{\mu}}$,
which are the field induced by the manifold centers, and $h_{i}^{\mu}=\left\Vert \mathbf{w}\right\Vert ^{-1}y^{\mu}\mathbf{w}\cdot\mathbf{u}_{i}^{\mu}\:i=1,...,D$
. Using these fields we can express the constraints \ref{eq:ineq-1}
by Eqn \ref{eq:phi} and \ref{eq:minS} corresponding to the point
on the manifolds with the smallest projection on the margin hyperplane
of $\mathbf{w}$. The evaluation of $\tilde{s}(\vec{h})$ requires
the differentiation of$\sum_{i=1}^{D}s_{i}h_{i}+\lambda f(\vec{s}$)
with respect to $s_{i}$, where $\lambda$ is a Lagrange multiplier
enforcing the manifold constraint, yielding, 

\begin{equation}
h_{i}=-\lambda\partial_{s_{i}}f(\vec{s}),\ f(\vec{s})=0\label{eq:minSgenM}
\end{equation}
which needs to be solved for $\tilde{s}=\vec{s}$ and substitute in
$\Phi(\vec{h})=\tilde{s}(\vec{h})\cdot\vec{h}-\kappa$. The relation
between the vector $\vec{h}$and $\tilde{s}(\vec{h})$ is shown in
Fig. \ref{fig:GeneralSmooth_Illustration1}(a). 

\subsection{Mean Field Equations of the Capacity}

We use the replica theory to evaluate the limit where the volume of
solutions vanishes. Similar to the replica calculation of ellipsoids,
above, the equation capacity is given by 

\begin{equation}
\alpha_{M}^{-1}(\kappa)=\langle F(t_{0},\vec{t})\rangle_{t_{0},\vec{t}}
\end{equation}

where $M$ stands for manifolds 

\begin{equation}
F(t_{0,}t)=\min_{z_{0},\vec{z},t_{0}+z_{0}+\Phi\left(\vec{t}+\vec{z}\right)>0}\left[z_{0}^{2}+\left\Vert \vec{z}\right\Vert ^{2}\right]
\end{equation}

where 
\begin{equation}
\Phi(\vec{h})=\tilde{s}(\vec{h})\cdot\vec{h}-\kappa
\end{equation}
and $\tilde{s}$ is the parameterization of the point on the manifold
that is closest to the solution hyperplane characterized by $\vec{h}$
(minimizing $\Phi)$, given by \ref{eq:minSgenM} (Figure . \ref{fig:GeneralSmooth_Illustration1}(a)) 

\subsubsection{Regime 1 (Interior Manifolds): $t_{0}-\kappa>-\vec{t}\cdot\vec{s}(\hat{t})$}

In this regime, $z_{0}=z=0$ so that the fields $h_{0}$and $\vec{h}$are
simply $t_{0}$ and $\vec{t}$, $\vec{s}$=$\tilde{s}(\vec{t})$ and
$\Phi(\vec{t}+\vec{z})=\Phi(\vec{t})$. This regime corresponds to
the case where all manifolds are interior and do not contribute to
$F$ . The regime exists until the inequality $t_{0}+\Phi(\vec{t})\geq0$
becomes equality, i.e., $t_{0}-\kappa=$$-\vec{t}\cdot\vec{s}(\hat{t}).$ 

\subsubsection{Regime 2 (Touching Manifolds): $t_{C}<t_{0}-\kappa<-\vec{t}\cdot\vec{s}(\hat{t})$}

Here, $t_{0}+z_{0}+\Phi\left(\vec{t}+\vec{z}\right)=0$ but $h_{0}=t_{0}+z_{0}\neq\kappa$,
implying that the manifold' center is an interior point; in other
words, the manifold touches the margin plane only at a single point.
Thus, for a given $t_{0}$and $\vec{t}$ we need to solve

\begin{equation}
\min_{\vec{z}}\left[z_{0}^{2}+\left\Vert \vec{z}\right\Vert ^{2}\right]
\end{equation}

where $z_{0}=-t_{0}-\Phi\left(\vec{t}+\vec{z}\right)$. Differentiating
with respect to \textcolor{black}{$\vec{z}$ yields, $\vec{z}=z_{0}\partial_{\vec{z}}\Phi=z_{0}\partial_{\vec{h}}\{\vec{s}\cdot\vec{h}\},$
namely,}

\textcolor{black}{
\begin{equation}
\vec{z}=z_{0}\vec{s}
\end{equation}
}

\textcolor{black}{(where we have changed notation from $\tilde{s}$to
$\vec{s}$). This yields also, $\Phi=\vec{s}\cdot(\vec{t}+z_{0}\vec{s})-\kappa$,
hence}

\textcolor{black}{
\begin{equation}
z_{0}=(\kappa-t_{0}-\vec{t}\cdot\vec{s})(1+s^{2})^{-1}\label{eq:phi2-1}
\end{equation}
}

\textcolor{black}{Finally, $z_{0}^{2}+z^{2}=z_{0}^{2}(1+s^{2})$ yielding,}

\textcolor{black}{
\begin{equation}
F(t_{0,}t)=\frac{(\kappa-\vec{t}\cdot\vec{s}-t_{0})^{2}}{1+s^{2}}\label{eq:F-1}
\end{equation}
}

\textcolor{black}{
\begin{equation}
z_{0}=-\left(t_{0}-\kappa\right)+\Phi(\vec{z}+\vec{t})
\end{equation}
}

This regime holds as long as the interior fields $\vec{t}+\vec{z}$
are non zero. The lower limit of this regime is when $t_{0}$ is such
that these fields vanish, i.e.,$\vec{z}\rightarrow-\vec{t}$, hence,

\begin{equation}
\vec{t}+z_{0}\vec{s}(\vec{h})\rightarrow0
\end{equation}
so that $\vec{s}$itself is antiparallel to $\vec{t}$, 

\begin{equation}
\vec{s}(\vec{h})=-z_{0}^{-1}\vec{t}\label{eq:regime3_s}
\end{equation}
and where $z_{0}=\kappa-t_{0}$ , hence $t=-z_{0}s$ yielding for
the lower limit of this regime, 

\begin{equation}
t_{C}=\kappa-\frac{t}{s_{C}(\vec{t})}
\end{equation}
where, $s_{C}(\vec{t})$ is the magnitude of the point $\vec{s}^{*}(\vec{h})$
defined by a vector $\vec{h}$ such that $\vec{s}^{*}(\vec{h})$ is
parallel to $\vec{t}$. Thus, $s_{C}(\vec{t})$ is simply the magnitude
of the intersection of $\vec{t}$ with the manifold, Figure \ref{fig:GeneralSmooth_Illustration1}(b). 

\begin{equation}
F=\frac{(-\vec{t}\cdot\vec{s}-\left(t_{0}-\kappa\right))^{2}}{1+s^{2}}
\end{equation}

where the $D$-dim vector $\vec{s}$ has to be calculated self-consistently
through, 
\begin{equation}
\vec{s}=\tilde{s}(\vec{t}-z_{0}\vec{s})
\end{equation}
\begin{equation}
z_{0}=\frac{-\vec{t}\cdot\vec{s}+\kappa-t_{0}}{1+s^{2}}
\end{equation}

\subsubsection{Regime 3 (Embedded Manifolds): $t_{0}-\kappa<-t_{C}$}

Here $\vec{h}=\vec{t}+\vec{z}=0$, and $h_{0}=t_{0}+z_{0}=\kappa,$
implying that the center point as well as the entire manifold is on
the margin plane, hence 

\begin{equation}
F(t_{0},\vec{t})=\left(t_{0}-\kappa\right)^{2}+\left\Vert \vec{t}\right\Vert ^{2}
\end{equation}

Putting results from the two regimes, we get: 

\noindent\fbox{\begin{minipage}[t]{1\columnwidth \fboxsep \fboxrule}%
\begin{equation}
\alpha_{M}^{-1}=\int D\vec{t}\int_{\kappa-t/s_{C}}^{\kappa-\vec{t}\cdot\vec{s}(\vec{t})}Dt_{0}\frac{(-\vec{t}\cdot\vec{s}-t_{0}+\kappa)^{2}}{1+s^{2}}+\int D\vec{t}\int_{-\infty}^{\kappa-t/s_{C}}Dt_{0}([t_{0}-\kappa]^{2}+t^{2})\label{eq:alpha_M}
\end{equation}

where,
\begin{equation}
\vec{s}=\tilde{s}(\vec{t}-z_{0}\vec{s})\label{eq:s_min}
\end{equation}
\begin{equation}
z_{0}=\frac{-\vec{t}\cdot\vec{s}-t_{0}+\kappa}{1+s^{2}}\label{eq: z0_smin}
\end{equation}

and $s^{2}=||\vec{s}||^{2}.$$s_{C}$ is the magnitude of the intersection
of $\vec{t}$ with the manifold. 

Here, $\tilde{s}(\vec{h})$ is defined via

\begin{equation}
\tilde{s}(\hat{h})=\arg\text{min}_{s'}\vec{s}'\cdot\hat{h},\:f(\vec{s}')=0\label{eq:s_htilde}
\end{equation}
\end{minipage}}

\begin{figure}[h]
\begin{centering}
\includegraphics[width=0.95\textwidth]{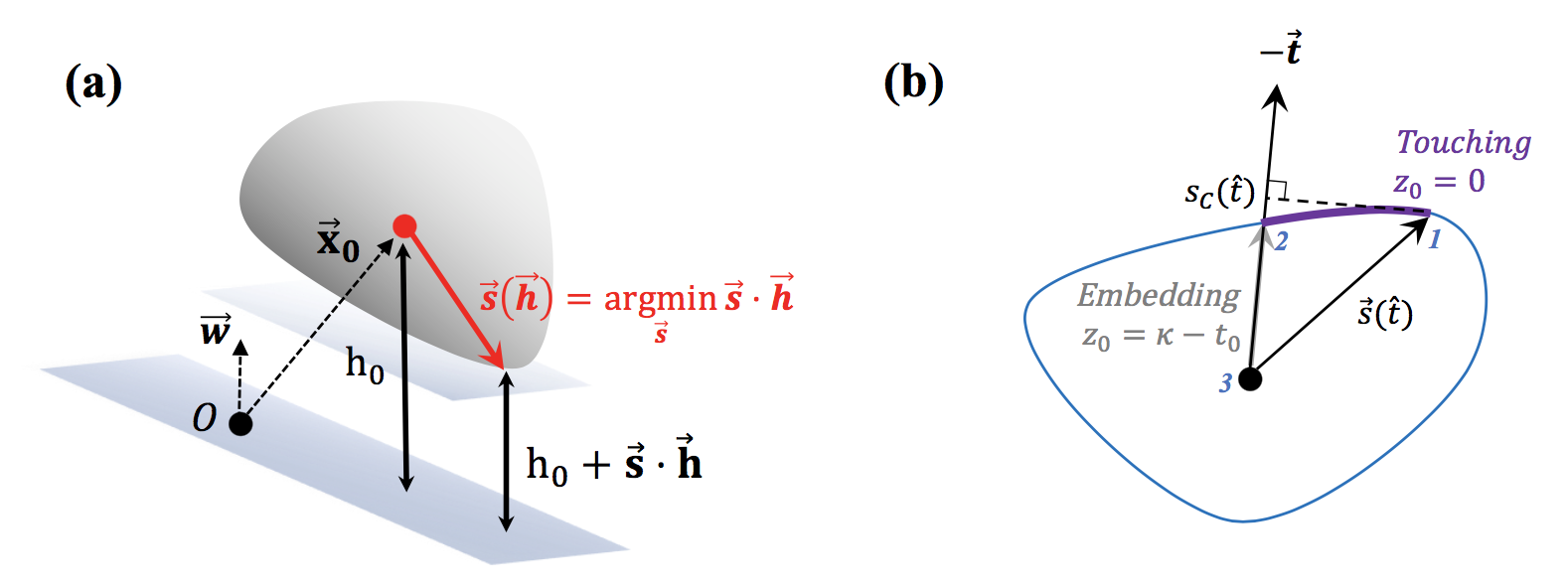}
\par\end{centering}
\caption{\textbf{Geometrical Interpretation. }(a) Relationship between different
fields. $h_{0}=\left\Vert \mathbf{w}\right\Vert ^{-1}y\mathbf{w}\cdot\mathbf{x_{0}}$:
field induced by the center of the manifold $\mathbf{x}_{0}$, i.e.
the distance between the center $\mathbf{x}_{0}$ and the solution
hyperplane characterized by $\mathbf{w}$ . $\vec{h}$ is the vector
of fields induced by basis vectors, i.e. $h_{i}^{\mu}=\left\Vert \mathbf{w}\right\Vert ^{-1}y^{\mu}\mathbf{w}\cdot\mathbf{u}_{i}^{\mu}$.
Together, $h_{0}+\vec{s}\cdot\vec{h}$ determines the distance between
the solution hyperplane and the closest point on the manifold characterized
by the manifold shape constraint $f(\vec{s})=0$. (b) Geometric interpretation
of different regimes. Purple line denotes the range of$\vec{s}(\hat{t})$
when the manifold is in the touching regime, from the point 1 to point
2. From point 2 to 3 denotes the range of $\vec{s}(\hat{t})$ when
the manifold is in the embedded regime. \label{fig:GeneralSmooth_Illustration1}}
\end{figure}

\subsection{Large D Limit}

We assume that $D$ is large but the size of the manifold is such
that $s\ll\sqrt{D}$, 

\begin{equation}
\alpha_{M}^{-1}=\int D\vec{t}\int_{-\infty}^{\kappa-\vec{t}\cdot\vec{s}}Dt_{0}\frac{(\kappa-\vec{t}\cdot\vec{s}-t_{0})^{2}}{1+s^{2}},\:\text{when}\:D\gg1\label{eq:largeD1}
\end{equation}

\begin{equation}
\vec{s}=\tilde{s}(\vec{t}-z_{0}\vec{s})
\end{equation}
\begin{equation}
z_{0}=\frac{-\vec{t}\cdot\vec{s}-t_{0}}{1+s^{2}}
\end{equation}

Since $\vec{t}\cdot\vec{s}$ is large we can approximate 
\begin{equation}
z_{0}=\frac{-\vec{t}\cdot\vec{s}}{1+s^{2}}=\vec{h}\cdot\vec{s}
\end{equation}

and assume self averaging,

\begin{equation}
z_{0}=\frac{-\langle\vec{t}\cdot\vec{s}\rangle}{1+\langle s^{2}\rangle}=\langle\vec{h}\cdot s\rangle
\end{equation}
We can introduce manifold dimensions and radii, 

\paragraph{Manifold Radius and Dimension }

We can now express the above results in terms of the effective manifold
dimensionality and radius. In the limit of large $D$, we can define 

\begin{equation}
R_{M}^{2}=\langle s^{2}\rangle_{\vec{t}}\label{eq:RM}
\end{equation}

\begin{equation}
D_{M}=\frac{\langle\vec{t}\cdot\vec{s}\rangle_{\vec{t}}^{2}}{R_{M}^{2}}\label{eq:DM}
\end{equation}

where$\text{\ensuremath{\vec{s}}}$ is defined by the coupled equations

\begin{equation}
\vec{s}=\tilde{s}(\vec{t}-z_{0}\vec{s})\label{eq: s_gen}
\end{equation}
\begin{equation}
z_{0}=\frac{-\langle\vec{t}\cdot\vec{s}\rangle}{1+\langle s^{2}\rangle}\label{eq:z0_gen}
\end{equation}
 so that the capacity can be expressed as 

\begin{equation}
\alpha_{M}=(1+R_{M}^{2})\alpha_{0}(\kappa+R_{M}\sqrt{D_{M}})\label{eq:alpha_M_smooth_largeD}
\end{equation}

where $R_{M}\sqrt{D_{M}}$ behaves like an additional margin $\kappa_{M}$
introduced by the manifold structure. 

\subsubsection{Scaling Regime}

In the scaling regime, $s_{i}=O(D^{-1})$ and $z_{0}$is $O(1)$ so
$\vec{t}-z_{0}\vec{s}\approx\vec{t}$. In this regime, Eqns. \ref{eq:RM}-\ref{eq:alpha_M_smooth_largeD}
hold but the expression for the manifold radius and dimension are
simpler, since $\vec{s}$ simply becomes 

\begin{equation}
\vec{s}=\tilde{s}(\vec{t})\label{eq:s_gen_scale}
\end{equation}
Here, the expressions for effective radius and dimension is given
with Eqn. \ref{eq:s_gen_scale}, 

\begin{equation}
R_{W}^{2}=\langle s^{2}\rangle_{\vec{t}},\;\text{in scaling regime}\label{eq:R_W}
\end{equation}

\begin{equation}
D_{W}=\frac{\langle\vec{t}\cdot\vec{s}\rangle_{\vec{t}}^{2}}{R_{M}^{2}},\:\text{in scaling regime}\label{eq:D_W}
\end{equation}

with the excess margin

\begin{equation}
\kappa_{W}=R_{W}\sqrt{D_{W}}\:\text{in scaling regime}\label{eq:kappa_W}
\end{equation}

where $W$ stands for \textit{widths}, see Figure \ref{fig:MeanWidthOfManifold}. 

\paragraph*{Mean Width}

Interestingly, the excess margin $\kappa_{W}$ is related to a well
known measure of a size of convex manifolds, known as the \emph{mean
width}s\cite{vershynin2015estimation}, and is defined as 

\begin{equation}
\text{Gaussian Mean Width}=\langle\max_{\vec{s}_{1},\vec{s}_{2}\in M}\left[\vec{t}\cdot\left(\vec{s_{1}}-\vec{s_{2}}\right)\right]\rangle_{\vec{t}}\sim2R_{W}\sqrt{D_{W}}=2\kappa_{W}\label{eq:GaussianMW}
\end{equation}

where $\vec{s}_{1}$ and $\vec{s}_{2}$ are points on a given manifold
$M$ in $\mathbb{R}^{N}$ and $\vec{t}$ is a Gaussian random vector
$\sim I(0,I_{D_{\vec{t}}})$ . 

The relationship between the manifold dimension $D_{W}$ and manifold
radius $R_{W}$ and Gaussian Mean Width is illustrated in Fig. \ref{fig:MeanWidthOfManifold}. 

\begin{figure}[h]
\begin{centering}
\includegraphics[width=0.85\textwidth]{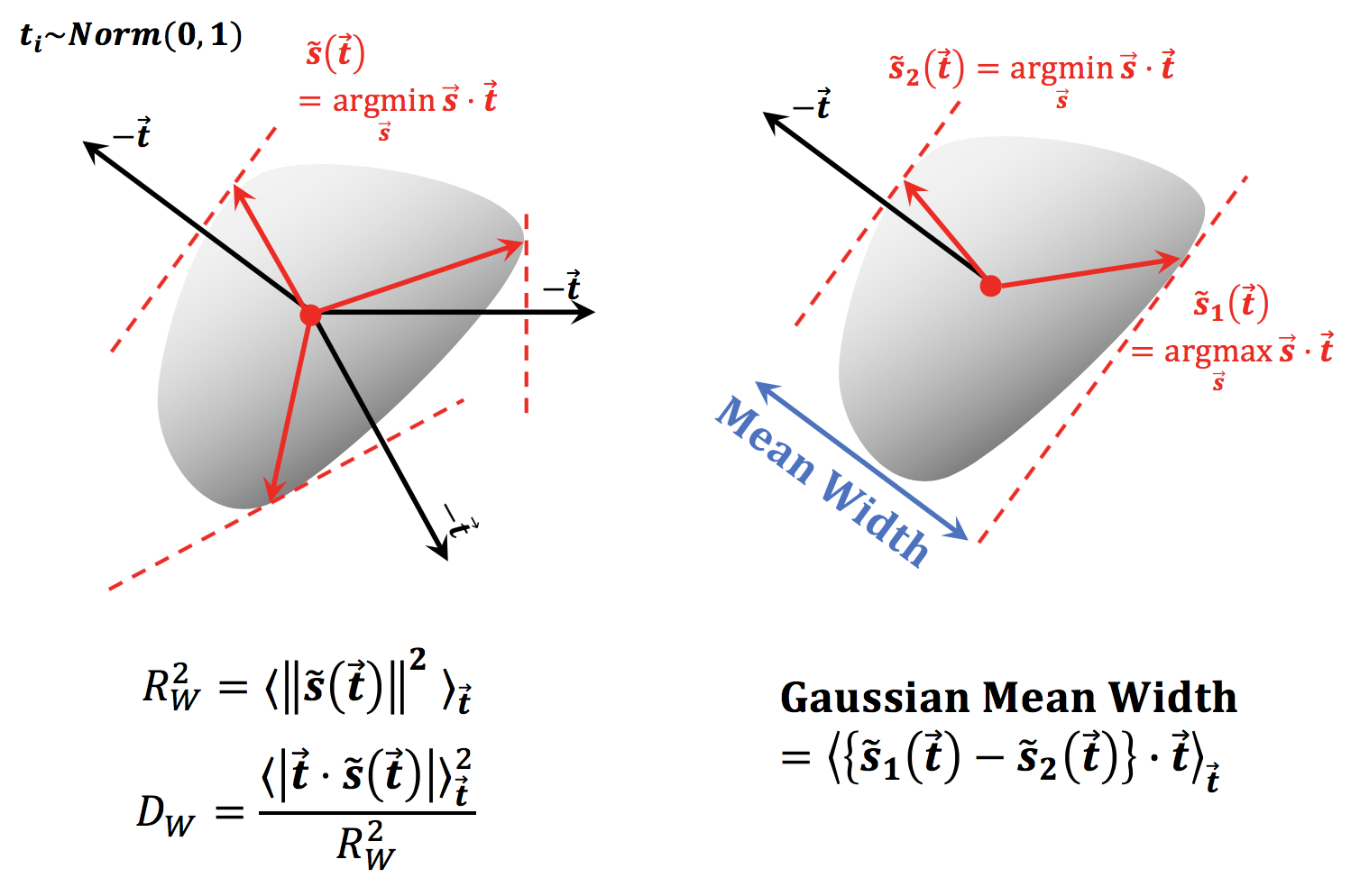}
\par\end{centering}
\caption{\textbf{Relationship between the Gaussian Mean Width the Effective
Manifold Radius and Dimension in the Scaling Regime. }(a) Effective
Radius $R_{W}=\langle\left\Vert \tilde{s}(\vec{t})\right\Vert ^{2}\rangle_{\vec{t}}$
is the mean of max projection points $\tilde{s}(\vec{t})$ along the
random directions $\vec{t}$, while the effective dimension $D_{W}$
is defined as $\frac{\langle\left|\vec{t}\cdot\tilde{s}(\vec{t})\right|\rangle_{\vec{t}}^{2}}{R_{W}^{2}}$.
(b) Gaussian Mean Width is defined as $GMW=\langle\left\{ \tilde{s_{1}}(\vec{t})-\tilde{s_{2}}(\vec{t})\right\} \cdot\vec{t}\rangle_{\vec{t}}$,
and in this definition, $GMW=R_{W}\sqrt{D_{W}}$. \label{fig:MeanWidthOfManifold} }
\end{figure}

\section{General Manifolds }

Smooth networks are simple in that they can touch a hyperplane by
a single point or be fully embedded in it. This is not true for non-smooth
manifolds, as there are many facets that can be partially embedded
due to the non-smoothness. In other words, a non-smooth manifold can
touch the hyperplane by a point, line segment, a facet, or multiple
facets (Fig. \ref{fig:L1embedding}). 

\begin{figure}
\begin{centering}
\includegraphics[width=0.9\textwidth]{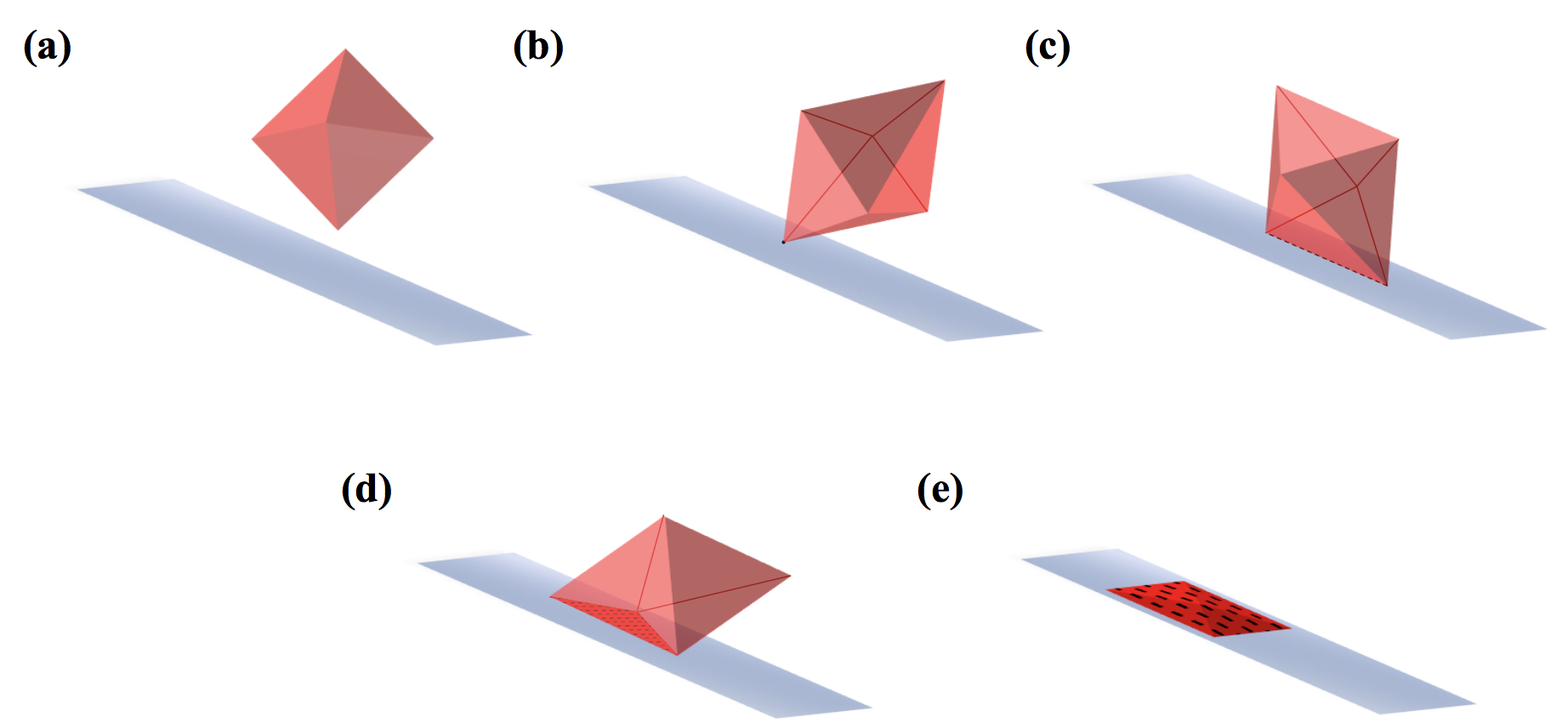}
\par\end{centering}
\caption{\textbf{Embedding (support) structures of non-smooth manifolds: $L_{1}$
manifolds. }(a) Interior manifolds. (b) Touching with a point. (c)
Touching with a line. (d) Touching with a facet. (e) Embedded in the
margin plane. \label{fig:L1embedding}}
\end{figure}

\subsection{Capacity of Smooth and Non-smooth Manifolds }

Given the complicated geometric relations between non smooth manifolds
and the margin planes, explicit expression for the capacity that delineates
the different regimes in $\vec{t}$ and $t_{0}$ is cumbersome and
depends on the specific details of the manifold at hand. Here we note
that for any manifold, we can write down the capacity in the following
\emph{universal} form

\noindent\fbox{\begin{minipage}[t]{1\columnwidth \fboxsep \fboxrule}%
Universal mean field equation for manifold capacity
\begin{equation}
\alpha_{M}^{-1}=\int D\vec{t}\int_{-\infty}^{\kappa-\vec{t}\cdot\vec{s}(\vec{t})}Dt_{0}\frac{(-\vec{t}\cdot\vec{s}-t_{0}+\kappa)^{2}}{1+s^{2}}\label{eq:alphaM_gen}
\end{equation}

where, $\vec{s}$ is defined via 

\begin{equation}
\vec{s}=\tilde{s}(\vec{t}-z_{0}\vec{s})\label{eq:s_min_gen}
\end{equation}
\begin{equation}
z_{0}=\frac{-\vec{t}\cdot\vec{s}-t_{0}+\kappa}{1+s^{2}}\label{eq: z0_smin_gen}
\end{equation}
\end{minipage}}

The key point is that the solution for $\vec{s}$ changes its nature
as $t_{0}$ decreases (for a given $\vec{t}$) and automatically dissects
the range of integration over $t_{0}$ to the specific domains (touching
with points, lines, facets, etc). Note that the fully embedded regime
is also incorporated in \ref{eq:alphaM_gen}. In this regime, $\vec{t}+\vec{z}=\vec{t}-z_{0}\vec{s}=0,$
and $z_{0}=\kappa-t_{0}$, hence, $\vec{s}=\vec{t}/(\kappa-t_{0})$
, which in the embedding regime will be a point \emph{inside }the
convex manifold in the direction of $\vec{t}$, see Fig. \ref{fig:GeneralSmooth_Illustration1}(b). 

\subsection{Large $D$ Approximation for a General Manifold }

In the case of smooth manifolds, we have shown that in the limit of
large $D$ the capacity can be approximated by Eqns. \ref{eq:alphaM_gen}-\ref{eq:z0_gen}.
Here we note that the same approximation applies to general, smooth
as well as non-smooth manifolds. Specifically, we approximate the
capacity as 

\noindent\fbox{\begin{minipage}[t]{1\columnwidth \fboxsep \fboxrule}%
Capacity for General Manifolds in high $D$: 

\begin{equation}
\alpha_{M}=(1+R_{M}^{2})\alpha_{0}(\kappa+\kappa_{M})\label{eq:alphaM_gen_LargeD}
\end{equation}

\begin{equation}
\kappa_{M}=R_{M}\sqrt{D_{M}}\label{eq:kappaM_LargeD}
\end{equation}

\begin{equation}
R_{M}^{2}=\langle s^{2}\rangle\label{eq:RM_LargeD}
\end{equation}

\begin{equation}
D_{M}=\frac{\langle\vec{t}\cdot\vec{s}\rangle^{2}}{R_{M}^{2}}\label{eq:DM_LargeD}
\end{equation}

where 

\begin{equation}
\vec{s}=\tilde{s}(\vec{t}-z_{0}\vec{s})\label{eq: s_gen-1}
\end{equation}

\begin{equation}
z_{0}=\frac{-\langle\vec{t}\cdot\vec{s}\rangle}{1+\langle s^{2}\rangle}\label{eq:z0_gen-1}
\end{equation}
\end{minipage}}

and averages are over gaussian $D-dimensional$ vectors $\vec{t}$.
As is the case of smooth manifolds, in the scaling regime where $R_{M}$
is $O(D^{-1/2})$, $R_{M}$ and $D_{M}$ are given via $\vec{s}$
where $\vec{s}$ is simply $\tilde{s}(\vec{t})$ , hence they coincide
with $R_{W}$ and $D_{W}$ and are related to the Gaussian Mean Width
as in the smooth case above (Figure \ref{fig:MeanWidthOfManifold}). 

\section{Numerical Investigations}

\subsection{Numerical Solutions of the Mean Field Equations}

In simple cases analytical expressions can be used to solve numerically
the mean field equations, as is the case of ellipsoids discussed above.
Here we show how to use the analytical formulae to solve the mean
field equations for the ellipsoids. For a general manifold, calculating
$\vec{s}\cdot\vec{t}$ and $s^{2}$ for $\vec{s}$ on the manifold
for each $\vec{t}$ and $t_{0}$ can be done by iterative methods.
Here we present such an algorithm, adequate for general manifolds
in the large $D$ regime. In the limit of large $D$, $\left|\vec{t}\cdot\vec{s}\right|\gg\left|t_{0}\right|,\kappa$,
and we can use Eqns \ref{eq:alphaM_gen_LargeD}-\ref{eq:z0_gen-1}.
Furthermore, in this limit, $\vec{t}\cdot\vec{s}$, and $s^{2}$ are
self averaged with respect to $\vec{t}$. The pseudocode for the algorithm
is given in Alg. \ref{alg:alpha_iter_algo} below. 

\subsubsection{Iterative Numerical Solution}

In order to solve for $s$ and $z_{0}$ (Eqn \ref{eq: s_gen-1}- \ref{eq:z0_gen-1}),
we use an algorithm for each $\vec{t}$, that essentially iterates
between updating $z_{0}$ given the current estimate of $\vec{s}$,
using Eqn \ref{eq:z0_gen-1} and updating the estimate of $\vec{s}$
given the new estimate of $z_{0}$ and the current estimate of $\vec{s}$,
Eqn \ref{eq: s_gen-1}. 

Solving eq. \ref{eq: s_gen-1}: First we note that evaluating the
$min$ operation in Eqn \ref{eq: s_gen-1} can be done explicitly
in simple cases (in particular, for convex smooth manifolds with known
parametrization). In general, one can search numerically for the max
projection points (or signed min projection points). If the manifold
is non-smooth and has a finite number of vertices, then one can simply
iterate over all vertices. Otherwise, a local search using a gradient
can be done, and since the search is on the convex hull, the local
search guarantees the convergence to the global optimum. This appears
as a $maxproj$ function in Alg. \ref{alg:alpha_iter_algo}. 

Note that $\vec{s}$ we are solving is not simply a max projection
point on $\vec{t}$, but a max projection on $\vec{h}=\vec{t}-z_{0}\vec{s}$
(Eqn. \ref{eq: s_gen-1}). The solution $\vec{s}$ may come from anywhere
inside of the convex hull or the surface of the manifold, hence we
allow the search on$\vec{s}$ to be a linear combination of the vertices.
To search for $\vec{s}$, in the next step we define $\vec{s}_{k}=\eta\vec{s_{h}}+(1-\eta)\vec{s}_{k-1}$
which is a linear sum of $\vec{s}_{k-1}$ in the previous step $k-1$
and the max projection in the direction of $\vec{h}$ at time $k$,
$\vec{s}_{h}$. If the difference between $\vec{s_{k}}$ and $\vec{s}_{k-1}$
is smaller than a given tolerance $\epsilon_{0}$, then $\vec{s}$
converged, as well as $z_{0}$. Otherwise, continue to the $k+1$
th step, where the new $\vec{h}$ is computed with the new $\vec{s}$.
The algorithm for this is summarized in the pseudocode (Alg. \ref{alg:alpha_iter_algo}). 

\begin{algorithm}
\textbf{{[}$\alpha_{M}$,$R_{M}$, $D_{M}${]} = function manifold\_capacity($D$,
$n_{t}$,$\eta$, $n_{max}$ , $\epsilon_{0}$, $S$) }

\textbf{Input: }\{Manifold dimension $D$, number of $t's$ $n_{t}$,
learning rate $\eta$, max iteration $n_{max}$, tolerance $\epsilon_{0}$,
the manifold data$S$ defined in $\mathbb{R}^{D\times M}$\}

\textbf{for }i=1 \textbf{to}$n_{t}$ \textbf{do}

$\:$$\vec{t}=\vec{t}^{(i)}\sim Norm(0,\mathbb{I}_{D}])$

$\:$Set $k=0$, $\epsilon_{ts}=\infty$ $ $

$\:$$\vec{s}_{k}=maxproj(\vec{t}$,$S$)

$\:$\textbf{while }$k<n_{max}$ \textbf{and $\epsilon_{s}>\epsilon_{0}$}

$\:$$\:$ $k=k+1$

$\:$$\:$ $z_{0}=\frac{-\vec{t}\cdot\vec{s}}{1+s^{2}}$

$\:$$\:$$\vec{h}=\vec{t}-z_{0}\vec{s}$

$\:$$\:$$\vec{s}_{h}=maxproj(\vec{h},S)$

$\:$$\:$$\vec{s_{k}}=\eta\vec{s_{h}}+(1-\eta)\vec{s}_{k-1}$

$\:$$\:$$\epsilon_{s}=||\vec{s}_{k}-\vec{s}_{k-1}||/||\vec{s}_{k}||$

$\:$\textbf{end}

$\:$$\vec{s}^{(i)}=\vec{s}_{k}$

\textbf{end}

$R_{M}=sqrt\left\{ \langle\left\Vert \vec{s}^{(i)}\right\Vert ^{2}\rangle_{i}\right\} $

$D_{M}=\frac{\left\{ \langle\vec{t}^{(i)}\cdot\vec{s}^{(i)}\rangle_{i}\right\} ^{2}}{R_{M}^{2}}$

$\alpha_{M}=(1+R_{M}^{2})\alpha_{0}(\kappa+R_{M}\sqrt{D_{M}})$

\textbf{Output }= {[}$\alpha_{M}$,$R_{M}$, $D_{M}${]} 

\caption{Iterative method for approximating capacity of general manifolds in
high $D$ \label{alg:alpha_iter_algo}}
\end{algorithm}

In the following sections, we show the specific manifold examples
and how the input $S$ of $maxproj(\vec{t},S)$ , as well as the details
of the max projection search differ based on the type of the problem. 

\subsection{Simulation Results }

\subsubsection{Ellipsoids }

Consider the ellipsoids give by Eqn. \ref{eq:xpoint} and constraint,
Eqn. \ref{eq:constraint_ellipsoid}. In this case, the manifold parameterization
is known, then $\vec{s}_{k}=\underset{\vec{s}}{\text{argmax}}\vec{t}\cdot\vec{s}$
has an analytical solution given by the minus of Eqn. \ref{eq:argminS}
and Eqn. \ref{eq:stilde} for ellipsoids. In this case, the input
to the max projection operation for ellipsoid, called $maxproj^{ellipsoid}$should
include the radii vector. Pseudocode for $maxproj^{ellipsoid}$ is
given in Alg. \ref{alg:maxproj_ellipsoid}. Furthermore, in the case
of ellipsoids, it is possible to solve for $z_{0}$ analytically in
the of the large $D$ (i.e. Eqn.\ref{eq:z0_LargeD_RO1}). However,
to test the effectiveness of the iterative algorithm, we proceed to
test the iterative algorithm in the simulations below. 

\begin{algorithm}
\textbf{$\tilde{s}$ = function $maxproj^{ellipsoid}(\vec{t},S)$ }

Input: $D$-dimensional direction vector $\vec{t}$, $S=\{\text{\ensuremath{D}-dimensional radii vector }\vec{R}\}$

\textbf{for }i=1 \textbf{to}$D$ \textbf{do}

$\tilde{s}_{i}=\frac{t_{i}R_{i}^{2}}{\sqrt{\sum_{j}t_{j}^{2}R_{j}^{2}}}$

\textbf{end}

Output: The point on the manifold with max projection in $\vec{t}$
, $\tilde{s}$

\caption{Maximum Projection Point on Ellipsoid $S$ \label{alg:maxproj_ellipsoid}}
\end{algorithm}

Using the iterative methods described above, we calculated the theoretical
estimate of linear classification capacity of ellipsoids, with embedding
dimension $D=50$ in ambient dimension $N=100$ . The radii for each
ellipsoid is given by the $R_{i}=Unif[0.5R_{0},1.5R_{0}]$ for each
$i=1,...,D$ . $R_{0}$ is shown in the $x$ axis of the Fig \ref{fig:EllipsoidsResults}
as $\langle R_{i}\rangle_{i}$. We compare the capacity estimated
by the iterative algorithm using the mean field approximation (noted
as $\alpha_{iter}$), with the capacities estimated by the expressions
of effective radius and effective dimension in different regimes (in
the regime of $R_{i}\sim O(1)$, Eqns. \ref{eq:RE_Rorder1}-\ref{eq:DE_Rorder1},
and in the scaling regime, Eqns \ref{eq:RE_elps_scale}-\ref{eq:DE_elps_scale}.
We also calculated the simulation capacity in the similar manner to
$M_{4}$ algorithms for $L_{2}$ balls described in Chapter \ref{cha:spheres}-\ref{cha:m4},
with the worst point analytically calculated by Eqn. \ref{eq:argminS}-\ref{eq:stilde}.
Using the ambient dimension of $N=100$ for the ellipsoids, the critical
value of $P$ was determined by finding the value of $P$ such that
the average number of being separable was half the times of the total
number of $50$ repetitions. 

The iterative algorithm capacity matches the estimated capacities
using effective radii and dimensions well, as well as simulation capacities
described above. In the limit of large $R_{i}$, the capacity approaches
$1/D$, where the hyperplane is orthogonal to all embedded dimensions
of an ellipsoid.

We also compare the measures of dimensions relevant for the classification
capacity of ellipsoids (Fig. \ref{fig:EllipsoidsResults}). The dimensions
calculated from the iterative algorithm ($D_{M}$, M for manifolds)
is compared with the approximate effective ellipsoid dimensions $D_{E}$
in different regimes (Eqn \ref{eq:DE_Rorder1} and \ref{eq:DE_elps_scale}),
and they agree well. Furthermore, these estimated effective ellipsoid
dimensions match the participation ratio (\ref{eq:Dsvd}), $D_{svd}$
, when $R_{i}$ is in the scaling regime, and match the actual full
embedding dimension $D$ when $R_{i}$ is large. Intuitively, this
means that $D_{svd}$ is a relevant measure for linear classification
capacity when $R_{i}$'s are small ($\sim D^{-1/2}$). In the case
where $R_{i}$ is large, the embedding dimension $D$ is the relevant
measure for the capacity, because the solution has to orthogonalize
all embedding dimensions independent of the structure. 

We also compare the ratio between effective radius and the actual
scale of the ellipsoid (in this case $R_{0}=\langle R_{i}\rangle_{i}$),
for $R_{M}$ calculated from the iterative algorithm and $R_{eff}$
calculated from approximations in each regime of $R_{i}\sim D^{-1/2}$
(scaling) and $R_{i}\sim O(1)$ (Fig. \ref{fig:EllipsoidsResults}).
As before, the agreement between the $R_{M}$, $R_{eff}$'s are good.
Furthermore, the ratio between effective radius (for capacity) and
the mean of the radii start out above 1, and decreases with increasing
$R_{0}$ . This means that the larger the ellipsoids get, the more
fraction of ellipsoids get embedded, hence effective radius $R_{eff}$
contributing to the effective increase in the margin $\kappa_{eff}=R_{eff}\sqrt{D_{eff}}$
gets smaller due to the embedding configuration. 

\begin{figure}
\begin{centering}
\includegraphics[width=1\textwidth]{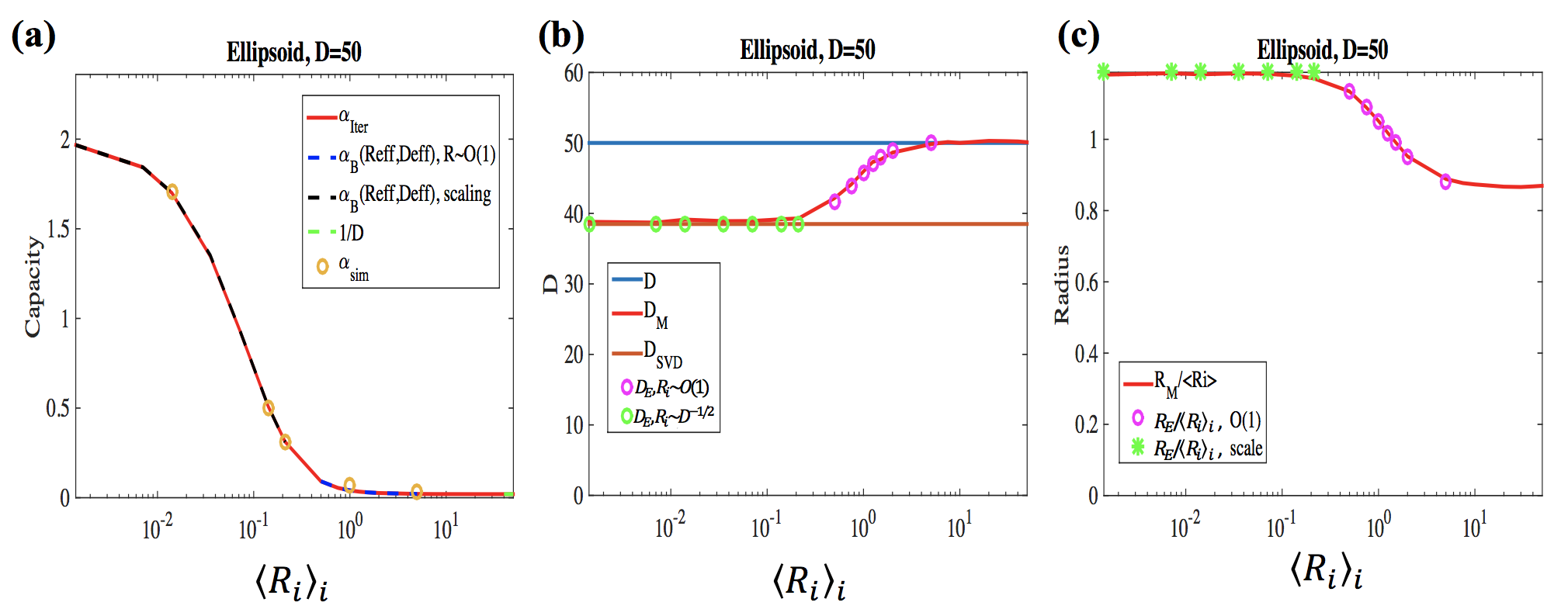}
\par\end{centering}
\caption{\textbf{Linear Classification of Ellipsoids.} (a) Linear Classification
Capacity of $D$-dimensional ellipsoids, $\alpha_{E}=P/N$, where
$N$ is the ambient dimension, and $P$ is maximum the number of ellipsoids
such that $P$ ellipsoids are linearly separable. The embedding dimension
of ellipsoids used was $D=50$ and $R_{i}\sim Unif\left[0.5R_{0},1.5R_{0}\right]$
for $i=1,...,D$ and $R_{0}=\langle R_{i}\rangle_{i}$ is shown in
the $x$-axis. (Red) Mean field approximation capacity $\alpha_{iter}$
, evaluated by the iterative algorithm given in Alg. \ref{alg:alpha_iter_algo}
and Alg. \ref{alg:maxproj_ellipsoid}. (Blue dashed) Approximation
of the ellipsoid capacity as the capacity of a ball using a large
$D$ and $R_{i}\sim O(1)$ approximation for $R_{E}$ and $D_{E}$
given by Eqns. \ref{eq:RE_Rorder1}- \ref{eq:DE_Rorder1}. (black
dashed) Ellipsoid capacity approximated using $R_{E}$and $D_{E}$
approximation when $R_{i}$ is in the scaling regime, given by Eqns.
\ref{eq:RE_elps_scale}-\ref{eq:DE_elps_scale}. (Green) Capacity
approximation for large $R_{i}$, $1/D$, where all of the ellipsoid
embedding dimensions are orthogonalized by the solution. (Yellow)
Simulation capacity computed with $N=100$ and $50$ repetitions.
(b) Dimensions of the ellipsoid. (Blue) Embedding dimension $D$.
(Red) Dimension of the ellipsoid evaluated by the iterative algorithm,
$D_{M}$. (Orange) Participation ratio, $D_{SVD}$, given by Eqn.
\ref{eq:Dsvd} using $R_{i}$ as eigenvalues. (Pink) $D_{E}$ approximation
in large $D$ $R_{i}\sim O(1)$ regime given by Eqn.\ref{eq:DE_Rorder1}
(Green) $D_{E}$ approximation in large $D$ and scaling regime given
by Eqn. \ref{eq:DE_elps_scale}. (c) Size (Radius) of Ellipsoids.
(Red) Effective manifold radius $R_{M}$ evaluated by the iterative
algorithm, divided by the overall scale $\langle R_{i}\rangle_{i}$.
(Pink) $R_{E}$ approximation in large $D$, $R_{i}\sim O(1)$ regime
given by Eqn. \ref{eq:RE_Rorder1}(Green) $R_{E}$ approximation in
large $D$ and scaling regime given by Eqn. \ref{eq:RE_elps_scale}
\label{fig:EllipsoidsResults} }
\end{figure}

\subsubsection{$D$ Dimensional $L_{1}$ Manifolds }

Consider the problem of linearly classifying $P$ of $D$-dimensional
$L_{1}$ manifolds where the point on the $L_{1}$ manifold is given
by Eqn. \ref{eq:xpoint} where $f(\vec{s})=\sum_{i=1}^{D}|s_{i}/R_{i}|-1=0$
. The explicit expression for classifying $L_{1}$ manifolds is considered
in \cite{chung2016linear}, and in this section we focus on finding
their perceptron capacity as an example of manifolds that are defined
by their vertices (Fig. \ref{fig:L1embedding}). In this case, there
are only $2D$ vertices (2 extreme points along the direction vector
$\mathbf{u}_{i}$), and finding the max projection point in the direction
of $\vec{t}$ from the set of points $S$ is given by $\vec{s}_{k}=\underset{l}{\text{argmax}}\vec{t}\cdot S_{l}$
, simply the search over all $2D$ vertices whose computation time
is linear in the number of points. In this case, the input for the
max projection operation in the iterative algorithm, called $maxproj^{setpoints}$,
should include the set of points. Pseudocode for $maxproj^{setpoints}$
is given by Alg. \ref{alg:maxproj_setpoints}. 

\begin{algorithm}
\textbf{$\tilde{s}$ = function $maxproj^{setpoints}(\vec{t},S)$ }

\textbf{Input:} $D$-dimensional direction vector$\vec{t}$, A set
of $M$ points in $D$ dimensional basis which define the vertices
of the convex hull$S=\{X\in\mathbb{R}^{D\times M}\}$

$i_{l}=argmax_{l}\vec{t}\cdot X(:,l)$

$\tilde{s}=X(:,i_{l}$) 

\textbf{Output: }The point on the manifold with max projection in
$\vec{t}$ , $\tilde{s}$

\caption{Maximum Projection Point on a set of points $S$\label{alg:maxproj_setpoints} }
\end{algorithm}

Using the iterative methods described above, we calculated the linear
classification of capacity of $L_{1}$ manifolds, with the embedding
dimension $D=20$ in the ambient dimension $N=100.$ The radii for
each direction is set to be equal, i.e. $R_{i}=R$ (all vertices are
distance $R$ away from the center ) for all $i=1,...D$. $R$ is
shown in the x axis of the Fig. \ref{fig:L1results}. 

We compare the capacity estimated by the iterative algorithm using
the mean field approximation (noted as $\alpha_{Iter}$), with the
capacities estimated as that of a ball using the effective radius
and effective dimension of $L_{1}$ manifolds in the scaling regime.
In the scaling regime, the replica analysis gives us $R_{M}=R$ and
$D_{M}=2\text{log}(D)$, and the derivations are given in the appendix
to the chapter (section \ref{subsec:L1effD_2logD}). The estimated
capacity using the iterative algorithm agrees well with the simulation
capacity, as well as approximations in the scaling regime, and large
R regime. 

We also compare the measures of dimensions relevant for the classification
capacity of $L_{1}$ manifolds. The dimension estimated by the iterative
algorithm matches the approximation of $D_{M}(L_{1})\sim2\text{log}(D)$
in the scaling regime (due to the extreme value theory, details in
the Appendix to the chapter), and in the regime of large $R$ it matches
the embedding dimension $D$, which in this case is equivalent to
the participation ratio $D_{svd}$ (as all $R_{i}=R$ ). 

Furthermore, in the scaling regime, the effective manifold radius
found by the iterative algorithm $R_{M}$ is close to $R$, as predicted
by the theory (details in the appendix). $R_{M}/R$ transitions from
1 (in the scaling regime) to a value much smaller than 1 (in the large
$R$ regime), due to the increased fraction of $L_{1}$ manifolds
that are embedded.

\begin{figure}
\begin{centering}
\includegraphics[width=1\textwidth]{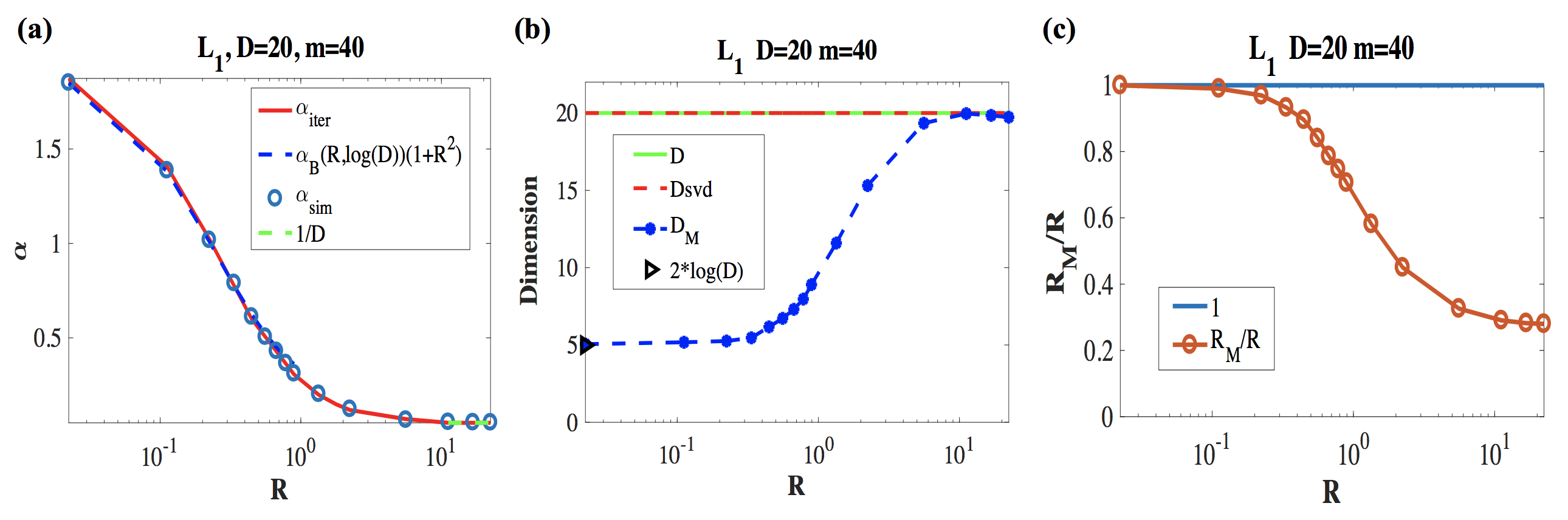}
\par\end{centering}
\caption{\textbf{Linear Classification of Non-smooth manifolds: $D$-dimensional
$L_{1}$ Manifolds of Radius $R$. }(a) Linear Classification Capacity
of $D$-dimensional $L_{1}$ manifolds, $\alpha=P/N$, where $N$
is the ambient dimension, and $P$ is the maximum number of manifolds
such that $P$ $L_{1}$manifolds are linearly separable. The embedding
dimension of $L_{1}$manifolds used was $D=20$ and the number of
subsamples used for testing was $m=40$, two endpoints in each basis
vector direction. $R$ is shown in the $x$-axis. (Red) Mean field
approximation capacity $\alpha_{iter}$ , evaluated by the iterative
algorithm given in Alg. \ref{alg:alpha_iter_algo} and Alg. \ref{alg:maxproj_setpoints}.
(Blue dashed) Approximation of the $L_{1}$ capacity as the capacity
of a ball using $R$ as the actual $R$ and $D_{M}=\text{log}D$,
which is the approximation of effective manifold properties in the
large $D$ regime. (blue markers) Simulation capacity calculated with
$N=100$ and $50$ iterations to compute the fraction of linear separability.
(Green) Capacity approximation for large $R$, $1/D$, where all of
the $L_{1}$ manifold embedding dimensions are orthogonalized by the
solution. (b) Dimensions of the $L_{1}$ manifolds. (Green) Embedding
dimension $D$. (Red dashed) Participation ratio, $D_{SVD}$, given
by Eqn. \ref{eq:Dsvd} using $R_{i}=R$ as eigenvalues. (Blue) Dimension
of the $L_{1}$ manifolds evaluated by the iterative algorithm, $D_{M}$.
(Black marker) $D_{B_{1}}$ approximation in large $D$ regime, $2\text{log}D$
(Derivation in appendix). (c) Size (Radius) of $L_{1}$ manifolds
divided by $R$. (Red) Effective manifold radius $R_{M}$ evaluated
by the iterative algorithm, divided by $R$, compared with unity (Blue).
\label{fig:L1results} }
\end{figure}

\subsubsection{Random Strings }

Consider the problem of linearly classifying $P$ of random strings,
whose intrinsic dimension is 1, but the embedding dimension is $D$,
and the ambient dimension of $N$ . Each point on the random string
is parameterized by the vector $\vec{s}$ whose components are

\begin{equation}
s_{2k}=R_{k}cos\left\{ k\left(\theta-\phi_{k}\right)\right\} ;\label{eq:s_k_randstr_even}
\end{equation}
 
\begin{equation}
s_{2k+1}=R_{k}sin\left\{ k\left(\theta-\phi_{k}\right)\right\} \label{eq:s_k_randstr_odd}
\end{equation}

This can be re-written as 

\begin{equation}
x_{i}^{\mu}=\left(x_{0}\right)_{i}^{\mu}+\sum_{n=1}^{D/2}R_{n}e^{jn(\theta-\phi_{n})}u_{i}\label{eq:Model_RandStr}
\end{equation}

(where the $j$ is used as an imaginary $\sqrt{-1}$ to distinguish
from the index $i$ ). 

\begin{figure}
\begin{centering}
\includegraphics[width=1\textwidth]{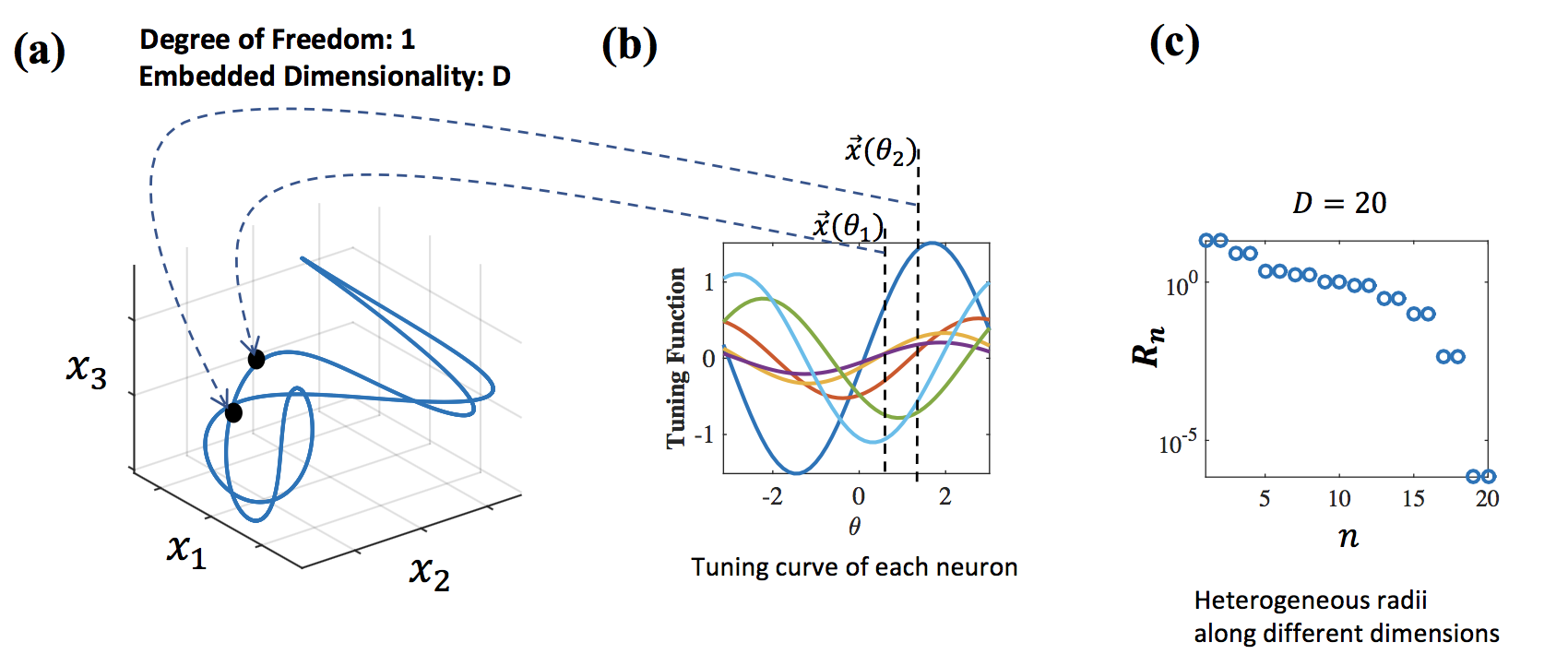}
\par\end{centering}
\caption{\textbf{Random Strings (Illustration). }(a) Random String in Neural
State Space. The values in each ambient dimension is given by each
neuron's activity. The random string's degree of freedom is 1, while
the embedding dimension is $D$. Each point on the sample manifold
represents neural activity of a same object, with different latent
variable such as orientation of an object. (b) Neural Interpretation
of Random Strings. The random string manifolds given by Eqns. \ref{eq:s_k_randstr_even}-\ref{eq:Model_RandStr}
can be interpreted as orientation tuned neurons with different amplitudes
and frequencies. (c) Illustration of $R_{n}$ for $n$th basis vector,
which is similar to the amplitude of the neural tuning curve for $n$
th neuron. \label{fig:randstrings_illu} }
\end{figure}

Figure \ref{fig:randstrings_illu} illustrates an example of a random
string. This definition has an interesting analogy with the activity
patterns of the population of orientation tuned neurons. For instance,
the value of point $\mathbf{x}$ in $i$ th dimension, $x_{i}$ can
be thought of as $i$ th neuron's activity, where each neuron is tuned
to a different orientation (middle panel). The heterogenous $R_{i}$
can be thought of as different amplitudes of the neural activity.
In this analogy, if you take an object at one angle, then all neurons
will have different levels of activations, and the slice of the activity
patterns correspond to a point on a random string, parametrized by
$\theta$, which is essentially the angle of an object. If you change
the orientation of the stimulus, then the activity patterns correspond
to a different slice, which corresponds to a different point on a
$\theta$-parametrized random string. Once you rotate the stimulus
the full $360\mathring{}$, then the activity pattern comes back to
the original slice, and you come back to the original point on the
random string. Note that this particular random string lies on the
surface of a ball, whose radius is $R_{string}=\sqrt{\sum_{n=1}^{D/2}R_{n}^{2}}$. 

What should be an effective dimension of this string with heterogenous
scales $\vec{R}$? And, if these strings are in random positions and
directions, what should be their capacity? To test this, we calculated
the classification capacity of samples of random strings in embedding
dimension $D$, ambient dimension $N=100$, where the number of samples
used was $m=200$, such that $\theta_{m}=\frac{2\pi}{m}$, and $R_{n}=R$
for all $n$. 

\begin{figure}
\begin{centering}
\includegraphics[width=1\textwidth]{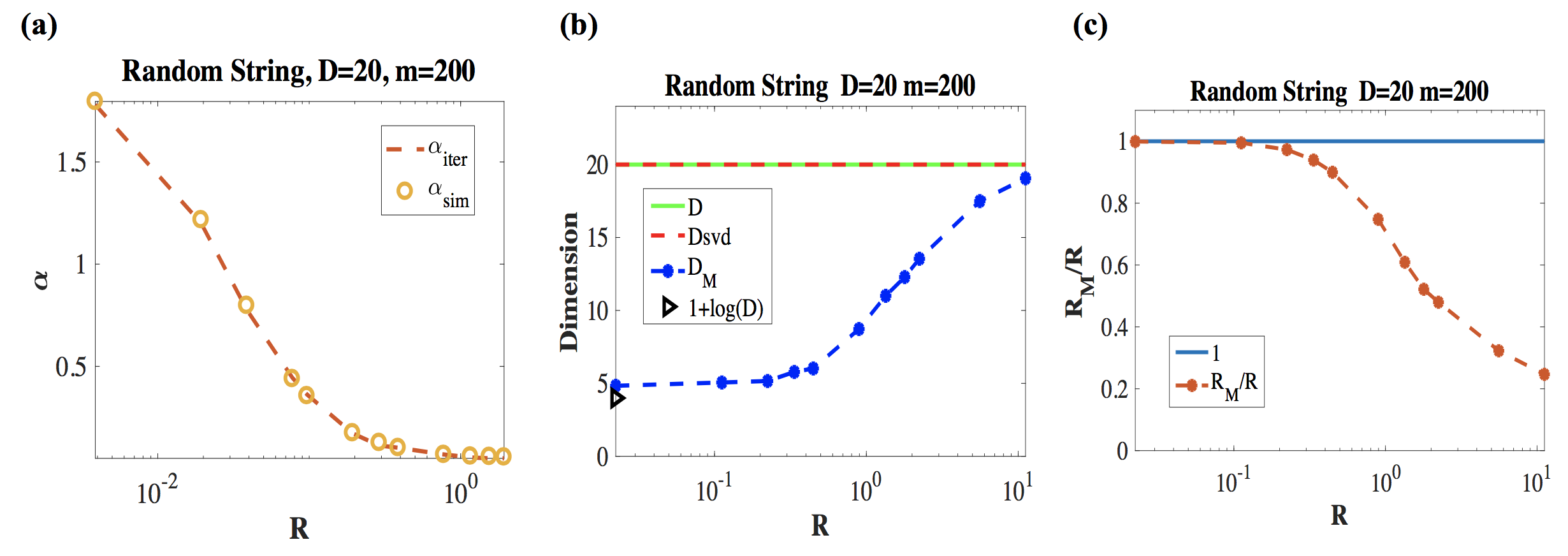}
\par\end{centering}
\caption{\textbf{Linear Classification of $D$-dimensional Random Strings.
}(a) Linear Classification Capacity of $D$-dimensional Random Strings,
$\alpha=P/N$, where $N$ is the ambient dimension, and $P$ is the
maximum number of random strings such that $P$ random strings are
linearly separable. The embedding dimension of random strings used
was $D=20$ and the number of subsamples used for testing was $m=200$.
$R_{string}$, overall scale of the random string, and the radius
of the $D$-dimensional ball that the random string is on, is shown
in the $x$-axis. (Red dashed) Mean field approximation capacity $\alpha_{iter}$
, evaluated by the iterative algorithm given in Alg. \ref{alg:alpha_iter_algo}
and Alg. \ref{alg:maxproj_setpoints}. (Yellow marker) Simulation
capacity calculated with $N=100$ and $50$ iterations to compute
the fraction of linear separability. (b) Dimensions of the Random
Strings. (Green) Embedding dimension $D$. (Red dashed) Participation
ratio, $D_{SVD}$, given by Eqn. \ref{eq:Dsvd} using $R_{i}=R$ as
eigenvalues. (Blue) Dimension of the random strings evaluated by the
iterative algorithm, $D_{M}$. (Black marker) $\text{1+log}D$. (c)
Size (Radius) of random strings divided by $R$. (Red) Effective manifold
radius $R_{M}$ evaluated by the iterative algorithm, divided by $R$,
compared with unity (Blue).\label{fig:RandString_ARD} }
\end{figure}

We find that the classification capacity $\alpha_{M}=\alpha_{iter}$
found by an iterative algorithm matches the simulation capacity of
random strings (Figure \ref{fig:RandString_ARD}(a)). In the case
of random strings, the manifold effective dimension $D_{M}$ (also
found via the iterative algorithm) has a very low effective dimension
in the scaling regime, due to the fact that it is a string whose intrinsic
dimension is merely $1$ and is not filling the space spanned by $D$
basis vectors, although when $R$ is large, the solution has to orthogonalize
the manifolds and the effective dimension approaches the embedding
dimension $D$, and it is reflected in the $D_{M}$ found via the
iterative algorithm (Figure \ref{fig:RandString_ARD}(b)). Note that
in this case, since all $R_{i}$ has the same size, $D_{svd}=D$ .
Furthermore, $R_{M}$ is $R$ in the scaling regime, and $R_{M}$goes
to a value much smaller than $R$ in the large $R$ regime, reflecting
the fact that many of the random strings must be orthogonalized by
the solution, and therefore more of them are embedded, resulting in
smaller $R_{M}/R$ seen by the centers. 

\begin{figure}
\begin{centering}
\includegraphics[width=1\textwidth]{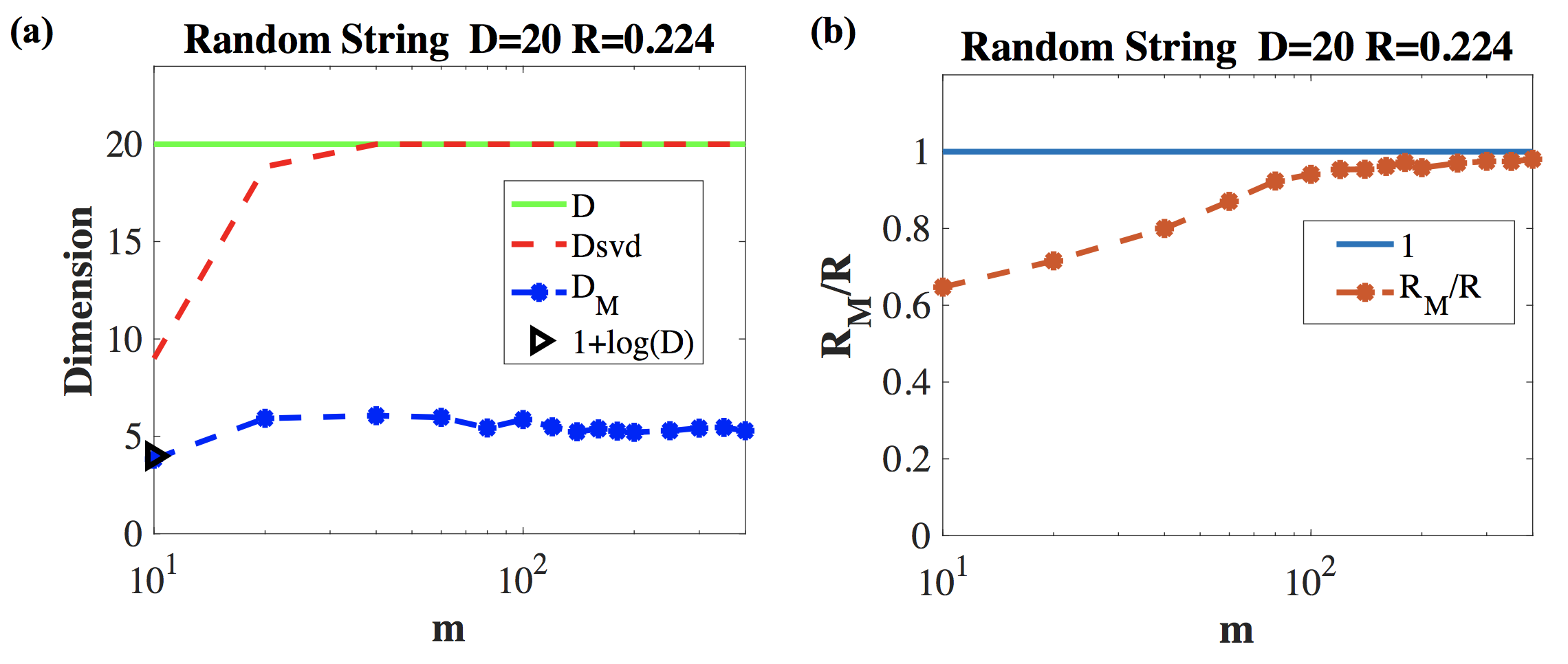}
\par\end{centering}
\caption{\textbf{Effective Properties Random Strings for Different Number of
Samples Per Manifold}. (a) Random String Dimension versus the Number
of Samples Per Manifold ($m$), for $N=100$, $D=20$, $R=0.224$($~1/\sqrt{D})$).
Both participation ratio, $D_{svd}$ (red) and the effective manifold
dimension found by the Mean Field iterative algorithm, $D_{M}$, saturates
around $m=50$, indicating that $m=200$ used in the Figure \ref{fig:RandString_ARD}
is already a saturated, $m$ -independant properties. (b) Random String
Radius versus the Number of Samples Per Manifold (m). The effective
manifold radius ($R_{M}$) found by the Mean Field iterative algorithm
saturates around $m=100$. \label{fig:RandString_M}}
\end{figure}

In the simulations for Figure \ref{fig:RandString_ARD}, we used $m=200$
training samples per each manifold. In the above figure we find that
the manifold dimensions $D_{M}$and radius $R_{M}$ is in the saturated
regime as a function of $m$, due to the fact that the samples are
coming from a string, and the sampling is already dense at $m=200.$
This is essentially like doing a local search for the worst points
(although not explicit), because we are in the densely sampled regime. 

\section{Discussion}

We have presented a mean field theoretic calculation of linear classification
of general manifolds, extending Gardner's replica theory of classification
of random isolated points\cite{gardner87,gardnerEPL} and the recently
developed theory of classification of random balls \cite{chung2016linear}.
This theory characterize the capacity as the inverse rate of reduction
in the entropy of the weight vector space by the separability constraints
per manifold. These constraints are expressed in terms of special
points on the manifolds that have minimum projections on self consistent
$D-$dimensional field vector. Importantly, we were able to derive
a set of universal mean field equations applicable to all low dimensional
convex manifolds, \ref{eq:alphaM_gen}-\ref{eq: z0_smin_gen}. The
key point is that for a given manifold geometry, the position of the
worst point on the manifold represented by $\vec{s}=\tilde{s}(\vec{t}-z_{0}\vec{s})$
changes as the field on the center represented by $t_{0}$ spanning
the sequence of increasingly large overlap between the manifold and
the the margin plane. This sequence depends of course on the details
of the geometry of the convex manifold. These equations cannot be
solved analytically except for the simplest geometries. We have developed
an iterative algorithm to solve these self consistent equations, and
in our experience, their converge is remarkably fast even when dealing
with $D$ dimensional manifolds with $D$ in the range of $10-100$. 

Of particular interest is the case of manifolds with high dimension
($D\gg1$) In this case, the key parameters are Manifold Radius, $R_{M}$
and Dimensionality $D_{M}$. The manifold capacity is equivalent to
that of $L_{2}$ balls with dimensionality and radius equal to $D_{M}$
and $R_{M}$ respectively. These quantities appear in the capacity
mainly through the excess margin $\kappa_{M}=R_{M}\sqrt{D_{M}}$.
The reason for this combination is the following. Consider first,
an $N$ dimensional ball. The margin is not dimensionless but depends
on the norm of the inputs. Thus, if we demand a margin $\kappa$this
is equivalent to demanding a margin $\kappa+R\sqrt{N}$from the centers,
since, if the distance of the points on the circumference of the ball
from the center is $R\sqrt{N}$. Now when the ball is low dimensional,
a reduction in the required excess margin occurs due to the tilt of
the ball with respect to the hyperplane, so that the projection of
the center should increase only by a factor $N\sqrt{D/N}=R\sqrt{D}$.
A similar argument holds for a general manifold in this limit. 

We have noted that the geometric parameters $R_{M}$ and $D_{M}$
are not intrinsic geometric measures but depend also on the overall
size of the manifolds relative to the center norm. Thus, if the manifold
is increased by scaling all the points by a global factor, say $r$,
relative to the center, then when $r$ grows, eventually $D_{M}$
approaches the full embedding dimension $D$, reflecting the need
of a solution to be orthogonal to the entire manifold subspace. Conversely,
when $r$ decreases, eventually the manifold reaches the scaling regime
where capacity is finite despite the high dimension, and $R_{M}\sqrt{D_{M}}$
approach the Gaussian Mean Width value $R_{W}\sqrt{D_{W}}=0.5MW$
, see \ref{fig:MeanWidthOfManifold} The change in the Manifold Dimension
as $r$ increases may be dramatic. For instance, in the random string
example (as well as $L_{1}$balls), $D_{M}$ increases from $D_{M}\sim\log D$
for small $r$ (the scaling regime) to $D_{M}=D$ for large $r$,
see \ref{fig:RandString_ARD}b. 

In conclusion, the generality of the theory developed in this Chapter
opens the door for applications of the derived results and methods
for the investigation of neuronal representations of perceptual manifolds
in biological as well as artificial neuronal networks. However, in
order to do so, some limitations of the current theory might need
to be relaxed. For instance, the present results deal with random
labels where the two classes are of roughly equal size. In many real
problems this may not be the case. Another issue is the assumption
of random orientation of the manifolds. It would be important to understand
the role of correlation between the manifolds. Also, it will be interesting
to explore extensions to manifold representations which are not linearly
separable. These issues and others are discussed in the next Chapter. 

 \textcolor{red}{}

\section{Appendix }

\subsection{Equivalence between $R_{i}$ and $\lambda_{i}$ }
\begin{lem}
If $\langle s_{i}^{2}R_{i}^{-2}\rangle=\frac{1}{D}$ then, $R_{i}\propto\lambda_{i}$.
\label{lem:RiLambdai}
\end{lem}
\begin{proof}
From the definition we know that $\langle s_{i}^{2}\rangle=\frac{R_{i}^{2}}{D}$.
Consider the covariance $C$ of the manifold data $X$ from many realizations
of $\vec{s}$ 

\[
X=\mathbf{x}_{0}+\sum_{i}s_{i}\mathbf{u_{i}}
\]

Then, 

\[
C=\langle\left\{ X-\langle X\rangle\right\} \left\{ X-\langle X\rangle\right\} ^{T}\rangle=\sum_{ij}\langle s_{i}s_{j}\rangle\mathbf{u}_{i}\mathbf{u}_{j}^{T}
\]

$\lambda_{i}$ are by definition, 

\[
\frac{1}{N}\lambda_{i}^{2}=\langle s_{i}^{2}\rangle
\]

Hence, 

\[
\lambda_{i}^{2}=\frac{N}{D}R_{i}^{2}
\]
\end{proof}

\subsection{Effective Dimension of $L_{1}$ manifolds \label{subsec:L1effD_2logD}}

In the scaling regime, $\vec{s}(\vec{t})$ for $L_{1}$ manifolds
is the $i$-th vertex where $i\underset{i}{=\text{argmax}\:t_{i}}$.
So, $\vec{s}\cdot\vec{t}=R\underset{i}{\text{max}\:t_{i}}$. Given
$D$ normally distributed $t_{i}$'s, the max value is centered in
large $D$ 

\begin{equation}
\underset{i}{\text{max}\:t_{i}}\sim\sqrt{2\log D}
\end{equation}

due to the extreme value theory. Hence, we obtain

\begin{equation}
D_{M}=2\log D\label{eq:L1DM_2LogD-2}
\end{equation}
\begin{equation}
R_{M}=R
\end{equation}

for $L_{1}$ manifolds in the scaling regime.

%% file: chapters/ch5_extensions_v7.tex
\chapter{Extensions \label{cha:extensions}}

\section{Correlated Manifolds}

So far, we have considered randomly oriented manifolds. In real world
data we expect that the manifolds will be correlated; in particular,
that the subspaces spanned by the different manifolds will be partially
aligned. We first consider the simple case of $D$ dimensional spheres
that all share the same subspace in the ambient dimension $\mathbb{R}^{N}$. 

\subsection{Parallel Spheres \label{subsec:ParallelSpheres}}

Consider a perceptron classifying parallel $D$-dimensional discs,
embedded in $N$ dimensions. The training data is given by:

\begin{equation}
\mathbf{x_{0}}^{\mu}+R\sqrt{\frac{1}{D}}\sum_{i=1}^{D}\mathbf{u}_{i}s_{i}\quad\forall\left|\vec{s}\right|^{2}\le1\label{eq:modelParallel}
\end{equation}

where as before the components of $\mathbf{x_{0}}^{\mu}$ and $\mathbf{u}_{i}$
are i.i.d. normally distributed random variables. For reasons that
will be clear later on, it is convenient to scale $R$ to $R/\sqrt{D}.$
It is also convenient to rotate the axes so that $\mathbf{u}_{i}$
are along the standard first $D$ axes, so that the corresponding
fields are 
\begin{equation}
h^{\mu}(\vec{s})=\frac{1}{\sqrt{N}}y^{\mu}\left(\mathbf{w}^{T}\mathbf{x}_{0}^{\mu}+R\sqrt{\frac{N}{D}}\sum_{i=1}^{D}w_{i}s_{i}\right)>0\quad\forall\left|\vec{s}\right|^{2}\le1\label{eq:fieldsParallel}
\end{equation}

where here we restrict ourselves to zero (fixed) margin. 

Minimizing with respect to $s_{i}$, we get

\begin{equation}
s_{i}=-y^{\mu}w_{i}/\sqrt{\sum_{j}w_{j}^{2}}
\end{equation}

Therefore, the minimum of LHS of \ref{eq:fieldsParallel} is 

\begin{equation}
\min_{s}h^{\mu}(s)=\frac{1}{\sqrt{N}}y^{\mu}\mathbf{w}^{T}\mathbf{x}_{0}^{\mu}-R\sqrt{\frac{1}{D}\sum_{i=1}^{D}w_{i}^{2}}
\end{equation}

Thus, the constraints are reduced to

\begin{equation}
\frac{1}{\sqrt{N}}y^{\mu}\mathbf{w}^{T}\mathbf{x}_{0}^{\mu}>\kappa_{\rho}\label{eq:parallelConstraint}
\end{equation}

where

\begin{equation}
\kappa_{\rho}=R\sqrt{\rho}\label{eq:marginParallel}
\end{equation}

and 
\begin{equation}
\rho=\frac{1}{D}\sum_{i=1}^{D}w_{i}^{2}\label{eq:p_overlap}
\end{equation}

which is an average of dot products between the solution \textbf{$\mathbf{w}$
}and direction vectors (in a rotated coordinate). Let us call $\rho$
the overlap parameter. 

Note that the total variance of each manifold is $R^{2}N/D$ while
the square distance between center pairs is: $2N$. Thus, their ratio
is $R^{2}/2D$ . In contrast, in the random spheres, the total variance
of each manifold is $R^{2}N$ whereas the square distance is as before
$2N$ . Hence the ratio is $R^{2}/2$ .

It is important to note that if $D\ll N$ then the solution vector
\textbf{$\mathbf{w}$} can be orthogonal to all manifolds by zeroing
the corresponding $D$ components for any size of the manifold $R$,
without sacrificing significant degrees of freedom. In this case,
we expect the capacity to be the same as for the capacity with the
centers only, in the reduced dimension of $N-D$. 

Thus, the problem of parallel manifolds is interesting only when $D\approx N$
($D$ scales with $N$). Let us define the relative dimension parameter
$d$, such that 

\begin{equation}
D=dN
\end{equation}
In this regime, zeroing all $D$ components of $\mathbf{w}$ is costly,
so we expect that the nature of \textbf{$\mathbf{w}$} depends on
$d$ and $R$. Note that in this case (since $D\gg1$) the relevant
scale of the radius should be radius over $\sqrt{D}$. In our normalization
above it means that $R$ is of order $1$. 

\subsubsection{Capacity}

The basic constraint \ref{eq:parallelConstraint} is equivalent to
a Gardner's theory with margin $\kappa$ which however is not a fixed
parameter but assumes self consistent value, set by the order parameter
$\rho$ which measures the overlap between $\mathbf{w}$ and the manifold
subspaces. To evaluate $\rho$ (Eqn. \ref{eq:p_overlap}) we need
to evaluate the entropy of the solution space given the constraint
that the solutions' average projection on the common manifold subspaces
is $\rho$ . We denote this entropy (per $N$) by$S_{0}(\rho)$. The
analog of \ref{eq:logVLargeQ} is 

\begin{equation}
\frac{1}{N}\langle\log V\rangle=S_{0}(\rho)-\alpha\alpha_{0}^{-1}(\kappa_{\rho})\label{eq:logVParallel}
\end{equation}

where $\alpha_{0}$ refers to the capacity for points. As before,
in the capacity limit, $\langle\log V\rangle$ vanishes and from \ref{eq:logVParallel}
it follows that the capacity is

\begin{equation}
\alpha_{||}=\alpha_{0}(\kappa_{\rho})S_{0}^{-1}(\rho)\label{eq:AlphaParallel}
\end{equation}

where the symbol $||$ stands for parallel manifolds and the entropy
term is 

\begin{equation}
S_{0}(\rho)=\frac{dx^{2}}{(x-1+d)^{2}+d(1-d)}\label{eq:ParallelEntropicTerm}
\end{equation}

where $x$ is related to $\rho$ through,

\begin{equation}
\rho=\frac{(x-1+d)^{2}}{d[(x-1+d)^{2}+d(1-d)]}\label{eq:RelateRhoWithXD}
\end{equation}

Different overlap $\rho$ yields different capacities, so optimizing
the capacity $\alpha_{||}$ with respect to $\rho$ yields the following
equation (for $\rho$ or $x$), as $\alpha_{||}$ is the maximum possible
capacity

\begin{equation}
\frac{x(1-x)\sqrt{d}}{\sqrt{(x-1+d)^{2}+d(1-d)}}=\alpha_{||}R\left(\kappa_{\rho}H(-\kappa_{\rho})+\frac{\exp-\frac{1}{2}\kappa_{\rho}^{2}}{\sqrt{2\pi}}\right)\label{eq:optimizeAlphaWithRho}
\end{equation}

Since there are 3 equations (Eqns \ref{eq:AlphaParallel}-\ref{eq:ParallelEntropicTerm},
\ref{eq:RelateRhoWithXD} and \ref{eq:optimizeAlphaWithRho}), and
3 unknowns ($\alpha_{||}$, $x$ , $\rho$), one can solve for $\alpha_{||}$
or $\rho$. 

\subsubsection{Phase Transition}

The solution for the capacity above shows dependence on the overlap
parameter $\rho$, which needs to be solved via $x$ as a function
of $R$. It turns out that there are two regimes of solutions for
$\rho$, one where $\rho$ decreases with increasing $R$, up to $R<R_{c}$,
and another where $\rho=0$ for $R>R_{c}$. Qualitatively, this means
that for the parallel manifolds whose radii are smaller than $R_{c}$
, the overlap between $\mathbf{w}$ and manifold subspace is nonzero,
and the overlap increases with increasing $R$. However, when the
manifold $R$ is beyond a critical value $R_{c}$, all manifold subspaces
are orthogonal to $\mathbf{w}$ and the overlap becomes $0$, due
to the tradeoff between orthogonalizing and sacrificing the degrees
of freedom. The geometric intuition for two different phases is given
in Fig. \ref{fig:Parallel_illu_ov_alpha} (a)-(b). 

Let us consider the value of $R_{c}(d)$ such that $\rho=\kappa\rightarrow0$
as $R\rightarrow R_{c}^{+}$. At this value, the solution must be
orthogonal to the manifold, therefore $\alpha=2(1-d)$ (Fig. \ref{fig:Parallel_illu_ov_alpha}(c),
Regime $R>R_{c}$ ). Also, for $\rho$ to vanish, $x=1-d$, yielding,

\noindent\fbox{\begin{minipage}[t]{1\columnwidth \fboxsep \fboxrule}%
\begin{equation}
R_{c}(d)=\sqrt{\frac{\pi d^{2}}{2(1-d)}}\label{eq:Rc(d)_parallel}
\end{equation}
\end{minipage}}

Thus for $R>R_{c}(d)$, $\alpha_{||}=2(1-d)$ and $\kappa_{\rho}=\rho=0$.
See Figure \ref{fig:Parallel_illu_ov_alpha}-(c). As the figure \ref{fig:Parallel_illu_ov_alpha}shows,
the predictions agree well with numerical simulations. 

\subsubsection{Field Distribution}

Since the field distribution is determined by the set of constraints
on the fields, in our case they should be equivalent to the distribution
of fields in the Gardner's theory with margin given by \ref{eq:marginParallel}
(see \ref{eq:logVParallel}). This means that as long as $\rho>0$
i.e. $R<R_{c}$ , $\mathbf{w}$ is not in the null space of any manifold.
The fraction of manifolds that touch the margin is given (as in the
Gardner theory) by

\begin{equation}
B=H(\kappa)
\end{equation}

and the fraction that are interior is $H(-\kappa).$ When $R>R_{c}$
, $\mathbf{w}$ is in the null space of the manifolds. Half of the
manifold center are on the margin plane and half are not. See Figure
\ref{fig:Parallel_illu_ov_alpha}(d). 

\begin{figure}[H]
\begin{centering}
\includegraphics[width=1\textwidth]{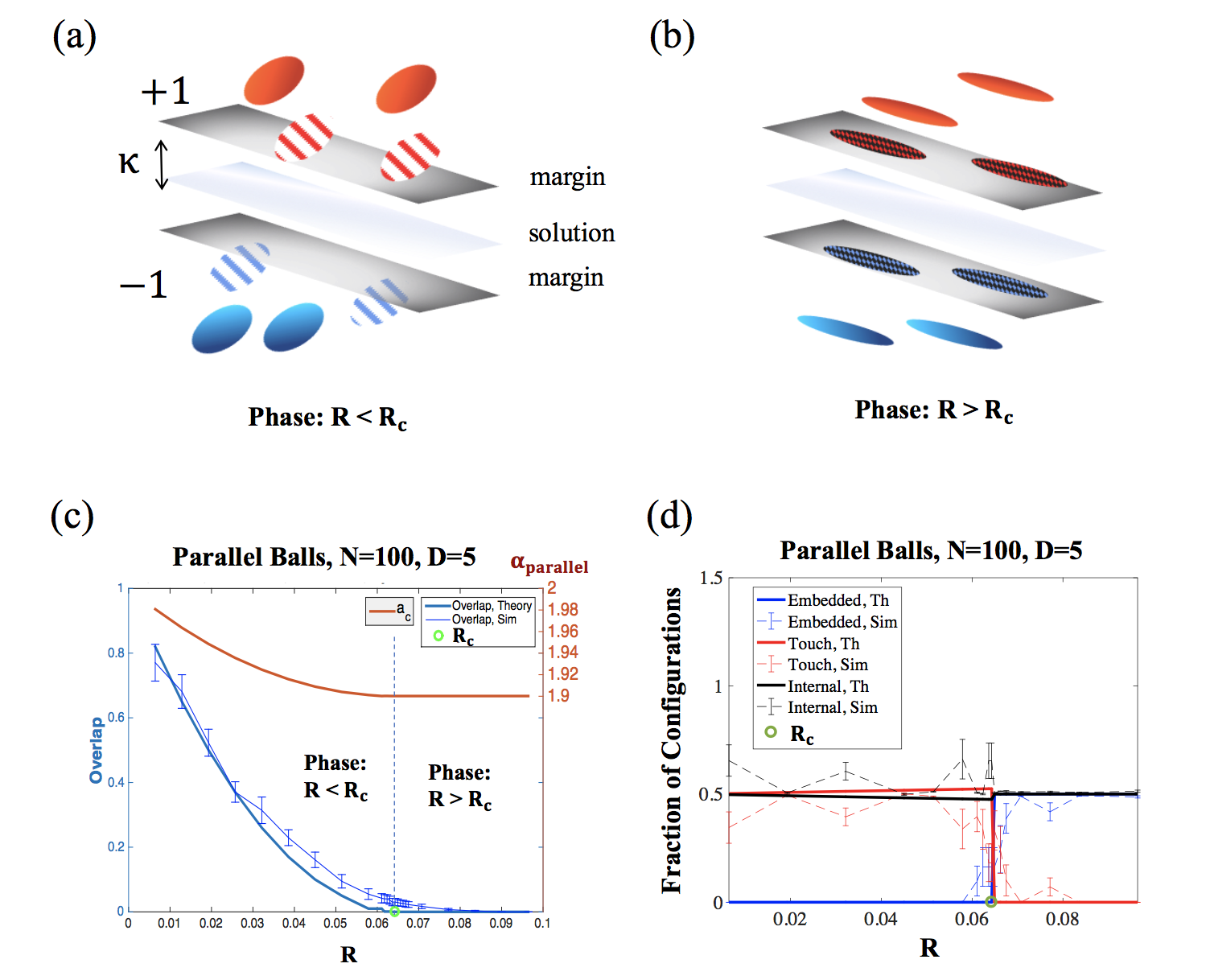}
\par\end{centering}
\caption{\textbf{The Phase Transition of Linear Classification of Parallel
Manifolds. }(a) When $R<R_{c}$ , manifolds are either interior (solid)
or touching the margin planes (striped). (b) When $R>R_{c}$ , manifolds
are either interior or embedded in the margin planes (diamonds). (c)
Phase transition of overlap and capacity at critical radius $R_{c}$,
for classification of parallel balls with $D=5$ in ambient dimension
$N=100$. The overlap ($\rho$ , Eqn. \ref{eq:p_overlap}) between
the manifold axes and the $\mathbf{w}$ denoting the solution hyperplane
vanishes as the ball's radius $R$ approaches $R_{c}$. The capacity
$\alpha_{parallel}$ also goes through phase transition as well, however
the value of capacity is stays large, as the embedding dimensions
by all manifolds are limited to $D$. (d) Phase transition of manifold
configurations at crucial radius $R_{c}$. When $R<R_{c}$ , most
manifolds are either interior or touching the margin plane (as shown
in (a)), and when $R>R_{c}$ , most manifolds are either embedded
or touching the margin plane (as shown in (b)). The simulations and
theory show good agreement. \label{fig:Parallel_illu_ov_alpha}}
\end{figure}

\subsection{Discussion}

We can generalize this analysis to partially parallel correlated balls,
where only a fraction of dimensions ($D_{1}=d_{1}N<D$) are shared
between them and the rest of the subspaces ($D_{2}=D-D_{1}$) are
random. In this case, in the regime where the shared subspaces are
orthogonalized, the problem remains as the classification of balls
in the null space of the shared subspaces. In this case, the solution
will take the form of 

\begin{equation}
\alpha_{||}=\alpha_{B}(\kappa_{\rho},R_{\rho},D_{\rho})S_{0}^{-1}(\rho)
\end{equation}

where $\alpha_{B}$ is the capacity of balls and $\kappa_{\rho}$
is the effective margin in the null space due to the radii in the
parallel subspaces and $R_{\rho}$is the effective radius in the null
space due to the random directions and $D_{\rho}$ is the effective
ball dimension in the null space due to the random directions. Furthermore,
the problem can be generalized to other types of correlations (i.e.
correlation between the manifold subspace and the center of the manifold,
or the correlation between centers.) We hope to explore these issues
of various types of correlations to take into account the structures
in the realistic data. 

\section{Mixtures of Shapes \label{sec:MixShapes}}

\begin{figure}[H]
\begin{centering}
\includegraphics[scale=0.6]{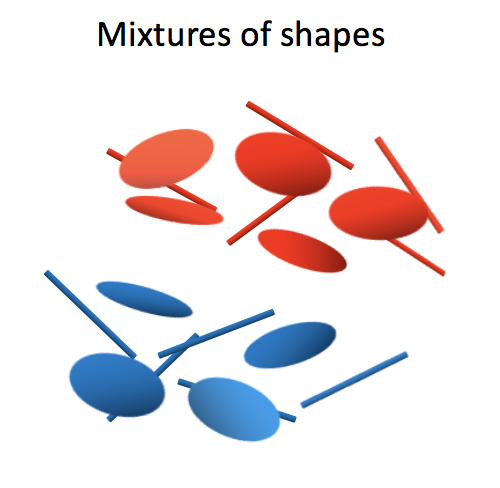}
\par\end{centering}
\caption{\textbf{Linear Classification of Mixtures of Shapes} (Illustration).
The linear classification capacity of mixtures of shapes is given
by Eqn. \ref{eq:CapacityMixture}. }
\end{figure}

So far, we have discussed the classification of manifolds with same
shapes and sizes. In this section, we generalize the problem to the
classification of manifolds with \emph{different} shapes and sizes.
Suppose there are $C$ \emph{different} manifold types, and each of
these types have the capacity of $\alpha_{s}$ where $s=1,...,C$
($s$ for shape). 

In this problem, the self-consistent term $G_{1}$ (the free energy
term) is an average of each $G_{1}$ of the classification problem
of manifolds of the shape $s$. Recall that $G_{1}$ of each shape
determines the capacity for each shape, through $\alpha_{s}^{-1}$. 

Then, the linear classification capacity of the mixtures of $C$ different
manifold types can be simplified to

\begin{equation}
\alpha_{mixture}^{-1}=\langle\alpha_{s}^{-1}\rangle{}_{s}\label{eq:CapacityMixture}
\end{equation}

where $s=1,...,C$ refers to the index for each manifold type. 

This remarkably simple and general theoretical result opens doors
to the treatment of a vastly diverse set of manifold classification
problems, from classification of manifolds different shapes and sizes
to different ratios of labels. 

\section{Class Imbalance \label{sec:Sparse-Coding}}

So far, we have covered the binary classification of manifolds where
the number of positively labeled manifolds is equal to the number
of negatively labeled manifolds. Here we consider the class imbalance problem, where the number of positive labels is far less than the number of negative labels, or vice versa. This is also known as classification with sparsity in the labels. In the theory of classification of
points, increasing the sparsity of labels has been known to increase
the point classification capacity by orders of magnitude (\cite{gardner1988space}). 
In this section, we ask the question whether increasing the sparsity
of manifold labels also improve the manifold classification capacity. 

Note that classification of manifolds with sparse labels is an important
example of classification with inhomogeneous manifolds (Section \ref{sec:MixShapes}).
Notice that sparsity term $f$ is defined as the fraction of positively
(or negatively) labeled manifolds out of the total number of manifolds,
so a large sparsity actually refers to a small $f$ (Fig. \ref{fig:SparseBalls}). 

One important thing to note is that in the balanced binary classification
case where the number of positive and negative labels are equal ($f=0.5$),
the bias term, $b$, of the linear classification 

\begin{equation}
y^{\mu}\left\{ \mathbf{w}^{T}\left(\mathbf{x_{\mathbf{0}}}^{\mu}+\sum_{i=1}^{D}s_{i}\mathbf{u}_{i}^{\mu}\right)+b\right\} \geq\kappa||\mathbf{w}||
\end{equation}

was ignored because the optimal bias term which maximizes the classification
capacity with balanced labels is $b=0$. However, in the sparse label
case, this is no longer true, and the nonzero bias term needs to be
included in the evaluation of the capacity, and the bias term needs
to be optimized. 

\subsection{Sparse General Manifolds }

Using the similar analysis as the calculation of perceptron capacity
for mixtures of shapes (Section \ref{sec:MixShapes}), we average
the inverse of capacities for the positive and negative labels, and
arrive at the expression for the manifold capacity with sparse labels. 

\paragraph{Sparse General Manifolds }

With this, we can express capacity of general manifolds with label
sparsity $f$ as

\textcolor{black}{
\begin{equation}
\alpha_{M}^{-1}(\kappa,f)=f\alpha_{M}^{-1}(\kappa+b_{max})+(1-f)\alpha_{M}^{-1}(\kappa-b_{max}).\label{eq:alpha_sparse_M}
\end{equation}
}

where 

\textcolor{black}{
\[
b_{max}=\mbox{argmax}_{b}\alpha_{M}(\kappa,f,b)
\]
}

and 

\[
\alpha_{M}^{-1}(\kappa,f,b)=f\alpha_{M}^{-1}(\kappa+b)+(1-f)\alpha_{M}^{-1}(\kappa-b)
\]

\paragraph{Sparse D-dimensional Balls }

As an example, we can express capacity of $D$-dimensional $L_{2}$
balls of radius $R$ with label sparsity $f$ as 

\textcolor{black}{
\begin{equation}
\alpha_{B}^{-1}(\kappa,R,D,f)=f\alpha_{B}^{-1}(\kappa+b_{max},R,D)+(1-f)\alpha_{B}^{-1}(\kappa-b_{max},R,D).\label{eq:alpha_sparse_balls}
\end{equation}
}

where $b_{max}$ needs to be found so that it maximizes the capacity
with sparsity $f$ , i.e. 

\textcolor{black}{
\begin{equation}
b_{max}=\mbox{argmax}_{b}\alpha_{B}(\kappa,R,D,f,b)
\end{equation}
}

where 

\begin{equation}
\alpha_{B}^{-1}(\kappa,R,D,f,b)=f\alpha_{B}^{-1}(\kappa+b,R,D)+(1-f)\alpha_{B}^{-1}(\kappa-b,R,D).
\end{equation}

This result implies that similarly to the case of the points, the
manifold capacity increases significantly with the increased sparsity
(reduced $f$) (Fig. \ref{fig:SparseBalls}). 

\paragraph{Fraction of Support Structures }

In a similar manner to the replica calculation of fraction of support
structures in \cite{chung2016linear}, the fraction of support structures
for sparse labels can be calculated. Note that due to the asymmetry
in the number of positive and negative labels and the non-zero bias,
the terms with majority labels and non-majority labels are different.
We give here the expression for the fraction of support structures
for $D$-dimensional balls. First, manifolds with non-majority labels,
which consists of $f$ of the total manifolds can be either embedded,
touching, or in the interior side of the shattered space. All together,
they consist of the first (non-majority) term of the capacity expression
with coefficient $f$, in Eqn. \ref{eq:alpha_sparse_balls}. 

With this, we can derive the fraction of support structures of manifolds
with sparse labels. Unlike the problem with dense labels, the minority
manifolds (red manifolds in Fig. \ref{fig:SparseBalls}(a)) and the
majority manifolds (blue manifolds in Fig. \ref{fig:SparseBalls}(a))
have different behaviors. 

First, fraction of embedded manifolds with non-majority labels is 

\begin{equation}
p_{emb}^{minor}=f\int_{0}^{\infty}dt\chi{}_{D}(t)\left[\int_{-\infty}^{\kappa+b-\frac{1}{R}t}Dt_{0}\right]\label{eq: p-sparse-emb}
\end{equation}

Then, fraction of manifolds with non-majority labels that touch the
margin plane is 

\begin{equation}
p_{touch}^{minor}=f\int_{0}^{\infty}dt\chi{}_{D}(t)\left[\int_{\kappa+b-\frac{1}{R}t}^{\kappa+b+Rt}Dt_{0}\right]\label{eq:p-sparse-touch}
\end{equation}

The rest of the non-majority manifolds are those in the interior space
shattered by the margin planes

\begin{equation}
p_{interior}^{minor}=f\int_{0}^{\infty}dt\chi{}_{D}(t)\left[\int_{\kappa+b+Rt}^{\infty}Dt_{0}\right]\label{eq:p-sparse-interior}
\end{equation}

Note that $p_{emb}^{minor}+p_{touch}^{minor}+p_{interior}^{minor}=f$.
Similarly, the fraction of embedded manifolds with majority labels
is 

\begin{equation}
p_{emb}^{major}=(1-f)\int_{0}^{\infty}dt\chi{}_{D}(t)\left[\int_{-\infty}^{\kappa-b-\frac{1}{R}t}Dt_{0}\right]\label{eq:p-nonsparse-emb}
\end{equation}

The fraction of manifolds that touch the margin plane with majority
labels is 

\begin{equation}
p_{touch}^{major}=(1-f)\int_{0}^{\infty}dt\chi{}_{D}(t)\left[\int_{\kappa-b-\frac{1}{R}t}^{\kappa-b+Rt}Dt_{0}\right]\label{eq:p-nonsparse-touch}
\end{equation}

The fraction of the manifolds with majority labels that are in the
interior space shattered by the margin planes is 

\begin{equation}
p_{interior}^{major}=(1-f)\int_{0}^{\infty}dt\chi{}_{D}(t)\left[\int_{\kappa-b+Rt}^{\infty}Dt_{0}\right]\label{eq:p-nonsparse-interior}
\end{equation}

Like with the sparse case, Note that $p_{emb}^{major}+p_{touch}^{major}+p_{interior}^{major}=1-f$.
The example of this theoretical prediction is tested in Fig. \ref{fig:SparseBalls}(c),
where we show that it matches well the fraction of sparse manifold
structures in the case of $1D$ manifolds (lines). 

\paragraph{Simulations }

The linear classification capacity for $2$ dimensional balls with
sparsity $f$ has been evaluated numerically and theoretically. For
the numerical evaluation, $M_{4}$ algorithm (Chapter \ref{cha:m4})
has been used with sparse labels. For the fraction of support structures,
$1D$ line manifolds with sparse labels were used. For the theoretical
evaluation, the Eqn. \ref{eq:alpha_sparse_balls} and Eqn.\ref{eq: p-sparse-emb}
- \ref{eq:p-nonsparse-interior} has been used, and agree well with
the simulations. 

\begin{figure}[H]
\begin{centering}
\includegraphics[clip,width=1\textwidth]{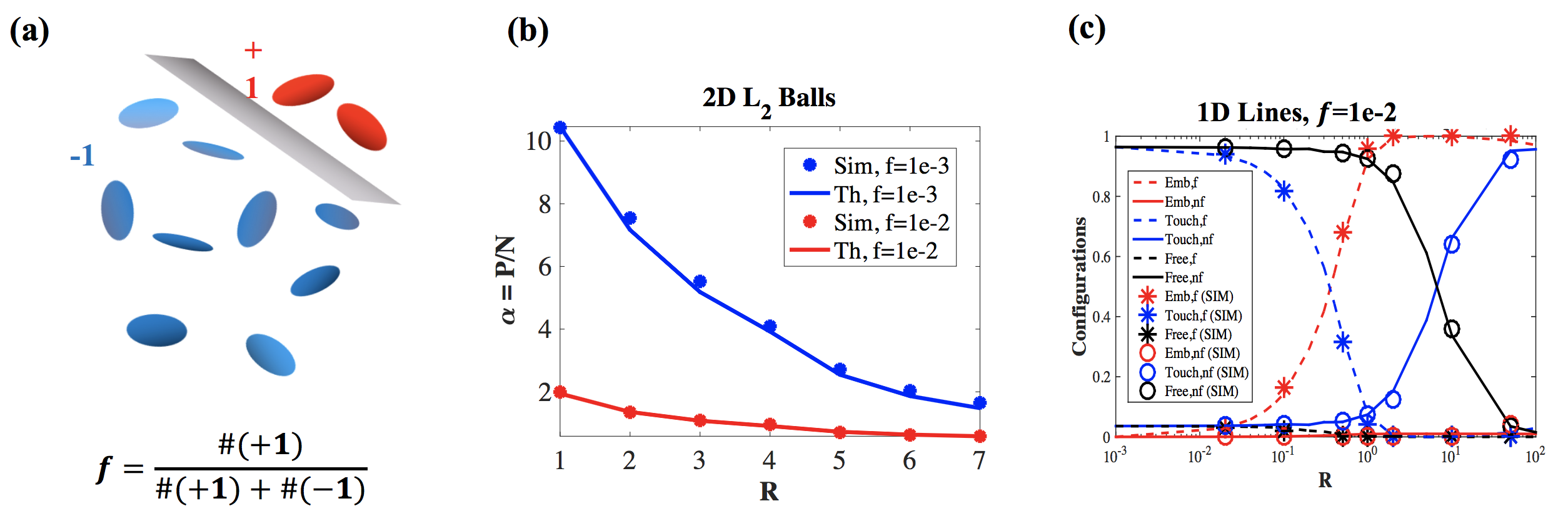}
\par\end{centering}
\caption{\textbf{Linear Classification of Balls with Sparse Labels. }(a) (Illustration)
The solution hyperplane (grey) separates manifolds, where the fraction
of positively labeled manifolds out of the total number of manifolds
are given by $f\ll1$ . (b) Capacity of $2D$ $L_{2}$ balls with
sparsity $f=0.001$(blue) and $f=0.01$ (red). Theory (line) matches
simulations (markers) well. (c) Support configurations of $1D$ line
segments with sparsity $f=0.01$, for majority labels (denoted as
$nf$, solid line in the legend) and non-majority (sparse) labels
(denoted as $f$ , dashed line in the legend). Theory matches simulations
well. As $R$ is increased, the fraction of embedded line segments
becomes 1, and the transition happens at smaller $R$ in manifolds
with non-majority label compared with manifolds with majority labels.
\label{fig:SparseBalls}}

\end{figure}

Note that in Fig. \ref{fig:SparseBalls}(c), when $R$ starts out
small, most of the non-majority manifolds are touching, and most of
the majority manifolds are interior, and as $R$ is increased, the
phase transition where most of the sparse manifolds become embedded
happens first, and then, the phase transition where most of the non-sparse
manifolds become touching happens. 

\subsection{Small $f$ Regime }

Let us focus on the $D=1$, the classification of lines with sparsity
$f$ . What is the behavior of the line capacity $\alpha_{L}$ , in
the case of extreme sparsity, i.e. $f\rightarrow0$? The dominant
term analysis in different regimes of $R$ gives the following analytical
approximations for capacity with $\kappa=0$. 

1. $R=O(1)$

\begin{equation}
\alpha_{L}^{-1}(R,f)=2(1+R^{2})f|\log f|
\end{equation}

2. $f^{-1/2}\gg R\gg1$ 
\begin{equation}
\alpha_{L}^{-1}(R,f)=2R^{2}f|\log R^{2}f|
\end{equation}

3. $f^{-1}\gg R\gg f^{-1/2}$
\begin{equation}
\alpha_{L}^{-1}(R,f)\approx1-\frac{2}{\pi R^{2}f}
\end{equation}

4. $R\gg f^{-1}$, $R\rightarrow\infty$
\begin{equation}
\alpha_{L}^{-1}(R,f)=1+2f|\log f|
\end{equation}

It is interesting to note that in the limit of large $R,$ the capacity
does not depend on the $R$ any more. 

\subsubsection{Object Recognition Limit, $f=1/P$ }

Particularly interesting regime is when the sparsity $f$ is equivalent
to $1/P$ where $P$ is number of manifolds. This is analogous to
the \emph{one-vs-all} task in the multi-class classification problem,
where the output unit is activated only when the input comes from
the correct class out of the possible $P$ classes. Can we estimate
the capacity of object manifolds, in this relevant sparse object recognition
limit? 

Using our theory, the minimum network size $N^{*}$required in order
to classify one manifold out of $P$ given manifolds can be estimated
to be

\begin{equation}
N^{*}=P/\alpha_{M}\left(f=\frac{1}{P}\right)
\end{equation}

where we can use $1/P$ as sparsity $f$ . 

One can also estimate the largest allowed size of the manifolds if
one is given with the network size $N$ and the number of object manifolds
$P$. That is, solve for $R$ and $D$ such that 

\begin{equation}
\frac{P}{N}=\alpha_{\text{B}}(\kappa,R,D,\frac{1}{P})
\end{equation}

which can be found numerically. 

\paragraph{Simulation Results and Discussion}

We tested below the capacity of line segments, in the one-vs-all object
recognition task limit. In the simulations where critical $P^{*}$(number
of classes/manifolds) had to be found, we fixed the network size,
$N=100$ and $N=200$, and $P$ was varied to find $P^{*}$ at which
the probability of finding a linear classifier goes from 1 to 0, given
100 repetitions. The theoretical prediction matches the simulation
capacity well (Fig. \ref{fig:obj_rec_lim}(a)).

\begin{figure}[H]
\begin{centering}
\includegraphics[width=1\textwidth]{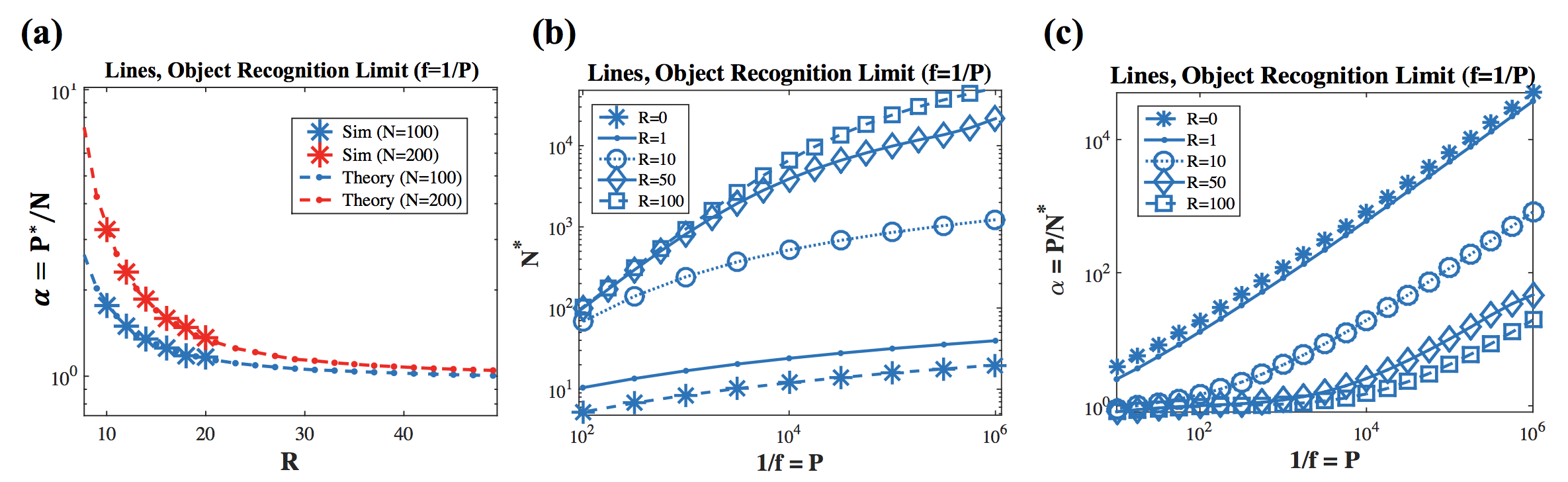}
\par\end{centering}
\caption{\textbf{Line Classification of Line Segments in Object Recognition
Limit ($f=1/P$)}. (a) Capacity $(\alpha_{L}(f=1/P))$ as a function
of size $R$ for different network sizes. (b) Minimum network size
required for linear separability ($N^{*}$) as a function of number
of manifolds/objects ($P$), for different sizes $R$ (c) Capacity
as a function of number of manifolds/objects ($P$). \label{fig:obj_rec_lim}}
\end{figure}

If we want to classify $P=1000$ classes, whose object manifolds are
1 dimensional, with one-versus-all task, how many neurons are required
for the problem to be linearly separable? Figure \ref{fig:obj_rec_lim}(b)
provides an answer to that question with our theoretical estimate
of minimum required network size $N^{*}$. In the limit of point,
$R=0$, roughly $N^{*}=10$ is enough to classify 1000 objects. However,
if $R>50$, at least $N^{*}=1000$ is required to be linearly separable. 

Notice that the capacity in the units of load $(P/N)$ shows an interesting
behavior in the limit of large $R$ in Figure \ref{fig:obj_rec_lim}(c).
Notice that in large $R$ , the capacity is dominated by $1/D$ and
the improvement due to sparsity is smaller than when $R$ is small.
In other words, if the manifold sizes are large, making the labels
sparse does not improve the capacity as much as when the manifold
sizes are small. Therefore, the effect of size of the manifold $(R)$
on the capacity is more dramatic in the case of classification with
sparsity in the object recognition limit (compared to the dense label
classification). 

\section{Building Multi Layer Networks of Sparse Classifiers\label{sec:Multilayer}}

\textcolor{black}{In the section \ref{sec:Sparse-Coding}, we showed
how introducing the sparsity in the manifold labels improves substantially
the classification capacity of manifolds. Here we show that we can
use this feature to solve }\textcolor{black}{\emph{dense }}\textcolor{black}{classification
task. The general idea is as follows. Suppose we have an input neural
layer with size $N$ representing $P$ manifolds and a dense classification
task (i.e. label sparsity $f$ is $\sim0.5$), such that a linear
classifier applied directly to this input layer fails to classify
all stimuli correctly, namely, the manifolds in the input representations
are not linearly separable. }

\textcolor{black}{To solve the task we add a single hidden layer with
$M$ binary units (Fig. \ref{fig:sparseMulti}(a)). We would like
to generate a hidden layer representation of the manifolds that is
invariant, namely that all inputs from a given manifold are mapped
to a single activity pattern in the hidden layer. If we can achieve
this, the invariant representations in the hidden layer can easily
be linearly separable. In order to generate this invariant representations
in the intermediate layer, we generate $M$ random sparse labels for
each manifold, and learn the connection from the input layer to the
intermediate layer as a sparse linear classification in each unit
in the intermediate layer, which we assume is below the capacity and
therefore can be implemented wi}thout error. In the following subsection
below, we analyze the range of parameters and the performance of this
two layer network. 

\noindent \textcolor{black}{}
\begin{figure}[H]
\noindent \begin{centering}
\textcolor{black}{\includegraphics[width=1\textwidth]{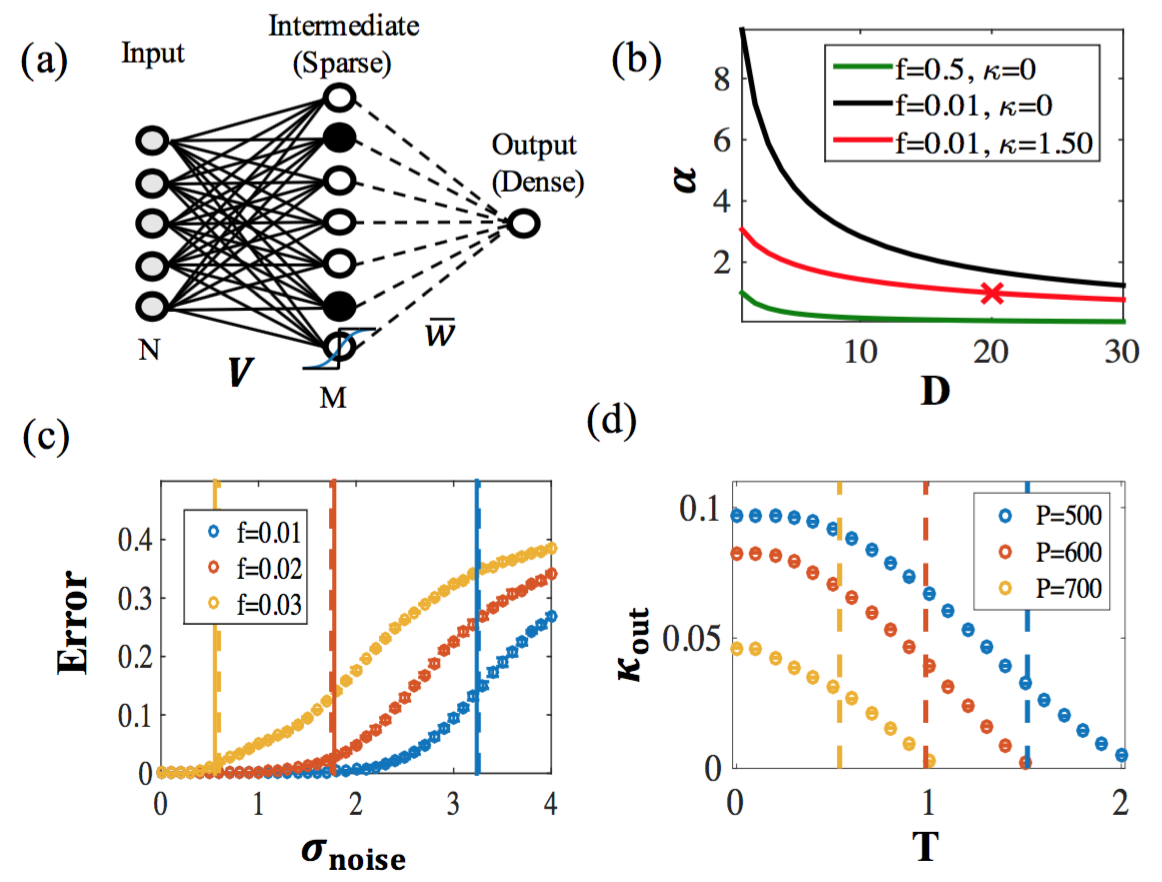} }
\par\end{centering}
\textcolor{black}{\caption{\textbf{\textcolor{black}{Sparse Intermediate Representation Enhances
the Invariant Processing of Manifolds.}}\textcolor{black}{{} (a) Input
layer with dimension $N$, where input vectors are drawn from a set
of $P$ manifolds. The weight matrix, $\boldsymbol{V}$, from the
input to the intermediate layer is constructed so that each node in
the intermediate layer is activated in invariant manner by randomly
chosen}\textcolor{black}{\emph{ fP}}\textcolor{black}{{} manifolds (where
$f$ is sparsity). This yields a sparse representation where all inputs
from the same manifold activates the same }\textcolor{black}{\emph{fM
}}\textcolor{black}{intermediate nodes. Output node classifies the
manifolds with desired dense binary labels. (b) Perceptron capacity
of manifold classification, $\alpha=P/N$ versus manifold ball dimension
$D$ at $R=1$, for different sparsity ($f$) and margin ($\kappa$).
The X marker denotes the working point for simulation in (c-d). (c-d):
Robustness to noise. (c) Probability of error at output layer, versus
the standard deviation of the input additive Gaussian noise $(\sigma_{\mbox{noise}})$
for different intermediate layer sparsity $f$. Input manifold dimension
and radius are: $D=20$ and $R=1$. For simulation, $N=M=500$, and
number of manifolds $P=250$ (which is five times the single layer
capacity, $P_{\mbox{SL}}=\alpha_{\mbox{SL}}N\approx48$, for $\kappa=0$,
for these manifolds). Markers indicate the simulation results for
the error for different sparsities. Robustness to noise is achieved
by ensuring a significant margin $\kappa_{\mbox{int}}$ at the intermediate
nodes; the margins $\kappa_{\mbox{int}}$ are shown as vertical lines.
(dashed) simulation (solid) analytical prediction. Higher sparsity
($f=0.01)$ ensures larger margin, and more robustness to input layer
variability. (d) Output margin ($\kappa_{\mbox{out}}$) versus the
smoothness parameter ($T$) of the sigmoid in the intermediate layer
$(1+e^{-x/T})^{-1}$, for different number of manifolds, $P$. $T=0$
is the binary limit. For the simulation, $N=500$, $M=500$, $f=0.01$,
$R=1$, $D=20$ was used. Values of $P$ are an order of magnitude
larger than the single layer capacity of the input layer which is
$P_{\mbox{SL}}\approx48$. (vertical lines) intermediate layer margins.
For the same $f$, smaller number of manifolds ($P=500$) allows larger
margin, and higher robustness to the smoothing of the intermediate
layer responses and the resultant manifold variability. }\textcolor{black}{\footnotesize{}\label{fig:sparseMulti}}}
}
\end{figure}

\noindent \textcolor{black}{}%

\subsection{Capacity of Two-Layer Network with Sparse Invariant Representation }

\noindent \textcolor{black}{As a simple example, we focus on input
manifolds of $D$ dimensional balls with radius $R$ such that $P$
(number of balls) is larger than the linear classification capacity
of such balls, $N\alpha_{B}(R,D)$. We first compute the perceptron
capacity for manifolds with sparse labels, parametrized as above by
$\kappa$ $D$, $R$ and $f$ , $\alpha_{B}(\kappa,R,D,f)$ , where
$f$ is the fraction of positive examples. Similar to classification
of points\cite{gardner1988space,babadi2014sparseness}, the perceptron
capacity of manifolds with sparse $f\ll1,$ is much higher than when
the labels are dense ($f=0.5)$ (Fig. \ref{fig:sparseMulti}(a)).
Consider now the task of invariant classification of manifolds where
the task labels are dense, (e.g., $f=0.5,$ and $\kappa=0$). If the
number of manifolds $P$ relative to the size of the input layer,
$N$ is above $\alpha_{B}(\kappa,R,D,0.5$), then the single layer
architecture will be unable to solve the linear classification task.
However, we can use the improved perceptron capacity for sparse labels
(Fig. \ref{fig:sparseMulti}(b)), to construct an intermediate representation
of $M$ nodes, each one of them trained for a randomly chosen sparse
labeling of the manifolds. These sparse labels are unrelated to the
task labels; they are used solely for building the intermediate representation.
As long as $P/N$ is below the perceptron capacity for sparse classification
of manifolds, the resultant intermediate layer generates an invariant
representation. The subsequent single layer readout at the output
layer can then perform the required dense classification as long as
$P/M<\alpha_{0}(\kappa,f=0.5)$, namely the perceptron capacity for
dense labels of }\textcolor{black}{\emph{points}}\textcolor{black}{.
The overall capacity for classification of manifolds of this two-layer
network is given by }$\alpha_{B}^{TL}(0,R,f=0.5)=\mbox{min}\left(2M,\:N\alpha_{D}(0,R,f)\right)/N$,
where $TL$ stands for two layer, much higher than that of the single
layer ($SL$), $\alpha_{\mbox{SL}}=\alpha_{B}(0,R,D,0.5)$ (for $M\geq N$)
\textcolor{black}{(Fig. \ref{fig:sparseMulti}(c-d) captions)}. In
addition, below the zero-margin capacity, $\alpha_{\mbox{D}}^{TL}(0,R,f=0.5)$,
the maximum margin $\kappa_{\mbox{out}}$ at the output node is given
by $\alpha_{0}(\kappa_{\mbox{out}},0.5)=P/M$.

\subsection{Robustness to Noise }

\noindent \textcolor{black}{This two-layer architecture shows not
only enhanced capacity, but can also enhance the system's robustness
to noise. To achieve robustness, the sparsity of the intermediate
representation should be sufficiently large to have a significant
margin, $\kappa_{\mbox{int}}>0$ in their representation (where subscript
'int' stands for intermediate layer), i.e, $P/N<\alpha_{D}(\kappa_{\mbox{int}},R,f)$.
Then adding noise may cause the sparse representation in the intermediate
layer to be only approximately invariant to the input manifold degrees
of freedom, nevertheless the effect on the performance will be small
provided the noise level is small compared with $\kappa_{\mbox{int}}$.
We demonstrate this in two cases. First, additive full-rank Gaussian
noise was introduced to the input layer. As shown in (Fig. \ref{fig:sparseMulti}(c))
the two layer network is robust to a range of noise values, and networks
with sparser intermediate representations have output error probabilities
close to zero for a larger range of noise.} For output noise in the
intermediate layer we haven't analyzed the performance explicitly,
but we present the heuristic analysis based on the idea that the level
of robustness to noise in each unit should be determined by the size
of each intermediate unit's classification margin $\kappa_{\text{int}}$.
The details of the numerics are given in the appendix to the chapter
(Section \ref{subsec:Multilayer_Appendix}). 

\noindent \textcolor{black}{Next, stochasticity in the activation
of the intermediate binary units is modeled by changing their activation
function to smooth sigmoidal units with gain parameter $1/T$. Smooth
rate functions are also more representative of biologically realistic
rate-based models of neural networks. Fig. \ref{fig:sparseMulti}(d)
shows that for a broad range of $T$ the output readout was able to
correctly classify the manifolds. The critical value of $T$ above
which classification fails is roughly given by the margin of the intermediate
layer. These examples show how to construct a network able to classify
manifolds with small error even when the intermediate layer is not
completely invariant to the manifold representation.}

\subsection{Discussion }

\noindent \textcolor{black}{In this section, we have shown that by
using classifiers of manifolds with sparse labels, a two layer network can be constructed
with enhanced manifold processing capacity and robustness to noise.
Thus, our theory provides a biologically plausible simple feedforward
network model that is capable of processing object relation information
in an invariant manner. The current theory can be extended to in several
important ways. Here we focused on training the network weights with
full manifolds and adding an additive noise after training, but the
network weights can be trained with subsamples of manifolds, or with
noisy realizations of manifolds. Here we focused on the intermediate
layer nonlinearity to be a sigmoid function, but other types of nonlinearities
such as ReLu can be considered, which can have a different effect
on the reformatted shape of the manifold. Intermediate layer invariance
can be achieved by different methods such as max-pooling or polynomial
nonlinearities, and we hope to explore the role of such processing
on the reformatting of the manifolds using the manifold capacity framework
in the future. }

\section{Kernel Extensions}

In section \ref{sec:Multilayer}, we showed how an additional intermediate
layer with sparsity can improve the output readout capacity for manifolds.
This is an example where a nonlinearity in each unit in the intermediate
layer created a new representation that is easier to be read out by
the output unit. In this section, we show another example of how nonlinear
processing reformats the input manifolds so that the output linear
separability is improved, by using kernels. 

Traditionally, nonlinear kernels have been used in the SVM dual framework
to allow for nonlinear classification of points \cite{vapnik1998statistical}.
Here we show that when input patterns are on manifolds, how nonlinear
kernels achieve the non-linear classification can be analyzed as improved
manifold classification capacity of reformatted manifolds in the kernel
feature space (Fig. \ref{fig:kernelmanifold}\cite{burges1999geometry}).
We also extend $M^{4}$ algorithm provided in Chapter \ref{cha:m4}
to show that an iterative method with the same principle can be used
to find a kernel-SVM solution for manifold classification (kernel-$M^{4}$). 

\subsection{Manifold Capacity under a Quadratic Kernel }

The effect of the kernel operation on the geometric properties of
the manifolds depends on the kernels. As a simple example of a non-linear
kernel, we study the improved classification capacity of manifolds
with quadratic kernels. We extend our theory to provide upper and
lower bounds for the classification capacity of manifolds with dimensionality
$D$ in input space after their transformation to a quadratic feature
space.

Consider arbitrary manifolds embedded in $D$-dimension, where each
point in the manifold can be expressed as $\mathbf{x}=\mathbf{x}_{0}+\sum_{i=1}^{D}s_{i}\mathbf{u}^{i}$,
where $\mathbf{x}_{0}$ is a $N$-dimensional center vector and $\mathbf{u}^{i}$
($i=1,...,D$) are the basis vectors, and parametrized by $\vec{s}$,
where $s_{i}$ corresponds to the $i$ th basis vector. The feature
space of a homogenous quadratic kernel is $\left\{ x_{i}x_{j},i\leq j\right\} $
which has $N_{f}=(N+1)N/2$ unique components. The feature space,
$\left(\mathbf{x}_{0}+\sum_{k=1}^{D}s_{k}\mathbf{u}^{k}\right)_{i}\left(\mathbf{x}_{0}+\sum_{l=1}^{D}s_{l}\mathbf{u}^{l}\right)_{j}$,
can be expanded as $x_{0i}x_{0j}+\sum_{k}s_{k}\left(u_{i}^{k}x_{0j}+x_{0i}u_{j}^{k}\right)+\sum_{k=1}^{D}\sum_{l=1}^{D}s_{k}s_{l}u_{i}^{k}u_{j}^{l}$.
This is $N_{f}$-dimensional vector with a center $x_{0i}x_{0j}$
and $(D+3)D/2$ basis vectors. The basis vectors consist of the two
classes: $u^{k}u^{l}$ where there are $(D+1)D/2$ of them, and $\left(u_{i}^{k}x_{0j}+x_{0i}u_{j}^{k}\right)$
where there are $D$ of them. Therefore, for input space dimension
$N$, the ambient dimension in this feature space is $N_{f}=(N+1)N/2$.
On the other hand, the dimensionality of the manifolds in feature
space is $D(D+3)/2$. 

In order to illustrate the effect of kernels on the learning of classifier
of the manifolds, we present 2 simple examples. Quadratic kernels
applied to $1D$ lines and $2D$ circles. The geometric illustration
of how$1D$ lines and $2D$ circles with radius $R$ map from input
space to quadratic kernel's feature space is provided in Fig. \ref{fig:kernelmanifold}(a),(c).
For $R=0$, the (zero margin) capacity $P$ increases from $2N$ to
$(N+1)N$, as given by the capacity for classifying points. On the
other hand, for $R=\infty$, the weight vector has to be orthogonal
to all the dimensions spanned by the reformatted manifolds, yielding
the capacity $P=\frac{N(N+1)}{(1+D(D+3))}$ in this limit. For an
intermediate $R$, the capacity will be affected by the extent and
geometry of the manifolds in feature space. The predicted bounds are
compared with numerical simulations in Fig. \ref{fig:kernelmanifold}.
These considerations can be easily generalized to polynomial kernels
with higher degree.

\begin{figure}[h]
\begin{centering}
\includegraphics[width=1\textwidth]{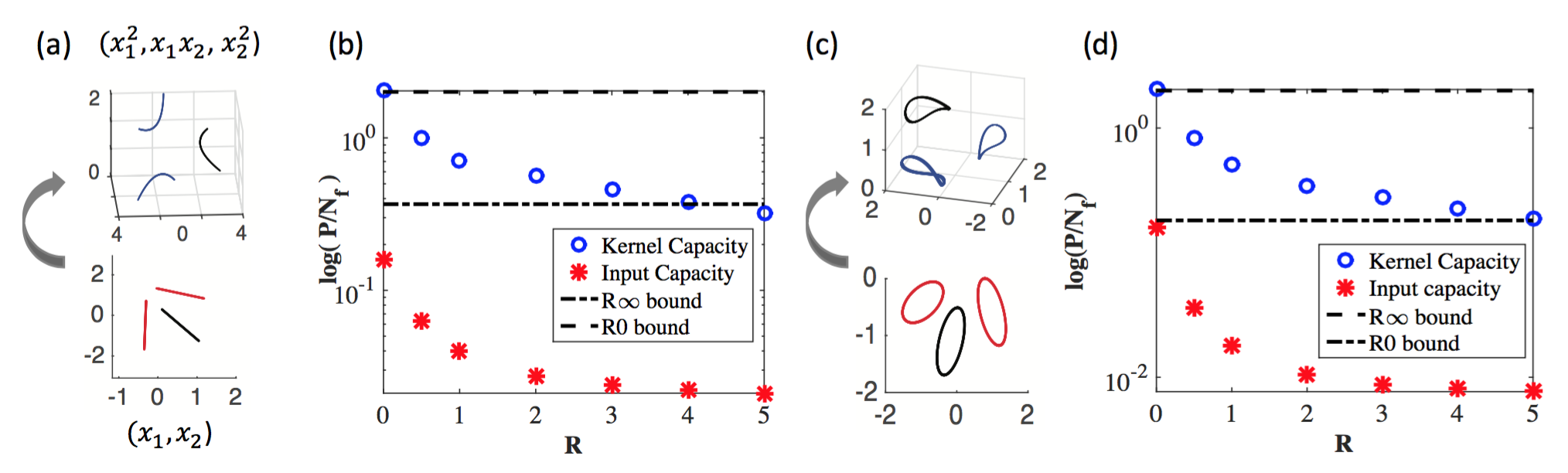}
\par\end{centering}
\caption{\textbf{Manifold Classification with a Quadratic Kernel.} (a) Classification
of Lines embedded in $N$-dim input space (black versus red, bottom)
maps to $2D$ curves (black vs. blue, top) embedded in $N_{f}=(N+1)N/2$
dim kernel feature space. (b) Line capacity in input space (red),
and quadratic kernel space (blue), shown as $P$ (number of manifolds)
over $N_{f}$, and the bounds on the kernel capacity: $\alpha=2$
($R=0$) and $\alpha=2/5$($R=\infty,D=2$) (dashed lines). (c) Classification
of $2D$ circles (black v. red, bottom) maps to $5D$ manifolds (black
vs. blue, top). (d) Manifold capacity of $2D$ circles in input space
(red) and quadratic kernel space (blue), and the bounds on the kernel
capacity: $\alpha=2$ ($R=0$) and $\alpha=2/11$($R=\infty,D=5$).
In both the line and 2D circles, the manifold capacities are improved
by the quadratic kernel operation. \label{fig:kernelmanifold}}
\end{figure}

\subsection{Kernel-$M^{4}$ Algorithm }

\paragraph*{General Framework }

The Kernel-$M^{4}$ algorithm applies the same logic as the $M^{4}$
algorithm in the chapter \ref{cha:m4}, but in the dual SVM framework.
The separating hyperplane is represented implicitly by dual coefficients
$\vec{\alpha}$ and point examples $\left(\mathbf{x}_{l},y_{l}\right)$
and a bias $b$, and the field induced by an input $\mathbf{x}$ is
given as $\sum_{l}\alpha_{l}y_{l}K\left(\mathbf{x},\mathbf{x}_{l}\right)+b$.
The Kernel-$M^{4}$ algorithm iteratively calls a quadratic optimization
solver on a finite number of labeled examples in the dual framework.
Given a current estimate of $\vec{\alpha}$ and $b$ , the algorithm
searches for the point on the manifolds with the worst margin. If
the margin of the new point is worse than the previous estimate, the
point set is augmented, i.e. the kernel matrix is increased by one
column and row, and the dual SVM solver is run to update values for
$\vec{\alpha}$ and $b$. The pseudocode for kernel-$M^{4}$is given
by Alg. \ref{alg:kernelM4}. 

\begin{algorithm}
{[}$\vec{\alpha}$,$b${]} = \textbf{function kernel-M4}($K$ , $S$
,$\epsilon$ )

\textbf{Input:} kernel type $K$ , data manifold parameters $S$ ,
tolerance $\epsilon$

Initialize: $\vec{\alpha}$,\textbf{$b$, $\delta$}

$t=0$

\textbf{while $\delta>\epsilon$ do }

$\quad$1. $t=t+1$

$\quad$2.$\alpha^{t}=$ $\max_{\alpha_{l}\geq0}\sum_{l=1}^{M}\alpha_{l}-\frac{1}{2}\sum_{j=1}^{M}\sum_{k=1}^{M}\alpha_{j}\alpha_{k}y_{j}y_{k}K\left(\mathbf{x}_{j},\mathbf{x}_{k}\right)$
$\mbox{ for \ensuremath{\forall\nu<t}}$ and and $\alpha^{tT}y=0$. 

$\quad$3. Compute $b$ and $h_{min}^{t}$. 

$\quad$4. Search for the new pattern such that $\frac{\left(h_{min}^{t}\right)_{new}}{||w||}=\left\{ \mbox{min}_{\vec{s},\mu}y_{\mu}\left(\sum_{l}\alpha_{l}y_{l}K\left(\mathbf{x}_{l},\mathbf{x}_{\mu}(\vec{s})\right)+b\right)\right\} $. 

$\quad$5. $\delta=\left\Vert \left(h_{min}^{t}\right)_{new}-h_{min}^{t}\right\Vert $

\textbf{end}

\textbf{Output:} $\alpha$, $b$

\caption{Pseudocode for kernel-$M^{4}$\label{alg:kernelM4}}
\end{algorithm}

In general, finding the point with the worst margin in the kernel
feature space may be hard as the convexity of the input manifolds
may be lost by the nonlinear kernel operation. If the manifolds are
given by finite sets of points, then, the search over all points to
find the worst point can be performed, where each search is upper
bounded by the number of examples. If the input manifolds are uncountable
sets of points, where the complete parameterization for the shape
$f(\vec{s})=0$ is given, then the search for the worst point is limited
to finding the worst $\vec{s}^{*}$, and sometimes $\vec{s}^{*}$
may be found analytically. If the parameterization is not available,
then one may need to find the worst point with a local search using
a gradient, but if there is no estimate of the convex hull, this operation
is not necessarily a convex problem, which may be investigated further
in the future.

With certain manifolds and kernel functions, the worst point operation
can still be done efficiently. As an illustrative example, we demonstrate
an example of the maximum margin classification of line segments and
$2D$ circles under a homogenous quadratic kernel, $K\left(\mathbf{x}_{j},\mathbf{x}_{k}\right)=\left(\mathbf{x}_{j}^{T}\mathbf{x}_{k}\right)^{2}$.
For these examples, we can reduce the worst point operation to a finite
set of analytical solutions, which is as efficient as regular$M^{4}$operation.
This computation can also be generalized to $D$-dimensional balls.

\subsubsection{$D$-dimensional Balls and Quadratic Kernel $M^{4}$ }

Here we show an example of quadratic kernel-$M^{4}$with $D$-dimensional
balls with radius $R$. In this case, finding the smallest distance
to the solution plane from each point on the manifold in the kernel
space is: 

\begin{equation}
\underset{\mbox{\ensuremath{\vec{s}}},f_{\mu}(\vec{s})=0}{\text{argmin}}h^{*}(\vec{s})=\underset{\mbox{\ensuremath{\vec{s}}},f_{\mu}(\vec{s})=0}{\text{argmin}}y_{\mu}\left(\sum_{j=1}^{M}\alpha_{j}y_{j}K\left(\mathbf{x}_{j},\mathbf{x}(\vec{s})\right)+b\right)\label{eq:minH_Kernel_Ball}
\end{equation}

where $\mathbf{x}=\mathbf{x}_{0}^{\mu}+\sum_{i=1}^{D}s_{i}\mathbf{u}_{i}^{\mu}$,
and $f_{\mu}(\vec{s})=0$ is the shape constraint of the $\mu$th
ball. 

Now, the closet point in the $D$-dimensional ball manifold $\mu$
can be found by considering

\begin{equation}
\vec{s}^{*}=\mbox{argmin}_{\vec{s}}\left[y_{\mu}\left(\sum_{l=1}^{\nu}\alpha_{l}y_{l}K\left(\mathbf{x_{l}},\mathbf{x}_{0}^{\mu}+R\sum_{j=1}^{D}s_{j}\mathbf{u}_{j}^{\mu}\right)+b\right)+\lambda\left\{ \sum_{j=1}^{D}s_{j}^{2}-1\right\} \right]
\end{equation}

This is in general hard, but analytically solvable for quadratic kernel
and a ball. For a homogenous quadratic kernel $K(\mathbf{x}_{n},\mathbf{x}_{m})=\left(\mathbf{x}_{n}^{T}\mathbf{x}_{m}\right)^{2}$,
we solve for $\vec{s}^{*}$ by taking a derivative of $h(\vec{s})$. 

The worst point $\vec{s}^{*}$ on the $\mu$th $D$-dimensional balls
are found to be 

\[
\lambda\left\{ s_{j}^{\mu}\right\} ^{*}=A_{j}^{\mu}+\sum_{j'=1}^{D}B_{jj'}^{\mu}s_{j'}^{\mu k}
\]

where 

\[
A_{j}^{\mu}=y_{\mu}\sum_{l=1}^{\nu}\alpha_{l}y_{l}\left[\left\{ \mathbf{x}_{l}^{T}\mathbf{x}_{0}^{\mu}\right\} \left(\mathbf{x}_{l}^{T}\mathbf{u}_{j}^{\mu}\right)\right]
\]

\[
B_{jj'}^{\mu}=Ry_{\mu}\left[\sum_{l=1}^{\nu}\alpha_{l}y_{l}\left(\mathbf{x}_{l}^{T}\mathbf{u}_{j}^{\mu}\right)\left(\mathbf{x}_{l}^{T}\mathbf{u}_{j'}^{\mu}\right)\right]
\]

with normalization on $\left\Vert \vec{s}^{*}\right\Vert =1$, and
$j,j'=\{1,...,D\}$, and $\vec{\alpha}$ , $\vec{y}$ are given. $l=1,...,\nu$
is an index of all the training points added so far, 

And for line segments, the worst point $s^{*}$ on the $\mu$ th line
is 

\[
a^{\mu}+b^{\mu}\left(s^{\mu}\right)^{*}=\lambda\left(s^{\mu}\right)^{*}
\]

where 

\[
a=y_{\mu}\sum_{l=1}^{\nu}\alpha_{l}y_{l}\left[\left\{ \mathbf{x}_{l}^{T}\mathbf{x}_{0}^{\mu}\right\} \left(\mathbf{x}_{l}^{T}\mathbf{u}^{\mu}\right)\right]
\]

\[
b=Ry_{\mu}\left[\sum_{l=1}^{\nu}\alpha_{l}y_{l}\left(\mathbf{x}_{l}^{T}\mathbf{u}^{\mu}\right)\left(\mathbf{x}_{l}^{T}\mathbf{u}^{\mu}\right)\right]
\]

with normalization on $\left\Vert \vec{s}^{*}\right\Vert =1$. 

In general, the solution for $\vec{s}^{*}$ has two classes. One is
where the closest point to the hyperplane in the kernel space comes
from the interior area of the convex hulls of the manifold in the
input space ($\lambda=0$ ), and the other is when $\vec{s}^{*}$
comes from the convex hull in the input space (nonzero $\lambda$).
Generally, one can solve for both cases, and check which $\vec{s}^{*}$
gives a smaller field, by computing Eqn. \ref{eq:minH_Kernel_Ball}
for each $\mu$ and do this over all manifolds, and with $P$ candidate
points, find the smallest one again. This is the step 4 of Alg. \ref{alg:kernelM4},
for $D$-dimensional balls. 

\subsection{Discussion }

In this section, we showed how the role of nonlinear kernel on data
manifolds with small manifold capacity (in other words, linearly non-separable
manifolds) in the input layer can be viewed as reformatting them to
increase the manifolds capacity in the nonlinear feature space, using
quadratic kernels and simple $1D$ and $2D$ balls in the input space
as examples. Which kernels are best suited for the classification
problem depends on the types of the data, which, in our case, are
manifold structures. We laid the ground here for future analysis with
manifolds capacity in kernel feature space, by formulating the iterative
algorithm for finding the max margin solution for manifolds with kernels
and demonstrating the simple examples of $1D$ and $2D$ balls, as
well as providing bounds on their manifold capacities in input and
feature space. Enabled with our theory of capacity of general manifolds
(Chapter \ref{cha:genManifolds}), we hope to extend our manifold
capacity analysis for more complex manifolds to study the role of
nonlinear kernels. 

\section{Generalization Error }

So far, we have considered different aspects and extensions of manifold
classification capacity, mainly motivated by a linear classifier that
achieves zero training error. However, another important aspect of
the linear separability of data manifolds is the generalization error
of a linear classifier. This is particularly relevant in more realistic
settings, where the manifolds given for training are not full manifolds,
but only a subset of the manifolds. An example of this would be when
the training data are convex hulls of subsamples of underlying manifolds.
Another relevant case is when there is noise in the input. By focusing
on the distribution of fields arising from manifold structures, here
we show how the our theory allows for the estimation of generalization
error for the manifold classification problem. 

\subsection{Generalization Error for General Manifolds, Given Weights. }

Exact analytical expression for the generalization error for general
manifolds is complicated; furthermore, the error depends on the assumed
sampling measure on the manifold (whereas the separability problem
is measure invariant). However, in the case of linearly separable
manifolds with high $D$ we can use the insight from the above theory
(the notions of effective dimensions and radii) to derive a particularly
simple approximation. 

Assume we have obtained a set of weights from learning manifolds either
from subsampling or from full unrealizable manifolds, so that we have
a vector $\mathbf{w}$ . Then, the generalization error can be expressed
as 

\begin{equation}
\epsilon_{G}^{\mu}=\left\langle \Theta\left(-h_{0}^{\mu}-\sum_{i}s_{i}h_{i}^{\mu}\right)\right\rangle _{s}
\end{equation}

where $h_{0}$are the field on the center $\mathbf{x}_{0}$and $h_{i}$are
the $D$ fields on the basis vector$\mathbf{u}_{i}$, generated by
the trained $\mathbf{w}$. The average over $\vec{s}$ is an integration
with constraint $f(\vec{s})=0$ . In other words, 

\begin{equation}
\langle F(\vec{s})\rangle_{\vec{s}}=\int d\vec{s}p(\vec{s})\delta(f(\vec{s}))F(\vec{s})
\end{equation}

and we assume $\int d\vec{s}p(\vec{s})\delta(f(\vec{s}))$=1 . An
important point to note is that the generalization error is in general
sensitive to the choice of measure to use in this average i.e., $p(\vec{s})$
.

The dependence on the weight $\mathbf{w}$ is via the fields $h_{0}$
and $\vec{h}$. Depending on the learning rule used to generate $\mathbf{w}$,
the above fields induced by $\mathbf{w}$ in general will not be the
same for all manifolds but will vary with a distribution $P(h_{0},\vec{h})$
even if all manifolds have the same geometry. For example if \textbf{$\mathbf{w}$}
is trained by the max margin classification of subsampled manifolds,
$P(h_{0})$ will be the distribution described in Chapter \ref{cha:genManifolds}
for finite point manifolds. Therefore, for a given $\mathbf{w}$,
the generalization error can be expressed as a double averaging 

\begin{equation}
\epsilon_{G}=\left\langle \left\langle \Theta\left(-h_{0}-\sum_{i}s_{i}h_{i}\right)\right\rangle _{\vec{s}}\right\rangle _{h_{0},h_{i}}
\end{equation}

\subsection{Gaussian Approximation in High Dimensional Manifolds.}

In high dimensional manifolds, we assume that $h_{i}$ are distributed
as projections of gaussian $\mathbf{w}$ on the $\mathbf{u}_{i}'s$
. Hence $h_{i}$ themselves are i.i.d. Gaussian distributed with norm
$1$. If we make this assumption we get, 

\begin{equation}
\epsilon_{G}^{\mu}=\left\langle \left\langle \Theta\left(-h_{0}^{\mu}-z\right)\right\rangle _{z}\right\rangle _{s}
\end{equation}

In other words, 

\begin{equation}
\epsilon_{G}=\left\langle H\left(\frac{h_{0}}{\sigma}\right)\right\rangle _{h_{0}}
\end{equation}

where the width of the distribution is roughly 
\begin{equation}
\sigma^{2}\approx\langle\vec{s}\cdot\vec{s}\rangle\approx R_{W}^{2}
\end{equation}

in the crude approximation. 

We might want to take into account the fact that the $h_{i}'s$ are
bounded so they are not unbounded Gaussians. We can use the mean width
of the manifold $R_{W}\sqrt{D_{W}}$ (Fig. \ref{fig:GenErr-Illu}(a))
as a measure of the bound. Finally, we can estimate the generalization
error to be 
\begin{equation}
\epsilon_{G}=\left\langle \frac{H(\frac{h_{0}(w)}{R_{W}})-H\left(\sqrt{D_{W}}\right)}{1-2H\left(\sqrt{D_{W}}\right)}\right\rangle _{h_{0}}\label{eq:GenErrorManifold}
\end{equation}

where we added a normalization so that $\epsilon_{G}=0.5$ for $h_{0}=0$.
Note that $h_{0}$ is the field from the center, and $D_{W}$ and
$R_{W}$ are effective dimension and effective radius of a manifold
identified in Chapter \ref{cha:genManifolds}. Intuitively, the generalization
error shows the fraction of a manifold on the wrong side of the hyperplane,
which is approximated by the area under one end of the tail of a Gaussian
distribution outside of the size of the center field induced by the
solution \textbf{$\mathbf{w}$}, divided by the area under the Gaussian
distribution with tails cut at the size of the mean width from both
end (approximately 1 in high $D_{W}$ regime). This relation is shown
in the Figure \ref{fig:GenErr-Illu}. 

\begin{figure}
\begin{centering}
\includegraphics[width=1\textwidth]{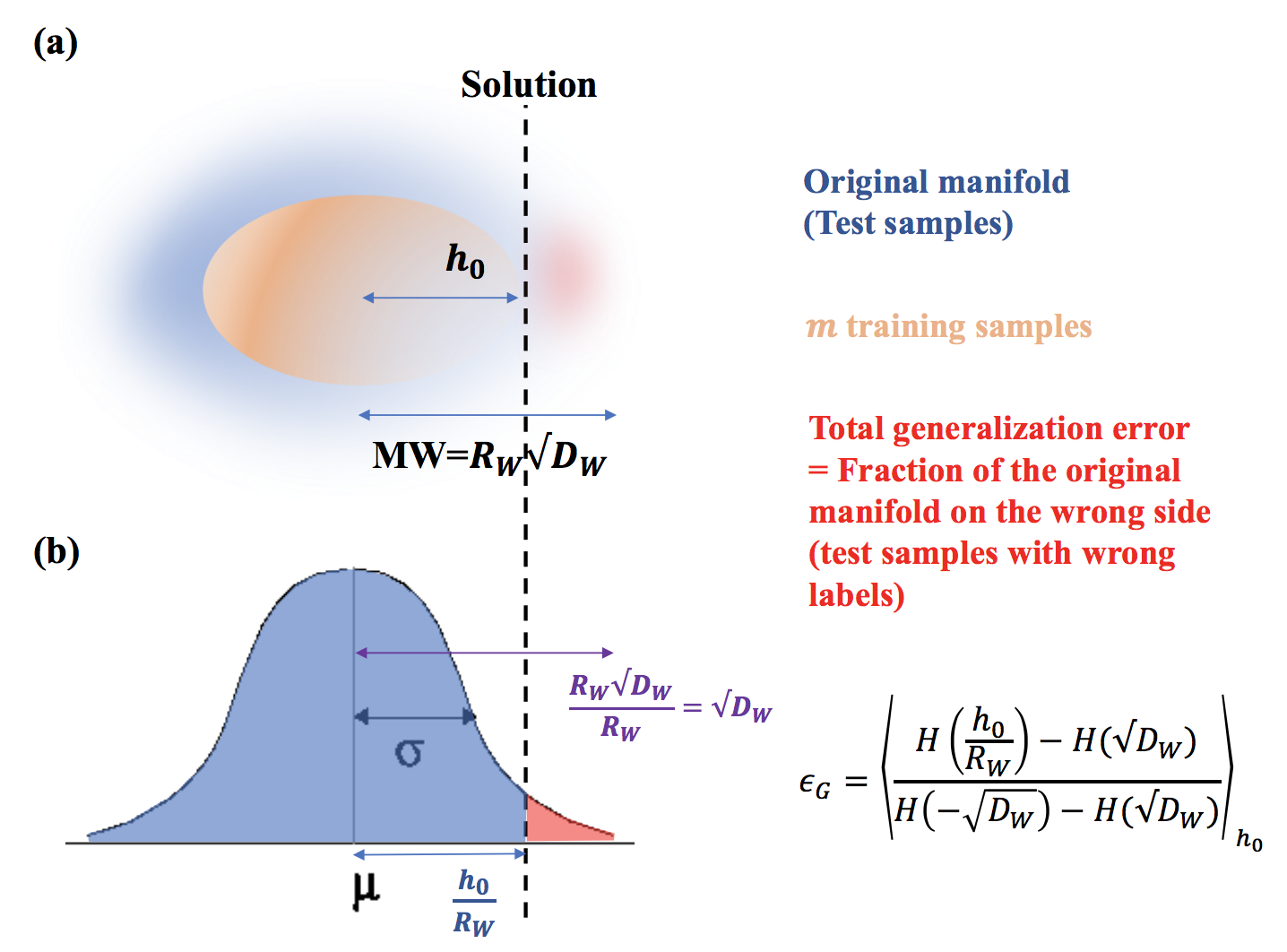}
\par\end{centering}
\caption{\textbf{Generalization Error from Each Manifold.} (a) The solution
hyperplane (dashed line) is determined by the training samples given
by the manifold (Orange manifold), and the distance between the center
of the original manifold (assuming the center of the testing manifold
and original manifold is the same) and the solution hyperplane is
given by $h_{0}$ (along the direction of $\mathbf{w}$ ), the field
induced by the testing manifold center. The size of the original (testing)
manifold, along the direction of $\mathbf{w}$ , is approximated by
the Mean Width, $R_{W}\sqrt{D_{W}}$ . (b) The generalization error
is the fraction of samples on the correct side of the hyperplane out
of the total samples (which is measure-dependent). Assuming the projections
of the manifold along $\mathbf{w}$ are Gaussian, we approximate generalization
error as ratio between the blue area and blue plus red area under
the gaussian. The width of the gaussian is estimated to be $\sigma=R_{W}$
as a crude approximation (assuming $h_{i}$'s are Gaussians with norm
1), with the tails truncated at$R_{W}\sqrt{D_{W}}$, and the distance
between the peak and the location of the hyperplane at $h_{0}$ (the
separation between correct and incorrect labels). Rescaling the $x$
axis such that $\sigma$ is 1, we get the expression $\epsilon_{G}\sim\langle\frac{H(\frac{h_{0}}{R_{W}})-H(\sqrt{D_{W}})}{H(-\sqrt{D_{W}})-H(\sqrt{D_{W}})}\rangle$.
\label{fig:GenErr-Illu}}
\end{figure}

The generalization error shows two regimes. In the case where full
(underlying) manifolds are linearly separable with margin $\kappa$,
Then the generalization error will eventually vanish as more samples
per manifold , $m$ , are presented during the training.

One thing to note is that as we increase the number of subsamples
for training $m$, the size of $h_{0}$ must grow, as the hyperplane
is always outside the mean width of the subsampled manifolds. According
to the extreme value theory, we expect the max distance between the
manifold center and extreme tail to grow like $\sqrt{2\text{log}m}$,
like the $m$th maximum value of samples of Gaussian iid distribution
\cite{bovier2005extreme}. In the limit of large $m$ , we obtain, 

\begin{equation}
\epsilon_{g}(m)\sim H(\kappa(m)+\sqrt{2\text{log}m})\propto\frac{\exp[-\kappa\sqrt{2\log m}]}{m}
\end{equation}

Interestingly, this decay is faster than the generic power law, $\epsilon_{g}(m)\propto m^{-1}$
of generalization bounds in linearly separable problem and reflects
the presence of finite margin of the entire manifold. This dependence
is demonstrated in Figure \ref{fig:genEtrE_vs_m_ellipsoids} (a1-a2)
using ellipsoids. 

In the case where the full manifolds are not separable and $\alpha$
is above the capacity, even subsampled manifolds with $m$ points
are not necessarily separable. Because the solution hyperplane may
intersect the sampled manifolds, $h_{0}$ no longer scales like $\sqrt{2\text{log}m}$.
Perhaps for this reason, the $m$ dependence of generalization error
given inseparable underlying manifolds is more like conventional general
power law, $\epsilon_{g}(m)-\epsilon_{g}(m=\infty)\propto m^{-1}$.
This is shown in \ref{fig:genEtrE_vs_m_ellipsoids} (b1-b2) using
ellipsoids.

\subsection{Numerical Investigations}

As a simple example, we computed the generalization error for a binary
dense classification of $P$ ellipsoids, where $R_{i}$ (radius in
$i$ th embedding dimension) is sampled from an uniform distribution
of $Unif[0.5R_{0},1.5R_{0}]$, and centers and axes are random Gaussian
distribution. From each ellipsoid, $m$ training samples and $m_{g}$
test samples were sampled, so that $s_{i}R_{i}^{-1}$ is from a uniform
spherical distribution. With these $mP$ finite training samples,
max margin solution was found using a standard slack-SVM solver, and
generalization error was computed using $m_{g}P$ test samples. Using
the centers of the ellipsoids and the max margin solution \textbf{$\mathbf{w}$},
theoretical estimation of the generalization error was computed using
Eqn. \ref{eq:GenErrorManifold}. We show the results of this simulation
in Fig. \ref{fig:genEtrE_vs_m_ellipsoids}. 

\begin{figure}[H]

\begin{centering}
\includegraphics[width=1\textwidth]{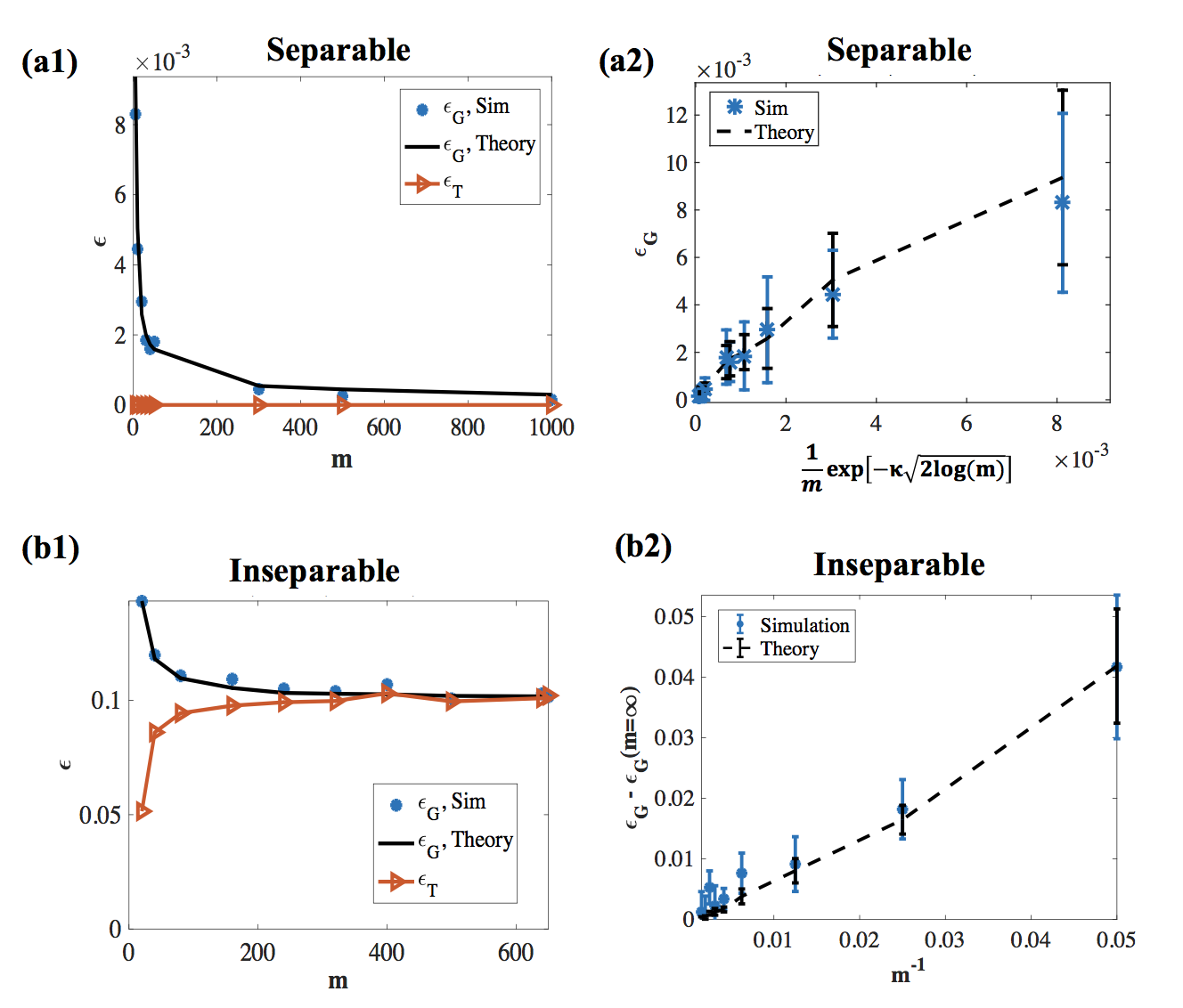}
\par\end{centering}
\caption{\textbf{Generalization Error of Ellipsoids Classification as a Function
of Number of Subsamples Per Manifold ($m$ ).} (a1-a2) In the regime
where the full manifolds are separable, the generalization error approaches
zero as $m$ is increased, at the rate of $\epsilon_{G}\sim\frac{1}{m}\text{exp}[-\kappa\sqrt{2\text{log}m}]$.
The manifolds used were $\text{\ensuremath{D}}$-dimensional ellipsoids
with $R_{i}\sim Unif[0.5R_{0},1.5R_{0}]$ and $P=10,N=50,R_{0}=1.00,D=50$.
(b1-b2) In the regime where the full manifolds are not separable,
the generalization error approaches at the rate of $\epsilon_{G}-\epsilon_{G}(m=\infty)\sim m^{-1}$
. The parameters for the ellipsoids used for the simulation were $P=20,N=50,R_{0}=2.00,D=50$.
In all simulations generalization error was tested with $m_{g}=1000$
per manifold. The theoretical predictions for generalization error
patches the generalization error calculated by the simulation well
in this regime. \label{fig:genEtrE_vs_m_ellipsoids}}

\end{figure}

We find that the Gaussian approximation of the generalization error
works quite well, and the estimation matches the generalization error
using simulations. We find that indeed in the separable case, the
$\epsilon_{G}$ is close to $\frac{1}{m}exp[-\kappa\sqrt{2logm}]$,
and in the inseparable case, $\epsilon_{g}(m)-\epsilon_{g}(m=\infty)$
is approximately $1/m.$

\subsection{Discussion }

Here we demonstrated how the insights from the manifolds capacity
theory can be used to compute the crude approximation of generalization
error in the high embedding dimension, which works surprisingly well
in the case of classification of ellipsoids. Clearly, we made some
assumptions for the sake of the approximation (i.e. $h_{i}$ are Gaussians).
The fully general replica theory of generalization error for manifolds
will require considerations of the actual manifold geometry, and the
measure of the samples on the manifold, $p(\vec{s})$, and using the
manifold-dependent distribution of fields $h_{0}$ and $h_{i}$ from
the replica theory.

\section{Analysis of GoogLeNet Manifolds }

So far, we have considered various extensions of the theory of manifold
classification for the analysis and application to real data. In this
section, we show how our manifold capacity theory can be applied to
realistic data, by analyzing manifold representations in conventional
deep networks as an illustrative example. In the recent years, the
performance of the artificial systems for visual classification tasks
has been focused on the generalization error of the final layer on
the test dataset. However, the underlying goal for training such system
is to create representations such that different objects are easy
to distinguish from each other. This idea is closely related to our
notion of manifold classification capacity. Using our theory, we analyze
how data representations reformat across different layers of GoogLeNet\cite{szegedy2015going}
, one of the widely used deep networks for a popular visual recognition
task, ImageNet \cite{deng2009imagenet} classification task. Using
different object classes of ImageNet dataset as manifolds, we show
how the quantities that contribute to the manifold classification
capacity, i.e. effective dimension, effective radius, Gaussian mean
width, and various correlations, change across the hierarchy of the
layers. 

\subsection{Methods}

We defined an object manifold as convex hulls of training samples
from different object classes from ImageNet classification task dataset
\cite{deng2009imagenet}. ImageNet Dataset has 1000 object classes
with roughly 1000 training samples in each object class. Here, we
computed the center of mass of each class (such that $\mu$th manifold's
center is $\mathbf{x}_{0}^{\mu}$ and mean of the centers are set
to be the origin), and selected a small set of object classes ($P$)
such that their center-to-center overlap, or center-to-center correlation,
($\rho_{center}=\hat{x}_{0}^{\mu}\cdot\hat{x}_{0}^{\nu}$) is smaller
than a threshold value. The correlation between centers of the image
object manifolds are surprisingly high, and there is a tradeoff between
a low threshold value $\rho$ and the number of object classes $P.$
In our simulation, we used $P=22$ and $\rho_{center}$ was roughly
0.3. 

To study how ImageNet object manifolds reformat in a network that
is guaranteed to achieve a high classification performance, we chose
GoogLeNet \cite{szegedy2015going}, a winner of ImageNet Large Scale
Visual Recognition Challenge (ILSVRC). We used pre-trained GoogLeNet
weights available via MatConvNet framework \cite{vedaldi2015matconvnet}
to extract and analyze the ImageNet Object Manifolds in different
layers. As the network size is extremely large, we focused on the
layers after the pooling layers. We randomly selected $N_{sub}$ units
from the total units from each layer, and computed $R_{W}$, $D_{W}$
using the set of points (defined in Chapter \ref{cha:genManifolds},
Eqns \ref{eq:R_W}-\ref{eq:D_W}). We also computed correlation coefficients
between the centers of manifolds in each layer, as well as the overlap
between centers and their own axis (self center-axis correlation)
and centers and axis of the rest of the object manifolds (cross center-axis
correlation). In the case of center-axis correlation, because object
manifolds are high dimensional and have different sizes of extents
along different embedding dimensions, the overlap measures were scaled
with the eigenvalue of the covariance matrix of the object manifold. 

\[
c_{c-c}=\langle\mathbf{\hat{x}}_{0}^{\mu}\cdot\mathbf{\hat{x}}_{0}^{\nu}\rangle_{\mu\nu}
\]

where 

\[
c_{c-u}^{self}=\left\langle \frac{\lambda_{i}^{\mu}}{||\vec{\lambda}^{\mu}||}\mathbf{\hat{x}}_{0}^{\mu}\cdot\mathbf{\hat{u}}_{i}^{\mu}\right\rangle _{i,\mu};\mu=1,...,Pi=1,...,D
\]

\[
c_{c-u}^{cross}=\left\langle \frac{\lambda_{i}^{\mu}}{||\vec{\lambda}^{\mu}||}\mathbf{\hat{x}}_{0}^{\nu}\cdot\mathbf{\hat{u}}_{i}^{\mu}\right\rangle ;\nu=1,...,P,\mu\neq\nu,i=1,...,D
\]

where $\lambda_{i}^{\mu}$ are the $i$ th eigenvalue of the covariance
matrix of the samples of the$\mu$ th manifold. 

\subsection{Results and Discussions}

In the figure below, we show the summary of the results from the analysis
described. We know from the perceptron capacity of random general
manifolds, that effective dimension and effective radius need to be
reduced in order to increase the manifold classification capacity.
We also have an ongoing theoretical work indicating that correlations
between manifolds reduce the effective ambient dimension of the data,
resulting in the reduced critical number of manifolds that can be
separated. 

\begin{figure}[H]
\begin{centering}
\includegraphics[width=1\textwidth]{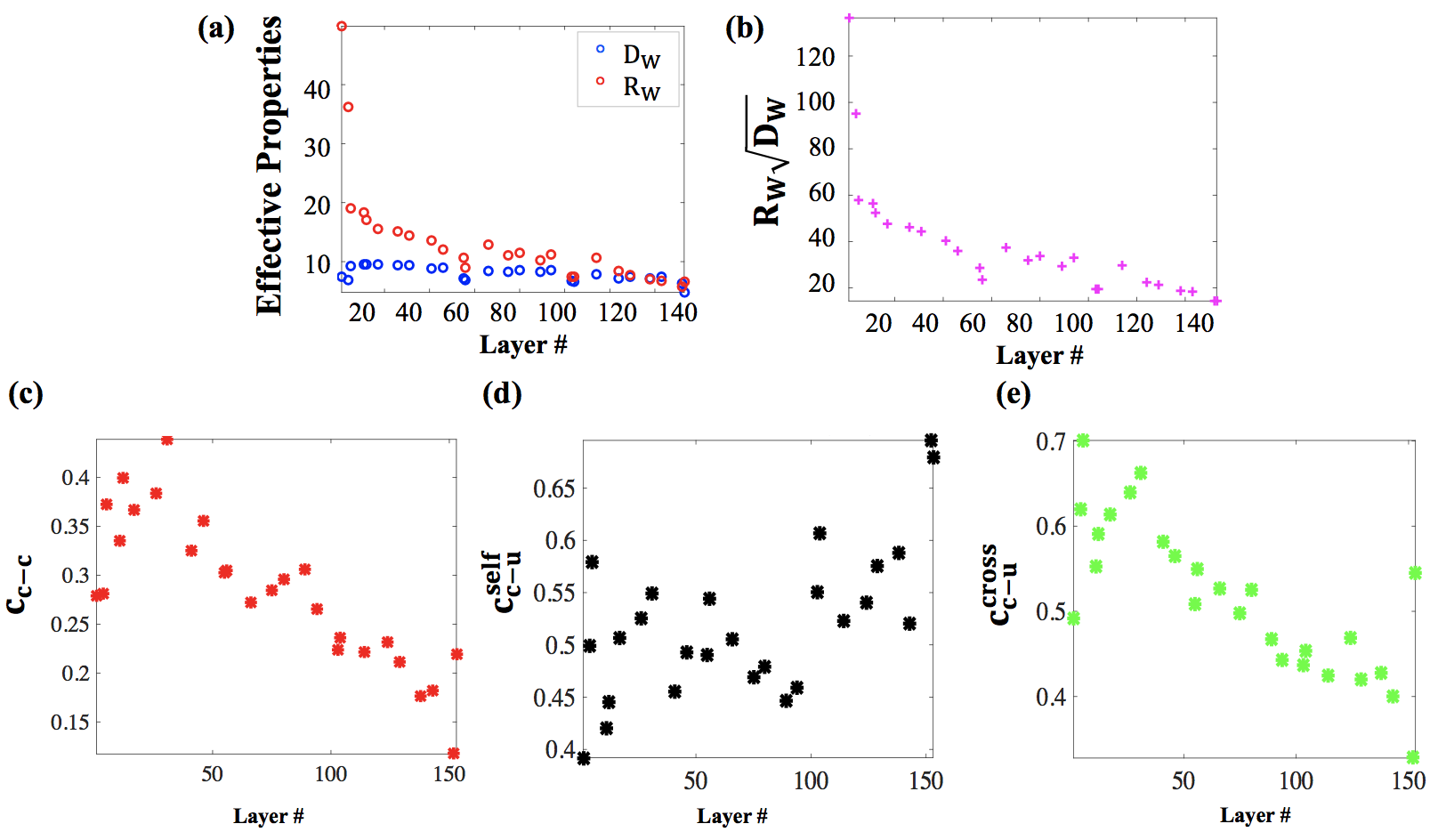}
\par\end{centering}
\caption{\textbf{Analysis of Manifold Properties in Different Layers of Deep
Networks.} Using the ImageNet Dataset for different class as object
manifolds, effective manifold properties were analyzed in different
layers of GoogLeNet. (a) Effective Dimension $D_{W}$ and Effective
Radius $R_{W}$ for different layers of GoogLeNet. (c) Half of Gaussian
Mean Width $(R_{W}\sqrt{D_{W}}$) for different layers of GoogLeNet.
(c) Correlation between manifold centers (d) Correlation between each
manifold center and its axes, averaged over manifolds. (e) Correlation
between each manifold center and the axes of the other manifolds,
averaged over manifolds. \label{fig:AnalysisDeepNets} }
\end{figure}

Interestingly, we find that the mean of the effective radius $R_{W}$
of ImageNet object manifolds decrease systematically while being processed
by the layers of the deep network (Fig \ref{fig:AnalysisDeepNets}(a)).
Particularly, the most dramatic improvement in $R_{W}$ appears in
the early part of the processing. On the contrary, the mean effective
dimension $D_{W}$ of ImageNet manifolds doesn't decrease as significantly
(Fig \ref{fig:AnalysisDeepNets}(a)). The values of $D_{W}$ at different
layers are quite close to $2\text{log}m$ throughout (where $m$ is
the number of training samples in each object class), and this may
imply that the effective dimension of object manifolds is determined
by the number of training samples of the data. Further investigations
on what determines $D_{W}$ of the realistic data is ongoing work. 

Recall that $R_{W}\sqrt{D_{W}}$, half of the Gaussian mean width,
is directly related to the capacity of manifolds. There is a systematic
drop in the size of $R_{W}\sqrt{D_{W}}$ across layers, where the
most dramatic improvement is in the early stage and there is a consistent
improvement in the rest of the layers (Fig \ref{fig:AnalysisDeepNets}(b)).
Supposedly the utility of the deep structure of the network is increasing
the manifold capacity, by reducing the Gaussian mean width. 

Note that in $R_{W}$ and $R_{W}\sqrt{D_{W}}$, the best values are
in the output (readout) layer, and the worst values are in the input
(pixel) layer (Fig \ref{fig:AnalysisDeepNets}(a)-(b)). 

Center-center correlations ($c_{c-c}$) also show systematic decrease
across layers (Fig \ref{fig:AnalysisDeepNets}(c)), implying that
deep network gradually decorrelates the centers of the object manifolds.
Note that when there is more correlation between the centers, the
total space spanned by the manifold centers will be more skewed (compared
to spherical), and the effective ambient dimension spanned by the
centers will be smaller, making the maximum critical number of manifolds
smaller \cite{litwin2017optimal}. Therefore, the deep network works
towards increasing the manifold capacity by reducing the correlation
between the centers. Other types of correlations (such as center to
axes, within the manifold, $c_{c-u}^{self}$, or across manifolds,
$c_{c-u}^{cross}$) show not as strong trends across the layers, and
we hope to address this further in the future. 

\subsection{Future Work }

In the future, we hope to connect the relationship between the theory
and the experimental manifold capacity, and take into account various
types of classifications, such as classification with sparse labels
in the object recognition limit. Many deep networks have similar computational
building blocks (i.e. convolution, max pooling, ReLu, dropouts, sparsity,
etc), and we hope to address the role of each computational building
block in terms of how they change the data manifolds. The change in
the manifold shapes and manifold capacity during the learning is also
an important future direction, in order to study what different types
of learning can achieve in each layer, with a quantitative measure
(manifold capacity). Finally, we hope to analyze the manifold representations
in neural data, particularly in different areas of the brain (i.e.
object representations in different stages of the visual hierarchy
in the cortex\cite{dicarlo2007untangling}). 

\section{Appendix }

\subsection{Multi-Layer Networks with Sparse Intermediate Layer \label{subsec:Multilayer_Appendix}}

In this section, we show the details for the numerical demonstrations
shown in Section. \ref{sec:Multilayer}. 

\subsubsection{\textcolor{black}{Training Two-Layer Network with Sparse Intermediate
Layer }}

\textcolor{black}{Input layer activity is $N$ dimensional and organized
as manifolds such that they can be described as $\mathbf{z_{in}^{\mu}}=\mathbf{x_{0}}^{\mu}+R\sum_{i=1}^{D}s_{i}\mathbf{u}_{i}^{\mu}\in\mathbb{R}^{N}$
where subscript 'in' denotes input layer, $\mathbf{x_{0}^{\mu}}$,$\mathbf{u}_{\mathbf{i}}^{\mathbf{\mu}}$,s,$R$
are from \ref{cha:spheres}. We draw a set of i.i.d., random sparse
binary labels $\{y_{i}^{\mu}\}$ (where $i=1,...,M$, $M$ being the
dimensionality of the intermediate layer) with probability $f$ of
being $1$. The output of the $i$-th nodes in the intermediate layer
is binary $0,1$ with $z_{i}=\Theta\left[\mathbf{V}^{(i)}\cdot\mathbf{z_{in}}+b^{i}\right]$
where $\Theta(x)$ is Heaviside step function. The weight vector to
the $i$ th node,$\mathbf{V}^{(i)}$ , and bias, $b^{i}$ are found
as a max margin (SVM) solution for $D$-dimensional balls to the set
of labeled pairs $\{\mathbf{z_{in}^{\mu}},y_{i}^{\mu}\}$, using the
same method as in Algorithm \ref{alg:Procedure-for-sampling}. In
Fig. }\ref{fig:sparseMulti}\textcolor{black}{, the network size used
was $N=M=500$. Initial training data for the SVM solver for spheres
had $m=2D$ samples per manifold, and then more points were added
until the iteration converged (Algorithm \ref{alg:Procedure-for-sampling}).
The output node is a linear readout of the intermediate layer representation,
$z_{out}=\mbox{sign}\left[\bar{w}\cdot\bar{z}+b_{f}\right]$ . {[}We
use overline to denote $M$-dim vectors.{]} The weight vector from
second layer to the output node, }\textbf{\textcolor{black}{$\bar{w}$}}\textcolor{black}{{}
, and bias, $b_{f}$ (where subscript $f$ stands for final), is trained
as an SVM solution to set of the labeled pairs, $\left\{ \mathbf{z}^{\mu},y_{out}^{\mu}\right\} $,
with the task dense labels $y_{out}^{\mu}=\pm1$ with $f=0.5$. For
details, see Algorithm \ref{alg:Pseudocode-for-Two-Layer}. }

\subsubsection{\textcolor{black}{Robustness to Input Noise }}

\textcolor{black}{Once the two layered network is fully trained (without
additive noise), we evaluate the probability of error in the output
by adding additive full-rank gaussian noise (with standard deviation
of $\sigma_{\mbox{noise}}$) to each input node and measuring the
output node classification error as a function of the variance of
the noise. For }\ref{fig:sparseMulti}\textcolor{black}{{} (c), about
100 samples per manifold ($m$) were used and number of trials was
5. For details, see Alg. \ref{alg:Pseudocode-for-Two-Layer}.}

\subsubsection{\textcolor{black}{Robustness to Noise in the Intermediate Layer }}

\textcolor{black}{We test the robustness to the introduction of smooth
sigmoidal units in the intermediate layer, by first evaluating the
intermediate weight matrix $\boldsymbol{V}$ to generate a sparse
binary representation in the intermediate layer. After training of
the intermediate layer, we replace the Heaviside step function of
the intermediate layer nodes with $\Phi_{T}(x)=\left(1+e^{-x/T}\right)^{-1}$,
which can be interpreted as the result of smoothing the binary function
by stochastic noise. The readout weights are calculated as SVM solutions
to the the densely labeled task pairs $\left\{ \mathbf{z}^{\mu},y_{out}^{\mu}\right\} $
where $\boldsymbol{\mathbf{z}^{\mu}}$ are the intermediate layer
sigmoidal responses to sampled inputs.}\textbf{\textcolor{black}{{}
}}\textcolor{black}{For }\ref{fig:sparseMulti}(d),\textcolor{black}{{}
about $100$ samples per manifold were used and number of trials was
5. The effect of the smoothness of the intermediate level on the classification
performance of the binary output unit is measured as a function of
the gain parameter $T$. For details, see Algorithm \ref{alg:Pseudocode-for-Two-Layer}.}

\begin{algorithm}
\textbf{Initialize: }

\textcolor{black}{$\quad$ $\mathbf{x_{0}}^{\mu},\mathbf{\vec{u}}^{\mu}\sim Norm(0,1)$
( $\mu=1,...,P$) {[}Sample centers and direction vectors{]}}

\textcolor{black}{$\quad$ $y_{\mbox{out}}^{\mu}\sim\mbox{sign}\left\{ \mbox{Unif}(-1,1)\right\} $
( $\mu=1,...,P$) {[}Sample dense labels for manifolds for output
layer{]}}\\

\textbf{Train V (Input to Intermediate Layer): }

\textcolor{black}{$\quad$ $s_{i}^{k,\mu}\sim\mbox{Unif}\left(-1,1\right)$
and $||\vec{s}^{k,\mu}||=1$ $i=1...D$,$k=1,...,m$, $\mu=1...P$.
{[}Sample $m$ coefficient vectors from each manifold{]}}

\textcolor{black}{$\quad$ }\textbf{\textcolor{black}{repeat:}}\textcolor{black}{{}
for$j=1,...,M$ {[}for each node in intermediate layer{]}}

\textcolor{black}{$\quad$ $\quad$ $\left(y^{\mu}\right)^{(j)}\sim P(y^{\mu})=f\delta(y^{\mu}-1)+(1-f)\delta(y^{\mu}+1).$
($\mu=1,...,P$) {[}Sample sparse labels for manifolds for intermediate
layer{]}}

\textcolor{black}{$\quad$ $\quad$ }$\mathbf{V}^{(j)},b^{(j)}$ =
$M_{B}^{4}$($\mathbf{x_{0}}^{\mu},\mathbf{\vec{u}}^{\mu},\left(y^{\mu}\right)^{(i)},R$,$m$)
{[}$M_{B}^{4}:$ $M^{4}$ algorithm for $L_{2}$ balls, see Chapter
\ref{cha:m4}{]}

\textcolor{black}{$\quad$ }\textbf{\textcolor{black}{end}}\\

\textbf{Train w (Intermediate to Output Layer): }

\textcolor{black}{$\quad$ $s_{i}^{k,\mu}\sim\mbox{Unif}\left(-1,1\right)$
and $||\vec{s}^{k,\mu}||=1$ $\mbox{for all}$ $i,k,\mu$}

\textcolor{black}{$\quad$ $\mathbf{z}_{in}^{k,\mu}=$$\mathbf{x_{0}}^{\mu}+R\sum_{i=1}^{D}s_{i}^{k,\mu}\mbox{\textbf{u}}_{i}^{\mu}$$\mbox{ for all \ensuremath{\mu,k}}$
{[}Sample $m$ points on each manifold (first layer activity){]}}

\textcolor{black}{$\quad$$z_{i}^{k,\mu}=\Theta\left[\mathbf{V}^{(i)}\cdot\mathbf{z_{in}}^{k,\mu}+b^{i}\right]$
for all $i=1,...,M$ and $k,\mu$ {[}Intermediate layer activity{]}}

\textcolor{black}{$\quad$$\bar{w},b_{f}=\mbox{svmsolver}(\bar{z}^{k,\mu},y_{out}^{k,\mu})$
{[}Find SVM solution with Intermediate layer activity and dense labels.{]}}\\

\textbf{Test Robustness to Input Noise (for $\sigma_{\mbox{noise}}$) }

\textcolor{black}{$\quad$ $s_{i}^{k,\mu}\sim\mbox{Unif}\left(-1,1\right)$
and $||\vec{s}^{k,\mu}||=1$ $\mbox{for all}$ $i,k,\mu$}

\textcolor{black}{$\quad$ $\mathbf{z}_{in}^{k,\mu}=$$\mathbf{x_{0}}^{\mu}+R\sum_{i=1}^{D}s_{i}^{k,\mu}\mbox{\textbf{u}}_{i}^{\mu}+\eta$
, where $\eta\sim Norm(0,\sigma_{\mbox{noise}}I)$ $\mbox{ for all \ensuremath{\mu,k}}$
{[}Sample $m$ points on each manifold with additive Gaussian noise
(input layer activity with noise){]}}

\textcolor{black}{$\quad$$z_{i}^{k,\mu}=\Theta\left[\mathbf{V}^{(i)}\cdot\mathbf{z_{in}}^{k,\mu}+b^{i}\right]$
for all $i,k,\mu$ }

\textcolor{black}{$\quad$$z_{out}^{k,\mu}=\mbox{sign}\left[\bar{w}\cdot\bar{z}^{k,\mu}+b_{f}\right]$ }

\textcolor{black}{$\quad$}\textbf{return}: $\epsilon_{G}=\left\langle \frac{1}{4}\left(z_{out}^{k,\mu}-y_{out}^{k,\mu}\right)^{2}\right\rangle _{k,\mu,\mbox{trials}}$\\

\textbf{Test Robustness to Noise in the Intermediate Layer (for $T\geq0$)}

\textcolor{black}{$\quad$ $s_{i}^{k,\mu}\sim\mbox{unif}\left(-1,1\right)$
and $||\vec{s}^{k,\mu}||=1$ for all $i,k,\mu$}

\textcolor{black}{$\quad$ $\mathbf{z}_{in}^{k,\mu}=$$\mathbf{x_{0}}^{\mu}+R\sum_{i=1}^{D}s_{i}^{k,\mu}\mbox{\textbf{u}}_{i}^{\mu}$
$\mbox{ for all \ensuremath{\mu,k}}$ }

\textcolor{black}{$\quad$$z_{i}^{k,\mu}=\Phi_{T}\left[\mathbf{V}^{(i)}\cdot\mathbf{z_{in}}^{k,\mu}+b^{i}\right]$
for all $i,k,\mu$ {[}Use smooth sigmoid for intermediate layer{]}}

\textcolor{black}{$\quad$$\bar{w},b_{f},\kappa_{\mbox{margin}}=\mbox{svmsolver}(\bar{z}^{k,\mu},y_{out}^{k,\mu})$
{[}SVM solution for intermediate layer activity and labels{]}}

\textcolor{black}{$\quad$}\textbf{\textcolor{black}{return:}}\textcolor{black}{{}
$\kappa_{\mbox{output}}\leftarrow\kappa_{\mbox{margin}}$ {[}Output
margin with SVM solution{]}}

\caption{Pseudocode for Two-Layer Network for Invariant Manifolds Classification\label{alg:Pseudocode-for-Two-Layer} }
\end{algorithm}

%% file: chapters/chz_appendix0_notations.tex
\chapter{Appendix A: Symbols and Notations }

\section{Notations}
\begin{description}[leftmargin=8em,style=nextline]
\item [{$N$}] Ambient Dimension for Data 
\item [{$P$}] Number of Manifolds 
\item [{$D$}] Embedding Dimension 
\item [{$\kappa$}] Margin 
\item [{$R$}] Radius of a ball 
\item [{$R_{i}$}] $i$ th radius of an ellipsoid ($i=1,...D$ for $D$
dimensional ellipsoid) 
\item [{$\alpha=P/N$}] Load (Number of Manifolds, $P$ /Ambient Dimension,
$N$) 
\item [{$\alpha_{C}$}] Critical capacity (general expression)
\item [{$\alpha_{G}$,$\alpha_{0}$}] Gardner's perceptron capacity for
points 
\item [{$\alpha_{L}$}] Capacity for Line Segments 
\item [{$\alpha_{B}$,$\alpha_{B_{2}}$}] Capacity for $L_{2}$ Balls 
\item [{$\alpha_{B_{p}}$}] Capacity for $L_{p}$ Balls 
\item [{$\alpha_{E}$}] Capacity for Ellipsoids 
\item [{$\alpha_{M}$}] Capacity for General Manifolds 
\item [{$\alpha_{iter}$}] Capacity for General Manifolds $\alpha_{M}$
found via iterative algorithm 
\item [{$\alpha_{||}$}] Capacity for Parallel Manifolds 
\item [{$\mu,p$}] Index of the $\mu$th ($p$th) manifold 
\item [{$\mathbf{x}$}] Point on a manifold 
\item [{$\mathbf{x}_{0}$}] Center of a manifold 
\item [{$\mathbf{u_{i}}$}] $i$ th basis vector of a manifold ($i=1,...,D$) 
\item [{$\mathbf{v}$}] $N$-dimensional vector $\mathbf{v}$ for an arbitrary
vector 
\item [{$\vec{v}$}] $D$-dimensinoal vector $\vec{v}$ for an arbitrary
vector 
\item [{$\mathbf{w}$}] Solution of a linear separating problem. 
\item [{$b$}] bias of a linear perceptron 
\item [{$y$}] binary ($\pm1$) labels 
\item [{$V$}] Volume of space of solutions 
\item [{$h_{0}$}] field induced by a center of a manifold $\mathbf{x_{0}}$
\item [{$h_{i}$}] field induced by the $i$ th axis of a manifold $\mathbf{u}_{i}$
\item [{$h$}] field induced by a data point $\mathbf{x}$
\item [{$R_{E}$}] Effective radius for an ellipsoid 
\item [{$D_{E}$}] Effective dimension for an ellispoid.
\item [{$R_{M}$}] Effective manifold capacity radius for a general manifold
(using self-consistent equations)
\item [{$D_{M}$}] Effective manifold capacity dimension for a general
manifold (using self-consistent equations)
\item [{$R_{W}$}] Effective manifold capacity radius for a general manifold
(using max projections on Gaussian vectors, related to the mean width) 
\item [{$D_{W}$}] Effective manifold capacity dimension for a general
manifold (using max projections on Gaussian vectors) 
\item [{$\epsilon_{g}$,$\epsilon_{G}$}] Generalization error 
\item [{$\epsilon_{t}$,$\epsilon_{T}$}] Training error 
\item [{$m$}] number of subsamples per manifold 
\item [{$\xi_{p}$}] Slack parameter for the $p$th manifold
\item [{$s_{min}(\vec{v}),\tilde{s}(\vec{v})$}] $\underset{\vec{s},f(\vec{s})=0}{argmin}\vec{s}\cdot\vec{v}$
\item [{$s^{*}$}] $s$ solved via self-consistent equations on $z_{0}$
and $s$
\item [{$\kappa_{E}$}] Excess margin for ellipsoids defined via $R_{E}\sqrt{D_{E}}$
\item [{$\kappa_{M}$}] Excess margin for general manifolds defined via
$R_{M}\sqrt{D_{M}}$
\item [{$\kappa_{W}$}] Excess margin for general manifolds in scaling
regime defined via $R_{W}\sqrt{D_{W}}$
\item [{$R_{c}$}] Critical radius for phase transition 
\item [{$d$}] Embedding dimension per ambient dimension $D/N$
\item [{$\rho$}] Overlap between the solution $\mathbf{w}$ and axes $\mathbf{u}_{i}$,
$\frac{1}{D}\sum_{i=1}^{D}\left(\mathbf{w}\cdot\mathbf{u}_{i}\right)^{2}$
\item [{$f$}] Sparsity of labels (Number of posive (or negative) labels/
Total number of labels) 
\item [{$\kappa_{\text{int }}$}] Margin in the intermediate layer unit
(if intermediate unit is like a binary classifier) 
\item [{$m$}] Number of subsamples (training samples) per manifold 
\item [{$M$}] Intermediate layer size 
\end{description}

\section{Mathematical Conventions}
\begin{description}[leftmargin=8em,style=nextline]
\item [{$\left\Vert \mathbf{v}\right\Vert $}] $L_{2}$ norm of a vector
$\mathbf{v}$
\item [{$\left\Vert \mathbf{v}\right\Vert _{p}$}] $L_{p}$ norm of a vector
$\mathbf{v}$
\item [{$\Theta(x)$}] Heaviside step function 
\item [{$\langle X\rangle$}] Average of a random variable $X$ 
\item [{$\langle\vec{x},\vec{y}\rangle$}] Inner product between two vectors
$\vec{x}$, $\vec{y}$
\item [{$\vec{x}\circ\vec{y}$}] Hadamard (element-wise) product 
\item [{$\left|S\right|$}] Cardinality (number of elements) of a set $S$ 
\item [{$\left|x\right|$}] Absolute value of a scalar $x$ 
\item [{$Dx$}] Gaussian measure, $\frac{1}{\sqrt{2\pi}}e^{-\frac{1}{2}x^{2}}dx$
\item [{$\chi_{D}(t)$}] Chi distribution, $\frac{2^{1-\frac{D}{2}}}{\Gamma(\frac{D}{2})}t^{D-1}e^{-\frac{1}{2}t^{2}}$
\item [{$H(x)$}] $\int_{x}^{\infty}Dz=\frac{1}{\sqrt{2\pi}}\int_{x}^{\infty}dze^{-\frac{z^{2}}{2}}$
\item [{$\hat{v}$}] $\vec{v}/||\vec{v}||$
\item [{$v$}] $||\vec{v}||$
\end{description}

%% file: chapters/chz_appendix1_gardner.tex
\chapter{Appendix B: Gardner's Replica Theory of Isolated Points }

Consider a perceptron with $P$ input points $\mathbf{x}^{\mu}\in\mathbb{R}^{N}$
and $P$ labels $y^{\mu}=\pm1$ , $\mu=1,...,P$. Assume that each
component of$\mathbf{x}^{\mu}$ are Gaussian i.i.d. The weight plane
$\mathbf{w\in\mathbb{R}^{N}}$ needs to classify all the input points
such that all points are at least $\kappa$ away from the solution
hyperplane. For simplicity, consider the regime where the number of
positive labels and negative labels are equal. {[}This regime is called
``dense classification'' regime, where sparsity is $f=0.5$.{]}
In this regime, optimal bias for maximum capacity is $b=0$ by symmetry,
so let us ignore the bias term $b$ for now. 

The problem is now to find $\mathbf{w}$ such that 

\begin{equation}
h_{\mu}=\frac{y^{\mu}\mathbf{w}^{T}\mathbf{x}^{\mu}}{\sqrt{N}}>\kappa\label{eq:G_constraint}
\end{equation}

where $\kappa$ is a margin. $h_{\mu}$, which we refer to as a field
from a pattern $\mathbf{x}^{\mu}$ is a measure of distance between
the pattern $\mathbf{x}^{\mu}$ and the solution hyperplane denoted
by $\mathbf{w}$. Note that the denominator $\sqrt{N}$ is introduced
so that $h_{\mu}$ does not grow with the network size $N$. In general,
if $\mathbf{w}$ and $\mathbf{x}$ are random, then the scale of $\mathbf{w}^{T}\mathbf{x}$
is $\sqrt{N}$ . For now let us consider $\kappa=0$ case. 

Gardner calculated $V$, which is the volume of solutions (weight
vectors $\mathbf{w}$), which satisfy the classification constraint
Eqn. \ref{eq:G_constraint} with $\kappa=0$. If $V$ goes to 0, then
there is no solution and the network is beyond capacity. In order
to compute $V$, each component of $w_{i}$ needs to be integrated
with a term that is $0$ when the constraint is not satisfied, hence 

\begin{equation}
V=\int d\mathbf{w}^{N}\delta(\mathbf{w}^{2}-N)\prod_{\mu=1}^{P}\Theta(h^{\mu})\label{eq:G_Volumeh}
\end{equation}

where $\delta$ is a delta function, $\Theta$ is a Heaviside step
fuction, and the norm on $\mathbf{w}$, $\delta(w^{2}-N)$, is imposed
to count $\mathbf{w}$ in the same direction only once. Using the
expression for the field (Eqn. \ref{eq:G_constraint}), we get 

\begin{equation}
V=\int dw^{N}\delta(w^{2}-N)\prod_{\mu=1}^{P}\Theta(\frac{y^{\mu}w^{T}x^{\mu}}{\sqrt{N}})\label{eq:G_volumeW}
\end{equation}

where $\Theta(h^{\mu})$ is one if $\mathbf{w}$ is a separating solution,
and zero if it is not a solution, as the field (input to the $\Theta(x)$)
will be negative.

Note that $V$ involves a \emph{product} of many random contributions.
Products of independent random numbers are known to possess distributions
with long tails for which the average and the most probable value
are markedly different. However, the logarithm of such quantities
is a large \emph{sum }of independent terms, hence is expected to have
a Gaussian distribution so that its average and the most probable
value match. Therefore, the most typical value of $V$ is expected
to be 

\begin{equation}
V_{typical}\sim\text{exp}\left[\langle\text{log}V\rangle\right]
\end{equation}

Therefore, in order to get the typical behavior, we are interested
in the average of log$V$ \cite{engel2001statistical}. Hence, we
need to calculate

\begin{equation}
\langle\log V\rangle_{x^{\mu}}
\end{equation}

which is the average of $\log V$ over the quenched distribution of
patterns. We can do this by the following formula:

\begin{equation}
\langle\log V\rangle=\lim_{n\rightarrow0}\frac{\langle V^{n}\rangle-1}{n}\label{eq:ReplicaTrick}
\end{equation}

called ``replica trick'', originally developed to calculate quenched
averages in the theory of disordered solides. 

Let us then calculate the expectation of $V^{n}$, which, in the case
of a natural $n$, can be expanded as

\begin{equation}
\langle V^{n}\rangle=\langle\int\prod_{\alpha=1}^{n}dw^{\alpha}\delta(w_{\alpha}^{2}-N)\prod_{\mu}^{P}\prod_{\alpha}^{n}\Theta(h_{\alpha}^{\mu})\rangle_{x^{\mu}}
\end{equation}

where $\alpha$ is an index for each one of the $n$ replicas of the
original system with the same realization of random samples. We note
that the question of how to go from a natural $n$ to the limit of
$n\rightarrow0$ is in general a hard problem, and more on this matter
can be found in \cite{mezard1988spin}. 

We can replace $\Theta$ using the integral representation of the
$\Theta$ function, 

\begin{equation}
\langle\int\prod_{\alpha=1}^{n}dw^{\alpha}\delta(w_{\alpha}^{2}-N)\underbrace{\int_{\kappa}^{\infty}\prod_{\mu,\alpha}^{P,n}dh_{\mu}^{\alpha}\int_{-\infty}^{\infty}\prod_{\mu,\alpha}^{P,n}\frac{d\hat{h}_{\mu}^{\alpha}}{2\pi}e^{\frac{\left[i\hat{h}_{\mu}^{\alpha}h_{\mu}^{\alpha}-i\hat{h}_{\mu}^{\alpha}w^{\alpha T}x^{\mu}y^{\mu}\right]}{\sqrt{N}}}}_{C}\rangle_{x^{\mu}}
\end{equation}

But because $\kappa=0$ and $\hat{h}>0$, we can change the range
of the integrals to 

\begin{equation}
\langle\int\prod_{\alpha=1}^{n}dw^{\alpha}\delta(w_{\alpha}^{2}-N)\underbrace{\int_{0}^{\infty}\prod_{\mu,\alpha}^{P,n}dh_{\mu}^{\alpha}\int_{0}^{\infty}\prod_{\mu,\alpha}^{P,n}\frac{d\hat{h}_{\mu}^{\alpha}}{2\pi}e^{\frac{\left[i\hat{h}_{\mu}^{\alpha}h_{\mu}^{\alpha}-i\hat{h}_{\mu}^{\alpha}w^{\alpha T}x^{\mu}y^{\mu}\right]}{\sqrt{N}}}}_{C}\rangle_{x^{\mu}}\label{eq:G_VtoN_avg_hh}
\end{equation}
\begin{equation}
\end{equation}

Now average over the random patterns $x^{\mu}$only affect the term
noted as $C$, that is, 

\begin{equation}
C=\prod_{\mu=1}^{P}\int\prod_{\alpha=1}^{n}d\hat{h}_{\mu}^{\alpha}\int dh_{\mu}^{\alpha}\langle\prod_{\alpha=1}^{n}e^{\frac{-i\hat{h}_{\mu}^{\alpha}w^{\alpha T}x^{\mu}y^{\mu}}{\sqrt{N}}}\cdot e^{\frac{-i\hat{h}_{\mu}^{\alpha}h_{\mu}^{\alpha}}{\sqrt{N}}}\rangle_{x^{\mu}}
\end{equation}

Now change the product of exponentials to the sum of powers, and taking
out terms that don't depend on $x^{\mu}$, we get 

\begin{equation}
C=\prod_{\mu=1}^{P}\int\prod_{\alpha=1}^{n}d\hat{h}_{\mu}^{\alpha}\int dh_{\mu}^{\alpha}\langle e^{\sum_{\alpha=1}^{n}\frac{-i\hat{h}_{\mu}^{\alpha}w^{\alpha T}x^{\mu}y^{\mu}}{\sqrt{N}}}\rangle_{x^{\mu}}\cdot e^{\frac{-i\hat{h}_{\mu}^{\alpha}h_{\mu}^{\alpha}}{\sqrt{N}}}
\end{equation}

Since $w^{\alpha}$ is a vector, we expand the vector dot product
$w^{\alpha T}x^{\mu}$ as sum of each vector's elements, then

\begin{equation}
C=\prod_{\mu=1}^{P}\int\prod_{\alpha=1}^{n}d\hat{h}_{\mu}^{\alpha}\int dh_{\mu}^{\alpha}\langle e^{\frac{\sum_{\alpha=1}^{n}-i\hat{h}_{\mu}^{\alpha}\sum_{j}w_{j}^{\alpha}x_{j}^{\mu}y^{\mu}}{\sqrt{N}}}\rangle_{x^{\mu}}\cdot e^{\frac{-i\hat{h}_{\mu}^{\alpha}h_{\mu}^{\alpha}}{\sqrt{N}}}
\end{equation}

Then the sum $\sum_{j}$ in the power can be the product of exponentials,
then

\begin{equation}
C=\prod_{\mu=1}^{P}\int\prod_{\alpha=1}^{n}d\hat{h}_{\mu}^{\alpha}\int dh_{\mu}^{\alpha}\langle\prod_{j}e^{\frac{-i\sum_{\alpha=1}^{n}\hat{h}_{\mu}^{\alpha}w_{j}^{\alpha}x_{j}^{\mu}y^{\mu}}{\sqrt{N}}}\rangle_{x^{\mu}}\cdot e^{\frac{-i\hat{h}_{\mu}^{\alpha}h_{\mu}^{\alpha}}{\sqrt{N}}}
\end{equation}

Becasue each element is independent, we can take out product over
j

\begin{equation}
C=\prod_{\mu=1}^{P}\int\prod_{\alpha=1}^{n}d\hat{h}_{\mu}^{\alpha}\int dh_{\mu}^{\alpha}\prod_{j}\langle e^{\frac{-i\sum_{\alpha=1}^{n}\hat{h}_{\mu}^{\alpha}w_{j}^{\alpha}x_{j}^{\mu}y^{\mu}}{\sqrt{N}}}\rangle_{x_{j}^{\mu}}\cdot e^{\frac{-i\hat{h}_{\mu}^{\alpha}h_{\mu}^{\alpha}}{\sqrt{N}}}
\end{equation}

Now use the fact that $\int Dxe^{ixA}=e^{-A^{2}/2}$ , we get rid
of $\langle\rangle_{x}\sim\int Dx$ . Note that $A^{2}$ is like 

\begin{equation}
\left(\sum_{\alpha}\left(-y^{\mu}\right)\frac{\hat{h}_{\mu}^{\alpha}w_{j}^{\alpha}}{\sqrt{N}}\right)^{2}=\left(\sum_{\alpha}\frac{\hat{h}_{\mu}^{\alpha}w_{j}^{\alpha}}{\sqrt{N}}\right)^{2}=\frac{\left(\sum_{\alpha}\hat{h}_{\mu}^{\alpha}w_{j}^{\alpha}\right)^{2}}{N}
\end{equation}

So we have: 

\begin{equation}
C=\prod_{\mu=1}^{P}\int\prod_{\alpha=1}^{n}d\hat{h}_{\mu}^{\alpha}\int dh_{\mu}^{\alpha}\prod_{j=1}^{N}e^{\frac{-\left(\sum_{\alpha=1}^{n}\hat{h}_{\mu}^{\alpha}w_{j}^{\alpha}\right)^{2}}{2N}}\cdot e^{\frac{-i\hat{h}_{\mu}^{\alpha}h_{\mu}^{\alpha}}{\sqrt{N}}}
\end{equation}

which can be re-written as 

\begin{equation}
C=\prod_{\mu=1}^{P}\int\prod_{\alpha=1}^{n}d\hat{h}_{\mu}^{\alpha}\int dh_{\mu}^{\alpha}\prod_{j=1}^{N}e^{\frac{-\sum_{\alpha,\beta=1}^{n}\hat{h}_{\mu}^{\alpha}\hat{h}_{\mu}^{\beta}w_{j}^{\alpha}w_{j}^{\beta}}{2N}}\cdot e^{\frac{-i\hat{h}_{\mu}^{\alpha}h_{\mu}^{\alpha}}{\sqrt{N}}}
\end{equation}

Define $ $$q_{\alpha\beta}=\frac{1}{N}\sum_{j}w_{j}^{\alpha}w_{j}^{\beta}$
which is the replica symmetric order parameter. 

\begin{equation}
C=\prod_{\mu=1}^{P}\int\prod_{\alpha=1}^{n}d\hat{h}_{\mu}^{\alpha}\int dh_{\mu}^{\alpha}\underbrace{e^{-\frac{1}{2}\left(\sum_{\alpha=1}^{n}\sum_{\beta=1}^{n}\hat{h}_{\mu}^{\alpha}q_{\alpha\beta}^{-1}\hat{h}_{\mu}^{\beta}\right)}}_{X}\cdot e^{\frac{-i\hat{h}_{\mu}^{\alpha}h_{\mu}^{\alpha}}{\sqrt{N}}}
\end{equation}

Define 

\begin{equation}
X(\hat{h}_{\mu}^{\alpha})=e^{-\frac{1}{2}\left(\sum_{\alpha,\beta=1}^{n}\hat{h}_{\mu}^{\alpha}q_{\alpha\beta}^{-1}\hat{h}_{\mu}^{\beta}\right)}\label{eq:G_delta_X}
\end{equation}

\begin{equation}
C=\prod_{\mu=1}^{P}\int_{0}^{\infty}\prod_{\alpha=1}^{n}dh_{\mu}^{\alpha}\left[\int_{0}^{\infty}\prod_{\alpha=1}^{n}d\hat{h}_{\mu}^{\alpha}X(\hat{h}_{\mu}^{\alpha},....)e^{\frac{-i\hat{h}_{\mu}^{\alpha}h_{\mu}^{\alpha}}{\sqrt{N}}}\right]
\end{equation}

(Here we changed the orders of $h$ and $\hat{h}$.) 

By integrating out $d\hat{h}$, we get 

\begin{equation}
C=\prod_{\mu=1}^{P}\int_{0}^{\infty}\prod_{\alpha=1}^{n}dh_{\mu}^{\alpha}\left[X(\hat{h}_{\mu}^{\alpha}=h_{\mu}^{\alpha})\frac{1}{\sqrt{\det q}}\right]\label{eq:G_C_int}
\end{equation}

where we used the delta function identity, Eqn. \ref{eq:G_delta_X},(
$X(\hat{h}_{\mu}^{\alpha}=h_{\mu}^{\alpha})$) and n-dimensional Gaussian
integration with matrix $q$ to get $\frac{1}{\sqrt{\det q}}$ term.
We re-write Eqn. \ref{eq:G_C_int} such that: 

\begin{equation}
C=\prod_{\mu=1}^{P}\int_{0}^{\infty}\prod_{\alpha=1}^{n}dh_{\mu}^{\alpha}\left[X(\hat{h}_{\mu}^{\alpha}=h_{\mu}^{\alpha})e^{-\frac{1}{2}\log\det q}\right]
\end{equation}

which is: 

\begin{equation}
C=\prod_{\mu=1}^{P}\int_{0}^{\infty}\prod_{\alpha=1}^{n}dh_{\mu}^{\alpha}\left[[\underbrace{e^{-\frac{1}{2}\left(\sum_{\alpha,\beta=1}^{n}h_{\mu}^{\alpha}q_{\alpha\beta}^{-1}h_{\mu}^{\beta}\right)}}_{X}]e^{-\frac{1}{2}\log\det q}\right]
\end{equation}

where

\begin{equation}
X=e^{-\frac{1}{2}\left(\sum_{\alpha,\beta=1}^{n}h_{\mu}^{\alpha}q_{\alpha\beta}^{-1}h_{\mu}^{\beta}\right)}\label{eq:G_X}
\end{equation}

Then, $C$ can be simplied with $X$, 

\begin{equation}
C=\prod_{\mu=1}^{P}\int_{0}^{\infty}\prod_{\alpha=1}^{n}dh_{\mu}^{\alpha}\left[Xe^{-\frac{1}{2}\log\det q}\right]
\end{equation}

Now going back to $\langle V^{N}\rangle$, in Eqn. \ref{eq:G_VtoN_avg_hh},
is now (after the above calculations): 

\begin{equation}
\langle V^{N}\rangle=\left\langle \int\prod_{\alpha=1}^{N}dw^{\alpha}\delta(w_{\alpha}^{2}-N)\underbrace{\prod_{\mu=1}^{P}\left[\int_{0}^{\infty}\prod_{\alpha=1}^{P}dh_{\mu}^{\alpha}Xe^{-\frac{1}{2}\log\det q}\right]}_{C}\right\rangle \label{eq:G_VN_withC}
\end{equation}

where $X$ is given as Eqn. \ref{eq:G_X}. We note that the term $C$
has two parts 

\begin{equation}
C(X)=\prod_{\mu=1}^{P}\int_{0}^{\infty}\prod_{\alpha=1}^{P}dh_{\mu}^{\alpha}\underbrace{\left\{ e^{-\frac{1}{2}\left(\sum_{\alpha,\beta=1}^{n}h_{\mu}^{\alpha}q_{\alpha\beta}^{-1}h_{\mu}^{\beta}\right)}\right\} }_{X}\label{eq:G_C(X)}
\end{equation}

and 

\begin{equation}
C(0)=e^{-\frac{P}{2}\log\det q}\label{eq:G_C(0)}
\end{equation}

for later use. 

Because it's an integral of $q_{\alpha\beta}=\frac{1}{N}\sum_{j}w_{j}^{\alpha}w_{j}^{\beta}=q_{\alpha\beta}(w)$
which is a complex function of $w$, we write it in terms of integral
of deltas of $q_{\alpha\beta}$. Intuitively, $q$ is generally thought
of as a function of the overlaps between the solution of weights. 

Now, we can re-write the $\langle V^{N}\rangle$ as functions of $q_{\alpha\beta}$
, by introducing 
\begin{equation}
\int dq_{\alpha\beta}d\hat{q}_{\alpha\beta}e^{i\hat{q}_{\alpha\beta}\left(-q_{\alpha\beta}N+\sum_{j}w_{j}^{\alpha}w_{j}^{\beta}\right)}
\end{equation}

And we get: 

\begin{align}
\langle V^{n}\rangle & =\int dq_{\alpha\beta}d\hat{q}_{\alpha\beta}\int\prod_{\alpha=1}^{n}dw^{\alpha}\left[\delta(w_{\alpha}^{2}-N)\right.\label{eq:G_<V^N>}\\
 & \left.e^{-iq_{\alpha\beta}\hat{q}_{\alpha\beta}N}e^{i\hat{q}_{\alpha\beta}\sum_{j}w_{j}^{\alpha}w_{j}^{\beta}}e^{-\frac{P}{2}\log\det q}C(X)\right]
\end{align}

where $C(X)$ is defined in Eqn. \ref{eq:G_C(X)}. 

Let us use the replica symmetric ansatz 

\begin{equation}
q_{\alpha\beta}=(1-q)\delta_{\alpha\beta}+q\label{eq:G_q_RSansatz}
\end{equation}

Now going back to evaluating $C(X)$, 

\begin{equation}
C(X)=\prod_{\mu=1}^{P}\int_{0}^{\infty}\prod_{\alpha}dh_{\mu}^{\alpha}[e^{-\frac{1}{2}\left(\sum_{\alpha,\beta=1}^{n}h_{\mu}^{\alpha}q_{\alpha\beta}^{-1}h_{\mu}^{\beta}\right)}]\label{eq:G_C(X)_withqab}
\end{equation}

where, given our ansatz (Eqn. \ref{eq:G_q_RSansatz}), 

\begin{equation}
q^{-1}_{\alpha\beta}=\frac{1}{1-q}\delta_{\alpha\beta}+\left(\frac{-q}{\left(1-q\right)\left(1+(n-1)q\right)}\right)\label{eq:G_qab}
\end{equation}

where $n$ is a dimension of the matrix. However, we are in the limit
$n\rightarrow0$, therefore 

\begin{equation}
q^{-1}_{\alpha\beta}=\frac{1}{1-q}\delta_{\alpha\beta}+\frac{-q}{(1-q)^{2}}\label{eq:G_qab_n0}
\end{equation}

With Eqn. \ref{eq:G_qab_n0}, Eqn. \ref{eq:G_C(X)_withqab} becomes 

\begin{equation}
C(X)=\prod_{\mu=1}^{P}\int_{0}^{\infty}\prod_{\alpha=1}^{n}dh_{\mu}^{\alpha}[e^{-\frac{1}{2(1-q)}\sum_{\alpha}\left(h_{\mu}^{\alpha}\right)^{2}+\frac{q}{2(1-q)^{2}}\left(\sum_{\alpha}h_{\mu}^{\alpha}\right)^{2}}]
\end{equation}

Let us introduce $h'=\frac{h}{\sqrt{1-q}}$ (for simplicity), then 

\begin{equation}
C(X)=\prod_{\mu=1}^{P}\int_{0}^{\infty}\prod_{\alpha=1}^{n}dh{}_{\mu}^{\prime\alpha}[e^{-\frac{1}{2}\sum_{\alpha}\left(h_{\mu}^{\prime\alpha}\right)^{2}+\frac{q}{2(1-q)}\left(\sum_{\alpha}h{}_{\mu}^{\prime\alpha}\right)^{2}}]\label{eq:G_C(X)_changehprime}
\end{equation}

Using Hubbard\textendash Stratonovich transformation, $e^{A^{2}/2}=\int_{-\infty}^{\infty}\frac{dt}{\sqrt{2\pi}}e^{-\frac{1}{2}t^{2}+At}$,
we introduce additional expansion, 

\begin{equation}
C(X)=\prod_{\mu=1}^{P}\int Dt_{\mu}\prod_{\alpha=1}^{n}\int_{0}^{\infty}dh{}_{\mu}^{\prime\alpha}[e^{-\frac{1}{2}\sum_{\alpha}\left(h_{\mu}^{\prime\alpha}\right)^{2}+\sqrt{\frac{q}{1-q}}t_{\mu}\sum_{\alpha}h{}_{\mu}^{\prime\alpha}}]
\end{equation}

which, finally, can be simplified to 

\begin{equation}
C(X)=\prod_{\mu=1}^{P}\int Dt_{\mu}\left\{ \int_{0}^{\infty}dh{}_{\mu}^{\prime\alpha}[e^{-\frac{1}{2}\left(h_{\mu}^{\prime\alpha}\right)^{2}+\sqrt{\frac{q}{1-q}}t_{\mu}h{}_{\mu}^{\prime\alpha}}]\right\} ^{n}
\end{equation}

In other words, 

\begin{equation}
C(X)=\prod_{\mu=1}^{P}\int Dt_{\mu}\left[\ e^{\log Z(t_{\mu},q)}\right]^{n}
\end{equation}

where 
\begin{equation}
Z(t_{\mu},\:q)=\int_{0}^{\infty}\frac{dh}{\sqrt{2\pi}}[e^{-\frac{1}{2}h^{2}+\sqrt{\frac{q}{1-q}}t_{\mu}h}]\label{eq:G_Z(t,q)}
\end{equation}

Now 

\begin{equation}
C(X)=\left[\int Dt\ e^{n\log Z(t,q)}\right]^{P}
\end{equation}

And by simple power expansion, 

\begin{equation}
C(X)=e^{P\log\left(1+n\langle\log Z(t,q)\rangle\right)}
\end{equation}

Expanding log, we get: 

\begin{equation}
C(X)=e^{Pn\langle\log Z(t,q)\rangle}\label{eq:C(X)_withLogZ}
\end{equation}

Note, we are trying to evaluate

\begin{equation}
\langle\log V\rangle=\lim_{n\rightarrow0}\frac{\langle V^{n}\rangle-1}{n}\label{eq:G_ReplicaVN}
\end{equation}

by using Eqn. \ref{eq:G_VN_withC} and which has $C(X)$ (Eqn.\ref{eq:G_C(X)})
and $C(0)$ (Eqn. \ref{eq:G_C(0)}). We need to now evaluate $C(0)$
with Eqn. \ref{eq:G_C(0)}. 

Now, noting that $\mbox{\ensuremath{\log}det}q_{\alpha\beta}$ is
related to the sum of log of eigenvalues, 

\begin{equation}
\log\det q_{\alpha\beta}=\sum_{l}\log\lambda_{l}
\end{equation}

where $\lambda_{l}$, $l=1,...$ are eigenvalues of $q_{\alpha\beta}$.
With the replica symmetric ansatz of $q_{\alpha\beta}$ is given by
Eqn. \ref{eq:G_q_RSansatz}, the first eigenvalues are $1+(n-1)q$
and the rest of the $(n-1)$ eigenvalues are $1-q$ , 

Therefore 

\begin{equation}
\log\det q=n\log\left(1-q\right)+\frac{nq}{\left(1-q\right)}\label{eq:G_logDetq_withnq}
\end{equation}

where the linear term in $n$ is dominant. Then,using Eqn. \ref{eq:G_logDetq_withnq},
$C(0)$ from Eqn. \ref{eq:G_C(0)} is now, 

\begin{equation}
C(0)=\mbox{exp}\left(-\frac{1}{2}Pn\log\left(1-q\right)-\frac{1}{2}\frac{Pnq}{\left(1-q\right)}\right)
\end{equation}

Now, back to the original equation. 

\begin{align}
\langle V^{n}\rangle & =\int dq_{\alpha\beta}d\hat{q}_{\alpha\beta}\int\prod_{\alpha=1}^{n}dw^{\alpha}\left[\delta(w_{\alpha}^{2}-N)\right.\\
 & \left.e^{-iq_{\alpha\beta}\hat{q}_{\alpha\beta}N}e^{i\hat{q}_{\alpha\beta}\sum_{j}w_{j}^{\alpha}w_{j}^{\beta}}C(0)C(X)\right]
\end{align}

Expanding $\delta(w_{\alpha}^{2}-N)$ to an exponential form with
a new variable $\lambda_{\alpha}$: 

\begin{align}
\langle V^{n}\rangle & =\int dq_{\alpha\beta}d\hat{q}_{\alpha\beta}\int\prod_{\alpha}d\lambda_{\alpha}\int\prod_{\alpha=1}^{n}dw^{\alpha}\text{exp}\Biggl\{ i\lambda_{\alpha}(w_{\alpha}^{2}-N)\\
 & \left.-iq_{\alpha\beta}\hat{q}_{\alpha\beta}N+i\hat{q}_{\alpha\beta}\sum_{j}w_{j}^{\alpha}w_{j}^{\beta}\right\} C(0)C(X)
\end{align}

We change the definition of $C(X)$ slightly, so that 

\begin{equation}
C(X)=e^{Pn\langle\log Z(t,q)\rangle}=e^{\alpha Nn\langle\log Z(t,q)\rangle}\label{eq:G_C(X)_alpha}
\end{equation}

where 

\begin{equation}
\alpha=P/N\label{eq:G_alpha_firstTime}
\end{equation}

is the capacity we wish to get. 

Doing the integral over $w^{\alpha}$ gives the term: 

\begin{equation}
e^{\left\{ -\frac{N}{2}\log\det\left[\lambda_{\alpha}\delta_{\alpha\beta}+\hat{q}_{\alpha\beta}\right]\right\} }
\end{equation}

And absorbing $i$ into the new variable, that is, use $\lambda\leftarrow i\lambda$,
we get, 

\noindent\fbox{\begin{minipage}[t]{1\columnwidth \fboxsep \fboxrule}%
\begin{equation}
\langle V^{n}\rangle=\int dq_{\alpha\beta}d\hat{q}_{\alpha\beta}\int d\lambda_{\alpha}\text{exp}\left[F(q_{\alpha\beta},\hat{q}_{\alpha\beta},\lambda_{\alpha})\right]\label{eq:G_VN_F}
\end{equation}

where 

\begin{equation}
F=-\lambda_{\alpha}N+q_{\alpha\beta}\hat{q}_{\alpha\beta}N-\frac{N}{2}\log\det\left[\lambda_{\alpha}\delta_{\alpha\beta}+\hat{q}_{\alpha\beta}\right]+\alpha Nn\langle\log Z(t,q)\rangle\label{eq:G_VN_F_FFF}
\end{equation}
\end{minipage}}

Where the dependence is only on non-local variables like $N$, $P$,
$n\rightarrow0$. Let us now evaluate the integral using the saddle
point approximation, 

\begin{equation}
I=\int dx\text{exp}\left[-g(x)\right]\sim\text{exp}\left[-g(x_{0})\right]\sqrt{\frac{2\pi}{g''(x_{0})}}\label{eq:G_SaddlePoint}
\end{equation}

where$g(x)$ is at its minimum at $x_{0}$ . 

In the limit of $N\rightarrow\infty$, we can use the following ansatz
assuming replica symmetry, 

\begin{equation}
q_{\alpha\beta}=(1-q)\delta_{\alpha\beta}+q
\end{equation}

\begin{equation}
\hat{q}_{\alpha\beta}=\left(\hat{q}_{0}-\hat{q_{1}}\right)\delta_{\alpha\beta}+\hat{q_{1}}
\end{equation}

where we denote $\hat{q}_{0}-\hat{q_{1}}$ as $\Delta\hat{q}$, and
find the saddle point by taking the derivative of $F$ 

$\frac{\partial F}{\partial\lambda_{\alpha}}=0:$

\begin{equation}
1=\frac{1}{\lambda+\Delta\hat{q}}-\frac{\hat{q}_{1}}{\left(\lambda+\Delta\hat{q}\right)^{2}}\label{eq:G_dFdL}
\end{equation}

$\frac{\partial F}{\partial\hat{q}_{\alpha\beta}}=0:$

\begin{equation}
q_{\alpha\beta}=\left(1-q\right)\delta_{\alpha\beta}+q=\left(\frac{1}{\lambda+\Delta\hat{q}}\right)\delta_{\alpha\beta}-\frac{\hat{q}_{1}}{\left(\lambda+\Delta\hat{q}\right)^{2}}\label{eq:G_dF_dqhat}
\end{equation}

Thus, we get 

\begin{equation}
q=\frac{-\hat{q_{1}}}{\left(\lambda+\Delta\hat{q}\right)^{2}}
\end{equation}

Note that doing saddle points with $\lambda$ and $\hat{q}$ are easier,
and we did these operations first. Then, from Eqn. \ref{eq:G_dFdL},
we obtain

\noindent\fbox{\begin{minipage}[t]{1\columnwidth \fboxsep \fboxrule}%
\begin{equation}
1=\frac{1}{\lambda+\Delta\hat{q}}+q
\end{equation}
\end{minipage}}

If we go back to the original equation, Eqn. \ref{eq:G_VN_F}, and
plug in back the saddle points, then, we get 

\begin{equation}
\frac{Nn}{2}\left[\left(\lambda+\Delta\hat{q}\right)-\log\left(\lambda+\Delta\hat{q}\right)\right]+Nn\alpha\langle\log Z(t,q)\rangle
\end{equation}

which eventually leads us to the expression for the term 

\begin{equation}
\langle V^{n}\rangle=e^{Nn\left[G_{0}(q)+\alpha G_{1}(q)\right]}\label{eq:G_V^N}
\end{equation}

where 

\noindent\fbox{\begin{minipage}[t]{1\columnwidth \fboxsep \fboxrule}%
\begin{equation}
G_{0}(q)=\frac{1}{2}\left[\frac{1}{1-q}+\log\left(1-q\right)\right]\label{eq:G_G0}
\end{equation}

\begin{equation}
G_{1}(q)=\langle\log\underbrace{\int_{0}^{\infty}\frac{dh}{\sqrt{2\pi}}e^{-\frac{1}{2}\left(h-\sqrt{\frac{q}{1-q}}t\right)^{2}}}_{H(\sqrt{\frac{q}{1-q}}t)=H(\sqrt{Q}t)}\rangle_{t}=\langle\log H(\sqrt{Q}t)\rangle_{t}\label{eq:G_G1}
\end{equation}
\end{minipage}}

Note that $G_{0}$ is an entropic term and doesn't change with constraints,
but for classification with additional constraints (such as manifolds
classification), $G_{1}$ does change and and computing $\langle\log H(\sqrt{Q}t)\rangle_{t}$
becomes important. 

Therefore, by L'Hospital's rule, we get in the limit of $n\rightarrow0$, 

\begin{equation}
\langle\log V\rangle=N\left[G_{0}(q)+\alpha G_{1}(q)\right]\label{eq:G_LogV_NG0G1}
\end{equation}

Note that $G_{0}$ is an entropic term, and $G_{1}$ needs to be self-consistent. 

\begin{equation}
\frac{\partial G_{0}}{\partial q}=\frac{q}{2(1-q)^{2}}
\end{equation}

\begin{equation}
\alpha\frac{\partial G_{1}}{\partial q}=\alpha\frac{\partial Q}{\partial q}\frac{\partial G_{1}}{\partial Q}
\end{equation}

Now plugging the expressions back to the self-consistency requirement
($\frac{\partial G_{0}}{\partial q}=\alpha\frac{\partial G_{1}}{\partial q}$)
we get: 

\begin{equation}
1=\frac{\alpha(1-q)}{q}\langle\left(\frac{e^{-Qt^{2}/2}}{\sqrt{2\pi}H(\sqrt{Q}t)}\right)^{2}\rangle_{t}\label{eq:G_alpha_q_final}
\end{equation}

And the final step for capacity is we send $q\rightarrow1$ ($Q\rightarrow\infty$)
because there is only one solution and in this limit the overlap is
1. With such limit, we get $\alpha=2$ , which is the capacity of
isolated points for zero-margin solution. 

\emph{Remarks.} If $\alpha=0$ then $q$ has to be zero from Eqn \ref{eq:G_alpha_q_final}.
Intuitively, this means that as the number of points $P$ approach
zero (or the network size $N$ approaches infinity compared to the
number of points $P$), the problem becomes very easy and there is
a lot of solutions, making the overlap between solutions $q$ go to
0.